\documentclass[pra,twocolumn,superscriptaddress]{revtex4-2}%
\usepackage{amsfonts}
\usepackage{amsmath}
\usepackage{amssymb}
\usepackage{mathtools}
\usepackage{graphicx}
\usepackage{float}
\usepackage{physics}
\usepackage[colorlinks=true,linkcolor=blue,citecolor=red,plainpages=false,pdfpagelabels]%
{hyperref}
\usepackage{newtxtext,newtxmath}%
\providecommand{\U}[1]{\protect\rule{.1in}{.1in}}
\newtheorem{theorem}{Theorem}

\newtheorem{proposition}[theorem]{Proposition}
\newtheorem{remark}[theorem]{Remark}

\newenvironment{proof}[1][Proof]{\noindent\textbf{#1.} }{\ \rule{0.5em}{0.5em}}
\allowdisplaybreaks
\begin{document}
\preprint{ }
\title[ ]{Quantifying the performance of approximate teleportation and quantum error
correction via symmetric two-PPT-extendibility}
\author{Tharon Holdsworth}
\affiliation{Department of Physics, University of Alabama at Birmingham, Birmingham,
Alabama 35233, USA}
\affiliation{Hearne Institute for Theoretical Physics, Department of Physics and Astronomy,
and Center for Computation and Technology, Louisiana State University, Baton
Rouge, Louisiana 70803, USA}
\author{Vishal Singh}
\affiliation{Hearne Institute for Theoretical Physics, Department of Physics and Astronomy,
and Center for Computation and Technology, Louisiana State University, Baton
Rouge, Louisiana 70803, USA}
\affiliation{School of Applied and Engineering Physics, Cornell University, Ithaca, New York 14850, USA}
\author{Mark M.~Wilde}
\affiliation{Hearne Institute for Theoretical Physics, Department of Physics and Astronomy,
and Center for Computation and Technology, Louisiana State University, Baton
Rouge, Louisiana 70803, USA}
\affiliation{School of Electrical and Computer Engineering, Cornell University, Ithaca, New York 14850, USA}

\begin{abstract}
The ideal realization of quantum teleportation relies on having access to a maximally entangled state; however, in practice, such an ideal state is typically not available and one can instead only realize an approximate teleportation. With this in mind, we present a method to quantify the performance of approximate teleportation when using an arbitrary resource state. More specifically, after framing the task of approximate teleportation as an optimization of a simulation error over one-way local operations and classical communication (LOCC) channels, we establish a semi-definite relaxation of this optimization task by instead optimizing over the larger set of two-PPT-extendible channels. The main analytical calculations in our paper consist of exploiting the unitary covariance symmetry of the identity channel to establish a significant reduction of the computational cost of this latter optimization.
Next, by exploiting known connections between approximate teleportation and quantum error correction, we also apply these concepts to establish bounds on the performance of approximate quantum error correction over a given quantum channel. 
Finally, we evaluate our bounds for various  examples of resource states and channels.
\end{abstract}

\date{\today}
\startpage{1}
\endpage{10}
\maketitle
\tableofcontents

\section{Introduction}

Teleportation is one of the most basic protocols in quantum information
science \cite{PhysRevLett.70.1895}. By means of two bits of classical
communication and an entangled pair of qubits (a so-called resource state), it
is possible to transmit a qubit from one location to another. This protocol
demonstrates the fascinating possibilities available under the distant
laboratories paradigm of local operations and classical communication (LOCC), 
and it prompted the development of the resource theory of entanglement
\cite{BDSW96}. Teleportation is so ubiquitous in quantum information science
now, that nearly every subfield (fault-tolerant computing, error correction,
cryptography, communication complexity, Shannon theory, etc.) employs it in
some manner. A number of impressive teleportation experiments have been
conducted over the past few decades
\cite{Bouwmeester1997,furusawa1998unconditional,PhysRevLett.80.1121,Riebe2004,Ursin2004,Sherson2006,Ma2012,Ren2017}%
.

The teleportation protocol assumes an ideal resource state; however, if the resource state shared between the two parties is imperfect, then the
teleportation protocol no longer simulates an ideal quantum channel, but 
rather some approximation of it \cite{PhysRevLett.72.797,PhysRevA.60.1888}.
This problem has been studied considerably in the literature and is related to the well-known problem of entanglement distillation
\cite{BBPSSW96EPP,BDSW96}. Recently, it has been addressed in a precise and
general operational way, in terms of a meaningful channel distinguishability
measure \cite[Definition~19]{KW17}.

In the seminal work \cite{BDSW96}, a connection was forged between
entanglement distillation and approximate quantum error correction. There, it
was shown that certain one-way LOCC\ entanglement distillation protocols can be
converted to approximate quantum error correction protocols, and vice versa.
Thus, techniques for analyzing entanglement distillation can be used to
analyze quantum error correction and vice versa.

In this paper, we obtain bounds on the performance of teleportation when using
an imperfect resource state, and by exploiting the aforementioned connection, we address a related problem
for approximate quantum error correction. We thus consider our paper to offer
two distinct, yet related contributions. The conceptual approach that we take
here is linked to that of \cite{SW21}, which was concerned with a more
involved protocol called bidirectional teleportation; it is also linked to \cite{KDWW19,KDWW21}, which introduced the set of $k$-extendible channels as a semi-definite relaxation of the set of one-way LOCC channels. Our approach has strong
links as well with that taken in \cite{BBFS21}, the latter concerned with
bounding the performance of approximate quantum error correction by means of
$k$-PPT-extendible channels; these channels were introduced in \cite{BBFS21} as a semi-definite relaxation of the set of one-way LOCC channels that forms a tighter containment than $k$-extendible channels alone. In fact, our method applied to the
problem of approximate quantum error correction can be understood as
exploiting further symmetries available when simulating the identity channel,
in order to reduce the computational complexity required to calculate the
bounds given in~\cite{BBFS21}.

Let us discuss our first contribution in a bit more detail. Suppose that the
goal is to use a bipartite resource state $\rho_{AB}$ along with one-way
LOCC to simulate a perfect quantum channel of dimension $d$. It is not always
possible to perform this simulation exactly, and for most resource states, an
error will occur. We can quantify the simulation error either in terms of the
diamond distance \cite{Kit97} or the channel infidelity \cite{GLN04}.
However, we prove here that the simulation error is the same, regardless of
whether we use the channel infidelity or the diamond distance, when
quantifying the deviation between the simulation and an ideal quantum channel (note
that a similar result was found previously in \cite{SW21} and we exploit similar techniques to arrive at our conclusion here). Next, in order to obtain a
lower bound on the simulation error, and due to the fact that it is
computationally challenging to optimize over one-way LOCC channels, we
optimize the error over the larger set of two-PPT-extendible channels (defined in Section~\ref{sec:2pk-ch}) and show
that the resulting quantity can be calculated by means of a semi-definite
program. By exploiting the unitary covariance symmetry of the ideal quantum channel, we reduce
the computational complexity of the semi-definite program to depend only on the dimension of
the resource state $\rho_{AB}$ being considered.
This constitutes
our main contribution to the analysis of teleportation with an imperfect
resource state. We also provide a general formulation of the simulation
problem when trying to simulate an arbitrary channel using one-way LOCC and a resource state.

The second contribution of our paper employs a similar line of reasoning to obtain a lower bound on the simulation error of
approximate quantum error correction. In this setting, instead of a bipartite
state, two parties have at their disposal a quantum channel $\mathcal{N}%
_{A\rightarrow B}$, for which they can prepend an encoding and append a
decoding in order to simulate a perfect quantum channel of dimension $d$. This
encoding and decoding can be understood as a superchannel \cite{CDP08}\ that
transforms $\mathcal{N}_{A\rightarrow B}$ into an approximation of the perfect
quantum channel. It is clear that the simulation error cannot increase by
allowing for a superchannel realized by one-way local operations and common randomness (LOCR), and here, following the
approach outlined above, we find a lower bound on the simulation error by
optimizing instead over the larger class of two-PPT-extendible superchannels with an extra non-signaling constraint.
Critically, this lower bound can be calculated by means of a semi-definite
program. As indicated above, this problem was previously considered in
\cite{BBFS21}, but our contribution is that the semi-definite programming
lower bound reported here has a substantially reduced computational
complexity, depending only on the input and output
dimensions of the channel $\mathcal{N}_{A\rightarrow B}$ of interest.

\subsection{Organization of the paper}

Our paper is organized into two major parts, according to the contributions mentioned above. The first part (Sections~\ref{sec: II}-\ref{sec: IV}) details our contribution to quantifying the performance of approximate teleportation. The second part (Sections~\ref{sec:superchs}-\ref{sec: VII}) details our contribution to quantifying the performance of approximate quantum error correction.

The first part of our paper is organized as follows: Section~\ref{sec: II} provides some background on quantum states and channels, with an emphasis on LOCC and LOCR bipartite channels. Section~\ref{sec:quantify-approx-tele} establishes a measure for the performance of quantum channel simulation, namely, in terms of the normalized diamond distance and channel infidelity. We prove here that these two error measures are equal when the goal is to simulate the identity channel, following as a consequence of the unitary covariance symmetry of the identity channel. Section~\ref{sec: IV} presents the major contribution of the first part, a semi-definite program (SDP) that gives a  lower bound on the simulation error of approximate teleportation when using an arbitrary bipartite resource state and one-way LOCC channels. This SDP is further simplified by exploiting the aforementioned symmetry of the identity channel to reduce the computational cost of the optimization task significantly.

The second part of our paper is organized as follows: Section~\ref{sec:superchs} provides background on quantum superchannels to generalize the concepts of one-way LOCC and LOCR bipartite channels to superchannels. Section~\ref{sec: VI} explores the task of channel simulation, i.e., simulating a quantum channel from an arbitrary quantum channel and LOCR superchannels. The performance of channel simulation is again quantified with the normalized diamond distance and channel infidelity, and again the error measures are equal when the goal is to simulate the identity channel with the assistance of common randomness. Section~\ref{sec: VII} presents the major contribution of the second part, an SDP that gives a lower bound on the error in simulating a quantum channel with an arbitrary channel and LOCR superchannels. We detail a much simplified SDP for the simulation of an identity channel, the case of interest in approximate quantum error correction, by leveraging its unitary covariance symmetry. 

Section~\ref{sec: VIII} presents plots that result from numerical calculations of our SDP error bounds. The first example in Section~\ref{sec:ex-1} bounds the error in approximate teleportation using a certain mixed state as the resource state, demonstrating that two-PPT-extendiblity constraints can achieve tighter bounds when compared to PPT constraints alone. The second example in Section~\ref{sec:ex-3d-tele-2d-resource} considers the bounds when using a lower dimensional resource state to simulate a higher dimensional identity channel. The next example in Section~\ref{sec:ex-depol-ch} considers the bounds for qubit and qutrit depolarizing channels.  The penultimate example in Section~\ref{sec:ex-2} bounds the error in approximate teleportation when using two-mode squeezed states as the resource state. The final example in Section~\ref{sec:ex-3} bounds the error in simulating an identity channel when using the three-level amplitude damping channel \cite{Chessa2021Feb}, and it is thus an example of our bound applied to approximate quantum error correction.

Section~\ref{sec: Conclusion} concludes by discussing several open questions for future work. We note here that Python code for calculating the SDPs in our paper is available with its arXiv posting.

\section{Background on states, channels, and bipartite channels}\label{sec: II}

We recall some basic facts about quantum information theory in this section to fix our notation before proceeding; more detailed background can be found in \cite{H17,H13book,Watrous2018,W17,KW20book}.

\subsection{States and channels}

A quantum state or density operator, usually denoted by $\rho_{A}$, $\sigma_A$, etc., is a positive semi-definite, unit trace operator  acting on a Hilbert space $\mathcal{H}_{A}$.
The Heisenberg--Weyl operators are unitary transformations of quantum states, defined  for all $x,z\in\{0,1,\ldots,d-1\}$ as
\begin{align}
Z(z)  &  \coloneqq\sum_{k=0}^{d-1}e^{\frac{2\pi ikz}{d}}|k\rangle\!\langle
k|,\\
X(x)  &  \coloneqq\sum_{k=0}^{d-1}|k\oplus_{d}x\rangle\!\langle k|,\\
W^{z,x}  &  \coloneqq Z(z)X(x),\label{eq:HW-ops}
\end{align}
where $\oplus_{d}$ denotes addition modulo $d$.

A quantum channel is a completely positive (CP), trace-preserving (TP)\ map.
Let $\mathcal{N}_{A\rightarrow B}$ denote a quantum channel that accepts as
input a linear operator acting on a Hilbert space$~\mathcal{H}_{A}$ and
outputs a linear operator acting on a Hilbert space$~\mathcal{H}_{B}$. For
short, we say that the channel takes system $A$ to system$~B$, where systems are
identified with Hilbert spaces.\ Let $\Gamma_{RB}^{\mathcal{N}}$ denote the
Choi operator of a channel $\mathcal{N}_{A\rightarrow B}$:%
\begin{equation}
\Gamma_{RB}^{\mathcal{N}}\coloneqq\mathcal{N}_{A\rightarrow B}(\Gamma_{RA}),
\end{equation}
where%
\begin{equation}
\Gamma_{RA}\coloneqq\sum_{i,j=0}^{d_{A}-1}|i\rangle\!\langle j|_{R}%
\otimes|i\rangle\!\langle j|_{A}%
\end{equation}
is the unnormalized maximally entangled operator and $\{|i\rangle_{R}\}_{i=0}^{d_{A}-1}$
and $\{|i\rangle_{A}\}_{i=0}^{d_{A}-1}$ are orthonormal bases.

The Choi representation of a channel is isomorphic to the superoperator representation and provides a convenient means of characterizing a channel. Namely, a channel  $\mathcal{M}_{A\rightarrow
B}$ is completely positive if and only if its Choi operator $\Gamma
_{RB}^{\mathcal{M}}$ is positive semi-definite and a channel
$\mathcal{M}_{A\rightarrow B}$ is trace preserving if and only if its Choi
operator $\Gamma_{RB}^{\mathcal{M}}$ satisfies $\operatorname{Tr}_{B}%
[\Gamma_{RB}^{\mathcal{M}}]=I_{R}$.

\subsection{Bipartite channels}

A bipartite channel $\mathcal{N}_{AB\rightarrow A^{\prime}B^{\prime}}$ maps
input systems $A$ and $B$ to output systems $A^{\prime}$ and $B^{\prime}$. In this model,
we assume that a single party Alice has access to systems $A$ and $A^{\prime}%
$, while another party Bob has access to systems $B$ and $B^{\prime}$.\ The
Choi operator for a bipartite channel $\mathcal{N}_{AB\rightarrow A^{\prime
}B^{\prime}}$ is as follows:%
\begin{equation}
\Gamma_{\tilde{A}\tilde{B}A^{\prime}B^{\prime}}^{\mathcal{N}}=\mathcal{N}%
_{AB\rightarrow A^{\prime}B^{\prime}}(\Gamma_{\tilde{A}A}\otimes\Gamma
_{\tilde{B}B}).
\end{equation}

\subsubsection{One-way LOCC channels}

A bipartite channel $\mathcal{L}_{AB\rightarrow A^{\prime}B^{\prime}}$ is a
one-way LOCC (1WL)\ channel if it can be written as follows:%
\begin{equation}
\mathcal{L}_{AB\rightarrow A^{\prime}B^{\prime}}=\sum_{x}\mathcal{E}%
_{A\rightarrow A^{\prime}}^{x}\otimes\mathcal{D}_{B\rightarrow B^{\prime}}%
^{x}, \label{eq:1WL-def}%
\end{equation}
where $\{\mathcal{E}_{A\rightarrow A^{\prime}}^{x}\}_{x}$ is a set of
completely positive maps, such that the sum map $\sum_{x}\mathcal{E}%
_{A\rightarrow A^{\prime}}^{x}$ is trace preserving, and $\{\mathcal{D}%
_{B\rightarrow B^{\prime}}^{x}\}_{x}$ is a set of quantum channels. The idea
here is that Alice acts on her system $A$ with a quantum instrument described
by $\{\mathcal{E}_{A\rightarrow A^{\prime}}^{x}\}_{x}$, transmits the classical outcome $x$ of the measurement over a classical communication
channel to Bob, who subsequently applies the quantum channel $\mathcal{D}%
_{B\rightarrow B^{\prime}}^{x}$ to his system $B$. A key example of a one-way
LOCC channel is in the teleportation protocol:\ given that Alice and Bob share a maximally entangled state in systems $\hat{A}\hat{B}$ and Alice has prepared
the system $A_{0}$ that she would like to teleport, the one-way LOCC\ channel consists of Alice performing a Bell measurement on systems $A_{0}\hat{A}$ (quantum instrument),
sending the measurement outcome to Bob (classical communication), who then applies a Heisenberg--Weyl correction
operation on system$~B$ conditioned on the classical communication from Alice. One-way LOCC channels are central in our analysis of approximate teleportation.

\subsubsection{LOCR\ channels}

A subset of one-way LOCC\ channels consists of those that can be implemented
by local operations and common randomness (LOCR). These channels have the following form:%
\begin{equation}
\mathcal{C}_{AB\rightarrow A^{\prime}B^{\prime}}=\sum_{y}p(y)\mathcal{E}%
_{A\rightarrow A^{\prime}}^{y}\otimes\mathcal{D}_{B\rightarrow B^{\prime}}%
^{y}, \label{eq:LOCR-ch-def}%
\end{equation}
where $\{p(y)\}_{y}$ is a probability distribution and $\{\mathcal{E}%
_{A\rightarrow A^{\prime}}^{y}\}_{y}$ and $\{\mathcal{D}_{B\rightarrow
B^{\prime}}^{y}\}_{y}$ are sets of quantum channels. The main difference
between one-way LOCC and LOCR is that, in the latter case, the channel is
simply a probabilistic mixture of local channels. In order to simulate them,
classical communication is not needed, and only the weaker resource of common
randomness is required. Thus, the following containment holds:%
\begin{equation}
\operatorname{LOCR}\subset\operatorname*{1WL}.
\end{equation}
These channels play a role in our analysis of approximate quantum error
correction and channel simulation.

\subsubsection{Two-extendible channels}

A bipartite channel $\mathcal{N}_{AB\rightarrow A^{\prime}B^{\prime}}$ is
two-extendible \cite{KDWW19,KDWW21}, if there exists an extension channel
$\mathcal{M}_{AB_{1}B_{2}\rightarrow A^{\prime}B_{1}^{\prime}B_{2}^{\prime}}$
satisfying permutation covariance:%
\begin{equation}
\mathcal{M}_{AB_{1}B_{2}\rightarrow A^{\prime}B_{1}^{\prime}B_{2}^{\prime}%
}\circ\mathcal{F}_{B_{1}B_{2}}=\mathcal{F}_{B_{1}^{\prime}B_{2}^{\prime}}%
\circ\mathcal{M}_{AB_{1}B_{2}\rightarrow A^{\prime}B_{1}^{\prime}B_{2}%
^{\prime}} \label{eq:perm-cov-ch}%
\end{equation}
and the following non-signaling constraint:%
\begin{equation}
\operatorname{Tr}_{B_{2^{\prime}}}\circ\mathcal{M}_{AB_{1}B_{2}\rightarrow
A^{\prime}B_{1}^{\prime}B_{2}^{\prime}}=\mathcal{N}_{AB_{1}\rightarrow
A^{\prime}B_{1}^{\prime}}\otimes\operatorname{Tr}_{B_{2}}.
\label{eq:marg-ch-ch}%
\end{equation}
In the above, $\mathcal{F}_{B_{1}B_{2}}$ is the unitary swap channel that
permutes systems $B_{1}$ and $B_{2}$, and $\mathcal{F}_{B_{1}^{\prime}%
B_{2}^{\prime}}$ is defined similarly. Also, $\operatorname{Tr}$ denotes the
partial trace channel. Note that the two conditions in \eqref{eq:perm-cov-ch}
and \eqref{eq:marg-ch-ch} imply that the original channel $\mathcal{N}_{AB\rightarrow A^{\prime
}B^{\prime}}$ is non-signaling from Bob to Alice, i.e.,%
\begin{equation}
\operatorname{Tr}_{B^{\prime}}\circ\mathcal{N}_{AB\rightarrow A^{\prime
}B^{\prime}}=\operatorname{Tr}_{B^{\prime}}\circ\mathcal{N}_{AB\rightarrow
A^{\prime}B^{\prime}}\circ\mathcal{R}_{B}^{\pi},
\label{eq:non-sig-implied-by-two-ext}%
\end{equation}
where%
\begin{equation}
\mathcal{R}_{B}^{\pi}(\cdot)\coloneqq\operatorname{Tr}[\cdot]\pi_{B}
\label{eq:replacer-ch-def}%
\end{equation}
is a replacer channel that traces out its input and replaces it with the
maximally mixed state $\pi_{B}\coloneqq\frac{I}{d_{B}}$. We provide a proof of \eqref{eq:non-sig-implied-by-two-ext} in Appendix~\ref{app:proof-2ext-non-sig-prop}. 

More generally,
$k$-extendible channels were defined in \cite{KDWW19,KDWW21}, and a resource
theory was constructed based on them. However, we only make use of
two-extendible channels in this work, and we leave the study of our problem
using $k$-extendible channels for future work. See \cite{BBFS21} for an
alternative definition of $k$-extendible channels that appeared after the
original proposal of \cite{KDWW19}. A key insight of \cite{KDWW19,KDWW21} is
that the set of one-way LOCC channels is contained in the set of
two-extendible channels, and we make use of this observation in our paper.

A bipartite channel $\mathcal{N}_{AB\rightarrow A^{\prime}B^{\prime}}$ is
two-extendible if and only if its Choi operator $\Gamma_{ABA^{\prime}%
B^{\prime}}^{\mathcal{N}}$ is such that there exists a Hermitian operator
$\Gamma_{AB_{1}B_{2}A^{\prime}B_{1}^{\prime}B_{2}^{\prime}}^{\mathcal{M}}$
satisfying \cite{KDWW19,KDWW21}%
\begin{align}
(\mathcal{F}_{B_{1}B_{2}}\otimes\mathcal{F}_{B_{1}^{\prime}B_{2}^{\prime}%
})(\Gamma_{AB_{1}B_{2}A^{\prime}B_{1}^{\prime}B_{2}^{\prime}}^{\mathcal{M}})
&  =\Gamma_{AB_{1}B_{2}A^{\prime}B_{1}^{\prime}B_{2}^{\prime}}^{\mathcal{M}%
},\label{eq:perm-cov-choi}\\
\operatorname{Tr}_{B_{2}^{\prime}}[\Gamma_{AB_{1}B_{2}A^{\prime}B_{1}^{\prime
}B_{2}^{\prime}}^{\mathcal{M}}]  &  =\Gamma_{AB_{1}A^{\prime}B_{1}^{\prime}%
}^{\mathcal{N}}\otimes I_{B_{2}},\label{eq:marg-ch-choi}\\
\Gamma_{AB_{1}B_{2}A^{\prime}B_{1}^{\prime}B_{2}^{\prime}}^{\mathcal{M}}  &
\geq0,\label{eq:ext-ch-CP}\\
\operatorname{Tr}_{A^{\prime}B_{1}^{\prime}B_{2}^{\prime}}[\Gamma_{AB_{1}%
B_{2}A^{\prime}B_{1}^{\prime}B_{2}^{\prime}}^{\mathcal{M}}]  &  =I_{AB_{1}%
B_{2}}. \label{eq:ext-ch-TP}%
\end{align}
The condition in \eqref{eq:perm-cov-choi} holds if and only if
\eqref{eq:perm-cov-ch} does. The condition in \eqref{eq:marg-ch-choi}\ holds
if and only if \eqref{eq:marg-ch-ch} does. Finally, \eqref{eq:ext-ch-CP} holds
if and only if $\mathcal{M}_{AB_{1}B_{2}\rightarrow A^{\prime}B_{1}^{\prime
}B_{2}^{\prime}}$ is completely positive, and \eqref{eq:ext-ch-TP} holds if
and only if $\mathcal{M}_{AB_{1}B_{2}\rightarrow A^{\prime}B_{1}^{\prime}%
B_{2}^{\prime}}$ is trace preserving. Related to the above, the conditions in
\eqref{eq:perm-cov-choi} and \eqref{eq:marg-ch-choi} imply the following
non-signaling condition on the Choi operator of $\mathcal{N}_{AB\rightarrow
A^{\prime}B^{\prime}}$:%
\begin{equation}
\operatorname{Tr}_{B^{\prime}}[\Gamma_{ABA^{\prime}B^{\prime}}^{\mathcal{N}%
}]=\frac{1}{d_{B}}\operatorname{Tr}_{BB^{\prime}}[\Gamma_{ABA^{\prime
}B^{\prime}}^{\mathcal{N}}]\otimes I_{B},
\end{equation}
which is equivalent to \eqref{eq:non-sig-implied-by-two-ext}. 

\subsubsection{Completely positive-partial-transpose preserving channels}

\label{sec:C-PPT-P-def}

A bipartite channel $\mathcal{N}_{AB\rightarrow A^{\prime}B^{\prime}}$ is
completely positive-partial-transpose preserving (C-PPT-P) \cite{Rai99,Rai01}%
\ if the map $T_{B^{\prime}}\circ\mathcal{N}_{AB\rightarrow A^{\prime
}B^{\prime}}\circ T_{B}$ is completely positive. Here, $T_{B}$ is the partial
transpose map, defined as the following superoperator:%
\begin{equation}
T_{B}(\cdot)\coloneqq\sum_{i,j}|i\rangle\!\langle j|_{B}(\cdot)|i\rangle
\!\langle j|_{B}. \label{eq:transpose-map-def}%
\end{equation}
See also \cite{CDGG20}. The set of one-way LOCC channels is contained in the
set of C-PPT-P channels \cite{Rai99,Rai01}, and we also make use of this
observation in our paper. A bipartite channel $\mathcal{N}_{AB\rightarrow
A^{\prime}B^{\prime}}$ is C-PPT-P\ if and only if its Choi operator
$\Gamma_{ABA^{\prime}B^{\prime}}^{\mathcal{N}}$ satisfies%
\begin{align}
\Gamma_{ABA^{\prime}B^{\prime}}^{\mathcal{N}}  &  \geq
0,\label{eq:cond-CPPTP-1}\\
\operatorname{Tr}_{A^{\prime}B^{\prime}}[\Gamma_{ABA^{\prime}B^{\prime}%
}^{\mathcal{N}}]  &  =I_{AB},\\
T_{BB^{\prime}}(\Gamma_{ABA^{\prime}B^{\prime}}^{\mathcal{N}})  &  \geq0,
\label{eq:cond-CPPTP-3}%
\end{align}
where $T_{BB^{\prime}}$ is the partial transpose acting on systems $B$ and$~B^{\prime
}$. We note that the C-PPT-P constraint has been used in prior work on
bounding the simulation error in bidirectional teleportation \cite{SW21}. See
also \cite{LM15,WXD18,WD16,WD16pra,BW17} for other contexts.

\subsubsection{Two-PPT-extendible channels}

\label{sec:2pk-ch}

We can combine the above constraints in a non-trivial way to define the set of
two-PPT-extendible channels, and we note that this was considered recently in
\cite[Remark after Lemma~4.10]{BBFS21}, as a generalization of the concept
employed for bipartite states \cite{DPS02,DPS04}. Explicitly, a bipartite
channel $\mathcal{N}_{AB\rightarrow A^{\prime}B^{\prime}}$ is
two-PPT-extendible if there exists an extension channel $\mathcal{M}%
_{AB_{1}B_{2}\rightarrow A^{\prime}B_{1}^{\prime}B_{2}^{\prime}}$ satisfying
the following conditions of permutation covariance, non-signaling, and being
completely-PPT-preserving:%
\begin{equation}
\mathcal{M}_{AB_{1}B_{2}\rightarrow A^{\prime}B_{1}^{\prime}B_{2}^{\prime}%
}\circ\mathcal{F}_{B_{1}B_{2}}=\mathcal{F}_{B_{1}^{\prime}B_{2}^{\prime}}%
\circ\mathcal{M}_{AB_{1}B_{2}\rightarrow A^{\prime}B_{1}^{\prime}B_{2}%
^{\prime}}, \label{eq:perm-cov-2pptext}%
\end{equation}%
\begin{align}
\operatorname{Tr}_{B_{2^{\prime}}}\circ\mathcal{M}_{AB_{1}B_{2}\rightarrow
A^{\prime}B_{1}^{\prime}B_{2}^{\prime}}  &  =\mathcal{N}_{AB_{1}\rightarrow
A^{\prime}B_{1}^{\prime}}\otimes\operatorname{Tr}_{B_{2}}%
,\label{eq:no-sig-2pptext}\\
T_{B_{2}^{\prime}}\circ\mathcal{M}_{AB_{1}B_{2}\rightarrow A^{\prime}%
B_{1}^{\prime}B_{2}^{\prime}}\circ T_{B_{2}}  &  \in\operatorname*{CP}%
,\label{eq:pptb2-2pptext}\\
T_{A^{\prime}}\circ\mathcal{M}_{AB_{1}B_{2}\rightarrow A^{\prime}B_{1}%
^{\prime}B_{2}^{\prime}}\circ T_{A}  &  \in\operatorname*{CP}.
\label{eq:ppta-2pptext}%
\end{align}
It is redundant to demand further that the following constraints hold:%
\begin{align}
T_{B_{1}^{\prime}}\circ\mathcal{M}_{AB_{1}B_{2}\rightarrow A^{\prime}%
B_{1}^{\prime}B_{2}^{\prime}}\circ T_{B_{1}}  &  \in\operatorname*{CP}%
,\label{eq:pptb1-2pptext}\\
T_{A^{\prime}B_{1}^{\prime}}\circ\mathcal{M}_{AB_{1}B_{2}\rightarrow
A^{\prime}B_{1}^{\prime}B_{2}^{\prime}}\circ T_{AB_{1}}  &  \in
\operatorname*{CP},\\
T_{A^{\prime}B_{2}^{\prime}}\circ\mathcal{M}_{AB_{1}B_{2}\rightarrow
A^{\prime}B_{1}^{\prime}B_{2}^{\prime}}\circ T_{AB_{2}}  &  \in
\operatorname*{CP},\\
T_{B_{1}^{\prime}B_{2}^{\prime}}\circ\mathcal{M}_{AB_{1}B_{2}\rightarrow
A^{\prime}B_{1}^{\prime}B_{2}^{\prime}}\circ T_{B_{1}B_{2}}  &  \in
\operatorname*{CP},
\end{align}
because they follow as a consequence of \eqref{eq:pptb2-2pptext} and
\eqref{eq:perm-cov-2pptext}, \eqref{eq:pptb2-2pptext},
\eqref{eq:pptb1-2pptext}, and \eqref{eq:ppta-2pptext}, respectively. A
bipartite channel $\mathcal{N}_{AB\rightarrow A^{\prime}B^{\prime}}$ is
two-PPT-extendible if and only if its Choi operator $\Gamma_{ABA^{\prime
}B^{\prime}}^{\mathcal{N}}$ is such that there exists a Hermitian operator
$\Gamma_{AB_{1}B_{2}A^{\prime}B_{1}^{\prime}B_{2}^{\prime}}^{\mathcal{M}}$
satisfying%
\begin{align}
(\mathcal{F}_{B_{1}B_{2}}\otimes\mathcal{F}_{B_{1}^{\prime}B_{2}^{\prime}%
})(\Gamma_{AB_{1}B_{2}A^{\prime}B_{1}^{\prime}B_{2}^{\prime}}^{\mathcal{M}})
&  =\Gamma_{AB_{1}B_{2}A^{\prime}B_{1}^{\prime}B_{2}^{\prime}}^{\mathcal{M}%
},\label{eq:SDP-constr-2pe-1}\\
\operatorname{Tr}_{B_{2}^{\prime}}[\Gamma_{AB_{1}B_{2}A^{\prime}B_{1}^{\prime
}B_{2}^{\prime}}^{\mathcal{M}}]  &  =\Gamma_{AB_{1}A^{\prime}B_{1}^{\prime}%
}^{\mathcal{N}}\otimes I_{B_{2}},\label{eq:SDP-constr-2pe-2}\\
T_{B_{2}B_{2}^{\prime}}(\Gamma_{AB_{1}B_{2}A^{\prime}B_{1}^{\prime}%
B_{2}^{\prime}}^{\mathcal{M}})  &  \geq0,\\
T_{AA^{\prime}}(\Gamma_{AB_{1}B_{2}A^{\prime}B_{1}^{\prime}B_{2}^{\prime}%
}^{\mathcal{M}})  &  \geq0,\\
\Gamma_{AB_{1}B_{2}A^{\prime}B_{1}^{\prime}B_{2}^{\prime}}^{\mathcal{M}}  &
\geq0,\\
\operatorname{Tr}_{A^{\prime}B_{1}^{\prime}B_{2}^{\prime}}[\Gamma_{AB_{1}%
B_{2}A^{\prime}B_{1}^{\prime}B_{2}^{\prime}}^{\mathcal{M}}]  &  =I_{AB_{1}%
B_{2}}. \label{eq:SDP-constr-2pe-last}%
\end{align}

Observe that a bipartite channel $\mathcal{N}_{AB\rightarrow A^{\prime
}B^{\prime}}$ is C-PPT-P if it is two-PPT-extendible. This follows from
\eqref{eq:no-sig-2pptext} and \eqref{eq:ppta-2pptext}.

Every one-way\ LOCC channel of the form in \eqref{eq:1WL-def}\ is
two-PPT-extendible by considering the following extension channel:%
\begin{equation}
\sum_{x}\mathcal{E}_{A\rightarrow A^{\prime}}^{x}\otimes\mathcal{D}%
_{B_{1}\rightarrow B_{1}^{\prime}}^{x}\otimes\mathcal{D}_{B_{2}\rightarrow
B_{2}^{\prime}}^{x},
\end{equation}
which manifestly satisfies the constraints in \eqref{eq:perm-cov-2pptext}--\eqref{eq:ppta-2pptext}. We
thus employ two-PPT-extendible channels as a semi-definite relaxation of the
set of one-way LOCC channels.

\subsubsection{Two-PPT-extendible non-signaling channels}

We can add a further constraint to the channels discussed in the previous
section, i.e., a non-signaling constraint of the following form:%
\begin{equation}
\operatorname{Tr}_{A^{\prime}}\circ\mathcal{M}_{AB_{1}B_{2}\rightarrow
A^{\prime}B_{1}^{\prime}B_{2}^{\prime}}=\operatorname{Tr}_{A^{\prime}}%
\circ\mathcal{M}_{AB_{1}B_{2}\rightarrow A^{\prime}B_{1}^{\prime}B_{2}%
^{\prime}}\circ\mathcal{R}_{A}^{\pi}, \label{eq:def-A-to-Bs-nonsig}%
\end{equation}
which ensures that the extension channel $\mathcal{M}_{AB_{1}B_{2}\rightarrow
A^{\prime}B_{1}^{\prime}B_{2}^{\prime}}$ is also non-signaling from Alice to
both Bobs. The constraint on the Choi operator $\Gamma_{AB_{1}B_{2}A^{\prime}B_{1}^{\prime
}B_{2}^{\prime}}^{\mathcal{M}}$ is as follows:%
\begin{equation}
\operatorname{Tr}_{A^{\prime}}[\Gamma_{AB_{1}B_{2}A^{\prime}B_{1}^{\prime
}B_{2}^{\prime}}^{\mathcal{M}}]=\frac{1}{d_{A}}\operatorname{Tr}_{A^{\prime}%
A}[\Gamma_{AB_{1}B_{2}A^{\prime}B_{1}^{\prime}B_{2}^{\prime}}^{\mathcal{M}%
}]\otimes I_{A}.
\end{equation}

Every LOCR\ channel of the form in \eqref{eq:LOCR-ch-def}\ is
two-PPT-extendible non-signaling, as is evident by choosing the following
extension channel:%
\begin{equation}
\sum_{y}p(y)\mathcal{E}_{A\rightarrow A^{\prime}}^{y}\otimes\mathcal{D}%
_{B_{1}\rightarrow B_{1}^{\prime}}^{y}\otimes\mathcal{D}_{B_{2}\rightarrow
B_{2}^{\prime}}^{y}.
\end{equation}
We thus employ two-PPT-extendible non-signaling channels as a semi-definite
relaxation of the set of LOCR channels, and we note here that \cite{BBFS21} previously used this approach.

Let us state explicitly here that extensions of one-way LOCC channels of the
form in \eqref{eq:1WL-def} generally do not satisfy the non-signaling
constraint in \eqref{eq:def-A-to-Bs-nonsig}, due to the fact that each map
$\mathcal{E}_{A\rightarrow A^{\prime}}^{x}$ in \eqref{eq:1WL-def} is not
necessarily trace preserving.

\section{Quantifying the performance of approximate teleportation}\label{sec:quantify-approx-tele}

In approximate teleportation, Alice and Bob are allowed to make use of a fixed
bipartite state $\rho_{\hat{A}\hat{B}}$ and an arbitrary one-way LOCC channel
$\mathcal{L}_{A\hat{A}\hat{B}\rightarrow B}$, with the goal of simulating an
identity channel of dimension $d$. To be clear, the one-way LOCC channel
$\mathcal{L}_{A\hat{A}\hat{B}\rightarrow B}$ has the following form:%
\begin{equation}
\mathcal{L}_{A\hat{A}\hat{B}\rightarrow B}(\omega_{A\hat{A}\hat{B}})=\sum
_{x}\mathcal{D}_{\hat{B}\rightarrow B}^{x}(\operatorname{Tr}_{A\hat{A}%
}[\Lambda_{A\hat{A}}^{x}\omega_{A\hat{A}\hat{B}}]),
\end{equation}
where $\{\Lambda_{A\hat{A}}^{x}\}_{x}$ is a positive operator-valued measure
(satisfying $\Lambda_{A\hat{A}}^{x}\geq0$ for all $x$ and $\sum_{x}%
\Lambda_{A\hat{A}}^{x}=I_{A\hat{A}}$) and $\{\mathcal{D}_{\hat{B}\rightarrow
B}^{x}\}_{x}$ is a set of quantum channels. We assume that the dimension of
the systems $\hat{A}\hat{B}$ is finite, and we write the dimension of $\hat
{A}$ as $d_{\hat{A}}$ and the dimension of $\hat{B}$ as $d_{\hat{B}}$. The
approximate teleportation protocol realizes the following simulation channel
$\widetilde{\mathcal{S}}_{A\rightarrow B}$ \cite[Eq.~(11)]{PhysRevA.60.1888}:%
\begin{equation}
\widetilde{\mathcal{S}}_{A\rightarrow B}(\omega_{A})\coloneqq\mathcal{L}%
_{A\hat{A}\hat{B}\rightarrow B}(\omega_{A}\otimes\rho_{\hat{A}\hat{B}}).
\label{eq:sim-channel}%
\end{equation}
In the following subsections, we discuss two seemingly different ways of
quantifying the simulation error.

\subsection{Quantifying simulation error with normalized diamond distance}

\label{sec:quant-error-dia-dist}The standard metric for quantifying the
distance between quantum channels is the normalized diamond distance
\cite{Kit97}. See the related paper \cite{SW21}\ for discussions of the
operational significance of the diamond distance (see also \cite{KW20book}).
For channels $\mathcal{N}_{C\rightarrow D}$ and $\widetilde{\mathcal{N}%
}_{C\rightarrow D}$, the diamond distance is defined as%
\begin{multline}
\left\Vert \mathcal{N}_{C\rightarrow D}-\widetilde{\mathcal{N}}_{C\rightarrow
D}\right\Vert _{\diamond}\\
\coloneqq\sup_{\rho_{RC}}\left\Vert \mathcal{N}_{C\rightarrow D}(\rho
_{RC})-\widetilde{\mathcal{N}}_{C\rightarrow D}(\rho_{RC})\right\Vert _{1},
\end{multline}
where the optimization is over every bipartite state $\rho_{RC}$ with the
reference system $R$ arbitrarily large. The following equality is well known (see, e.g., \cite{KW20book})
\begin{multline}
\left\Vert \mathcal{N}_{C\rightarrow D}-\widetilde{\mathcal{N}}_{C\rightarrow
D}\right\Vert _{\diamond}\label{eq:diamond-dist-reduction}\\
=\sup_{\psi_{RC}}\left\Vert \mathcal{N}_{C\rightarrow D}(\psi_{RC}%
)-\widetilde{\mathcal{N}}_{C\rightarrow D}(\psi_{RC})\right\Vert _{1},
\end{multline}
where the optimization is over every pure bipartite state $\psi_{RC}$ with the
reference system $R$ isomorphic to the channel input system $C$. The
normalized diamond distance is then given by%
\begin{equation}
\frac{1}{2}\left\Vert \mathcal{N}_{C\rightarrow D}-\widetilde{\mathcal{N}%
}_{C\rightarrow D}\right\Vert _{\diamond},
\end{equation}
so that the resulting error takes a value between zero and one.
The reduction in \eqref{eq:diamond-dist-reduction} implies that it is a
computationally tractable problem to calculate the diamond distance, and in
fact, one can do so by means of the following semi-definite program
\cite{Wat09}:%
\begin{equation}
\inf_{\lambda,Z_{RD}\geq0}\left\{
\begin{array}
[c]{c}%
\lambda:\lambda I_{R}\geq\operatorname{Tr}_{D}[Z_{RD}],\\
Z_{RD}\geq\Gamma_{RD}^{\mathcal{N}}-\Gamma_{RD}^{\widetilde{\mathcal{N}}}%
\end{array}
\right\}  , \label{eq:diamond-d-SDP}%
\end{equation}
where $\Gamma_{RD}^{\mathcal{N}}$ and $\Gamma_{RD}^{\widetilde{\mathcal{N}}}$
are the Choi operators of $\mathcal{N}_{C\rightarrow D}$ and $\widetilde
{\mathcal{N}}_{C\rightarrow D}$, respectively.

The simulation error when using a bipartite state $\rho_{\hat{A}\hat{B}}$ and
a one-way LOCC\ channel to simulate an identity channel $\operatorname{id}%
_{A\rightarrow B}^{d}$ of dimension $d$ is given by%
\begin{equation}
e_{\operatorname{1WL}}(\rho_{\hat{A}\hat{B}},\mathcal{L}_{A\hat{A}\hat
{B}\rightarrow B})\coloneqq\frac{1}{2}\left\Vert \operatorname{id}%
_{A\rightarrow B}^{d}-\widetilde{\mathcal{S}}_{A\rightarrow B}\right\Vert
_{\diamond},
\end{equation}
where the simulation channel $\widetilde{\mathcal{S}}_{A\rightarrow B}$ is
defined in \eqref{eq:sim-channel}. Employing
\eqref{eq:diamond-dist-reduction}, we find that%
\begin{multline}
e_{\operatorname{1WL}}(\rho_{\hat{A}\hat{B}},\mathcal{L}_{A\hat{A}\hat
{B}\rightarrow B})\\
=\sup_{\psi_{RA}}\frac{1}{2}\left\Vert \psi_{RA}-\mathcal{L}_{A\hat{A}\hat
{B}\rightarrow B}(\psi_{RA}\otimes\rho_{\hat{A}\hat{B}})\right\Vert _{1},
\end{multline}
with $\psi_{RA}$ a pure bipartite state such that system $R$ is isomorphic to
system $A$. We are interested in the minimum possible simulation error, and so
we define%
\begin{equation}
e_{\operatorname{1WL}}(\rho_{\hat{A}\hat{B}})\coloneqq\inf_{\mathcal{L\in
}\operatorname{1WL}}e_{\operatorname{1WL}}(\rho_{\hat{A}\hat{B}}%
,\mathcal{L}_{A\hat{A}\hat{B}\rightarrow B}), \label{eq:diamond-err-sim-tele}%
\end{equation}
where we recall that $\operatorname{1WL}$ denotes the set of one-way LOCC
channels. The error $e_{\operatorname{1WL}}(\rho_{\hat{A}\hat{B}})$ is one
kind of simulation error on which we are interested in obtaining
computationally efficient lower bounds. Indeed, it is a computationally
difficult problem to calculate $e_{\operatorname{1WL}}(\rho_{\hat{A}\hat{B}})$
directly, and so we instead resort to calculating lower bounds.

\subsection{Quantifying simulation error with channel infidelity}

Another measure of the simulation error is by means of the channel infidelity. Let us recall that the fidelity of states $\omega$ and $\tau$ is
defined as \cite{Uhl76}
\begin{equation}
F(\omega,\tau)\coloneqq\left\Vert \sqrt{\omega}\sqrt{\tau}\right\Vert _{1}%
^{2},
\end{equation}
where $\left\Vert X\right\Vert _{1}\coloneqq\operatorname{Tr}[\sqrt{X^{\dag}%
X}]$. From this measure, we can define a channel fidelity measure for channels
$\mathcal{N}_{C\rightarrow D}$ and $\widetilde{\mathcal{N}}_{C\rightarrow D}$
as follows:%
\begin{equation}
F(\mathcal{N},\widetilde{\mathcal{N}})\coloneqq\inf_{\rho_{RC}}F(\mathcal{N}%
_{C\rightarrow D}(\rho_{RC}),\widetilde{\mathcal{N}}_{C\rightarrow D}%
(\rho_{RC})),
\end{equation}
where the optimization is over every bipartite state $\rho_{RC}$ with the
reference system $R$ arbitrarily large. Similar to the diamond distance, it
suffices to optimize the channel fidelity over every pure bipartite state $\psi_{RC}$ with
reference system $R$ isomorphic to the channel input system $C$ (see, e.g., \cite{KW20book}):%
\begin{equation}
F(\mathcal{N},\widetilde{\mathcal{N}})\coloneqq\inf_{\psi_{RC}}F(\mathcal{N}%
_{C\rightarrow D}(\psi_{RC}),\widetilde{\mathcal{N}}_{C\rightarrow D}%
(\psi_{RC})). \label{eq:ch-fid-simplified-1}%
\end{equation}
The square root of the channel fidelity can be calculated by means of the
following semi-definite program \cite{Yuan2017,KW20}:%
\begin{equation}
\sqrt{F}(\mathcal{N},\widetilde{\mathcal{N}})=\sup_{\lambda\geq0,Q_{RD}%
}\lambda\label{eq:SDP-ch-fid-1}%
\end{equation}
subject to%
\begin{align}
\lambda I_{R}  &  \leq\operatorname{Re}[\operatorname{Tr}_{D}[Q_{RD}]],\\%
\begin{bmatrix}
\Gamma_{RD}^{\widetilde{\mathcal{N}}} & Q_{RD}^{\dag}\\
Q_{RD} & \Gamma_{RD}^{\mathcal{N}}%
\end{bmatrix}
&  \geq0. \label{eq:SDP-ch-fid-3}%
\end{align}

An alternative method for quantifying the simulation error is to employ the
channel infidelity, defined as $1-F(\mathcal{N},\widetilde{\mathcal{N}})$.
Indeed, we can measure the simulation error as follows, when using a bipartite
state $\rho_{\hat{A}\hat{B}}$ and a one-way LOCC channel~$\mathcal{L}%
_{A\hat{A}\hat{B}\rightarrow B}$:%
\begin{equation}
e_{\text{$\operatorname{1WL}$}}^{F}(\rho_{\hat{A}\hat{B}},\mathcal{L}%
_{A\hat{A}\hat{B}\rightarrow B})\coloneqq1-F(\operatorname{id}_{A\rightarrow
B}^{d},\widetilde{\mathcal{S}}_{A\rightarrow B}),
\end{equation}
where the simulation channel $\widetilde{\mathcal{S}}_{A\rightarrow B}$ is
defined in \eqref{eq:sim-channel}. By employing
\eqref{eq:ch-fid-simplified-1}, we find that%
\begin{multline}
e_{\text{$\operatorname{1WL}$}}^{F}(\rho_{\hat{A}\hat{B}},\mathcal{L}%
_{A\hat{A}\hat{B}\rightarrow B})=\\
\sup_{\psi_{RA}}\left[  1-F(\psi_{RA},\mathcal{L}_{A\hat{A}\hat{B}\rightarrow
B}(\psi_{RA}\otimes\rho_{\hat{A}\hat{B}}))\right]  ,
\end{multline}
where the optimization is over every pure bipartite state $\psi_{RA}$ with
system $R$ isomorphic to the channel input system $A$. Since we are interested
in the minimum possible simulation error, we define%
\begin{equation}
e_{\text{$\operatorname{1WL}$}}^{F}(\rho_{\hat{A}\hat{B}})\coloneqq\inf
_{\mathcal{L}\in\text{$\operatorname{1WL}$}}e_{\text{$\operatorname{1WL}$}%
}^{F}(\rho_{\hat{A}\hat{B}},\mathcal{L}_{A\hat{A}\hat{B}\rightarrow B}).
\label{eq:fid-err-sim-tele}%
\end{equation}
This is the other kind of simulation error on which we are interested in
obtaining lower bounds.

\subsection{One-way LOCC\ simulation of general point-to-point channels}

Beyond the case of simulating an ideal channel, more generally we can consider
using a resource state $\rho_{\hat{A}\hat{B}}$ along with a one-way LOCC
channel $\mathcal{L}_{A\hat{A}\hat{B}\rightarrow B}$ in order to simulate a
general channel $\mathcal{N}_{A\rightarrow B}$. In this case, the simulation
channel has the following form:%
\begin{equation}
\widetilde{\mathcal{N}}_{A\rightarrow B}(\omega_{A})\coloneqq\mathcal{L}%
_{A\hat{A}\hat{B}\rightarrow B}(\omega_{A}\otimes\rho_{\hat{A}\hat{B}}).
\end{equation}
The simulation error when employing a specific one-way LOCC\ channel
$\mathcal{L}_{A\hat{A}\hat{B}\rightarrow B}$ is%
\begin{equation}
e_{\text{$\operatorname{1WL}$}}(\mathcal{N}_{A\rightarrow B},\rho_{\hat{A}%
\hat{B}},\mathcal{L}_{A\hat{A}\hat{B}\rightarrow B})\coloneqq\frac{1}%
{2}\left\Vert \mathcal{N}-\widetilde{\mathcal{N}}\right\Vert _{\diamond},
\end{equation}
and the simulation error minimized over all possible one-way LOCC channels is%
\begin{multline}
e_{\text{$\operatorname{1WL}$}}(\mathcal{N}_{A\rightarrow B},\rho_{\hat{A}%
\hat{B}})\coloneqq\label{eq:sim-err-diamond-1WL-gen}\\
\inf_{\mathcal{L}\in\text{$\operatorname{1WL}$}}e_{\text{$\operatorname{1WL}$%
}}(\mathcal{N}_{A\rightarrow B},\rho_{\hat{A}\hat{B}},\mathcal{L}_{A\hat
{A}\hat{B}\rightarrow B}).
\end{multline}
We note here that this is a special case of the simulation problem considered
in \cite[Section II]{FWTB18}.

Alternatively, we can employ the infidelity to quantify the simulation error
as follows:%
\begin{equation}
e_{\text{$\operatorname{1WL}$}}^{F}(\mathcal{N}_{A\rightarrow B},\rho_{\hat
{A}\hat{B}},\mathcal{L}_{A\hat{A}\hat{B}\rightarrow B}%
)\coloneqq 1-F(\mathcal{N},\widetilde{\mathcal{N}}),
\end{equation}%
\begin{multline}
e_{\text{$\operatorname{1WL}$}}^{F}(\mathcal{N}_{A\rightarrow B},\rho_{\hat
{A}\hat{B}})\coloneqq\\
\inf_{\mathcal{L}\in\text{$\operatorname{1WL}$}}e_{\text{$\operatorname{1WL}$%
}}^{F}(\mathcal{N}_{A\rightarrow B},\rho_{\hat{A}\hat{B}},\mathcal{L}%
_{A\hat{A}\hat{B}\rightarrow B}).
\end{multline}

\subsection{Equality of simulation errors when simulating the identity
channel}

Proposition~\ref{prop:sim-errs-equal-1WL}\ below states that the following
equality holds for every bipartite state $\rho_{\hat{A}\hat{B}}$:
\begin{equation}
e_{\text{$\operatorname{1WL}$}}(\rho_{\hat{A}\hat{B}}%
)=e_{\text{$\operatorname{1WL}$}}^{F}(\rho_{\hat{A}\hat{B}}).
\end{equation}
We
provide an explicit proof in Appendix~\ref{app:proof-errors-equal}.
This equality follows as a consequence of the unitary covariance symmetry of the identity
channel being simulated and the fact that an optimal simulating channel
should respect the same symmetries. Indeed, consider that the identity channel $\operatorname{id}_{A\rightarrow
B}^{d}$ possesses the following unitary covariance symmetry:%
\begin{equation}
\operatorname{id}_{A\rightarrow B}^{d}\circ\mathcal{U}_{A}=\mathcal{U}%
_{B}\circ\operatorname{id}_{A\rightarrow B}^{d},
\label{eq:identity-ch-symmetries}%
\end{equation}
which holds for every unitary channel $\mathcal{U}(\cdot)=U(\cdot)U^{\dag}$,
with $U$ a unitary operator. As a consequence, the theory simplifies in the sense that we need only focus on bounding the simulation error with respect to a single measure. We note here that a similar result was found in \cite{SW21} for the case of simulating the bipartite swap channel by means of LOCC. 

\begin{proposition}
\label{prop:sim-errs-equal-1WL}The optimization problems in
\eqref{eq:diamond-err-sim-tele} and \eqref{eq:fid-err-sim-tele}, for the error
in simulating the identity channel $\operatorname{id}_{A\rightarrow B}^{d}$,
simplify as follows:%
\begin{align}
e_{\text{$\operatorname{1WL}$}}(\rho_{\hat{A}\hat{B}})  &
=e_{\text{$\operatorname{1WL}$}}^{F}(\rho_{\hat{A}\hat{B}})\\
&  =1-\sup_{K_{\hat{A}\hat{B}},L_{\hat{A}\hat{B}}\geq0}\operatorname{Tr}%
[K_{\hat{A}\hat{B}}\rho_{\hat{A}\hat{B}}],
\end{align}
subject to $K_{\hat{A}\hat{B}}+L_{\hat{A}\hat{B}}=I_{\hat{A}\hat{B}}$ and the
following channel $\mathcal{L}_{A\hat{A}\hat{B}\rightarrow B}$ being a one-way
LOCC\ channel:%
\begin{multline}
\mathcal{L}_{A\hat{A}\hat{B}\rightarrow B}(\omega_{A\hat{A}\hat{B}%
})=\operatorname{id}_{A\rightarrow B}^{d}(\operatorname{Tr}_{\hat{A}\hat{B}%
}[K_{\hat{A}\hat{B}}\omega_{A\hat{A}\hat{B}}])\\
+\mathcal{D}_{A\rightarrow B}(\operatorname{Tr}_{\hat{A}\hat{B}}[L_{\hat
{A}\hat{B}}\omega_{A\hat{A}\hat{B}}]),
\end{multline}
where $\mathcal{D}_{A\rightarrow B}$ is the following channel:%
\begin{equation}\label{eq:randomizing channel}
\mathcal{D}_{A\rightarrow B}(\sigma_{A})\coloneqq\frac{1}{d^{2}-1}%
\sum_{(z,x)\neq(0,0)}W^{z,x}\sigma(W^{z,x})^{\dag},
\end{equation}
and $W^{z,x}$ is defined in \eqref{eq:HW-ops}. The constraint that
$\mathcal{L}_{A\hat{A}\hat{B}\rightarrow B}$ is a one-way LOCC channel is
equivalent to the existence of a positive operator-valued measure (POVM) $\{\Lambda_{B\hat{A}}^{x}\}_{x}$ and a
set $\{\mathcal{D}_{\hat{B}\rightarrow B}^{x}\}_{x}$\ of channels such that%
\begin{equation}
K_{\hat{A}\hat{B}}=\frac{1}{d^{2}}\sum_{x}\operatorname{Tr}_{B}[T_{B}%
(\Lambda_{B\hat{A}}^{x})\Gamma_{\hat{B}B}^{\mathcal{D}^{x}}],
\label{eq:one-way-LOCC-K-form}%
\end{equation}
where $\Gamma_{\hat{B}B}^{\mathcal{D}^{x}}$ is the Choi operator of the
channel $\mathcal{D}_{\hat{B}\rightarrow B}^{x}$.
\end{proposition}

\begin{proof}
See Appendix~\ref{app:proof-errors-equal}.
\end{proof}

\section{SDP lower bounds on the performance of approximate teleportation based on two-PPT-extendibility}\label{sec: IV}

\subsection{SDP lower bound on the error in one-way LOCC\ simulation of a
channel}\label{sec_SDP_lower_bound_tel}

It is difficult to compute the simulation error $e_{\text{$\operatorname{1WL}%
$}}(\mathcal{N}_{A\rightarrow B},\rho_{\hat{A}\hat{B}})$ defined in
\eqref{eq:sim-err-diamond-1WL-gen} because it is challenging to optimize over
the set of one-way LOCC channels \cite{Gur04,Ghar10}. Here we enlarge the set
of one-way LOCC channels to the set of two-PPT-extendible bipartite channels,
with the goal of simplifying the calculation of the simulation error. The
result is that we provide a lower bound on the one-way LOCC\ simulation error
in terms of a semi-definite program, which follows because the set of
two-PPT-extendible channels is specified by semi-definite constraints, as
indicated in \eqref{eq:SDP-constr-2pe-1}--\eqref{eq:SDP-constr-2pe-last}.

In more detail, recall that a bipartite channel is two-PPT-extendible if the
conditions in \eqref{eq:perm-cov-2pptext}--\eqref{eq:ppta-2pptext} hold. As
indicated previously at the end of Section~\ref{sec:2pk-ch}, every one-way
LOCC\ channel is a two-extendible channel, and the containment is strict.
Thus,%
\begin{equation}
\text{$\operatorname{1WL}$}\subset\text{$\operatorname*{2PE}\ $},
\label{eq:1wl-in-2ext}%
\end{equation}
where $\operatorname*{2PE}$ denotes the set of two-PPT-extendible channels, as
defined in Section~\ref{sec:2pk-ch}.

We can then define the simulation error under two-PPT-extendible channels, as a semi-definite relaxation of \eqref{eq:sim-err-diamond-1WL-gen},   as
follows:%
\begin{equation}
e_{\operatorname{2PE}}(\mathcal{N}_{A\rightarrow B},\rho_{\hat{A}\hat{B}%
})\coloneqq\inf_{\mathcal{K}\in\operatorname{2PE}}\frac{1}%
{2}\left\Vert \mathcal{N}-\widetilde{\mathcal{N}}\right\Vert _{\diamond},
\label{eq:sim-err-2-ext-gen-ch}%
\end{equation}
where%
\begin{equation}
\widetilde{\mathcal{N}}_{A\rightarrow B}(\omega_{A})\coloneqq\mathcal{K}%
_{A\hat{A}\hat{B}\rightarrow B}(\omega_{A}\otimes\rho_{\hat{A}\hat{B}})
\end{equation}
and $\mathcal{K}_{A\hat{A}\hat{B}\rightarrow B}$ is a two-PPT-extendible
channel, meaning that there exists an extension channel $\mathcal{M}_{A\hat
{A}\hat{B}_{1}\hat{B}_{2}\rightarrow B_{1}B_{2}}$ satisfying the following
conditions:%
\begin{align}
\operatorname{Tr}_{B_{2}}\circ\mathcal{M}_{A\hat{A}\hat{B}_{1}\hat{B}%
_{2}\rightarrow B_{1}B_{2}}  &  =\mathcal{K}_{A\hat{A}\hat{B}_{1}\rightarrow
B_{1}}\otimes\operatorname{Tr}_{\hat{B}_{2}},\\
\mathcal{M}_{A\hat{A}\hat{B}_{1}\hat{B}_{2}\rightarrow B_{1}B_{2}}%
\circ\mathcal{F}_{\hat{B}_{1}\hat{B}_{2}}  &  =\mathcal{F}_{B_{1}B_{2}}%
\circ\mathcal{M}_{A\hat{A}\hat{B}_{1}\hat{B}_{2}\rightarrow B_{1}B_{2}%
},\label{eq:perm-cov-sim-ch}\\
T_{B_{2}}\circ\mathcal{M}_{A\hat{A}\hat{B}_{1}\hat{B}_{2}\rightarrow
B_{1}B_{2}}\circ T_{\hat{B}_{2}}  &  \in\text{$\operatorname*{CP}$%
},\label{eq:C-PPT-P-constr-1-B2}\\
\mathcal{M}_{A\hat{A}\hat{B}_{1}\hat{B}_{2}\rightarrow B_{1}B_{2}}\circ
T_{A\hat{A}}  &  \in\text{$\operatorname*{CP}$}. \label{eq:C-PPT-P-constr-2-A}%
\end{align}
As a consequence of the containment in \eqref{eq:1wl-in-2ext}, the following
bound holds%
\begin{equation}
e_{\operatorname{2PE}}(\mathcal{N}_{A\rightarrow B},\rho_{\hat{A}\hat{B}})\leq
e_{\text{$\operatorname*{1WL}$}}(\mathcal{N}_{A\rightarrow B},\rho_{\hat
{A}\hat{B}}). \label{eq:lower-bnd-two-ext}%
\end{equation}

We now show that the simulation error in \eqref{eq:sim-err-2-ext-gen-ch}\ can
be calculated by means of a semi-definite program.

\begin{proposition}
\label{prop:two-ext-sim-err-gen-ch}The simulation error in
\eqref{eq:sim-err-2-ext-gen-ch} can be calculated by means of the following
semi-definite program:%
\begin{equation}
e_{\operatorname{2PE}}(\mathcal{N}_{A\rightarrow B},\rho_{\hat{A}\hat{B}%
})=\inf_{\substack{\mu\geq0,Z_{AB}\geq0,\\M_{A\hat{A}\hat{B}_{1}B_{1}\hat
{B}_{2}B_{2}}\geq0}}\mu, \label{eq:obj-func-gen-two-ext-sim}%
\end{equation}
subject to%
\begin{align}
\mu I_{A}  &  \geq Z_{A},\label{eq:general-SDP-two-ext-sim-2}\\
Z_{AB}  &  \geq\Gamma_{AB}^{\mathcal{N}}-\operatorname{Tr}_{\hat{A}\hat{B}%
_{1}}\!\left[  T_{\hat{A}\hat{B}_{1}}(\rho_{\hat{A}\hat{B}_{1}})\frac
{M_{A\hat{A}\hat{B}_{1}B_{1}}}{d_{\hat{B}}}\right]  ,
\label{eq:general-SDP-two-ext-sim}%
\end{align}%
\begin{align}
\operatorname{Tr}_{B_{1}B_{2}}[M_{A\hat{A}\hat{B}_{1}B_{1}\hat{B}_{2}B_{2}}]
&  =I_{A\hat{A}\hat{B}_{1}\hat{B}_{2}},\label{eq:sim-two-ext-TP}\\
(\mathcal{F}_{\hat{B}_{1}\hat{B}_{2}}\otimes\mathcal{F}_{B_{1}B_{2}}%
)(M_{A\hat{A}\hat{B}_{1}B_{1}\hat{B}_{2}B_{2}})  &  =M_{A\hat{A}\hat{B}%
_{1}B_{1}\hat{B}_{2}B_{2}},\label{eq:sim-two-ext-2EXT-constr}\\
\operatorname{Tr}_{B_{2}}[M_{A\hat{A}\hat{B}_{1}B_{1}\hat{B}_{2}B_{2}}]  &
=\frac{M_{A\hat{A}\hat{B}_{1}B_{1}}}{d_{\hat{B}}}\otimes I_{\hat{B}_{2}%
},\label{eq:sim-two-ext-constr-ext}\\
T_{A\hat{A}}(M_{A\hat{A}\hat{B}_{1}B_{1}\hat{B}_{2}B_{2}})  &  \geq0,\label{eq:sim-two-ext-constr-ppt1}\\
T_{\hat{B}_{2}B_{2}}(M_{A\hat{A}\hat{B}_{1}B_{1}\hat{B}_{2}B_{2}})  &  \geq0.
\label{eq:sim-two-ext-constr-ppt2}
\end{align}

\end{proposition}

The objective function in \eqref{eq:obj-func-gen-two-ext-sim} and the first
two constraints in \eqref{eq:general-SDP-two-ext-sim-2} and
\eqref{eq:general-SDP-two-ext-sim} follow from the semi-definite program in
\eqref{eq:diamond-d-SDP} for the normalized diamond distance. The quantity
$\operatorname{Tr}_{\hat{A}\hat{B}_{1}}\!\left[  T_{\hat{A}\hat{B}_{1}}%
(\rho_{\hat{A}\hat{B}_{1}})\frac{M_{A\hat{A}\hat{B}_{1}B_{1}}}{d_{\hat{B}}%
}\right]  $ in \eqref{eq:general-SDP-two-ext-sim} is the Choi operator
corresponding to the composition of the appending channel and the simulation
channel $\mathcal{K}_{A\hat{A}\hat{B}_{1}\rightarrow B_{1}}$, with Choi
operator $\frac{M_{A\hat{A}\hat{B}_{1}B_{1}}}{d_{\hat{B}}}$, where
$\mathcal{K}_{A\hat{A}\hat{B}_{1}\rightarrow B_{1}}$ is the marginal channel
of $\mathcal{M}_{A\hat{A}\hat{B}_{1}\hat{B}_{2}\rightarrow B_{1}B_{2}}$,
defined as%
\begin{multline}
\mathcal{K}_{A\hat{A}\hat{B}_{1}\rightarrow B_{1}}(\omega_{A\hat{A}\hat{B}%
_{1}})\\
\coloneqq\operatorname{Tr}_{B_{2}}[\mathcal{M}_{A\hat{A}\hat{B}_{1}\hat{B}%
_{2}\rightarrow B_{1}B_{2}}(\omega_{A\hat{A}\hat{B}_{1}}\otimes\pi_{\hat
{B}_{2}})].
\end{multline}
The constraint in \eqref{eq:sim-two-ext-TP} forces $\mathcal{M}_{A\hat{A}%
\hat{B}_{1}\hat{B}_{2}\rightarrow B_{1}B_{2}}$ to be trace preserving, that in
\eqref{eq:sim-two-ext-2EXT-constr} forces $\mathcal{M}_{A\hat{A}\hat{B}%
_{1}\hat{B}_{2}\rightarrow B_{1}B_{2}}$ to be permutation covariant with
respect to the $B$ systems (see \eqref{eq:perm-cov-sim-ch}), and that in
\eqref{eq:sim-two-ext-constr-ext} forces $\mathcal{M}_{A\hat{A}\hat{B}_{1}%
\hat{B}_{2}\rightarrow B_{1}B_{2}}$ to be the extension of a marginal channel
$\mathcal{K}_{A\hat{A}\hat{B}_{1}\rightarrow B_{1}}$. The final two
PPT\ constraints are equivalent to the C-PPT-P\ constraints in
\eqref{eq:C-PPT-P-constr-1-B2} and \eqref{eq:C-PPT-P-constr-2-A}, respectively.

\subsection{SDP lower bound on the simulation error of approximate teleportation}

The semi-definite program in Proposition~\ref{prop:two-ext-sim-err-gen-ch}%
\ can be evaluated for an important case of interest, i.e., when
$\mathcal{N}_{A\rightarrow B}=\operatorname{id}_{A\rightarrow B}^{d}$. Recall from Section~\ref{sec:quantify-approx-tele} that this special case corresponds to approximate teleportation. The
semi-definite program in Proposition~\ref{prop:two-ext-sim-err-gen-ch} is efficiently computable with respect to the dimensions
of the systems $A$, $\hat{A}$, $\hat{B}$, and $B$. However, it is in our
interest to reduce the computational complexity of these optimization tasks even further for
this important case, and we can do so by exploiting the unitary covariance symmetry of
the identity channel, as stated in \eqref{eq:identity-ch-symmetries}.

In this section, we provide a semi-definite program for evaluating the
simulation error%
\begin{equation}
e_{\operatorname{2PE}}(\rho_{\hat{A}\hat{B}})\equiv e_{\operatorname{2PE}%
}(\operatorname{id}_{A\rightarrow B}^{d},\rho_{\hat{A}\hat{B}}),
\end{equation}
with reduced complexity, i.e., only polynomial in the dimensions $d_{\hat{A}}$
and $d_{\hat{B}}$ of the resource state $\rho_{\hat{A}\hat{B}}$. We provide a
proof of Proposition~\ref{prop:simplified-SDP-two-ext-sim-id}\ in
Appendix~\ref{app:SDP-id-sim-simplify-sym}.

\begin{proposition}
\label{prop:simplified-SDP-two-ext-sim-id}The semi-definite program in
Proposition~\ref{prop:two-ext-sim-err-gen-ch}, for the special case of simulating the identity
channel $\operatorname{id}_{A\rightarrow B}^{d}$, simplifies as follows for $d\geq 3$:
\begin{align}
&  e_{\operatorname{2PE}}(\rho_{\hat{A}\hat{B}})\nonumber\\
&  =e_{\operatorname{2PE}}^{F}(\rho_{\hat{A}\hat{B}})\\
&  =1-\sup_{\substack{M^{+},M^{-},M^{0}\geq0,\\M^{1},M^{2},M^{3}\in\operatorname{LinOp}%
}}\operatorname{Tr}\!\left[  T_{\hat{A}\hat{B}_{1}}(\rho_{\hat{A}\hat{B}_{1}%
})\frac{P_{\hat{A}\hat{B}_{1}\hat{B}_{2}}}{d_{\hat{B}}}\right]  ,\label{eq:obj-funct-simple-sdp-tele}
\end{align}
subject to%
\begin{equation}%
\begin{bmatrix}
M^{0}+M^{3} & M^{1}-iM^{2}\\
M^{1}+iM^{2} & M^{0}-M^{3}%
\end{bmatrix}
\geq0,
\label{eq:M-matrix-constr-tele}
\end{equation}%
\begin{align}
I_{\hat{A}\hat{B}_{1}\hat{B}_{2}}  &  =M_{\hat{A}\hat{B}_{1}\hat{B}_{2}}%
^{+}+M_{\hat{A}\hat{B}_{1}\hat{B}_{2}}^{-}+M_{\hat{A}\hat{B}_{1}\hat{B}_{2}%
}^{0}, \label{eq:TP-constr-tele}\\
M_{\hat{A}\hat{B}_{1}\hat{B}_{2}}^{i}  &  =\mathcal{F}_{\hat{B}_{1}\hat{B}%
_{2}}(M_{\hat{A}\hat{B}_{1}\hat{B}_{2}}^{i})\quad\forall i\in\left\{
+,-,0,1\right\}  , \label{eq:perm-cov-constr-tele-1}\\
M_{\hat{A}\hat{B}_{1}\hat{B}_{2}}^{j}  &  =-\mathcal{F}_{\hat{B}_{1}\hat
{B}_{2}}(M_{\hat{A}\hat{B}_{1}\hat{B}_{2}}^{j})\quad\forall j\in\left\{
2,3\right\}  , \label{eq:perm-cov-constr-tele-2}\\
P_{\hat{A}\hat{B}_{1}\hat{B}_{2}}  &  =\frac{1}{d_{\hat{B}}}\operatorname{Tr}%
_{\hat{B}_{2}}[P_{\hat{A}\hat{B}_{1}\hat{B}_{2}}]\otimes I_{\hat{B}_{2}}, \label{eq:no-sig-cond-tele-simple-sdp}\\
P_{\hat{A}\hat{B}_{1}\hat{B}_{2}}  &  \coloneqq\frac{1}{2d}\left[
dM^{0}+M^{1}+\sqrt{d^{2}-1}\ M^{2}\right]  ,
\label{eq:P-op-no-sig-tele-simple-sdp}
\end{align}%
\begin{align}
T_{\hat{A}}\!\left(  \frac{2M_{\hat{A}\hat{B}_{1}\hat{B}_{2}}^{+}}%
{d+2}+M_{\hat{A}\hat{B}_{1}\hat{B}_{2}}^{0}+M_{\hat{A}\hat{B}_{1}\hat{B}_{2}%
}^{1}\right)   &  \geq0,\label{eq:PPTA-cond-1}\\
T_{\hat{A}}\!\left(  \frac{2M_{\hat{A}\hat{B}_{1}\hat{B}_{2}}^{-}}%
{d-2}+M_{\hat{A}\hat{B}_{1}\hat{B}_{2}}^{1}-M_{\hat{A}\hat{B}_{1}\hat{B}_{2}%
}^{0}\right)   &  \geq0,\label{eq:PPTA-cond-2}\\%
\begin{bmatrix}
G^{0}+G^{3} & G^{1}-iG^{2}\\
G^{1}+iG^{2} & G^{0}-G^{3}%
\end{bmatrix}
&  \geq0, \label{eq:PPTA-cond-3}%
\end{align}%
\begin{align}
&  G_{\hat{A}\hat{B}_{1}\hat{B}_{2}}^{0}\coloneqq T_{\hat{A}}\left(
M^{+}+M^{-}+\frac{M^{0}-dM^{1}}{2}\right)  ,\\
&  G_{\hat{A}\hat{B}_{1}\hat{B}_{2}}^{1}\coloneqq T_{\hat{A}}\left(
M^{+}-M^{-}+\frac{M^{1}-dM^{0}}{2}\right)  ,\\
&  G_{\hat{A}\hat{B}_{1}\hat{B}_{2}}^{2}\coloneqq\frac{\sqrt{3\left(
d^{2}-1\right)  }}{2}T_{\hat{A}}(M_{\hat{A}\hat{B}_{1}\hat{B}_{2}}^{2}),\\
&  G_{\hat{A}\hat{B}_{1}\hat{B}_{2}}^{3}\coloneqq\frac{\sqrt{3\left(
d^{2}-1\right)  }}{2}T_{\hat{A}}(M_{\hat{A}\hat{B}_{1}\hat{B}_{2}}^{3}).
\end{align}%
\begin{align}
T_{\hat{A}\hat{B}_{1}}\!\left(  \frac{dM^{+}}{d+2}+M^{-}+\frac{dM^{0}%
-M^{1}-\sqrt{d^{2}-1}\ M^{2}}{2}\right)   &  \geq0, \label{eq:PPTB-cond-1}\\
T_{\hat{A}\hat{B}_{1}}\!\left(  M^{+}+\frac{dM^{-}}{d-2}-\frac{dM^{0}%
-M^{1}-\sqrt{d^{2}-1}\ M^{2}}{2}\right)   &  \geq0, \label{eq:PPTB-cond-2}
\end{align}%
\begin{equation}%
\begin{bmatrix}
E^{0}+E^{3} & E^{1}-iE^{2}\\
E^{1}+iE^{2} & E^{0}-E^{3}%
\end{bmatrix}
\geq0, \label{eq:PPTB-cond-3}
\end{equation}%
\begin{align}
E_{\hat{A}\hat{B}_{1}\hat{B}_{2}}^{0}  &  \coloneqq\frac{T_{\hat{A}\hat{B}%
_{1}}\!\left(  d(M^{+}-M^{-})+\frac{L^{0}}{2}\right)  }{d^{2}-1},\\
E_{\hat{A}\hat{B}_{1}\hat{B}_{2}}^{1}  &  \coloneqq\frac{T_{\hat{A}\hat{B}%
_{1}}\!\left(  -M^{+}+M^{-}+\frac{L^{1}}{2}\right)  }{d^{2}-1},\\
E_{\hat{A}\hat{B}_{1}\hat{B}_{2}}^{2}  &  \coloneqq\frac{T_{\hat{A}\hat{B}%
_{1}}\!(M^{+}-M^{-}+\frac{L^{2}}{2})}{\sqrt{d^{2}-1}},\\
E_{\hat{A}\hat{B}_{1}\hat{B}_{2}}^{3}  &  \coloneqq T_{\hat{A}\hat{B}_{1}%
}(M_{\hat{A}\hat{B}_{1}\hat{B}_{2}}^{3}),\\
L^{0}  &  \coloneqq\left(  d^{2}-2\right)  M^{0}+dM^{1}+d\sqrt{d^{2}-1}\ 
M^{2},\\
L^{1}  &  \coloneqq dM^{0}+\left(  2d^{2}-3\right)  M^{1}-\sqrt{d^{2}-1}\ 
M^{2},\\
L^{2}  &  \coloneqq M^{1}-dM^{0}-\sqrt{d^{2}-1}\ M^{2}.
\label{eq:PPT-cond-final-tele}
\end{align}
For the case of $d=2$, the SDP is the same, with the exception that we set $M^-_{\hat{A}\hat{B}_1\hat{B}_2} = 0$ and the constraints in \eqref{eq:PPTA-cond-2} and \eqref{eq:PPTB-cond-2} are not used.
\end{proposition}

\begin{proof}
See Appendix~\ref{app:SDP-id-sim-simplify-sym}.
\end{proof}

\begin{remark}
\label{rem:SDP-explain}
    The SDP in the statement of Proposition~\ref{prop:simplified-SDP-two-ext-sim-id} is rather lengthy, and so we provide some explanation here. The constraint in \eqref{eq:M-matrix-constr-tele} and the constraints $M^+, M^-, M^0 \geq 0 $ in \eqref{eq:obj-funct-simple-sdp-tele} correspond to the constraint of complete positivity in \eqref{eq:obj-func-gen-two-ext-sim} (i.e., $M_{A\hat{A}\hat{B}_{1}B_{1}\hat
{B}_{2}B_{2}}\geq 0$). The constraint in \eqref{eq:TP-constr-tele} corresponds to the constraint of trace preservation in \eqref{eq:sim-two-ext-TP}. The constraints in \eqref{eq:perm-cov-constr-tele-1}--\eqref{eq:perm-cov-constr-tele-2} correspond to the permutation covariance constraint in \eqref{eq:sim-two-ext-2EXT-constr}. The constraint in \eqref{eq:no-sig-cond-tele-simple-sdp} corresponds to the non-signaling constraint in \eqref{eq:sim-two-ext-constr-ext}. The constraints in \eqref{eq:PPTA-cond-1}--\eqref{eq:PPTA-cond-3} correspond to the PPT constraint in \eqref{eq:sim-two-ext-constr-ppt1}, and the constraints in \eqref{eq:PPTB-cond-1}--\eqref{eq:PPTB-cond-3} correspond to the PPT constraint in \eqref{eq:sim-two-ext-constr-ppt2}.
\end{remark}

\begin{remark}
\label{rem:SDP-many-constr}
Even though the number of constraints in the SDP above appears to increase when compared with the SDP from Proposition~\ref{prop:two-ext-sim-err-gen-ch}, we note that the runtime of the SDP above is significantly reduced because the size of the matrices involved in each of the constraints is much smaller. This is the main advantage conferred by incorporating unitary covariance symmetry of the identity channel.

If we only optimized over the larger set of two-extendible chanels instead of the set of two-PPT-extendible channels, the SDP would be much simpler, given by \eqref{eq:obj-funct-simple-sdp-tele}--\eqref{eq:P-op-no-sig-tele-simple-sdp}. However, optimizing over the smaller set of two-PPT-extendible channels gives tighter bounds at a marginal increase in computational cost, and thus we also include the PPT constraints in \eqref{eq:PPTA-cond-1}--\eqref{eq:PPTA-cond-3} and \eqref{eq:PPTB-cond-1}--\eqref{eq:PPTB-cond-3}.
\end{remark}

\section{Background on superchannels}

\label{sec:superchs}

This section constitutes the beginning of the second contribution of our paper, regarding lower bounds on the error in channel simulation and approximate quantum error correction. We begin by reviewing the theory of superchannels, as well as particular examples of them relevant to the aforementioned applications.

\subsection{Basics of superchannels}

A superchannel $\Theta\equiv\Theta_{\left(  A\rightarrow B\right)
\rightarrow(C\rightarrow D)}$ is a physical transformation of a channel
$\mathcal{N}_{A\rightarrow B}$ that accepts as input the channel
$\mathcal{N}_{A\rightarrow B}$ and outputs a channel with input system $C$ and
output system $D$. Mathematically, a superchannel is a linear map that
preserves the set of quantum channels, even when the quantum channel is an
arbitrary bipartite channel with external input and output systems that are
arbitrarily large. Superchannels are thus completely CPTP preserving in this
sense. A general theory of superchannels was introduced in \cite{CDP08}\ and
developed further in \cite{CDP09,CDP08a,Gour18}.

In more detail, let us denote the output of a superchannel $\Theta$ by
$\mathcal{K}_{C\rightarrow D}$, so that%
\begin{equation}
\Theta_{\left(  A\rightarrow B\right)  \rightarrow(C\rightarrow D)}%
(\mathcal{N}_{A\rightarrow B})=\mathcal{K}_{C\rightarrow D}.
\label{eq:superch-in-out}%
\end{equation}
The superchannel $\Theta_{\left(  A\rightarrow B\right)  \rightarrow
(C\rightarrow D)}$ is completely CPTP\ preserving in the sense that the
following output map%
\begin{equation}
(\operatorname{id}_{\left(  R\right)  \rightarrow\left(  R\right)  }%
\otimes\Theta_{\left(  A\rightarrow B\right)  \rightarrow(C\rightarrow
D)})(\mathcal{M}_{RA\rightarrow RB})
\end{equation}
is a quantum channel for every input quantum channel $\mathcal{M}%
_{RA\rightarrow RB}$, where $\operatorname{id}_{\left(  R\right)
\rightarrow\left(  R\right)  }$ denotes the identity superchannel~\cite{CDP08}.

The fundamental theorem of superchannels from \cite{CDP08} is that
$\Theta_{\left(  A\rightarrow B\right)  \rightarrow(C\rightarrow D)}$ has a
physical realization in terms of a pre-processing channel $\mathcal{E}%
_{C\rightarrow AQ}$ and a post-processing channel $\mathcal{D}_{BQ\rightarrow
D}$ as follows:%
\begin{multline}
\Theta_{\left(  A\rightarrow B\right)  \rightarrow(C\rightarrow D)}%
(\mathcal{N}_{A\rightarrow B})\label{eq:superch-decomp-pre-post}\\
=\mathcal{D}_{BQ\rightarrow D}\circ\mathcal{N}_{A\rightarrow B}\circ
\mathcal{E}_{C\rightarrow AQ},
\end{multline}
where $Q$ is a quantum memory system. Furthermore, every superchannel
$\Theta_{\left(  A\rightarrow B\right)  \rightarrow(C\rightarrow D)}$ is in
one-to-one correspondence with a bipartite channel of the following form:%
\begin{equation}
\mathcal{P}_{CB\rightarrow AD}\coloneqq\mathcal{D}_{BQ\rightarrow D}%
\circ\mathcal{E}_{C\rightarrow AQ}. \label{eq:superch-bipart-factorization}%
\end{equation}
Note that $\mathcal{P}_{CB\rightarrow AD}$ is completely positive, trace
preserving, and obeys the following non-signaling constraint:%
\begin{equation}
\operatorname{Tr}_{D}\circ\mathcal{P}_{CB\rightarrow AD}=\operatorname{Tr}%
_{D}\circ\mathcal{P}_{CB\rightarrow AD}\circ\mathcal{R}_{B}^{\pi},
\label{eq:non-sig-bipartite-superch}%
\end{equation}
where the replacer channel $\mathcal{R}_{B}^{\pi}$ is defined in
\eqref{eq:replacer-ch-def}. Related to this, $\Gamma_{CBAD}^{\mathcal{P}}$ is
the Choi operator of a superchannel if and only if it satisfies the following
constraints:%
\begin{align}
\Gamma_{CBAD}^{\mathcal{P}}  &  \geq0,\\
\operatorname{Tr}_{AD}[\Gamma_{CBAD}^{\mathcal{P}}]  &  =I_{CB},\\
\operatorname{Tr}_{D}[\Gamma_{CBAD}^{\mathcal{P}}]  &  =\frac{1}{d_{B}%
}\operatorname{Tr}_{BD}[\Gamma_{CBAD}^{\mathcal{P}}]\otimes I_{B}.
\end{align}
The first two constraints correspond to complete positivity and trace
preservation, respectively, and the last constraint is a non-signaling constraint corresponding to
$\mathcal{P}_{CB\rightarrow AD}$ having the factorization in
\eqref{eq:superch-bipart-factorization}, so that $\mathcal{P}_{CB\rightarrow
AD}$ is in correspondence with a superchannel. To determine the Choi operator
for the output channel $\mathcal{K}_{C\rightarrow D}$ in
\eqref{eq:superch-in-out}, we can use the following propagation rule
\cite{CDP08,Gour18}:%
\begin{equation}\label{eq:propagation rule}
\Gamma_{CD}^{\mathcal{K}}=\operatorname{Tr}_{AB}[T_{AB}(\Gamma_{AB}%
^{\mathcal{N}})\Gamma_{CBAD}^{\mathcal{P}}],
\end{equation}
where $\Gamma_{CBAD}^{\mathcal{P}}$ is the Choi operator of $\mathcal{P}%
_{CB\rightarrow AD}$ and $\Gamma_{AB}^{\mathcal{N}}$ is the Choi operator of
$\mathcal{N}_{A\rightarrow B}$.

\subsection{One-way LOCC superchannels}

A superchannel $\Lambda\equiv\Lambda_{\left(  A\rightarrow B\right)
\rightarrow\left(  C\rightarrow D\right)  }$\ is implementable by one-way LOCC
if it can be written in the following form:%
\begin{equation}
\Lambda(\mathcal{N}_{A\rightarrow B})\coloneqq\sum_{x}\mathcal{D}%
_{B\rightarrow D}^{x}\circ\mathcal{N}_{A\rightarrow B}\circ\mathcal{E}%
_{C\rightarrow A}^{x},
\label{eq:LOCC-superch-Lam}
\end{equation}
where $\{\mathcal{E}_{C\rightarrow A}^{x}\}_{x}$ is a set of completely
positive maps such that the sum map $\sum_{x}\mathcal{E}_{C\rightarrow A}^{x}$
is trace preserving and $\{\mathcal{D}_{B\rightarrow D}^{x}\}_{x}$ is a set of
quantum channels. This is equivalent to the quantum memory system $Q$\ in
\eqref{eq:superch-decomp-pre-post} being a classical system $X$, with%
\begin{align}
\mathcal{E}_{C\rightarrow AX}(\rho_{C})  &  \coloneqq\sum_{x}\mathcal{E}%
_{C\rightarrow A}^{x}(\rho_{C})\otimes|x\rangle\!\langle x|_{X},\\
\mathcal{D}_{BX\rightarrow D}(\omega_{BX})  &  \coloneqq\sum_{x}%
\mathcal{D}_{B\rightarrow D}^{x}(\langle x|_{X}\omega_{BX}|x\rangle_{X}),
\end{align}
so that%
\begin{equation}
\Lambda(\mathcal{N}_{A\rightarrow B})=\mathcal{D}_{BX\rightarrow D}%
\circ\mathcal{N}_{A\rightarrow B}\circ\mathcal{E}_{C\rightarrow AX}.
\end{equation}
In this case, the bipartite channel in
\eqref{eq:superch-bipart-factorization}, but corresponding to $\Lambda$ in \eqref{eq:LOCC-superch-Lam},
becomes the following one-way LOCC channel:%
\begin{equation}
\mathcal{L}_{CB\rightarrow AD}\coloneqq\sum_{x}\mathcal{E}_{C\rightarrow
A}^{x}\otimes\mathcal{D}_{B\rightarrow D}^{x}.
\end{equation}
Thus, the set of one-way LOCC superchannels is in direct correspondence with
the set of one-way LOCC bipartite channels.

\subsection{LOCR\ superchannels}\label{sec:LOCR superchannels}

A superchannel $\Upsilon\equiv\Upsilon_{\left(  A\rightarrow B\right)
\rightarrow\left(  C\rightarrow D\right)  }$ is implementable by local
operations and common randomness (LOCR)\ if it can be written in the following
form:%
\begin{equation}
\Upsilon(\mathcal{N}_{A\rightarrow B})\coloneqq\sum_{y}p(y)\mathcal{D}%
_{B\rightarrow D}^{y}\circ\mathcal{N}_{A\rightarrow B}\circ\mathcal{E}%
_{C\rightarrow A}^{y},
\label{eq:LOCR-superch-Ups}
\end{equation}
where $\{p(y)\}_{y}$ is a probability distribution and $\{\mathcal{E}%
_{C\rightarrow A}^{y}\}_{y}$ and $\{\mathcal{D}_{B\rightarrow D}^{y}\}_{y}$
are sets of quantum channels. In more detail, the superchannel $\Upsilon
_{\left(  A\rightarrow B\right)  \rightarrow\left(  C\rightarrow D\right)  }$
can be realized as%
\begin{equation}
\Upsilon(\mathcal{N}_{A\rightarrow B})=\mathcal{D}_{BY_{B}\rightarrow D}%
\circ\mathcal{N}_{A\rightarrow B}\circ\mathcal{E}_{CY_{A}\rightarrow A}%
\circ\mathcal{P}_{Y_{A}Y_{B}},
\end{equation}
where $\mathcal{P}_{Y_{A}Y_{B}}$ is a preparation channel that prepares the
common randomness state%
\begin{equation}
\sum_{y}p(y)|y\rangle\!\langle y|_{Y_{A}}\otimes|y\rangle\!\langle y|_{Y_{B}},
\end{equation}
and the channels $\mathcal{E}_{CY_{A}\rightarrow A}$ and $\mathcal{D}%
_{BY_{B}\rightarrow D}$ are defined as%
\begin{align}
\mathcal{E}_{CY_{A}\rightarrow A}(\rho_{CY_{A}})  &  \coloneqq\sum
_{y}\mathcal{E}_{C\rightarrow A}^{y}(\langle y|_{Y_{A}}\rho_{CY_{A}}%
|y\rangle_{Y_{A}}),\\
\mathcal{D}_{BY_{B}\rightarrow D}(\omega_{BY_{B}})  &  \coloneqq\sum
_{y}\mathcal{D}_{B\rightarrow D}^{y}(\langle y|_{Y_{B}}\omega_{BY_{B}%
}|y\rangle_{Y_{B}}),
\end{align}
In this case, the bipartite channel in
\eqref{eq:superch-bipart-factorization}, but corresponding to $\Upsilon$ in \eqref{eq:LOCR-superch-Ups},
becomes the following LOCR bipartite channel:%
\begin{equation}\label{eq:LOCR_superchannel_bipartite_rep}
\mathcal{C}_{CB\rightarrow AD}\coloneqq\sum_{y}p(y)\mathcal{E}_{C\rightarrow
A}^{y}\otimes\mathcal{D}_{B\rightarrow D}^{y}.
\end{equation}
Thus, the set of LOCR\ superchannels is in direct correspondence with the set
of LOCR\ bipartite channels.

\subsection{Two-extendible superchannels}

A superchannel $\Theta_{\left(  A\rightarrow B\right)  \rightarrow\left(
C\rightarrow D\right)  }$ is defined to be two-extendible if there exists an
extension channel $\mathcal{M}_{CB_{1}B_{2}\rightarrow AD_{1}D_{2}}$ of its
corresponding bipartite channel $\mathcal{P}_{CB\rightarrow AD}$\ that obeys
the conditions in \eqref{eq:perm-cov-ch} and \eqref{eq:marg-ch-ch}.
Furthermore, due to the fact that \eqref{eq:perm-cov-ch} and
\eqref{eq:marg-ch-ch} imply \eqref{eq:non-sig-implied-by-two-ext}, there is no
need to explicitly indicate that \eqref{eq:non-sig-bipartite-superch} holds.
Two-extendible superchannels were considered in \cite{BBFS21}, but this
terminology was not employed there.

The specific constraints on the Choi operator of $\mathcal{M}_{CB_{1}%
B_{2}\rightarrow AD_{1}D_{2}}$ are precisely the same as those in
\eqref{eq:perm-cov-choi}--\eqref{eq:ext-ch-TP}, with the identifications
$C\leftrightarrow A$, $B\leftrightarrow B$, $A\leftrightarrow A^{\prime}$, and
$D\leftrightarrow B^{\prime}$. Explicitly, a superchannel $\Theta_{\left(
A\rightarrow B\right)  \rightarrow\left(  C\rightarrow D\right)  }$ is
two-extendible if the Choi operator $\Gamma_{CBAD}^{\mathcal{P}}$ of its
corresponding bipartite channel $\mathcal{P}_{CB\rightarrow AD}$ satisfies the
following conditions: there exists a Hermitian operator $\Gamma_{CB_{1}%
B_{2}AD_{1}D_{2}}^{\mathcal{M}}$ such that%
\begin{align}
(\mathcal{F}_{B_{1}B_{2}}\otimes\mathcal{F}_{D_{1}D_{2}})(\Gamma_{CB_{1}%
B_{2}AD_{1}D_{2}}^{\mathcal{M}})  &  =\Gamma_{CB_{1}B_{2}AD_{1}D_{2}%
}^{\mathcal{M}},\label{eq:two-ext-superch-1}\\
\operatorname{Tr}_{D_{2}}[\Gamma_{CB_{1}B_{2}AD_{1}D_{2}}^{\mathcal{M}}]  &
=\Gamma_{CB_{1}AD_{1}}^{\mathcal{P}}\otimes I_{B_{2}},\\
\Gamma_{CB_{1}B_{2}AD_{1}D_{2}}^{\mathcal{M}}  &  \geq0,\\
\operatorname{Tr}_{AD_{1}D_{2}}[\Gamma_{CB_{1}B_{2}AD_{1}D_{2}}^{\mathcal{M}%
}]  &  =I_{CB_{1}B_{2}}. \label{eq:two-ext-superch-4}%
\end{align}
Every one-way LOCC superchannel is two-extendible.

\subsection{Completely PPT\ preserving superchannels}

A superchannel $\Theta_{\left(  A\rightarrow B\right)  \rightarrow\left(
C\rightarrow D\right)  }$\ is C-PPT-P if its corresponding bipartite channel
$\mathcal{P}_{CB\rightarrow AD}$\ in
\eqref{eq:superch-bipart-factorization}\ is C-PPT-P and obeys the
non-signaling constraint in \eqref{eq:non-sig-bipartite-superch} \cite{LM15}.
This implies the following for its Choi operator $\Gamma_{CBAD}^{\mathcal{P}}%
$:%
\begin{align}
\Gamma_{CBAD}^{\mathcal{P}}  &  \geq0,\\
\operatorname{Tr}_{AD}[\Gamma_{CBAD}^{\mathcal{P}}]  &  =I_{CB},\\
\operatorname{Tr}_{D}[\Gamma_{CBAD}^{\mathcal{P}}]  &  =\frac{1}{d_{B}%
}\operatorname{Tr}_{BD}[\Gamma_{CBAD}^{\mathcal{P}}]\otimes I_{B},\\
T_{BD}(\Gamma_{CBAD}^{\mathcal{P}})  &  \geq0.
\end{align}

\subsection{Two-PPT-extendible superchannels}

A superchannel $\Theta_{\left(  A\rightarrow B\right)  \rightarrow\left(
C\rightarrow D\right)  }$\ is two-PPT-extendible if its corresponding
bipartite channel $\mathcal{P}_{CB\rightarrow AD}$\ in
\eqref{eq:superch-bipart-factorization}\ is two-PPT-extendible. Again, there
is no need to explicitly indicate that \eqref{eq:non-sig-bipartite-superch}
holds. The following conditions hold for the Choi operator $\Gamma
_{CBAD}^{\mathcal{P}}$ of a two-PPT-extendible superchannel: there exists a
Hermitian operator $\Gamma_{CB_{1}B_{2}AD_{1}D_{2}}^{\mathcal{M}}$ such that
\eqref{eq:two-ext-superch-1}--\eqref{eq:two-ext-superch-4} hold, as well as%
\begin{align*}
T_{B_{2}D_{2}}(\Gamma_{CB_{1}B_{2}AD_{1}D_{2}}^{\mathcal{M}})  &  \geq0,\\
T_{CA}(\Gamma_{CB_{1}B_{2}AD_{1}D_{2}}^{\mathcal{M}})  &  \geq0.
\end{align*}
Similar to what was already discussed in Section~\ref{sec:2pk-ch}, the
following constraints are redundant:%
\begin{align}
T_{B_{1}D_{1}}(\Gamma_{CB_{1}B_{2}AD_{1}D_{2}}^{\mathcal{M}})  &  \geq0,\\
T_{CAB_{2}D_{2}}(\Gamma_{CB_{1}B_{2}AD_{1}D_{2}}^{\mathcal{M}})  &  \geq0,\\
T_{CAB_{1}D_{1}}(\Gamma_{CB_{1}B_{2}AD_{1}D_{2}}^{\mathcal{M}})  &  \geq0,\\
T_{B_{1}D_{1}B_{2}D_{2}}(\Gamma_{CB_{1}B_{2}AD_{1}D_{2}}^{\mathcal{M}})  &
\geq0.
\end{align}

Note that every one-way LOCC superchannel is two-PPT-extendible.

\subsection{Two-PPT-extendible non-signaling superchannels}

\label{sec_2-ppt-ext-non-sig-superchannels}

Finally, we can impose an additional non-signaling constraint on
two-PPT-extendible superchannels, such that the extension of its corresponding
bipartite channel is non-signaling from Alice to both Bobs. The additional
constraint on the Choi operator $\Gamma_{CB_{1}B_{2}AD_{1}D_{2}}^{\mathcal{M}%
}$ of the extension channel $\mathcal{M}_{CB_{1}B_{2}\rightarrow AD_{1}D_{2}}$
is as follows:%
\begin{equation}
\operatorname{Tr}_{A}[\Gamma_{CB_{1}B_{2}AD_{1}D_{2}}^{\mathcal{M}}]=\frac
{1}{d_{C}}\operatorname{Tr}_{AC}[\Gamma_{CB_{1}B_{2}AD_{1}D_{2}}^{\mathcal{M}%
}]\otimes I_{C}.
\end{equation}

Every LOCR superchannel is non-signaling and two-PPT-extendible, which follows from definitions and the form of the corresponding bipartite channel in \eqref{eq:LOCR_superchannel_bipartite_rep}. This fact
plays an important role in our analysis of approximate quantum error
correction. In more detail, we obtain our tightest lower bound on the simulation error of
approximate quantum error correction by relaxing the set of
LOCR\ superchannels to the set of non-signaling and two-PPT-extendible
superchannels. We note here that this approach was already considered in
\cite{BBFS21}, and our main contribution here is to employ unitary covariance symmetry of
the identity channel to reduce the complexity of the SDPs from that work.

\section{Quantifying the performance of approximate quantum error correction} \label{sec: VI}

\subsection{Quantifying simulation error with normalized diamond distance and
channel infidelity}

In approximate quantum error correction \cite{qip2002schu} or quantum
communication \cite{BDSW96}, the resource available is a quantum channel
$\mathcal{N}_{\hat{A}\rightarrow\hat{B}}$ and the goal is to use it, along
with an encoding channel $\mathcal{E}_{A\rightarrow\hat{A}}$\ and a decoding
channel $\mathcal{D}_{\hat{B}\rightarrow B}$, to simulate a $d$-dimensional
identity channel $\operatorname{id}_{A\rightarrow B}^{d}$. We can use the
normalized diamond distance to quantify the error for a fixed encoding and
decoding, as follows:%
\begin{multline}
e(\mathcal{N}_{\hat{A}\rightarrow\hat{B}},(\mathcal{E}_{A\rightarrow\hat{A}%
},\mathcal{D}_{\hat{B}\rightarrow B}))\coloneqq\\
\frac{1}{2}\left\Vert \operatorname{id}_{A\rightarrow B}^{d}-\mathcal{D}%
_{\hat{B}\rightarrow B}\circ\mathcal{N}_{\hat{A}\rightarrow\hat{B}}%
\circ\mathcal{E}_{A\rightarrow\hat{A}}\right\Vert _{\diamond}.
\end{multline}
By minimizing over all encodings and decodings, we arrive at the error in
using the channel $\mathcal{N}_{\hat{A}\rightarrow\hat{B}}$ to simulate the
identity channel:%
\begin{equation}
e(\mathcal{N}_{\hat{A}\rightarrow\hat{B}})\coloneqq\inf_{(\mathcal{E}%
,\mathcal{D})}e(\mathcal{N}_{\hat{A}\rightarrow\hat{B}},(\mathcal{E}%
_{A\rightarrow\hat{A}},\mathcal{D}_{\hat{B}\rightarrow B})).
\end{equation}
We can alternatively employ channel infidelity to quantify the error:%
\begin{multline}
e^{F}(\mathcal{N}_{\hat{A}\rightarrow\hat{B}},(\mathcal{E}_{A\rightarrow
\hat{A}},\mathcal{D}_{\hat{B}\rightarrow B}))\coloneqq\\
1-F(\operatorname{id}_{A\rightarrow B}^{d},\mathcal{D}_{\hat{B}\rightarrow
B}\circ\mathcal{N}_{\hat{A}\rightarrow\hat{B}}\circ\mathcal{E}_{A\rightarrow
\hat{A}}),
\end{multline}%
\begin{equation}
e^{F}(\mathcal{N}_{\hat{A}\rightarrow\hat{B}})\coloneqq\inf_{(\mathcal{E}%
,\mathcal{D})}e^{F}(\mathcal{N}_{\hat{A}\rightarrow\hat{B}},(\mathcal{E}%
_{A\rightarrow\hat{A}},\mathcal{D}_{\hat{B}\rightarrow B})).
\end{equation}
Note that the transformation of the channel given by%
\begin{equation}
\mathcal{D}_{\hat{B}\rightarrow B}\circ\mathcal{N}_{\hat{A}\rightarrow\hat{B}%
}\circ\mathcal{E}_{A\rightarrow\hat{A}} \label{eq:basic-trans-q-comm}%
\end{equation}
is a superchannel, as discussed in Section~\ref{sec:superchs}, with
corresponding bipartite channel%
\begin{equation}
\mathcal{P}_{A\hat{B}\rightarrow\hat{A}B}\coloneqq\mathcal{E}_{A\rightarrow
\hat{A}}\otimes\mathcal{D}_{\hat{B}\rightarrow B}.
\end{equation}
As this bipartite channel is a product channel, it is contained within the set
of LOCR superchannels, which in turn is contained in the set of one-way LOCC superchannels.

By supplementing the encoding and decoding with common randomness, the resulting error
correction scheme $\Upsilon\equiv\Upsilon_{(\hat{A}\rightarrow\hat
{B})\rightarrow(A\rightarrow B)}$ realizes the following simulation channel:%
\begin{equation}
\Upsilon(\mathcal{N}_{\hat{A}\rightarrow\hat{B}})\coloneqq\sum_{y}%
p(y)\mathcal{D}_{\hat{B}\rightarrow B}^{y}\circ\mathcal{N}_{\hat{A}%
\rightarrow\hat{B}}\circ\mathcal{E}_{A\rightarrow\hat{A}}^{y},
\label{eq:1WL-trans-channel-N}%
\end{equation}
where $\{p(y)\}_{y}$ is a probability distribution and $\{\mathcal{E}%
_{A\rightarrow\hat{A}}^{y}\}_{y}$ and $\{\mathcal{D}_{\hat{B}\rightarrow
B}^{y}\}_{y}$ are sets of quantum channels. Recall from
Section~\ref{sec:LOCR superchannels} that $\Upsilon$ is an LOCR superchannel, and let
LOCR\ denote the set of all LOCR\ superchannels. Then we can quantify the
simulation error under LOCR in a manner similar to Section~\ref{sec:quant-error-dia-dist}: we can use the normalized diamond
distance to quantify the error for a fixed LOCR\ superchannel$~\Upsilon$, as
follows:%
\begin{multline}
e_{\operatorname{LOCR}}(\mathcal{N}_{\hat{A}\rightarrow\hat{B}},\Upsilon
_{(\hat{A}\rightarrow\hat{B})\rightarrow(A\rightarrow B)})\coloneqq\\
\frac{1}{2}\left\Vert \operatorname{id}_{A\rightarrow B}^{d}-\Upsilon
_{(\hat{A}\rightarrow\hat{B})\rightarrow(A\rightarrow B)}(\mathcal{N}_{\hat
{A}\rightarrow\hat{B}})\right\Vert _{\diamond}.
\end{multline}
By minimizing over all such superchannels, we arrive at the error in using the
channel $\mathcal{N}_{\hat{A}\rightarrow\hat{B}}$ to simulate the identity
channel:%
\begin{multline}
e_{\operatorname{LOCR}}(\mathcal{N}_{\hat{A}\rightarrow\hat{B}}%
)\coloneqq\label{eq:diamond-err-sim-QEC}\\
\inf_{\Upsilon\in\operatorname{LOCR}}e(\mathcal{N}_{\hat{A}\rightarrow\hat{B}%
},\Upsilon_{(\hat{A}\rightarrow\hat{B})\rightarrow(A\rightarrow B)}).
\end{multline}
As before, we can alternatively employ channel infidelity to quantify the
error:%
\begin{multline}
e_{\operatorname{LOCR}}^{F}(\mathcal{N}_{\hat{A}\rightarrow\hat{B}}%
,\Upsilon_{(\hat{A}\rightarrow\hat{B})\rightarrow(A\rightarrow B)})\coloneqq\label{eq:fidelity-err-sim-QEC}\\
1-F(\operatorname{id}_{A\rightarrow B}^{d},\Upsilon_{(\hat{A}\rightarrow
\hat{B})\rightarrow(A\rightarrow B)}(\mathcal{N}_{\hat{A}\rightarrow\hat{B}%
})),
\end{multline}%
\begin{multline}
e_{\operatorname{LOCR}}^{F}(\mathcal{N}_{\hat{A}\rightarrow\hat{B}%
})\coloneqq\label{eq:fid-err-sim-QEC}\\
\inf_{\Upsilon\in\operatorname{LOCR}}e_{\operatorname{LOCR}}^{F}%
(\mathcal{N}_{\hat{A}\rightarrow\hat{B}},\Upsilon_{(\hat{A}\rightarrow\hat
{B})\rightarrow(A\rightarrow B)}).
\end{multline}

However, we have the following:

\begin{proposition}
\label{prop:equality-sim-errs-LOCR-QEC}For a channel $\mathcal{N}_{\hat
{A}\rightarrow\hat{B}}$, the LOCR\ simulation errors defined from normalized
diamond distance and channel infidelity are equal to each other:%
\begin{equation}
e_{\operatorname{LOCR}}(\mathcal{N}_{\hat{A}\rightarrow\hat{B}}%
)=e_{\operatorname{LOCR}}^{F}(\mathcal{N}_{\hat{A}\rightarrow\hat{B}}).
\end{equation}

\end{proposition}

\begin{proof}
The proof of this equality is similar to the proof of
Proposition~\ref{prop:sim-errs-equal-1WL}, following again from the symmetry
of the target channel, which is an identity channel having the symmetry in
\eqref{eq:identity-ch-symmetries}, and the fact that a channel twirl can be
implemented by means of LOCR. Note that a channel twirl of a channel $\mathcal{M}_{A\to B}$ has the following form:
\begin{equation}
 \int dU\ \mathcal{U}_{B}^{\dag}\circ\mathcal{M}_{A\rightarrow B}\circ\mathcal{U}_{A},
\end{equation}
where $\mathcal{U}$ is a unitary channel.
\end{proof}

By exploiting the fact that a superchannel of the form in
\eqref{eq:basic-trans-q-comm} is contained in the set of LOCR\ superchannels,
the following inequality holds%
\begin{equation}
e_{\operatorname{LOCR}}(\mathcal{N}_{\hat{A}\rightarrow\hat{B}})\leq
\min\left\{  e(\mathcal{N}_{\hat{A}\rightarrow\hat{B}}),e^{F}(\mathcal{N}%
_{\hat{A}\rightarrow\hat{B}})\right\}  .
\end{equation}
It is unclear if $e(\mathcal{N}_{\hat{A}\rightarrow\hat{B}})$ is equal to
$e^{F}(\mathcal{N}_{\hat{A}\rightarrow\hat{B}})$ in general:\ a critical
aspect of the proof of Proposition~\ref{prop:equality-sim-errs-LOCR-QEC} is
the fact that LOCR superchannels are allowed for free, so that the
symmetrizing twirling superchannel can be used. In the unassisted setting, we
cannot use twirling because it is an LOCR superchannel and thus not allowed
for free.

Recall again that the identity channel $\operatorname{id}_{A\rightarrow
B}^{d}$ possesses the unitary covariance symmetry in \eqref{eq:identity-ch-symmetries}. Considering this leads to
the following proposition:

\begin{proposition}
\label{prop:sim-errs-equal-1WL-QEC}The optimization problems in
\eqref{eq:diamond-err-sim-QEC} and \eqref{eq:fid-err-sim-QEC}, for the error
in simulating the identity channel $\operatorname{id}_{A\rightarrow B}^{d}$,
simplify as follows:%
\begin{align}
e_{\operatorname{LOCR}}(\mathcal{N}_{\hat{A}\rightarrow\hat{B}})  &
=e_{\operatorname{LOCR}}^{F}(\mathcal{N}_{\hat{A}\rightarrow\hat{B}%
}) \label{eq:dd_fid_err_equality_QEC}\\
&  =1-\sup_{\mathcal{P}}E_{F}(\mathcal{N}_{\hat{A}\rightarrow\hat{B}%
};\mathcal{P}), \label{eq:e-fid-opt}%
\end{align}
where the optimization in \eqref{eq:e-fid-opt}\ is over every LOCR\ protocol~$\mathcal{P}$, defined as%
\begin{equation}
\mathcal{P}\coloneqq \{(p(y),\mathcal{E}_{A\rightarrow\hat{A}}^{y}%
,\mathcal{D}_{\hat{B}\rightarrow B}^{y})\}_{y},
\end{equation}
and $E_{F}(\mathcal{N}_{\hat{A}\rightarrow\hat{B}};\mathcal{P})\in\left[
0,1\right]  $ is the entanglement fidelity:%
\begin{align}
E_{F}  &  \equiv E_{F}(\mathcal{N}_{\hat{A}\rightarrow\hat{B}};\mathcal{P})\\
&  \coloneqq \sum_{y}p(y)\operatorname{Tr}[\Phi_{AB}^{d}(\mathcal{D}_{\hat
{B}\rightarrow B}^{y}\circ\mathcal{N}_{\hat{A}\rightarrow\hat{B}}%
\circ\mathcal{E}_{A\rightarrow\hat{A}}^{y})(\Phi_{AB}^{d})].
\end{align}
An optimal LOCR simulation channel for both $e_{\operatorname{LOCR}%
}(\mathcal{N}_{\hat{A}\rightarrow\hat{B}})$ and
$e_{\operatorname{LOCR}}^{F}(\mathcal{N}_{\hat{A}\rightarrow\hat{B}%
})$ has the following form:
\begin{equation}
E_{F}\operatorname{id}_{A\rightarrow B}^{d}+\left(  1-E_{F}\right)
\mathcal{D}_{A\rightarrow B},
\label{eq:ent-fid-ch-spec-form}
\end{equation}
where $\mathcal{D}_{A\rightarrow B}$ is the channel defined in \eqref{eq:randomizing channel}. Thus,
the LOCR\ simulation channel applies the identity channel $\operatorname{id}%
_{A\rightarrow B}^{d}$ with probability $E_{F}$ and the randomizing channel
$\mathcal{D}_{A\rightarrow B}$ with probability $1-E_{F}$.
\end{proposition}

\begin{proof}
    See Appendix~\ref{sec:sim-errs-equal-1WL-QEC-proof}.
\end{proof}

\subsection{LOCR\ simulation of general point-to-point channels}

We can use a point-to-point channel $\mathcal{N}_{\hat{A}\rightarrow\hat{B}}$,
along with LOCR, to simulate another general point-to-point channel
$\mathcal{O}_{A\rightarrow B}$. In this case, the simulation channel
$\widetilde{\mathcal{O}}_{A\rightarrow B}$ has the form%
\begin{equation}\label{eq:O-tilde-definition}
\widetilde{\mathcal{O}}_{A\rightarrow B}\coloneqq \Upsilon_{(\hat
{A}\rightarrow\hat{B})\rightarrow(A\rightarrow B)}(\mathcal{N}_{\hat
{A}\rightarrow\hat{B}}),
\end{equation}
where $\Upsilon_{(\hat{A}\rightarrow\hat{B})\rightarrow(A\rightarrow B)}$ is
an LOCR superchannel, as discussed in Section~\ref{sec:LOCR superchannels}. The simulation error when
employing a specific LOCR\ superchannel $\Upsilon_{(\hat{A}\rightarrow\hat
{B})\rightarrow(A\rightarrow B)}$ is%
\begin{multline}
e_{\operatorname{LOCR}}(\mathcal{O}_{A\rightarrow B},\mathcal{N}%
_{\hat{A}\rightarrow\hat{B}},\Upsilon_{(\hat{A}\rightarrow\hat{B}%
)\rightarrow(A\rightarrow B)})\\
\coloneqq \frac{1}{2}\left\Vert \mathcal{O}_{A\rightarrow B}-\widetilde
{\mathcal{O}}_{A\rightarrow B}\right\Vert _{\diamond},
\end{multline}
and the simulation error minimized over all possible LOCR\ superchannels is%
\begin{multline}\label{eq:QEC_sim_error_optimization}
e_{\operatorname{LOCR}}(\mathcal{O}_{A\rightarrow B},\mathcal{N}%
_{\hat{A}\rightarrow\hat{B}})\coloneqq\\
\inf_{\Upsilon\in\operatorname{LOCR}}e_{\text{$\operatorname*{LOCR}%
$}}(\mathcal{O}_{A\rightarrow B},\mathcal{N}_{\hat{A}\rightarrow\hat{B}%
},\Upsilon_{(\hat{A}\rightarrow\hat{B})\rightarrow(A\rightarrow B)}).
\end{multline}

Again we can alternatively consider quantifying simulation error in terms of
the channel infidelity:%
\begin{multline}
e_{\operatorname{LOCR}}^{F}(\mathcal{O}_{A\rightarrow B}%
,\mathcal{N}_{\hat{A}\rightarrow\hat{B}},\Upsilon_{(\hat{A}\rightarrow\hat
{B})\rightarrow(A\rightarrow B)})\\
\coloneqq 1-F(\mathcal{O}_{A\rightarrow B},\widetilde{\mathcal{O}%
}_{A\rightarrow B}),
\end{multline}%
\begin{multline}
e_{\operatorname{LOCR}}^{F}(\mathcal{O}_{A\rightarrow B}%
,\mathcal{N}_{\hat{A}\rightarrow\hat{B}})\coloneqq\\
\inf_{\Upsilon\in\operatorname{LOCR}}e_{\text{$\operatorname*{LOCR}%
$}}^{F}(\mathcal{O}_{A\rightarrow B},\mathcal{N}_{\hat{A}\rightarrow\hat{B}%
},\Upsilon_{(\hat{A}\rightarrow\hat{B})\rightarrow(A\rightarrow B)}).
\end{multline}

\section{SDP\ lower bounds on the performance of approximate quantum error
correction based on two-PPT extendibility and non-signaling constraints}

\label{sec: VII}

\subsection{SDP\ lower bound on the error in LOCR\ simulation of a channel}

Using \eqref{eq:QEC_sim_error_optimization} to calculate the simulation error, we again encounter an intractable optimization task. Employing the same idea from Section~\ref{sec_SDP_lower_bound_tel}, we enlarge the set of LOCR superchannels to two-PPT-extendible, non-signaling superchannels (abbreviated henceforth as 2PENS). As noted in Section~\ref{sec_2-ppt-ext-non-sig-superchannels}, the 2PENS set  strictly contains the set of LOCR superchannels. Thus, we can obtain a lower bound on the simulation error by optimizing over all 2PENS superchannels. We define the simulation error under 2PENS superchannels as
\begin{multline}\label{eq:QEC_sim_error_2PE}
e_{\operatorname{2PENS}}(\mathcal{O}_{A\rightarrow B},\mathcal{N}%
_{\hat{A}\rightarrow\hat{B}})\coloneqq
\inf_{\Upsilon\in\operatorname{2PENS}}\frac{1}{2}\left\Vert \mathcal{O}_{A\rightarrow B}-\widetilde
{\mathcal{O}}_{A\rightarrow B}\right\Vert _{\diamond},
\end{multline}
where $\widetilde{\mathcal{O}}_{A\rightarrow B}$ is defined in \eqref{eq:O-tilde-definition}.

As a result of the strict containment
\begin{equation}
    \operatorname{LOCR} \subset \operatorname{2PENS},
\end{equation}
we have the relation
\begin{equation}
    e_{\operatorname{2PENS}}(\mathcal{O}_{A\rightarrow B},\mathcal{N}%
    _{\hat{A}\rightarrow\hat{B}}) \le e_{\operatorname{LOCR}}(\mathcal{O}_{A\rightarrow B},\mathcal{N}%
    _{\hat{A}\rightarrow\hat{B}}).
\end{equation}

We now state that the simulation error in \eqref{eq:QEC_sim_error_2PE} can be calculated by means of a semi-definite program.
\begin{proposition}
\label{prop:two-ext-sim-err-gen-ch-ch}The simulation error in \eqref{eq:QEC_sim_error_2PE} can be
calculated by means of the following semi-definite program:%
\begin{equation}
e_{\operatorname{2PENS}}(\mathcal{O}_{A\rightarrow B},\mathcal{N}_{\hat
{A}\rightarrow\hat{B}})=\inf_{\substack{\mu\geq0,Z_{AB}\geq0,\\ M_{A\hat{A}%
\hat{B}_{1}B_{1}\hat{B}_{2}B_{2}}\geq0}}\mu,
\label{eq:obj-func-qec-gen}
\end{equation}
subject to%
\begin{align}
\mu I_{A}  &  \geq Z_{A},\\
Z_{AB}  &  \geq\Gamma_{AB}^{\mathcal{O}}-\operatorname{Tr}_{\hat{A}\hat{B}%
_{1}}[T_{\hat{A}\hat{B}_{1}}(\Gamma_{\hat{A}\hat{B}_{1}}^{\mathcal{N}%
})M_{A\hat{A}\hat{B}_{1}B_{1}}/d_{\hat{B}}],
\label{eq:general-SDP-two-ext-superch-sim}%
\end{align}%
\begin{align}
\operatorname{Tr}_{\hat{A}B_{1}B_{2}}[M_{A\hat{A}\hat{B}_{1}B_{1}\hat{B}%
_{2}B_{2}}]  &  =I_{A\hat{B}_{1}\hat{B}_{2}},
\label{eq:sim-two-ext-superch-TP}\\
(\mathcal{F}_{\hat{B}_{1}\hat{B}_{2}}\otimes\mathcal{F}_{B_{1}B_{2}}%
)(M_{A\hat{A}\hat{B}_{1}B_{1}\hat{B}_{2}B_{2}})  &  =M_{A\hat{A}\hat{B}%
_{1}B_{1}\hat{B}_{2}B_{2}},\label{eq:sim-two-ext-superch-2EXT-constr}\\
\operatorname{Tr}_{B_{2}}[M_{A\hat{A}\hat{B}_{1}B_{1}\hat{B}_{2}B_{2}}]  &
=\frac{M_{A\hat{A}\hat{B}_{1}B_{1}}}{d_{\hat{B}}}\otimes I_{\hat{B}_{2}%
},\label{eq:sim-two-ext-superch-constr-ext}\\
T_{A\hat{A}}(M_{A\hat{A}\hat{B}_{1}B_{1}\hat{B}_{2}B_{2}})  &  \geq0,
\label{eq:PPTA-gen-ch-sim}\\
T_{\hat{B}_{2}B_{2}}(M_{A\hat{A}\hat{B}_{1}B_{1}\hat{B}_{2}B_{2}})  &  \geq0,
\label{eq:PPTB-gen-ch-sim}
\end{align}%
\begin{equation}
\operatorname{Tr}_{\hat{A}}[M_{A\hat{A}\hat{B}_{1}B_{1}\hat{B}_{2}B_{2}%
}]=I_{A}\otimes\frac{1}{d_{A}}\operatorname{Tr}_{\hat{A}A}[M_{A\hat{A}\hat
{B}_{1}B_{1}\hat{B}_{2}B_{2}}].
\label{eq:non-sig-constr-gen-ch-sim}
\end{equation}

\end{proposition}

The objective function and the first two constraints follow from the
semi-definite program in \eqref{eq:diamond-d-SDP} for the normalized diamond
distance. The quantity
\begin{equation}
\operatorname{Tr}_{\hat{A}\hat{B}_{1}}[T_{\hat{A}%
\hat{B}_{1}}(\Gamma_{\hat{A}\hat{B}_{1}}^{\mathcal{N}})M_{A\hat{A}\hat{B}%
_{1}B_{1}}/d_{\hat{B}}]
\end{equation}
in \eqref{eq:general-SDP-two-ext-superch-sim}
is the Choi operator of the serial  composition of the available channel $\mathcal{N}_{\hat
{A}\rightarrow\hat{B}}$\ and the superchannel with corresponding bipartite
channel $\mathcal{K}_{A\hat{B}_{1}\rightarrow\hat{A}B_{1}}$, with Choi
operator $M_{A\hat{A}\hat{B}_{1}B_{1}}/d_{\hat{B}}$, where $\mathcal{K}%
_{A\hat{B}_{1}\rightarrow\hat{A}B_{1}}$ is the marginal channel of
$\mathcal{M}_{A\hat{B}_{1}\hat{B}_{2}\rightarrow\hat{A}B_{1}B_{2}}$, defined
as%
\begin{multline}
\mathcal{K}_{A\hat{B}_{1}\rightarrow\hat{A}B_{1}}(\omega_{A\hat{B}_{1}})\\
\coloneqq\operatorname{Tr}_{B_{2}}[\mathcal{M}_{A\hat{B}_{1}\hat{B}%
_{2}\rightarrow\hat{A}B_{1}B_{2}}(\omega_{A\hat{B}_{1}}\otimes\pi_{\hat{B}%
_{2}})].
\end{multline}
The constraint in \eqref{eq:sim-two-ext-superch-TP} forces $\mathcal{M}%
_{A\hat{B}_{1}\hat{B}_{2}\rightarrow\hat{A}B_{1}B_{2}}$ to be trace
preserving, that in \eqref{eq:sim-two-ext-superch-2EXT-constr} forces
$\mathcal{M}_{A\hat{B}_{1}\hat{B}_{2}\rightarrow\hat{A}B_{1}B_{2}}$ to be
permutation covariant with respect to the $B$ systems (see
\eqref{eq:perm-cov-sim-ch}), and that in
\eqref{eq:sim-two-ext-superch-constr-ext} forces $\mathcal{M}_{A\hat{B}%
_{1}\hat{B}_{2}\rightarrow\hat{A}B_{1}B_{2}}$ to be the extension of the
marginal channel $\mathcal{K}_{A\hat{B}_{1}\rightarrow\hat{A}B_{1}}$. The
final two PPT\ constraints are equivalent to the C-PPT-P\ constraints in
\eqref{eq:C-PPT-P-constr-1-B2} and \eqref{eq:C-PPT-P-constr-2-A}, respectively.

\subsection{SDP\ lower bound on the error of approximate quantum error
correction}

The semi-definite program in Proposition~\ref{prop:two-ext-sim-err-gen-ch-ch}%
\ can be simplified for the special case
$\mathcal{N}_{A\rightarrow B}=\operatorname{id}_{A\rightarrow B}^{d}$ by  exploiting the unitary covariance symmetry of
the identity channel, as stated in \eqref{eq:identity-ch-symmetries}.

\begin{proposition}
\label{prop:simplified-SDP-two-ext-sim-id-QEC}The semi-definite program in
Proposition~\ref{prop:two-ext-sim-err-gen-ch-ch}, for the special case of simulating the identity
channel $\operatorname{id}_{A\rightarrow B}^{d}$, simplifies as follows for $d\geq 3$:%
\begin{align}
&  e_{\operatorname{2PENS}}(\mathcal{N}_{\hat{A}\rightarrow\hat{B}})\nonumber\\
&  =e_{\operatorname{2PENS}}^{F}(\mathcal{N}_{\hat{A}\rightarrow\hat{B}})\\
&  =1-\sup_{\substack{M^{+},M^{-},M^{0}\geq
0,\\M^{1},M^{2},M^{3}\in\operatorname{LinOp}}}\operatorname{Tr}\!\left[T_{\hat{A}\hat{B}_{1}}(\Gamma
_{\hat{A}\hat{B}_{1}}^{\mathcal{N}})\frac{P_{\hat{A}\hat{B}_{1}\hat{B}_{2}}}{d_{\hat{B}}}\right],
\label{eq:obj-func-qec-simp}
\end{align}
subject to%
\begin{equation}%
\begin{bmatrix}
M^{0}+M^{3} & M^{1}-iM^{2}\\
M^{1}+iM^{2} & M^{0}-M^{3}%
\end{bmatrix}
\geq0,
\label{eq:M-matrix-PSD-qec-constr}
\end{equation}%
\begin{align}
I_{\hat{B}_{1}\hat{B}_{2}} &  =\operatorname{Tr}_{\hat{A}}[M_{\hat{A}\hat
{B}_{1}\hat{B}_{2}}^{+}+M_{\hat{A}\hat{B}_{1}\hat{B}_{2}}^{-}+M_{\hat{A}%
\hat{B}_{1}\hat{B}_{2}}^{0}],
\label{eq:TP-constr-simp-QEC}\\
M_{\hat{A}\hat{B}_{1}\hat{B}_{2}}^{i} &  =\mathcal{F}_{\hat{B}_{1}\hat{B}_{2}%
}(M_{\hat{A}\hat{B}_{1}\hat{B}_{2}}^{i})\quad\forall i\in\left\{
+,-,0,1\right\}  , \label{eq:perm-cov-1-constr-simp-QEC}\\
M_{\hat{A}\hat{B}_{1}\hat{B}_{2}}^{j} &  =-\mathcal{F}_{\hat{B}_{1}\hat{B}%
_{2}}(M_{\hat{A}\hat{B}_{1}\hat{B}_{2}}^{j})\quad\forall j\in\left\{
2,3\right\}  ,\label{eq:perm-cov-2-constr-simp-QEC}\\
P_{\hat{A}\hat{B}_{1}\hat{B}_{2}} &  =\frac{1}{d_{\hat{B}}}\operatorname{Tr}%
_{\hat{B}_{2}}[P_{\hat{A}\hat{B}_{1}\hat{B}_{2}}]\otimes I_{\hat{B}_{2}}, \label{eq:P-constr-no-sig-qec}\\
Q_{\hat{A}\hat{B}_{1}\hat{B}_{2}} 
&  =\frac{1}{d_{\hat{B}}}\operatorname{Tr}%
_{\hat{B}_{2}}[Q_{\hat{A}\hat{B}_{1}\hat{B}_{2}}]\otimes I_{\hat{B}_{2}}, \label{eq:Q-constr-no-sig-qec}\\
P_{\hat{A}\hat{B}_{1}\hat{B}_{2}} &  \coloneqq\frac{1}{2d}\left[  dM^{0}%
+M^{1}+\sqrt{d^{2}-1}\,M^{2}\right]  ,\\
Q_{\hat{A}\hat{B}_{1}\hat{B}_{2}} &  \coloneqq\frac{1}{2d}\left[
\begin{array}
[c]{c}%
2d\left(  M_{\hat{A}\hat{B}_{1}\hat{B}_{2}}^{+}+M_{\hat{A}\hat{B}_{1}\hat
{B}_{2}}^{-}\right)  +dM_{\hat{A}\hat{B}_{1}\hat{B}_{2}}^{0}\\
-M_{\hat{A}\hat{B}_{1}\hat{B}_{2}}^{1}-\sqrt{d^{2}-1}\,M_{\hat{A}\hat{B}_{1}%
\hat{B}_{2}}^{2}%
\end{array}
\right]  ,
\end{align}%
\begin{align}
T_{\hat{A}}\!\left(  \frac{2M_{\hat{A}\hat{B}_{1}\hat{B}_{2}}^{+}}%
{d+2}+M_{\hat{A}\hat{B}_{1}\hat{B}_{2}}^{0}+M_{\hat{A}\hat{B}_{1}\hat{B}_{2}%
}^{1}\right)   &  \geq0,\label{eq:PPT-begin-QEC}\\
T_{\hat{A}}\!\left(  \frac{2M_{\hat{A}\hat{B}_{1}\hat{B}_{2}}^{-}}%
{d-2}+M_{\hat{A}\hat{B}_{1}\hat{B}_{2}}^{1}-M_{\hat{A}\hat{B}_{1}\hat{B}_{2}%
}^{0}\right)   &  \geq0,\label{eq:PPT-begin-QEC-2}\\%
\begin{bmatrix}
G^{0}+G^{3} & G^{1}-iG^{2}\\
G^{1}+iG^{2} & G^{0}-G^{3}%
\end{bmatrix}
&  \geq0,
\label{eq:PPTA-end-QEC}
\end{align}%
\begin{align}
&  G_{\hat{A}\hat{B}_{1}\hat{B}_{2}}^{0}\coloneqq T_{\hat{A}}\left(
M^{+}+M^{-}+\frac{M^{0}-dM^{1}}{2}\right)  ,\\
&  G_{\hat{A}\hat{B}_{1}\hat{B}_{2}}^{1}\coloneqq T_{\hat{A}}\left(
M^{+}-M^{-}+\frac{M^{1}-dM^{0}}{2}\right)  ,\\
&  G_{\hat{A}\hat{B}_{1}\hat{B}_{2}}^{2}\coloneqq\frac{\sqrt{3\left(
d^{2}-1\right)  }}{2}T_{\hat{A}}(M_{\hat{A}\hat{B}_{1}\hat{B}_{2}}^{2}),\\
&  G_{\hat{A}\hat{B}_{1}\hat{B}_{2}}^{3}\coloneqq\frac{\sqrt{3\left(
d^{2}-1\right)  }}{2}T_{\hat{A}}(M_{\hat{A}\hat{B}_{1}\hat{B}_{2}}^{3}),
\end{align}%
\begin{align}
T_{\hat{A}\hat{B}_{1}}\!\left(  \frac{dM^{+}}{d+2}+M^{-}+\frac{dM^{0}%
-M^{1}-\sqrt{d^{2}-1}\,M^{2}}{2}\right)   &  \geq0,
\label{eq:PPTB-constr-qec-simp-1}\\
T_{\hat{A}\hat{B}_{1}}\!\left(  M^{+}+\frac{dM^{-}}{d-2}-\frac{dM^{0}%
-M^{1}-\sqrt{d^{2}-1}\,M^{2}}{2}\right)   &  \geq0,
\label{eq:PPTB-constr-qec-simp-2}
\end{align}%
\begin{equation}%
\begin{bmatrix}
E^{0}+E^{3} & E^{1}-iE^{2}\\
E^{1}+iE^{2} & E^{0}-E^{3}%
\end{bmatrix}
\geq0,
\label{eq:PPTB-constr-qec-simp-3}
\end{equation}%
\begin{align}
E_{\hat{A}\hat{B}_{1}\hat{B}_{2}}^{0} &  \coloneqq\frac{T_{\hat{A}\hat{B}_{1}%
}\!\left(  d(M^{+}-M^{-})+\frac{L^{0}}{2}\right)  }{d^{2}-1},\\
E_{\hat{A}\hat{B}_{1}\hat{B}_{2}}^{1} &  \coloneqq\frac{T_{\hat{A}\hat{B}_{1}%
}\!\left(  -M^{+}+M^{-}+\frac{L^{1}}{2}\right)  }{d^{2}-1},\\
E_{\hat{A}\hat{B}_{1}\hat{B}_{2}}^{2} &  \coloneqq\frac{T_{\hat{A}\hat{B}_{1}%
}\!(M^{+}-M^{-}+\frac{L^{2}}{2})}{\sqrt{d^{2}-1}},\\
E_{\hat{A}\hat{B}_{1}\hat{B}_{2}}^{3} &  \coloneqq T_{\hat{A}\hat{B}_{1}%
}(M_{\hat{A}\hat{B}_{1}\hat{B}_{2}}^{3}),\\
L^{0} &  \coloneqq\left(  d^{2}-2\right)  M^{0}+dM^{1}+d\sqrt{d^{2}-1}\,M^{2},\\
L^{1} &  \coloneqq dM^{0}+\left(  2d^{2}-3\right)  M^{1}-\sqrt{d^{2}-1}\,
M^{2},\\
L^{2} &  \coloneqq M^{1}-dM^{0}-\sqrt{d^{2}-1}\,M^{2},
\end{align}%
\begin{align}
\frac{\operatorname{Tr}_{\hat{A}}[2M^{+}]}{\left(  d+2\right)  \left(
d-1\right)  } &  =\frac{\operatorname{Tr}_{\hat{A}}[2M^{+}+M^{0}+M^{1}%
]}{d\left(  d+1\right)  },
\label{eq:non-sig-A-simp-QEC}
\\
\frac{\operatorname{Tr}_{\hat{A}}[2M^{-}]}{\left(  d-2\right)  \left(
d+1\right)  } &  =\frac{\operatorname{Tr}_{\hat{A}}[2M^{-}+M^{0}-M^{1}%
]}{d\left(  d-1\right)  }, \label{eq:non-sig-B-simp-QEC}\\
\frac{1}{2}\operatorname{Tr}_{\hat{A}}[M^{0}] &  =\frac{dI_{\hat{B}_{1}\hat
{B}_{2}}+\operatorname{Tr}_{\hat{A}}[M^{-}-M^{+}-M^{1}]}{d\left(
d^{2}-1\right)  },\\
\frac{1}{2}\operatorname{Tr}_{\hat{A}}[M^{1}] &  =\frac{-I_{\hat{B}_{1}\hat
{B}_{2}}+d\operatorname{Tr}_{\hat{A}}[M^{+}-M^{-}+M^{1}]}{d\left(
d^{2}-1\right)  },\\
\operatorname{Tr}_{\hat{A}}[M^{2}] & = \operatorname{Tr}_{\hat{A}}[M^{3}] = 0.
\label{eq:PPT-end-QEC}
\end{align}
For the case of $d=2$, the SDP is the same, with the exception that we set $M^-_{\hat{A}\hat{B}_1\hat{B}_2} = 0$ and the constraints in \eqref{eq:PPT-begin-QEC-2},  \eqref{eq:PPTB-constr-qec-simp-2}, and \eqref{eq:non-sig-B-simp-QEC} are not used.
\end{proposition}

\begin{proof}
See Appendix~\ref{app:proof-approx-QEC-SDP-simplify}.
\end{proof}

\bigskip 

We now provide expository remarks similar to Remarks~\ref{rem:SDP-explain} and \ref{rem:SDP-many-constr}, as well as an additional remark about approximate quantum error correction assisted by one-way LOCC.

\begin{remark}\label{rem:NS_LOCC}
    The SDP in the statement of Proposition~\ref{prop:simplified-SDP-two-ext-sim-id-QEC} is rather lengthy, and so we provide some explanation here. The constraint in \eqref{eq:M-matrix-PSD-qec-constr} and the constraints $M^+, M^-, M^0 \geq 0$  in \eqref{eq:obj-func-qec-simp} correspond to the constraint of complete positivity in \eqref{eq:obj-func-qec-gen} (i.e., $M_{A\hat{A}%
\hat{B}_{1}B_{1}\hat{B}_{2}B_{2}}\geq0$). The constraint in \eqref{eq:TP-constr-simp-QEC} corresponds to the constraint of trace preservation in \eqref{eq:sim-two-ext-superch-TP}. The constraints in \eqref{eq:perm-cov-1-constr-simp-QEC}--\eqref{eq:perm-cov-2-constr-simp-QEC} correspond to the constraint of permutation covariance in \eqref{eq:sim-two-ext-superch-2EXT-constr}. The constraints in \eqref{eq:P-constr-no-sig-qec}--\eqref{eq:Q-constr-no-sig-qec} correspond to the non-signaling constraint in \eqref{eq:sim-two-ext-superch-constr-ext}. The constraints in \eqref{eq:PPT-begin-QEC}--\eqref{eq:PPTA-end-QEC} correspond to the PPT constraint in \eqref{eq:PPTA-gen-ch-sim}, and the constraints in \eqref{eq:PPTB-constr-qec-simp-1}--\eqref{eq:PPTB-constr-qec-simp-3} correspond to the PPT constraint in \eqref{eq:PPTB-gen-ch-sim}. Finally, the constraints in \eqref{eq:non-sig-A-simp-QEC}--\eqref{eq:PPT-end-QEC} correspond to the non-signaling constraint in \eqref{eq:non-sig-constr-gen-ch-sim}.
\end{remark}

\begin{remark}
\label{rem:SDP-many-constr-QEC}
Even though the number of constraints in the SDP above appears to increase when compared with the SDP from Proposition~\ref{prop:two-ext-sim-err-gen-ch-ch}, we note that the runtime of the SDP above is significantly reduced because the size of the matrices involved in each of the constraints is much smaller. This is the main advantage that we get by incorporating unitary covariance symmetry of the identity channel.

If we only optimized over the larger set of two-extendible channels instead of the set of two-PPT-extendible non-signaling channels, the SDP would be much simpler, given by \eqref{eq:obj-func-qec-simp}--\eqref{eq:Q-constr-no-sig-qec}. However, optimizing over the smaller set of two-PPT-extendible non-signaling channels gives tighter bounds at a marginal increase in computational cost, and thus we also include the PPT constraints in \eqref{eq:PPT-begin-QEC}--\eqref{eq:PPTA-end-QEC} and \eqref{eq:PPTB-constr-qec-simp-1}--\eqref{eq:PPTB-constr-qec-simp-3} and the non-signaling constraints in \eqref{eq:non-sig-A-simp-QEC}--\eqref{eq:PPT-end-QEC}.
\end{remark}

\begin{remark}
By excluding the non-signaling constraints in \eqref{eq:non-sig-A-simp-QEC}--\eqref{eq:PPT-end-QEC}, the resulting SDP gives a lower bound on the simulation error of approximate quantum error correction assisted by a one-way LOCC channel. That is, the resulting SDP gives a lower bound on 
\begin{equation}
e_{\operatorname{1WL}}(\mathcal{N}_{\hat{A}\rightarrow\hat{B}}%
)\coloneqq
\inf_{\Lambda\in\operatorname{1WL}}e(\mathcal{N}_{\hat{A}\rightarrow\hat{B}%
},\Lambda_{(\hat{A}\rightarrow\hat{B})\rightarrow(A\rightarrow B)}),
\end{equation}
where
\begin{multline}e_{\operatorname{1WL}}(\mathcal{N}_{\hat{A}\rightarrow\hat{B}},\Lambda
_{(\hat{A}\rightarrow\hat{B})\rightarrow(A\rightarrow B)})\coloneqq\\
\frac{1}{2}\left\Vert \operatorname{id}_{A\rightarrow B}^{d}-\Lambda
_{(\hat{A}\rightarrow\hat{B})\rightarrow(A\rightarrow B)}(\mathcal{N}_{\hat
{A}\rightarrow\hat{B}})\right\Vert _{\diamond},
\end{multline}
with $\Lambda$ a one-way LOCC superchannel, as defined in \eqref{eq:LOCC-superch-Lam}. By the same reasoning given for Proposition~\ref{prop:equality-sim-errs-LOCR-QEC}, this error is no different if we use infidelity instead of normalized diamond distance.
\end{remark}



\section{Examples}\label{sec: VIII}

In this section we present some numerical results from our semi-definite programs. To perform these numerical calculations, we employed  CVXPY \cite{diamond2016cvxpy, agrawal2018rewriting} with the interior point optimizer MOSEK. All of our Python source code is available with the arXiv posting of our paper.

\subsection{Approximate teleportation and quantum error correction using special mixed states and channels}

\label{sec:ex-1}

First, we provide bounds on the performance of approximate teleportation (i.e.,  on the error in simulating an identity channel), when using a particular set of imperfect resource states.
In the past, PPT constraints alone (i.e., without two-extendibility) have been used to obtain bounds on objective functions involving an optimization over the set of LOCC channels (see, e.g., \cite{LM15,WXD18,WD16,WD16pra,BW17,SW21}). We can also use them to obtain a lower bound on the simulation error of approximate teleportation. By following techniques similar to those in \cite{LM15,SW21}, we find the following SDP gives a lower bound on the simulation error of approximate  teleportation:
\begin{equation}
    1-\sup_{K_{\hat{A}\hat{B}}\geq0}\left\{
\begin{array}
[c]{c}%
\operatorname{Tr}[K_{\hat{A}\hat{B}}\rho_{\hat{A}\hat{B}}]:\\
K_{\hat{A}\hat{B}}\leq I_{\hat{A}\hat{B}},\\
-I_{\hat{A}\hat{B}}\leq d\ T_{\hat{B}}\left(  K_{\hat{A}\hat{B}}\right)  \leq
I_{\hat{A}\hat{B}}%
\end{array}
\right\}.
\label{eq:sdp-approx-tele-just-ppt}
\end{equation}
See Appendix~\ref{app:approx-tele-just-PPT} for a proof. We note here that PPT constraints are implied by the two-PPT-extendibility constraints given in Proposition~\ref{prop:simplified-SDP-two-ext-sim-id}, so that the optimal value in \eqref{eq:sdp-approx-tele-just-ppt} is not smaller than the optimal value in \eqref{eq:obj-func-qec-simp}. We also note that an SDP bearing some similarities to that in \eqref{eq:sdp-approx-tele-just-ppt} was presented in \cite{FWTD19}, but that SDP calculates a bound on one-shot distillable entanglement, whereas the SDP in \eqref{eq:sdp-approx-tele-just-ppt} calculates a bound on the error of approximate teleportation.

In the following example, we show that two-PPT-extendibility gives strictly stronger bounds than PPT constraints alone, when optimizing over one-way LOCC channels.
Consider the following mixed state:
\begin{equation}
    p~\Phi_{\hat{A}\hat{B}} + (1-p)~\pi_{\hat{A}}\otimes\sigma_{\hat{B}},
    \label{eq:special-mixed-state}
\end{equation}
where $p \in [0,1]$, $\Phi_{\hat{A}\hat{B}}$ is the maximally entangled state of Schmidt rank three, $\pi_{\hat{A}}$ is the maximally mixed state of dimension three, and $\sigma_{\hat{B}}$ is a randomly selected $3\times 3$ density matrix. Using the state in \eqref{eq:special-mixed-state} as the resource for approximate teleportation, lower bounds on the simulation error, as given by two-PPT-extendibility, are stronger than those given by PPT constraints alone, for small values of $p$. Figure~\ref{fig:difference} compares the lower bounds obtained for different values of $p$ and randomly generated $\sigma_{\hat{B}}$. The state $\sigma_{\hat{B}}$ that was used to generate data for Figure~\ref{fig:difference} is as follows:
\begin{equation}
    \begin{bmatrix}
    0.140 & 0.043 + 0.024i & -0.143 + 0.028i\\
    0.043 - 0.024i & 0.222 & -0.257 + 0.006i\\
    -0.143 - 0.028i & -0.257 - 0.006i & 0.638
    \end{bmatrix}.
    \label{eq:sigma_B-state-example}
\end{equation}

\begin{figure}
    \centering
    \includegraphics[width = \linewidth]{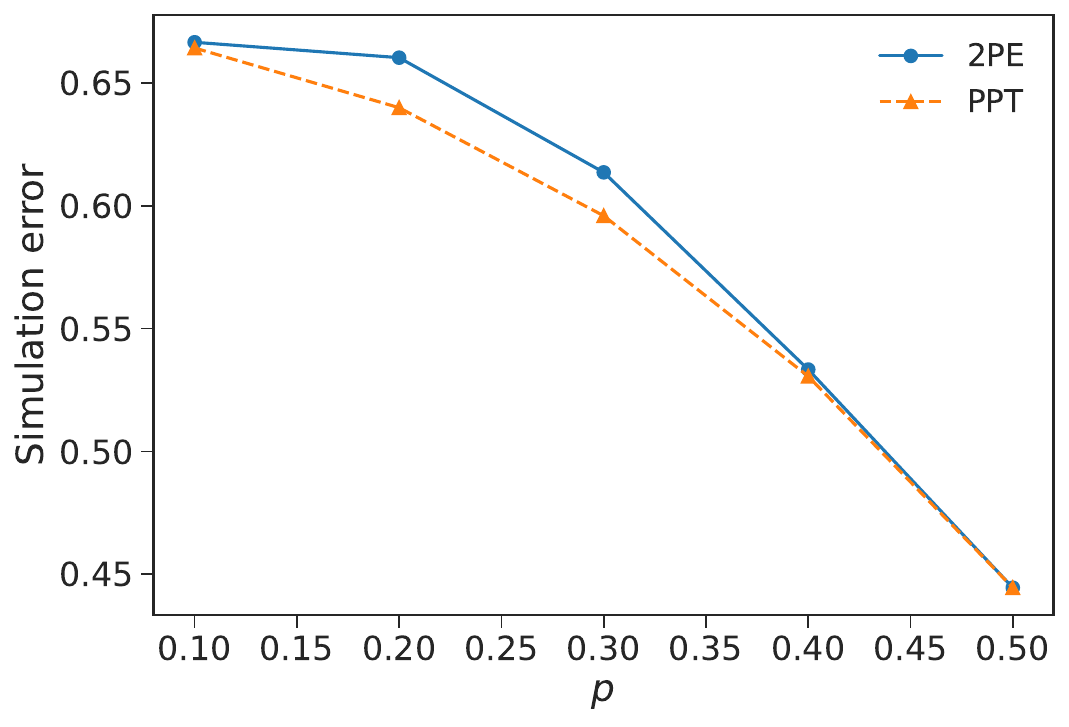}
    \caption{Comparison between two-PPT-extendiblity and PPT constraints for bounding the simulation error in approximate teleportation, when using the resource state $ p~\Phi_{\hat{A}\hat{B}} + (1-p)~\pi_{\hat{A}}\otimes\sigma_{\hat{B}}$, where $p\in [0,1]$. The plot shows that two-PPT-extendibility gives slightly better bounds for $p < 0.5$. For higher values of $p$, the two curves become indistinguishable.}
    \label{fig:difference}
\end{figure}

We note here that the SDP calculations depend on the choice of $\sigma_{\hat{B}}$.
For certain choices of $\sigma_{\hat{B}}$, the difference in the errors disappears for all values of $p$, e.g., when $\sigma_{\hat{B}}$ is a maximally mixed state. 
It still remains open to determine the full set  of resource states for which  two-PPT-extendibility  gives stronger bounds on the simulation error. Regardless, this example demonstrates that including two-PPT-extendibility constraints can improve the bounds obtained using the PPT constraints alone. 

One can consider the same comparison for approximate quantum error correction. Using similar techniques, we derive the following SDP lower bound on the simulation error of approximate quantum error correction for a channel $\mathcal{N}_{\hat{A}\to \hat{B}}$, when using PPT and non-signaling constraints only:
\begin{equation}
1-\sup_{K_{\hat{A}\hat{B}},\sigma_{\hat{A}}\geq0}\left\{
\begin{array}
[c]{c}%
\operatorname{Tr}[K_{\hat{A}\hat{B}}\Gamma_{\hat{A}\hat{B}}^{\mathcal{N}}]:\\
K_{\hat{A}\hat{B}}\leq\sigma_{\hat{A}}\otimes I_{\hat{B}},\\
d^{2}\operatorname{Tr}_{\hat{A}}[K_{\hat{A}\hat{B}}]=I_{\hat{B}},\\
\sigma_{\hat{A}}\otimes I_{\hat{B}}\pm d\ T_{\hat{B}}(K_{\hat{A}\hat{B}}%
)\geq0,\\
\operatorname{Tr}[\sigma_{\hat{A}}]=1.
\end{array}
\right\}  .
\label{eq:PPT-SDP-approx-QEC}
\end{equation}
See Appendix~\ref{sec:PPT-NS-alone-approx-QEC} for a proof.  We note here that essentially the same SDP was given in \cite{LM15} (up to a transpose in the objective function). The SDP in \cite{LM15} resulted from taking the error criterion to be in terms of entanglement fidelity when transmitting the maximally entangled state. Our proof here clarifies that essentially the same SDP results when using normalized diamond distance or channel infidelity as the error criterion. The second constraint in the SDP ($d^2\operatorname{Tr}_{\hat{A}}[K_{\hat{A}\hat{B}}]=I_{\hat{B}}$) corresponds to the non-signaling condition. Following the same reasoning as in Remark~\ref{rem:NS_LOCC}, removing this constraint leads to an SDP that provides  a lower bound on the simulation error of approximate quantum error correction assisted by one-way LOCC. 

The example state in \eqref{eq:special-mixed-state} can also serve as the Choi state of a channel, due to the fact that the reduced state of system~$\hat{A}$ is maximally mixed. In Figure~\ref{fig:difference_qec}, we plot the  lower bound in \eqref{eq:PPT-SDP-approx-QEC} and the lower bound from Proposition~\ref{prop:simplified-SDP-two-ext-sim-id-QEC} for the corresponding channel. Additionally, we also plot the simulation errors that result from excluding the non-signaling constraints from both SDPs. The resulting SDPs provide lower bounds on the errors in approximate quantum error correction assisted by one-way LOCC using PPT and two-PPT-extendibility, respectively. Figure~\ref{fig:difference_qec} demonstrates that the lower bound in Proposition~\ref{prop:simplified-SDP-two-ext-sim-id-QEC} improves upon \eqref{eq:PPT-SDP-approx-QEC} for one-way LOCC simulation but provides no advantage for LOCR simulation. The difference between all four curves becomes very small (less than $10^{-3}$) for higher values of $p$. 

\begin{figure}
    \centering
    \includegraphics[width = \linewidth]{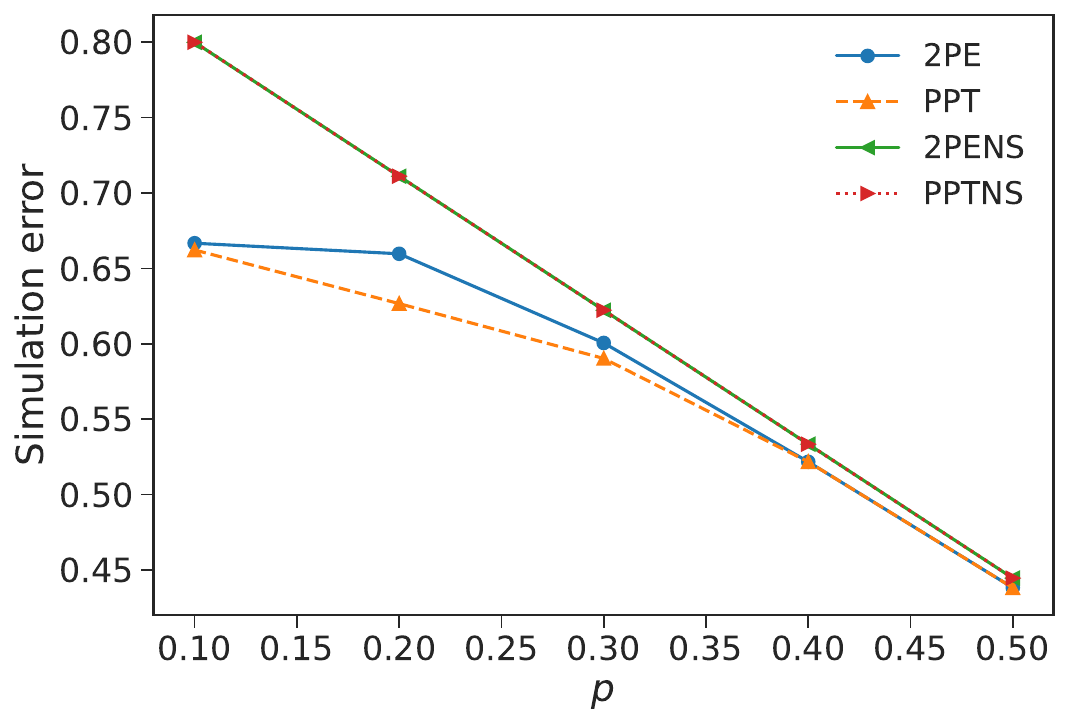}
    \caption{Comparison between two-PPT-extendiblity and PPT constraints for bounding the simulation error in approximate quantum error correction when using the resource channel with Choi state $ p~\Phi_{\hat{A}\hat{B}} + (1-p)~\pi_{\hat{A}}\otimes\sigma_{\hat{B}}$, where $p\in [0,1]$ and $\sigma_{\hat{B}}$ is defined in \eqref{eq:sigma_B-state-example}. PPTNS and 2PENS are the curves obtained using the SDPs in \eqref{eq:PPT-SDP-approx-QEC} and Proposition~\ref{prop:simplified-SDP-two-ext-sim-id-QEC}, respectively, giving lower bounds on the error in approximate quantum error correction. There is no signficant difference in the numercial values obtained from these two conditions. PPT and 2PE are the curves obtained using the same SDPs but without the non-signaling constraints, hence, giving lower bounds on the error in one-way LOCC-assisted approximate error correction.}
    \label{fig:difference_qec}
\end{figure}

\subsection{Three-dimensional approximate teleportation using two-dimensional special mixed states}

\label{sec:ex-3d-tele-2d-resource}

In this example, we investigate the simulation error in approximate teleportation when a lower dimensional imperfect resource state is used to teleport a higher dimensional state. We use a similar resource state as in \eqref{eq:special-mixed-state}:
\begin{equation}
    \rho_{\hat{A}\hat{B}} = p~\Phi_{\hat{A}\hat{B}} + (1-p)~\pi_{\hat{A}}\otimes\sigma^{\prime}_{\hat{B}},
    \label{eq:new-special-mixed-state}
\end{equation}
but the maximally entangled and maximally mixed states are two-dimensional. Additionally, $\sigma^{\prime}_{\hat{B}}$ was generated randomly and is taken as
\begin{equation}\label{eq:sigma_2d}
    \sigma^{\prime}_{\hat{B}} = \begin{bmatrix}
    0.287 & -0.347 + 0.132i \\
    -0.347 - 0.132i & 0.713
    \end{bmatrix}.
\end{equation}
In Figure~\ref{fig:higher_dim_tel}, we plot the bounds on the simulation error versus the parameter $p$ in \eqref{eq:new-special-mixed-state}, when using the 2PE constraints given in Proposition~\ref{prop:simplified-SDP-two-ext-sim-id} and the PPT constraints given in \eqref{eq:sdp-approx-tele-just-ppt}. We also compare this to the bounds on the  simulation error when using a three-dimensional special mixed state instead. The resource state used is the same as the state in \eqref{eq:special-mixed-state}, but $\sigma_{\hat{B}}$ is chosen as follows:
\begin{equation}\label{eq:sigma_3d}
    \begin{bmatrix}
    0.287 & -0.347 + 0.132i & 0 \\
    -0.347 - 0.132i & 0.713 & 0 \\
    0 & 0 & 0
    \end{bmatrix},
\end{equation}
in order to provide a closer comparison with the two-dimensional case in \eqref{eq:new-special-mixed-state}.

\begin{figure}
    \centering
    \includegraphics[width = \linewidth]{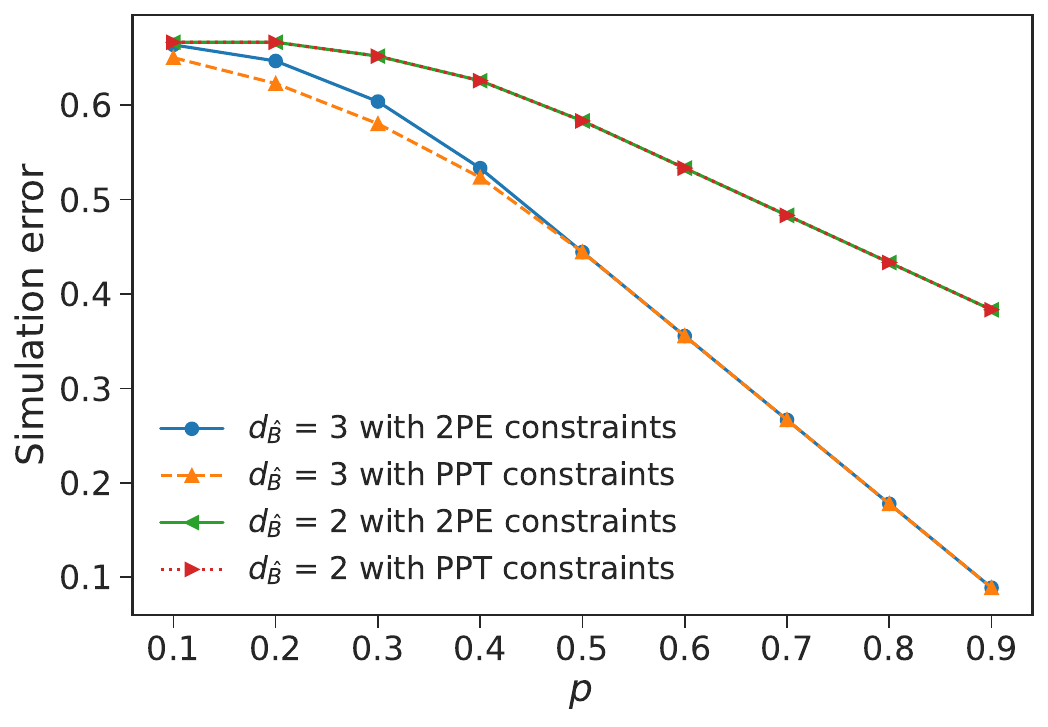}
    \caption{Comparison between bounds on the simulation error for approximate teleportation when using a two-dimensional special mixed state and a three-dimensional special mixed state as a resource. The resource state is of the form $ p~\Phi_{\hat{A}\hat{B}} + (1-p)~\pi_{\hat{A}}\otimes\sigma_{\hat{B}}$, where $p\in [0,1]$ and $\sigma_{\hat{B}}$ is chosen to be \eqref{eq:sigma_2d} when $d_{\hat{B}}=2$ and \eqref{eq:sigma_3d} when $d_{\hat{B}}=3$. The bounds on the simulation error are calculated using both the 2PE constraints given in Proposition~\ref{prop:simplified-SDP-two-ext-sim-id} and the PPT constraints given in~\eqref{eq:sdp-approx-tele-just-ppt}. There is no significiant difference in the numerical values obtained from both the constraints for $d_{\hat{B}} = 2$.}
    \label{fig:higher_dim_tel}
\end{figure}

We see from Figure~\ref{fig:higher_dim_tel} that a two-dimensional resource state with a small amount of imperfection can outperform a three-dimensional resource with higher amounts of imperfection for the task of three-dimensional approximate teleportation. We also notice that the 2PE constraints and the PPT constraints give the same error values when $d_{\hat{B}}=2$, but give different values when $d_{\hat{B}}=3$, as seen in Figure~\ref{fig:difference} as well.

\subsection{Approximate quantum error correction for depolarizing channels}

\label{sec:ex-depol-ch}

In this example, we investigate the simulation error in approximate error correction for qubit and qutrit depolarizing channels, with the goal of simulating a qubit identity channel. The Choi state of the depolarizing channel $\mathcal{D}_{\hat{A}\to\hat{B}}$ is given by
\begin{equation}
    \Phi_{\hat{A}\hat{B}}^{\mathcal{D}} \coloneqq p~\Phi_{\hat{A}\hat{B}} + (1-p)~\pi_{\hat{A}\hat{B}},
\end{equation}
where $p\in[0,1]$, $\Phi_{\hat{A}\hat{B}}$ is the maximally entangled state, and $\pi_{\hat{A}\hat{B}}$ is the maximally mixed state. For a qubit depolarizing channel, $d_{\hat{A}} = d_{\hat{B}} = 2$, and for a qutrit depolarizing channel, $d_{\hat{A}} = d_{\hat{B}} = 3$.
\begin{figure}
    \centering
    \includegraphics[width = \linewidth]{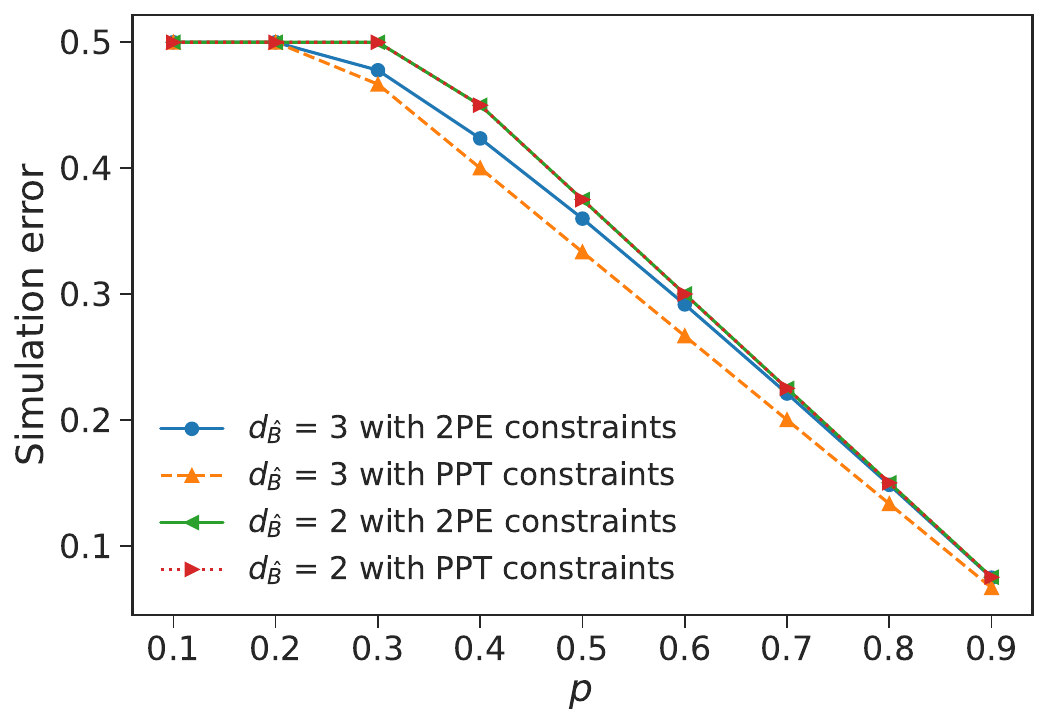}
    \caption{Lower bounds on the simulation error of approximate quantum error correction for depolarizing channels when simulating a two-dimensional identity channel. The bounds are calculated using the SDP in Proposition~\ref{prop:simplified-SDP-two-ext-sim-id-QEC} with two-PPT-extendibility constraints, and the SDP in \eqref{eq:PPT-SDP-approx-QEC} with PPT constraints only, for different dimensions of the depolarizing channel ($d_{\hat{B}}=2$ and $d_{\hat{B}}=3$). The bounds are obtained without the non-signaling constraints, hence, corresponding to one-way LOCC simulation.}
    \label{fig:depolarizing_qec}
\end{figure}

In Figure~\ref{fig:depolarizing_qec}, we plot the lower bounds on the simulation error of approximate error correction for a depolarizing channel, when simulating a qubit identity channel. The bounds are obtained using the two-PPT-extendibility conditions from Proposition~\ref{prop:simplified-SDP-two-ext-sim-id-QEC} and the PPT conditions from \eqref{eq:PPT-SDP-approx-QEC}. The bounds are calculated for the case of one-way LOCC assistance, i.e., by ignoring the non-signaling constraints in \eqref{eq:non-sig-A-simp-QEC}--\eqref{eq:PPT-end-QEC} and the constraint $d^{2}\operatorname{Tr}_{\hat{A}}[K_{\hat{A}\hat{B}}]=I_{\hat{B}}$ in \eqref{eq:PPT-SDP-approx-QEC}, respectively. We notice from Figure~\ref{fig:depolarizing_qec} that the two-PPT-extendibility constraints give better bounds compared to the PPT constraints when using a three-dimensional depolarizing channel to simulate a two-dimensional identity channel. However, both sets of constraints give the same bounds when a two-dimensional depolarizing channel is used to simulate a two-dimensional identity channel. This was also observed in the numerical calculations of \cite{BBFS21}. 

We also note that a three-dimensional depolarizing channel provides little advantage over a two-dimensional depolarizing channel for simulating a two-dimensional identity channel. Therefore, a two-dimensional depolarizing channel with a slightly higher value of the parameter $p$ can outperform a three-dimensional depolarizing channel with a lower value of $p$, for the purpose of approximating a qubit identity channel.

\subsection{Approximate teleportation using the two-mode squeezed vacuum state}

\label{sec:ex-2}

Two-mode squeezed vacuum states are easily prepared in laboratories and have entanglement content that can be parameterized by $\lambda \geq 0$. They are defined as \cite{S17}
\begin{equation}
    \sqrt{1-\lambda^2}\sum^{\infty}_{n=0} \lambda^n\ket{n}\ket{n}.
\end{equation}
They are used as a resource state in continuous-variable quantum teleportation \cite{braunstein1998teleportation} and have also been used as a resource in experiments on teleportation of photonic qubits \cite{furusawa1998unconditional,takeda2013deterministic}. Here we investigate bounds on the performance of qudit teleportation with the two-mode squeezed vacuum state as the resource state.

The parameter $\lambda$ denotes the strength of squeezing applied ($\lambda = \operatorname{tanh}(r)$, where $r$ is the squeezing parameter). For low squeezing strength, we can ignore higher order terms in $\lambda$ without inducing much error.   We use the following state in our calculations for qudit teleportation:
\begin{equation}
     \frac{1}{\sqrt{1+\lambda^2+\lambda^4}}\sum^{2}_{n=0} \lambda^n\ket{n}\ket{n}.
\end{equation}
However, for higher values of the squeezing strength (i.e., $\lambda$ near to one), we do not expect this approximation to be good.

\begin{figure}
    \centering
    \includegraphics[width = \linewidth]{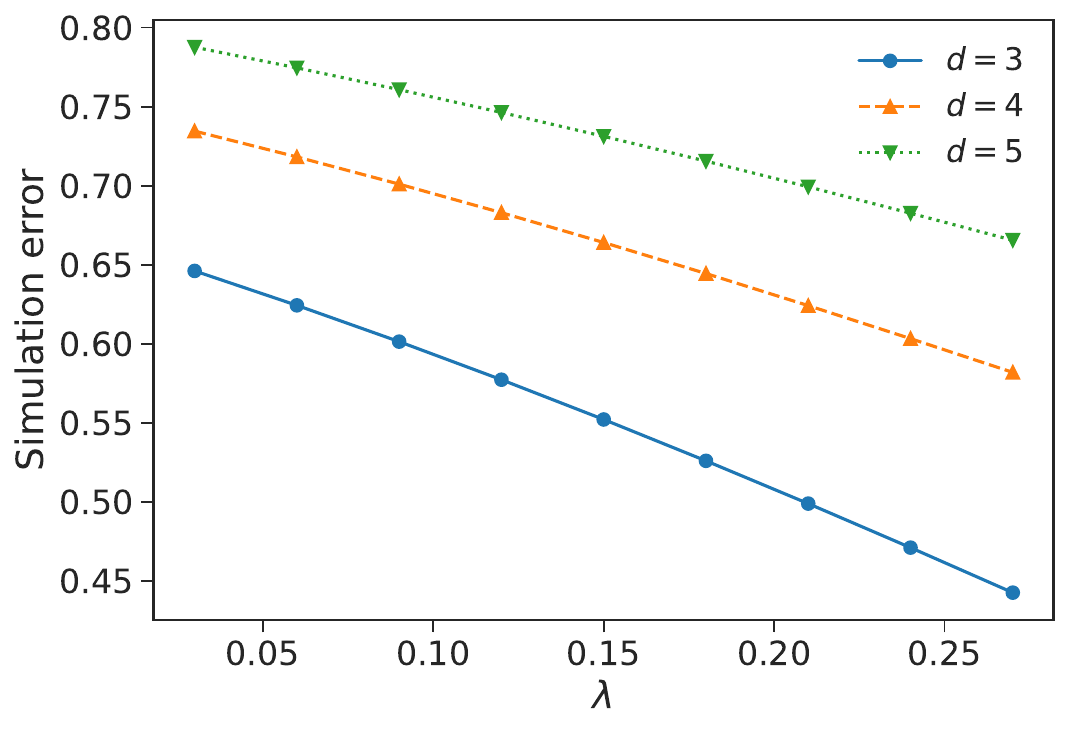}
    \caption{Lower bounds on the simulation error of unideal telportation when using the two-mode squeezed vacuum state as the resource state. The parameter $\lambda = \operatorname{tanh}(r)$, where $r$ is the squeezing parameter. Larger values of  $\lambda$ correspond to larger values of entanglement, which leads to a smaller error in simulating the identity channel.}
    \label{fig:teleportation}
\end{figure}

Figure~\ref{fig:teleportation} demonstrates that the simulation error increases with $d$ for fixed values of $\lambda$, where $d$ is the dimension of the target identity channel that the protocol is simulating. The simulation error does not go to zero for $d>3$, even for maximally entangled qutrit resource states. Therefore, projecting this trend further, we conclude that simulation of a higher-dimensional identity channel with a lower-dimensional resource state incurs larger errors in the simulation. We note here that we observed no difference in the values calculated by the SDPs in \eqref{eq:PPT-SDP-approx-QEC} and Proposition~\ref{prop:simplified-SDP-two-ext-sim-id-QEC}.

\subsection{Approximate quantum error correction for a three-level amplitude damping channel}

\label{sec:ex-3}

Here we present an example of our bound for the simulation error in approximate error correction. We consider a three-level amplitude damping channel, as defined in \cite{Chessa2021Feb}, to demonstrate our SDP in Proposition~\ref{prop:simplified-SDP-two-ext-sim-id-QEC}.

\begin{figure}
    \centering
    \includegraphics[width = \linewidth]{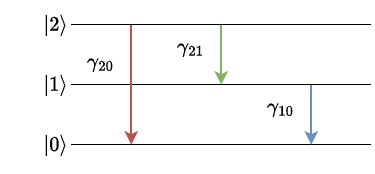}
    \caption{Action of an amplitude damping channel on a three-level quantum system. The parameters $\gamma_{10}$, $\gamma_{20}$, and $\gamma_{21}$ represent decay rates between the respective levels.}
    \label{fig:amp_damp}
\end{figure}

The channel can be defined using three decay parameters, labeled by the states involved: $(\gamma_{10}, \gamma_{21},\gamma_{20})$. See Figure~\ref{fig:amp_damp} for a depiction. 
The Kraus operators for the three-level amplitude damping channel are as follows:
\begin{align}
    K_0 & \coloneqq \begin{bmatrix}
    1 & 0 & 0\\
    0 & \sqrt{1-\gamma_{10}} & 0\\
    0 & 0 & \sqrt{1-\gamma_{21}-\gamma_{20}}
    \end{bmatrix},\\
    K_1 & \coloneqq \begin{bmatrix}
    0 & \sqrt{\gamma_{10}} & 0\\
    0 & 0 & 0\\
    0 & 0 & 0
    \end{bmatrix},\\
    K_2 & \coloneqq \begin{bmatrix}
    0 & 0 & 0\\
    0 & 0 & \sqrt{\gamma_{21}}\\
    0 & 0 & 0
    \end{bmatrix}, \\
    K_3 & \coloneqq \begin{bmatrix}
    0 & 0 & \sqrt{\gamma_{20}}\\
    0 & 0 & 0\\
    0 & 0 & 0
    \end{bmatrix},
\end{align}
so that its action on an input state $\rho$ is given by
$\sum_{i=0}^{3} K_i \rho K_i^{\dag}$. For the map to be completely positive and trace preserving, the decay parameters must obey
\begin{equation}
    \begin{cases}
    0\le \gamma_i \le 1 & \forall i\in\{10,21,20\}\\
    \gamma_{21} + \gamma_{20} \le 1
    \end{cases}.
\end{equation}

\begin{figure}
    \centering
    \includegraphics[width = \linewidth]{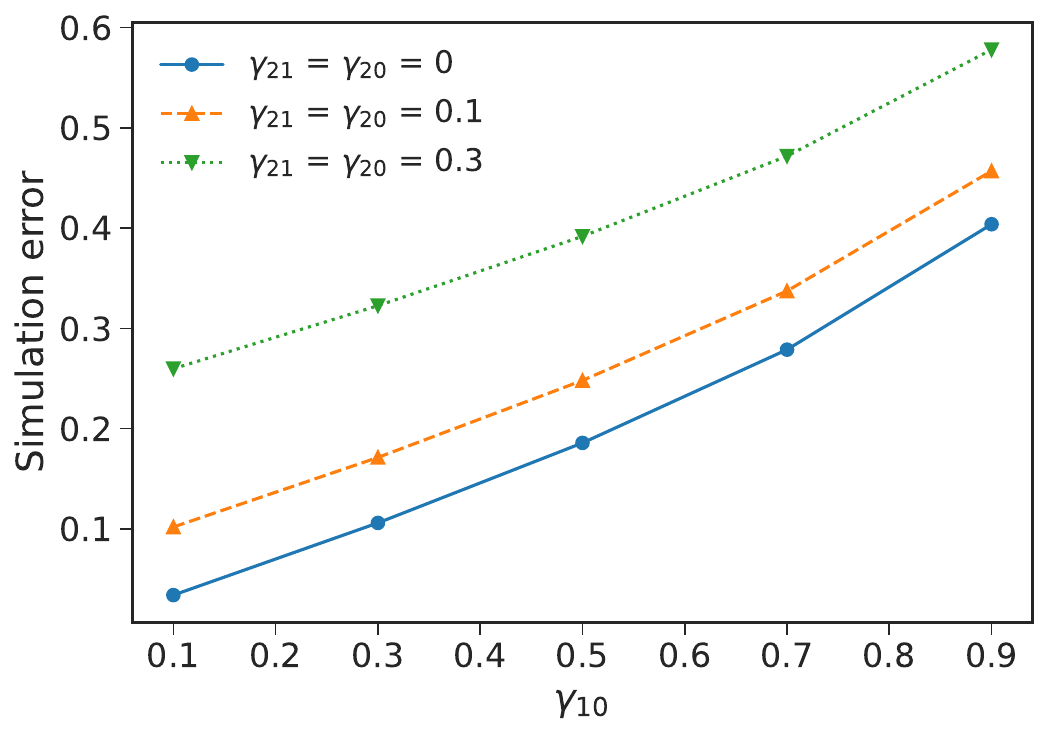}
    \caption{Lower bounds on the simulation error of approximate quantum error correction when using a three-level amplitude damping channel. The parameters $\gamma_{10}$, $\gamma_{21}$, and $\gamma_{20}$ are decay parameters for the labeled states. The dimension $d$ of the target identity channel is set equal to the input and output dimensions (equal to three) of  the amplitude-damping channel. }
    \label{fig:error_correction}
\end{figure}

Figure~\ref{fig:error_correction} plots the lower bound on the simulation error as a function of the decay parameter $\gamma_{10}$, for various values of the other decay parameters.
We notice in Figure~\ref{fig:error_correction} that the simulation error monotonically increases with the decay parameters. As all three decay parameters approach zero, the channel becomes close to an identity channel. This is reflected in the plot as the simulation error also approaches zero. We note here that we observed no difference in the values calculated by the SDPs in \eqref{eq:PPT-SDP-approx-QEC} and Proposition~\ref{prop:simplified-SDP-two-ext-sim-id-QEC}.

\subsection{Comparison of computational runtimes}

In this section we present the average runtime to execute various SDPs listed in this work. The calculations were performed on a computer with 16 GB RAM and an Intel i7-9750H processor. 

All calculations that generated the entries in Table~\ref{tab:runtime} employed the two-dimensional maximally entangled state. For approximate teleportation, the input is the maximally entangled state of Schmidt rank two, and for approximate error correction, the input is the qubit identity channel. The simulated channel is also the qubit identity channel in all cases.
\begin{table}
    \centering
    \begin{tabular}{|c|c|}
    \hline
    SDP & Runtime (seconds) \\
    \hline
    \hline
    Teleportation unsimplified 2PE & 253.03\\
    \hline
    Teleportation 2PE & 10.34 \\
    \hline
    Teleportation PPT & 0.19 \\
    \hline
    Error Correction unsimplified & 147.75\\
    \hline
    Error Correction unsimplified 2PENS & 158.22\\
    \hline
    Error Correction 2PENS & 5.65\\
    \hline
    Error Correction 2PE & 5.38\\
    \hline
    Error Correction PPTNS & 0.20\\
    \hline
    Error Correction PPT & 0.16\\
    \hline
\end{tabular}
    \caption{Comparing the runtime of different SDPs presented in this work. 2PE refers to two-PPT-extendibility constraints and NS indicates that non-signaling conditions were used. All calculations were done for a two-dimensional resource state and simulating the two-dimensional identity channel.}
    \label{tab:runtime}
\end{table}
The runtimes were calculated using time.time() function in Python. They are only presented for the purpose of comparison and can vary moderately. 

All runtimes are listed in Table~\ref{tab:runtime}, where we see that the unsimplified SDP for approximate teleportation with two-PPT-extendibility, given in Proposition~\ref{prop:two-ext-sim-err-gen-ch}, is around 25 times slower than the simplified SDP for the same in Proposition~\ref{prop:simplified-SDP-two-ext-sim-id}. The SDP for the simulation error in approximate teleportation using PPT constraints that is given in \eqref{eq:sdp-approx-tele-just-ppt} is several times faster than when two-PPT-extendibility constraints are employed, but we have seen in the examples that two-PPT-extendibility constraints can give tighter lower bounds on the simulation error.

Similarly, we see that the unsimplified SDP for approximate error correction when using two-PPT-extendibility constraints (Proposition~\ref{prop:two-ext-sim-err-gen-ch-ch}) is several times slower than the simplified SDP given in Proposition~\ref{prop:simplified-SDP-two-ext-sim-id-QEC}. Again, the SDP with PPT constraints given in \eqref{eq:PPT-SDP-approx-QEC} is much faster than the SDP with two-PPT-extendibility constraints, but we have demonstrated examples for which two-PPT-extendibility constraints provide a tighter lower bound on the simulation error.

\section{Conclusion}

\label{sec: Conclusion}

In this work, we developed a technique for quantifying the performance of approximate teleportation using an arbitrary resource state, by establishing a lower bound on the error in simulating a teleportation protocol that uses an imperfect resource state and one-way LOCC channels. We accomplished this by combining the notions of C-PPT-P channels and two-extendible channels to give a relaxation of the set of one-way LOCC channels, as was done previously in \cite{BBFS21} but for approximate quantum error correction. We significantly reduced the complexity of our semi-definite program by exploiting the unitary covariance symmetry  of the simulated identity channel. This symmetry is useful in semi-definite programs and can have much wider applications with respect to dynamical resource theories. As an example, we evaluated our lower bound when using a two-mode squeezed vacuum state as the resource state for approximate teleportation. 

We used related techniques to quantify the performance of approximate quantum error correction. Incorporating two-PPT-extendibility constraints again led to  computationally feasible semi-definite optimizations for evaluating lower bounds on the error in approximate quantum error correction. We further exploited the unitary covariance symmetry of the identity channel to give a less computationally taxing semi-definite program to calculate the error. Finally, we demonstrated some calculations for amplitude damping channels as the resource channels.

The SDPs in this work provide computational support to ongoing experimental research in quantum information by providing tools to analyse available resources and identify valuable states and channels.

Several directions for future work remain open: 
\begin{enumerate}
    \item We have only considered two-extendible channels; incorporation of $k$-extendible channels for $k > 2$ into our semi-definite optimization could offer tighter bounds on the measures we have described. The recent work of \cite{GO22} might be helpful for addressing this problem.
The notion of two-PPT-extendible channels is interesting in its own right via its connection with one-way LOCC channels.

\item It would also be interesting to find semi-definite constraints on one-way LOCC and LOCR channels, beyond those presented here, which include $k$-extendibility, PPT, and non-signaling. 

\item One could also try to find semi-definite tightenings of one-way LOCC and LOCR, which would lead to upper bounds on the simulation errors.

\item The paper \cite{SW21} shows that PPT constraints are sufficient to determine the exact simulation error in bidirectional teleportation for certain special states. Future work can identify a class of resource states that saturate the error bound using two-PPT-extendibility constraints, e.g., states that are PPT but two-unextendible. Such a class of states can offer insight not only in the study of teleportation protocols, but also to entanglement of states and channels.
\end{enumerate}
 
\bigskip

\begin{acknowledgments}

We thank Mario Berta, Eric Chitambar, Arshag Danageozian, Felix Leditzky, and Aliza Siddiqui for helpful discussions. This material is based upon work supported by the National Science Foundation under award OAC-1852454 with additional support from the Center for Computation \& Technology at Louisiana State University. VS and MMW acknowledge support from the National Science Foundation under grant no.~1907615.
\end{acknowledgments}

\bibliographystyle{alpha}
\bibliography{Ref}

\newcommand{\etalchar}[1]{$^{#1}$}
\begin{thebibliography}{BBDM{\etalchar{+}}98}

\bibitem[AVDB18]{agrawal2018rewriting}
Akshay Agrawal, Robin Verschueren, Steven Diamond, and Stephen Boyd.
\newblock A rewriting system for convex optimization problems.
\newblock {\em Journal of Control and Decision}, 5(1):42--60, 2018.

\bibitem[BBC{\etalchar{+}}93]{PhysRevLett.70.1895}
Charles~H. Bennett, Gilles Brassard, Claude Cr\'epeau, Richard Jozsa, Asher
  Peres, and William~K. Wootters.
\newblock Teleporting an unknown quantum state via dual classical and
  {Einstein-Podolsky-Rosen} channels.
\newblock {\em Physical Review Letters}, 70(13):1895--1899, March 1993.

\bibitem[BBDM{\etalchar{+}}98]{PhysRevLett.80.1121}
D.~Boschi, S.~Branca, F.~De~Martini, L.~Hardy, and S.~Popescu.
\newblock Experimental realization of teleporting an unknown pure quantum state
  via dual classical and {Einstein-Podolsky-Rosen} channels.
\newblock {\em Physical Review Letters}, 80(6):1121--1125, February 1998.

\bibitem[BBFS21]{BBFS21}
Mario Berta, Francesco Borderi, Omar Fawzi, and Volkher Scholz.
\newblock Semidefinite programming hierarchies for constrained bilinear
  optimization.
\newblock {\em Mathematical Programming}, 194:781--829, 2021.
\newblock arXiv:1810.12197.

\bibitem[BBP{\etalchar{+}}96]{BBPSSW96EPP}
Charles~H. Bennett, Gilles Brassard, Sandu Popescu, Benjamin Schumacher,
  John~A. Smolin, and William~K. Wootters.
\newblock Purification of noisy entanglement and faithful teleportation via
  noisy channels.
\newblock {\em Physical Review Letters}, 76(5):722--725, January 1996.
\newblock arXiv:quant-ph/9511027.

\bibitem[BDSW96]{BDSW96}
Charles~H. Bennett, David~P. DiVincenzo, John~A. Smolin, and William~K.
  Wootters.
\newblock Mixed-state entanglement and quantum error correction.
\newblock {\em Physical Review A}, 54(5):3824--3851, November 1996.
\newblock arXiv:quant-ph/9604024.

\bibitem[BK98]{braunstein1998teleportation}
Samuel~L. Braunstein and H.~Jeff Kimble.
\newblock Teleportation of continuous quantum variables.
\newblock {\em Physical Review Letters}, 80(4):869, January 1998.

\bibitem[BPM{\etalchar{+}}97]{Bouwmeester1997}
Dik Bouwmeester, Jian-Wei Pan, Klaus Mattle, Manfred Eibl, Harald Weinfurter,
  and Anton Zeilinger.
\newblock Experimental quantum teleportation.
\newblock {\em Nature}, 390(6660):575--579, December 1997.
\newblock arXiv:1901.11004.

\bibitem[BW18]{BW17}
Mario Berta and Mark~M. Wilde.
\newblock Amortization does not enhance the max-{Rains} information of a
  quantum channel.
\newblock {\em New Journal of Physics}, 20(5):053044, May 2018.
\newblock arXiv:1709.00200.

\bibitem[CDP08a]{CDP08a}
Giulio Chiribella, Giacomo~Mauro D'Ariano, and Paolo Perinotti.
\newblock Memory effects in quantum channel discrimination.
\newblock {\em Physical Review Letters}, 101(18):180501, October 2008.
\newblock arXiv:0803.3237.

\bibitem[CDP08b]{CDP08}
Giulio Chiribella, Giacomo~Mauro D'Ariano, and Paolo Perinotti.
\newblock Transforming quantum operations: Quantum supermaps.
\newblock {\em Europhysics Letters}, 83(3):30004, August 2008.
\newblock arXiv:0804.0180.

\bibitem[CDP09]{CDP09}
Giulio Chiribella, Giacomo~Mauro D'Ariano, and Paolo Perinotti.
\newblock Theoretical framework for quantum networks.
\newblock {\em Physical Review A}, 80(2):022339, August 2009.
\newblock arXiv:0904.4483.

\bibitem[CdVGG20]{CDGG20}
Eric Chitambar, Julio~I. de~Vicente, Mark~W. Girard, and Gilad Gour.
\newblock Entanglement manipulation beyond local operations and classical
  communication.
\newblock {\em Journal of Mathematical Physics}, 61(4):042201, April 2020.
\newblock arXiv:1711.03835.

\bibitem[CG21]{Chessa2021Feb}
Stefano Chessa and Vittorio Giovannetti.
\newblock Quantum capacity analysis of multi-level amplitude damping channels.
\newblock {\em Communications Physics}, 4(22):1--12, February 2021.
\newblock arXiv:2008.00477.

\bibitem[DB16]{diamond2016cvxpy}
Steven Diamond and Stephen Boyd.
\newblock {CVXPY}: {A} {P}ython-embedded modeling language for convex
  optimization.
\newblock {\em Journal of Machine Learning Research}, 17(83):1--5, 2016.

\bibitem[DPS02]{DPS02}
Andrew~C. Doherty, Pablo~A. Parrilo, and Federico~M. Spedalieri.
\newblock Distinguishing separable and entangled states.
\newblock {\em Physical Review Letters}, 88(18):187904, April 2002.
\newblock arXiv:quant-ph/0112007.

\bibitem[DPS04]{DPS04}
Andrew~C. Doherty, Pablo~A. Parrilo, and Federico~M. Spedalieri.
\newblock Complete family of separability criteria.
\newblock {\em Physical Review A}, 69(2):022308, February 2004.
\newblock arXiv:quant-ph/0308032.

\bibitem[EW01]{EW01}
Tilo Eggeling and Reinhard~F. Werner.
\newblock Separability properties of tripartite states with
  $u\ensuremath{\otimes}u\ensuremath{\otimes}u$ symmetry.
\newblock {\em Physical Review A}, 63(4):042111, March 2001.
\newblock arXiv:quant-ph/0010096.

\bibitem[FSB{\etalchar{+}}98]{furusawa1998unconditional}
Akira Furusawa, Jens~Lykke S{\o}rensen, Samuel~L. Braunstein, Christopher~A.
  Fuchs, H.~Jeff Kimble, and Eugene~S. Polzik.
\newblock Unconditional quantum teleportation.
\newblock {\em Science}, 282(5389):706--709, October 1998.

\bibitem[FWTB20]{FWTB18}
Kun Fang, Xin Wang, Marco Tomamichel, and Mario Berta.
\newblock Quantum channel simulation and the channel's smooth max-information.
\newblock {\em IEEE Transactions on Information Theory}, 66(4):2129--2140,
  April 2020.
\newblock arXiv:1807.05354.

\bibitem[FWTD19]{FWTD19}
Kun Fang, Xin Wang, Marco Tomamichel, and Runyao Duan.
\newblock Non-asymptotic entanglement distillation.
\newblock {\em IEEE Transactions on Information Theory}, 65(10):6454--6465,
  October 2019.
\newblock arXiv:1706.06221.

\bibitem[Gha10]{Ghar10}
Sevag Gharibian.
\newblock Strong {NP}-hardness of the quantum separability problem.
\newblock {\em Quantum Information and Computation}, 10(3):343--360, March
  2010.
\newblock arXiv:0810.4507.

\bibitem[GLN05]{GLN04}
Alexei Gilchrist, Nathan~K. Langford, and Michael~A. Nielsen.
\newblock Distance measures to compare real and ideal quantum processes.
\newblock {\em Physical Review A}, 71(6):062310, June 2005.
\newblock arXiv:quant-ph/0408063.

\bibitem[GO22]{GO22}
Dmitry Grinko and Maris Ozols.
\newblock Linear programming with unitary-equivariant constraints, July 2022.
\newblock arXiv:2207.05713.

\bibitem[Gou19]{Gour18}
Gilad Gour.
\newblock Comparison of quantum channels with superchannels.
\newblock {\em IEEE Transactions on Information Theory}, 65(9):5880--5904,
  September 2019.
\newblock arXiv:1808.02607.

\bibitem[Gur04]{Gur04}
Leonid Gurvits.
\newblock Classical complexity and quantum entanglement.
\newblock {\em Journal of Computer and System Sciences}, 69(3):448--484, 2004.
\newblock arXiv:quant-ph/0303055.

\bibitem[Hay17]{H17}
Masahito Hayashi.
\newblock {\em Quantum Information Theory: Mathematical Foundation}.
\newblock Springer, second edition, 2017.

\bibitem[HH99]{Horodecki99}
{Micha\l{}} Horodecki and {Pawe\l{}} Horodecki.
\newblock Reduction criterion of separability and limits for a class of
  distillation protocols.
\newblock {\em Physical Review A}, 59(6):4206--4216, June 1999.
\newblock arXiv:quant-ph/9708015.

\bibitem[HHH99]{PhysRevA.60.1888}
Micha\l{} Horodecki, Pawe\l{} Horodecki, and Ryszard Horodecki.
\newblock General teleportation channel, singlet fraction, and
  quasidistillation.
\newblock {\em Physical Review A}, 60(3):1888--1898, September 1999.
\newblock arXiv:quant-ph/9807091.

\bibitem[Hol19]{H13book}
Alexander~S. Holevo.
\newblock {\em Quantum Systems, Channels, Information: A Mathematical
  Introduction}.
\newblock Walter de Gruyter, second edition, 2019.

\bibitem[JV13]{JV13}
Peter~D. Johnson and Lorenza Viola.
\newblock Compatible quantum correlations: Extension problems for {Werner} and
  isotropic states.
\newblock {\em Physical Review A}, 88(3):032323, September 2013.
\newblock arXiv:1305.1342.

\bibitem[KDWW19]{KDWW19}
Eneet Kaur, Siddhartha Das, Mark~M. Wilde, and Andreas Winter.
\newblock Extendibility limits the performance of quantum processors.
\newblock {\em Physical Review Letters}, 123(7):070502, August 2019.
\newblock arXiv:1803.10710.

\bibitem[KDWW21]{KDWW21}
Eneet Kaur, Siddhartha Das, Mark~M. Wilde, and Andreas Winter.
\newblock Resource theory of unextendibility and nonasymptotic quantum
  capacity.
\newblock {\em Physical Review A}, 104(2):022401, August 2021.
\newblock arXiv:2108.03137.

\bibitem[Kit97]{Kit97}
Alexei Kitaev.
\newblock Quantum computations: algorithms and error correction.
\newblock {\em Russian Mathematical Surveys}, 52(6):1191--1249, 1997.

\bibitem[KW17]{KW17}
Eneet Kaur and Mark~M. Wilde.
\newblock Amortized entanglement of a quantum channel and approximately
  teleportation-simulable channels.
\newblock {\em Journal of Physics A: Mathematical and Theoretical},
  51(3):035303, December 2017.
\newblock arXiv:1707.07721.

\bibitem[KW20]{KW20book}
Sumeet Khatri and Mark~M. Wilde.
\newblock {\em Principles of Quantum Communication Theory: A Modern Approach}.
\newblock November 2020.
\newblock arXiv:2011.04672v1.

\bibitem[KW21]{KW20}
Vishal Katariya and Mark~M. Wilde.
\newblock Geometric distinguishability measures limit quantum channel
  estimation and discrimination.
\newblock {\em Quantum Information Processing}, 20:78, February 2021.
\newblock arXiv:2004.10708.

\bibitem[LM15]{LM15}
Debbie Leung and William Matthews.
\newblock On the power of {PPT}-preserving and non-signalling codes.
\newblock {\em IEEE Transactions on Information Theory}, 61(8):4486--4499,
  August 2015.
\newblock arXiv:1406.7142.

\bibitem[MHS{\etalchar{+}}12]{Ma2012}
Xiao-Song Ma, Thomas Herbst, Thomas Scheidl, Daqing Wang, Sebastian
  Kropatschek, William Naylor, Bernhard Wittmann, Alexandra Mech, Johannes
  Kofler, Elena Anisimova, Vadim Makarov, Thomas Jennewein, Rupert Ursin, and
  Anton Zeilinger.
\newblock Quantum teleportation over 143 kilometres using active feed-forward.
\newblock {\em Nature}, 489(7415):269--273, September 2012.

\bibitem[Pop94]{PhysRevLett.72.797}
Sandu Popescu.
\newblock Bell's inequalities versus teleportation: What is nonlocality?
\newblock {\em Physical Review Letters}, 72(6):797--799, February 1994.

\bibitem[Rai99]{Rai99}
Eric~M. Rains.
\newblock Bound on distillable entanglement.
\newblock {\em Physical Review A}, 60(1):179--184, July 1999.
\newblock arXiv:quant-ph/9809082.

\bibitem[Rai01]{Rai01}
Eric~M. Rains.
\newblock A semidefinite program for distillable entanglement.
\newblock {\em IEEE Transactions on Information Theory}, 47(7):2921--2933,
  November 2001.
\newblock arXiv:quant-ph/0008047.

\bibitem[RHR{\etalchar{+}}04]{Riebe2004}
M.~Riebe, H.~H\"affner, C.~F. Roos, W.~H\"ansel, J.~Benhelm, G.~P.~T.
  Lancaster, T.~W. K\"orber, C.~Becher, F.~Schmidt-Kaler, D.~F.~V. James, and
  R.~Blatt.
\newblock Deterministic quantum teleportation with atoms.
\newblock {\em Nature}, 429(6993):734--737, June 2004.

\bibitem[RXY{\etalchar{+}}17]{Ren2017}
Ji-Gang Ren, Ping Xu, Hai-Lin Yong, Liang Zhang, Sheng-Kai Liao, Juan Yin,
  Wei-Yue Liu, Wen-Qi Cai, Meng Yang, Li~Li, Kui-Xing Yang, Xuan Han,
  Yong-Qiang Yao, Ji~Li, Hai-Yan Wu, Song Wan, Lei Liu, Ding-Quan Liu, Yao-Wu
  Kuang, Zhi-Ping He, Peng Shang, Cheng Guo, Ru-Hua Zheng, Kai Tian, Zhen-Cai
  Zhu, Nai-Le Liu, Chao-Yang Lu, Rong Shu, Yu-Ao Chen, Cheng-Zhi Peng, Jian-Yu
  Wang, and Jian-Wei Pan.
\newblock Ground-to-satellite quantum teleportation.
\newblock {\em Nature}, 549(7670):70--73, August 2017.

\bibitem[Ser17]{S17}
Alessio Serafini.
\newblock {\em Quantum Continuous Variables: A Primer of Theoretical Methods}.
\newblock CRC Press, 2017.

\bibitem[SKO{\etalchar{+}}06]{Sherson2006}
Jacob~F. Sherson, Hanna Krauter, Rasmus~K. Olsson, Brian Julsgaard, Klemens
  Hammerer, Ignacio Cirac, and Eugene~S. Polzik.
\newblock Quantum teleportation between light and matter.
\newblock {\em Nature}, 443(7111):557--560, October 2006.

\bibitem[SW02]{qip2002schu}
Benjamin Schumacher and Michael~D. Westmoreland.
\newblock Approximate quantum error correction.
\newblock {\em Quantum Information Processing}, 1(1/2):5--12, April 2002.
\newblock arXiv:quant-ph/0112106.

\bibitem[SW20]{SW21}
Aliza~U. Siddiqui and Mark~M. Wilde.
\newblock Quantifying the performance of bidirectional quantum teleportation.
\newblock October 2020.
\newblock arXiv:2010.07905v2.

\bibitem[TMF{\etalchar{+}}13]{takeda2013deterministic}
Shuntaro Takeda, Takahiro Mizuta, Maria Fuwa, Peter Van~Loock, and Akira
  Furusawa.
\newblock Deterministic quantum teleportation of photonic quantum bits by a
  hybrid technique.
\newblock {\em Nature}, 500(7462):315--318, 2013.

\bibitem[Uhl76]{Uhl76}
Armin Uhlmann.
\newblock The ``transition probability'' in the state space of a *-algebra.
\newblock {\em Reports on Mathematical Physics}, 9(2):273--279, April 1976.

\bibitem[UJA{\etalchar{+}}04]{Ursin2004}
Rupert Ursin, Thomas Jennewein, Markus Aspelmeyer, Rainer Kaltenbaek, Michael
  Lindenthal, Philip Walther, and Anton Zeilinger.
\newblock Quantum teleportation across the {D}anube.
\newblock {\em Nature}, 430(7002):849--849, August 2004.

\bibitem[Wat09]{Wat09}
John Watrous.
\newblock Semidefinite programs for completely bounded norms.
\newblock {\em Theory of Computing}, 5(11):217--238, November 2009.
\newblock arXiv:0901.4709.

\bibitem[Wat18]{Watrous2018}
John Watrous.
\newblock {\em The Theory of Quantum Information}.
\newblock Cambridge University Press, 2018.

\bibitem[WD16a]{WD16pra}
Xin Wang and Runyao Duan.
\newblock An improved semidefinite programming upper bound on distillable
  entanglement.
\newblock {\em Physical Review A}, 94(5):050301, November 2016.
\newblock arXiv:1601.07940.

\bibitem[WD16b]{WD16}
Xin Wang and Runyao Duan.
\newblock A semidefinite programming upper bound of quantum capacity.
\newblock {\em 2016 IEEE International Symposium on Information Theory (ISIT)},
  pages 1690--1694, July 2016.
\newblock arXiv:1601.06888.

\bibitem[Wer89]{W89}
Reinhard~F. Werner.
\newblock Quantum states with {Einstein-Podolsky-Rosen} correlations admitting
  a hidden-variable model.
\newblock {\em Physical Review A}, 40(8):4277--4281, October 1989.

\bibitem[Wil17]{W17}
Mark~M. Wilde.
\newblock {\em Quantum Information Theory}.
\newblock Cambridge University Press, second edition, 2017.
\newblock arXiv:1106.1445.

\bibitem[WXD18]{WXD18}
Xin Wang, Wei Xie, and Runyao Duan.
\newblock Semidefinite programming strong converse bounds for classical
  capacity.
\newblock {\em IEEE Transactions on Information Theory}, 64(1):640--653,
  January 2018.
\newblock arXiv:1610.06381.

\bibitem[YF17]{Yuan2017}
Haidong Yuan and Chi-Hang~Fred Fung.
\newblock Fidelity and {{Fisher}} information on quantum channels.
\newblock {\em New Journal of Physics}, 19(11):113039, November 2017.
\newblock arXiv:1506.00819.

\end{thebibliography}

\appendix

\section{Proof of Equation \eqref{eq:non-sig-implied-by-two-ext}}

\label{app:proof-2ext-non-sig-prop}

We provide a proof of \eqref{eq:non-sig-implied-by-two-ext} here. Consider that%
\begin{align}
& \operatorname{Tr}_{B_{1}^{\prime}}\circ\mathcal{N}_{AB_{1}\rightarrow
A^{\prime}B_{1}^{\prime}}\nonumber\\
& =\operatorname{Tr}_{B_{1}^{\prime}B_{2}}\circ(\mathcal{N}_{AB_{1}\rightarrow
A^{\prime}B_{1}^{\prime}}\otimes\mathcal{P}_{B_{2}}^{\pi}) \\
& =\operatorname{Tr}_{B_{1}^{\prime}B_{2}^{\prime}}\circ\mathcal{M}%
_{AB_{1}B_{2}\rightarrow A^{\prime}B_{1}^{\prime}B_{2}^{\prime}}%
\circ\mathcal{P}_{B_{2}}^{\pi}\\
& =\operatorname{Tr}_{B_{1}^{\prime}B_{2}^{\prime}}\circ\mathcal{F}%
_{B_{1}^{\prime}B_{2}^{\prime}}\circ\mathcal{M}_{AB_{1}B_{2}\rightarrow
A^{\prime}B_{1}^{\prime}B_{2}^{\prime}}\circ\mathcal{F}_{B_{1}B_{2}}%
\circ\mathcal{P}_{B_{2}}^{\pi}\\
& =\operatorname{Tr}_{B_{1}^{\prime}B_{2}^{\prime}}\circ\mathcal{M}%
_{AB_{1}B_{2}\rightarrow A^{\prime}B_{1}^{\prime}B_{2}^{\prime}}%
\circ\mathcal{P}_{B_{1}}^{\pi}\circ\operatorname{id}_{B_{1}\rightarrow B_{2}%
}\\
& =\operatorname{Tr}_{B_{1}^{\prime}}\circ\mathcal{N}_{AB_{1}\rightarrow
A^{\prime}B_{1}^{\prime}}\circ\mathcal{P}_{B_{1}}^{\pi}\circ\operatorname{Tr}%
_{B_{2}}\circ\operatorname{id}_{B_{1}\rightarrow B_{2}}\\
& =\operatorname{Tr}_{B_{1}^{\prime}}\circ\mathcal{N}_{AB_{1}\rightarrow
A^{\prime}B_{1}^{\prime}}\circ\mathcal{R}_{B_{1}}^{\pi}.
\end{align}
The first equality follows because $\mathcal{P}_{B_{2}}^{\pi}$ is a
preparation channel that prepares the maximally mixed state $\pi_{B_{2}}$ on
system $B_{2}$, and then we trace it out. The second equality follows by using
the non-signaling property in \eqref{eq:marg-ch-ch}. The third equality follows from
permutation covariance of the channel $\mathcal{M}_{AB_{1}B_{2}\rightarrow
A^{\prime}B_{1}^{\prime}B_{2}^{\prime}}$ (i.e., the assumption that \eqref{eq:perm-cov-ch} holds). The fourth equality follows because
$\mathcal{F}_{B_{1}^{\prime}B_{2}^{\prime}}$ is\ a unitary channel, so that
\begin{equation}
\operatorname{Tr}_{B_{1}^{\prime}B_{2}^{\prime}}\circ\mathcal{F}%
_{B_{1}^{\prime}B_{2}^{\prime}}=\operatorname{Tr}_{B_{1}^{\prime}B_{2}%
^{\prime}}.    
\end{equation}
Additionally, we used the fact that 
\begin{equation}
\mathcal{F}_{B_{1}B_{2}%
}\circ\mathcal{P}_{B_{2}}^{\pi}=\mathcal{P}_{B_{1}}^{\pi}\circ
\operatorname{id}_{B_{1}\rightarrow B_{2}},    
\end{equation}
 where $\operatorname{id}%
_{B_{1}\rightarrow B_{2}}$ is an identity channel that transforms system
$B_{1}$ to $B_{2}$. The fifth equality again invokes the non-signaling
property in \eqref{eq:marg-ch-ch}. The last equality follows because 
\begin{equation}
\mathcal{P}_{B_{1}}%
^{\pi}\circ\operatorname{Tr}_{B_{2}}\circ\operatorname{id}_{B_{1}\rightarrow
B_{2}}=\mathcal{R}_{B_{1}}^{\pi}.    
\end{equation}
That is, $\operatorname{Tr}_{B_{2}}%
\circ\operatorname{id}_{B_{1}\rightarrow B_{2}}$ is equivalent to
$\operatorname{Tr}_{B_{1}}$, so that this action combined with $\mathcal{P}_{B_{1}%
}^{\pi}$ realizes a replacer channel.

\section{Proof of Proposition~\ref{prop:sim-errs-equal-1WL}}

\label{app:proof-errors-equal}

The proof given here is similar to that of the proof of Proposition~1 in
\cite{SW21}. We include it for completeness.

\subsection{Exploiting unitary covariance symmetry of the identity channel}

\label{app:exploiting-sym}

The main idea behind this proof is to simplify the optimization problems in
\eqref{eq:diamond-err-sim-tele} and \eqref{eq:fid-err-sim-tele} by exploiting
the symmetries of the identity channel, as stated in
\eqref{eq:identity-ch-symmetries}. To begin, let us note that the Choi
operator of the identity channel is simply $\Gamma_{AB}$, where%
\begin{equation}
\Gamma_{AB}\coloneqq\sum_{i,j}|i\rangle\!\langle j|_{A}\otimes|i\rangle
\!\langle j|_{B},
\end{equation}
and $\{|i\rangle_{A}\}_{i}$ and $\{|i\rangle_{B}\}_{i}$ are orthonormal bases.
Defining the unitary channel $\mathcal{U}(\cdot)=U(\cdot)U^{\dag}$, let us
recall from \eqref{eq:identity-ch-symmetries} the following covariance
property of the identity channel:%
\begin{equation}
\operatorname{id}_{A\rightarrow B}^{d}=\mathcal{U}_{B}^{\dag}\circ
\operatorname{id}_{A\rightarrow B}^{d}\circ\mathcal{U}_{A},
\label{eq-app:cov-prop-id}%
\end{equation}
which holds for every unitary channel $\mathcal{U}$. Let $\mathcal{A}_{\hat
{A}\hat{B}}^{\rho}$ denote the channel that appends the bipartite state
$\rho_{\hat{A}\hat{B}}$ to its input:%
\begin{equation}
\mathcal{A}_{\hat{A}\hat{B}}^{\rho}(\omega_{A})=\omega_{A}\otimes\rho_{\hat
{A}\hat{B}}.
\end{equation}

Let us begin our analysis with the diamond distance, but we note here that all
of the reasoning employed in the first part of our proof applies to the
channel infidelity error measure as well. Let $\mathcal{L}_{A\hat{A}\hat
{B}\rightarrow B}$ be an arbitrary one-way LOCC\ channel to consider for the
optimization problem in \eqref{eq:diamond-err-sim-tele}. Exploiting the
unitary invariance of the diamond distance with respect to input and output
unitaries \cite[Proposition~3.44]{Watrous2018}, we find that%
\begin{align}
&  \left\Vert \mathcal{L}_{A\hat{A}\hat{B}\rightarrow B}\circ\mathcal{A}%
_{\hat{A}\hat{B}}^{\rho}-\operatorname{id}_{A\rightarrow B}^{d}\right\Vert
_{\diamond}\nonumber\\
&  =\left\Vert \mathcal{U}_{B}^{\dag}\circ\mathcal{L}_{A\hat{A}\hat
{B}\rightarrow B}\circ\mathcal{A}_{\hat{A}\hat{B}}^{\rho}\circ\mathcal{U}%
_{A}-\mathcal{U}_{B}^{\dag}\circ\operatorname{id}_{A\rightarrow B}^{d}%
\circ\mathcal{U}_{A}\right\Vert _{\diamond}\nonumber\\
&  =\left\Vert \mathcal{U}_{B}^{\dag}\circ\mathcal{L}_{A\hat{A}\hat
{B}\rightarrow B}\circ\mathcal{U}_{A}\circ\mathcal{A}_{\hat{A}\hat{B}}^{\rho
}-\operatorname{id}_{A\rightarrow B}^{d}\right\Vert _{\diamond}.
\label{eq:-diamond-d-unitary-equiv}%
\end{align}
Thus, the channels $\mathcal{L}_{A\hat{A}\hat{B}\rightarrow B}$ and
$\mathcal{U}_{B}^{\dag}\circ\mathcal{L}_{A\hat{A}\hat{B}\rightarrow B}%
\circ\mathcal{U}_{A}$ perform equally well for the optimization. Now let us
exploit the convexity of diamond distance with respect to one of the channels
\cite{Watrous2018}, as well as the Haar probability measure, to conclude that%
\begin{align}
&  \left\Vert \mathcal{L}_{A\hat{A}\hat{B}\rightarrow B}\circ\mathcal{A}%
_{\hat{A}\hat{B}}^{\rho}-\operatorname{id}_{A\rightarrow B}^{d}\right\Vert
_{\diamond}\nonumber\\
&  =\int dU\ \left\Vert \mathcal{U}_{B}^{\dag}\circ\mathcal{L}_{A\hat{A}%
\hat{B}\rightarrow B}\circ\mathcal{U}_{A}\circ\mathcal{A}_{\hat{A}\hat{B}%
}^{\rho}-\operatorname{id}_{A\rightarrow B}^{d}\right\Vert _{\diamond}
\label{eq:twirl-optimal-best-1}\\
&  \geq\left\Vert \widetilde{\mathcal{L}}_{A\hat{A}\hat{B}\rightarrow B}%
\circ\mathcal{A}_{\hat{A}\hat{B}}^{\rho}-\operatorname{id}_{A\rightarrow
B}^{d}\right\Vert _{\diamond}, \label{eq:convexity-dd-1wl}%
\end{align}
where%
\begin{equation}
\widetilde{\mathcal{L}}_{A\hat{A}\hat{B}\rightarrow B}\coloneqq\int
dU\ \mathcal{U}_{B}^{\dag}\circ\mathcal{L}_{A\hat{A}\hat{B}\rightarrow B}%
\circ\mathcal{U}_{A}. \label{eq:twirled-channel-1W-LOCC}%
\end{equation}
Thus, we conclude that it suffices to optimize \eqref{eq:diamond-err-sim-tele}
over one-way LOCC\ channels that possess this symmetry. It is important to
note that the channel twirl in \eqref{eq:twirled-channel-1W-LOCC} can be
realized by one-way LOCC, so that $\widetilde{\mathcal{L}}_{A\hat{A}\hat
{B}\rightarrow B}$ is a one-way LOCC channel.

Let us determine the form of one-way LOCC\ channels that possess this
symmetry. Let $L_{A\hat{A}\hat{B}B}$ denote the Choi operator of the channel
$\mathcal{L}_{A\hat{A}\hat{B}\rightarrow B}$. Then the Choi operator
$\widetilde{L}_{A\hat{A}\hat{B}B}$ of the twirled channel $\widetilde
{\mathcal{L}}_{A\hat{A}\hat{B}\rightarrow B}$ is as follows:%
\begin{equation}
\widetilde{L}_{A\hat{A}\hat{B}B}=\int dU\ (\mathcal{U}_{A}\otimes
\overline{\mathcal{U}}_{B})(L_{A\hat{A}\hat{B}B}),
\end{equation}
where $\overline{\mathcal{U}}(\cdot)\coloneqq\overline{U}(\cdot)\overline
{U}^{\dag}$, with the overbar denoting the entrywise complex conjugate. Now
let us recall the following identity from \cite{W89,Horodecki99,Watrous2018}:%
\begin{align}
&  \widetilde{\mathcal{T}}_{CD}(X_{CD})\nonumber\\
&  \coloneqq\int dU\ (\mathcal{U}_{C}\otimes\overline{\mathcal{U}}_{D}%
)(X_{CD})\\
&  =\Phi_{CD}\operatorname{Tr}_{CD}[\Phi_{CD}X_{CD}]\nonumber\\
&  \qquad+\frac{I_{CD}-\Phi_{CD}}{d^{2}-1}\operatorname{Tr}_{CD}[\left(
I_{CD}-\Phi_{CD}\right)  X_{CD}].
\label{eq:twirl-ch-def}
\end{align}
We can apply this identity to find that%
\begin{align}
\widetilde{L}_{A\hat{A}\hat{B}B}  &  =\widetilde{\mathcal{T}}_{AB}(L_{A\hat
{A}\hat{B}B})\\
&  =\Phi_{AB}\operatorname{Tr}_{AB}[\Phi_{AB}L_{A\hat{A}\hat{B}B}]\nonumber\\
&  \qquad+\frac{I_{AB}-\Phi_{AB}}{d^{2}-1}\operatorname{Tr}_{AB}[\left(
I_{AB}-\Phi_{AB}\right)  L_{A\hat{A}\hat{B}B}].
\end{align}
Now defining%
\begin{align}
K_{\hat{A}\hat{B}}^{\prime}  &  \coloneqq\operatorname{Tr}_{AB}[\Phi
_{AB}L_{A\hat{A}\hat{B}B}],\\
L_{\hat{A}\hat{B}}^{\prime}  &  \coloneqq\operatorname{Tr}_{AB}[\left(
I_{AB}-\Phi_{AB}\right)  L_{A\hat{A}\hat{B}B}],
\end{align}
we can write%
\begin{equation}
\widetilde{L}_{A\hat{A}\hat{B}B}=\Phi_{AB}\otimes K_{\hat{A}\hat{B}}^{\prime
}+\frac{I_{AB}-\Phi_{AB}}{d^{2}-1}\otimes L_{\hat{A}\hat{B}}^{\prime}.
\label{eq:Kprime-Lprime-channel}%
\end{equation}
Let us determine the conditions on $K_{\hat{A}\hat{B}}^{\prime}$ and
$L_{\hat{A}\hat{B}}^{\prime}$ for $\widetilde{L}_{A\hat{A}\hat{B}B}$ to be the
Choi operator of a channel. Consider that $\widetilde{L}_{A\hat{A}\hat{B}B}$
is a Choi operator if and only if
\begin{align}
\widetilde{L}_{A\hat{A}\hat{B}B}  &  \geq0,\label{eq:choi-cond-01}\\
\operatorname{Tr}_{B}[\widetilde{L}_{A\hat{A}\hat{B}B}]  &  =I_{A\hat{A}%
\hat{B}}. \label{eq:choi-cond-02}%
\end{align}
This implies that%
\begin{equation}
K_{\hat{A}\hat{B}}^{\prime},\ L_{\hat{A}\hat{B}}^{\prime}\geq0
\label{eq:psd-KL-prime}%
\end{equation}
and%
\begin{equation}
\pi_{A}\otimes K_{\hat{A}\hat{B}}^{\prime}+\pi_{A}\otimes L_{\hat{A}\hat{B}%
}^{\prime}=I_{A\hat{A}\hat{B}},
\end{equation}
which is equivalent to the following condition:%
\begin{equation}
K_{\hat{A}\hat{B}}^{\prime}+L_{\hat{A}\hat{B}}^{\prime}=dI_{\hat{A}\hat{B}}.
\label{eq:dim-scaled-completeness-KL-cond}%
\end{equation}
Let us define%
\begin{equation}
K_{\hat{A}\hat{B}}\coloneqq\frac{1}{d}K_{\hat{A}\hat{B}}^{\prime},\qquad
L_{\hat{A}\hat{B}}\coloneqq\frac{1}{d}L_{\hat{A}\hat{B}}^{\prime}.
\label{eq:rescale-K-L}%
\end{equation}
Then the conditions in \eqref{eq:psd-KL-prime} and
\eqref{eq:dim-scaled-completeness-KL-cond} are equivalent to%
\begin{align}
K_{\hat{A}\hat{B}},\ L_{\hat{A}\hat{B}}  &  \geq0,\label{eq:K-L-conds-1}\\
K_{\hat{A}\hat{B}}+L_{\hat{A}\hat{B}}  &  =I_{\hat{A}\hat{B}},
\label{eq:K-L-conds-2}%
\end{align}
and we can write the Choi operator $\widetilde{L}_{A\hat{A}\hat{B}B}$ as%
\begin{equation}
\widetilde{L}_{A\hat{A}\hat{B}B}=\Gamma_{AB}\otimes K_{\hat{A}\hat{B}}%
+\frac{dI_{AB}-\Gamma_{AB}}{d^{2}-1}\otimes L_{\hat{A}\hat{B}}.
\end{equation}
The Choi operator of the composed channel $\widetilde{\mathcal{L}}_{A\hat
{A}\hat{B}\rightarrow B}\circ\mathcal{A}_{\hat{A}\hat{B}}^{\rho}$ is given by%
\begin{align}
&  \operatorname{Tr}_{\hat{A}\hat{B}}[T_{\hat{A}\hat{B}}(\rho_{\hat{A}\hat{B}%
})\widetilde{L}_{A\hat{A}\hat{B}B}]\nonumber\\
&  =\Gamma_{AB}\operatorname{Tr}[T_{\hat{A}\hat{B}}(\rho_{\hat{A}\hat{B}%
})K_{\hat{A}\hat{B}}]\nonumber\\
&  \qquad+\frac{dI_{AB}-\Gamma_{AB}}{d^{2}-1}\operatorname{Tr}[T_{\hat{A}%
\hat{B}}(\rho_{\hat{A}\hat{B}})L_{\hat{A}\hat{B}}]\\
&  =\Gamma_{AB}\operatorname{Tr}[\rho_{\hat{A}\hat{B}}T_{\hat{A}\hat{B}%
}(K_{\hat{A}\hat{B}})]\nonumber\\
&  \qquad+\frac{dI_{AB}-\Gamma_{AB}}{d^{2}-1}\operatorname{Tr}[\rho_{\hat
{A}\hat{B}}T_{\hat{A}\hat{B}}(L_{\hat{A}\hat{B}})].
\end{align}
We can make the substitutions $T_{\hat{A}\hat{B}}(K_{\hat{A}\hat{B}%
})\rightarrow K_{\hat{A}\hat{B}}$ and $T_{\hat{A}\hat{B}}(L_{\hat{A}\hat{B}%
})\rightarrow L_{\hat{A}\hat{B}}$ without changing the optimization problem
because the conditions in \eqref{eq:K-L-conds-1} and \eqref{eq:K-L-conds-2}
hold if and only if the following conditions hold%
\begin{align}
T_{\hat{A}\hat{B}}(K_{\hat{A}\hat{B}}),\ T_{\hat{A}\hat{B}}(L_{\hat{A}\hat{B}%
})  &  \geq0,\\
T_{\hat{A}\hat{B}}(K_{\hat{A}\hat{B}})+T_{\hat{A}\hat{B}}(L_{\hat{A}\hat{B}})
&  =I_{\hat{A}\hat{B}}.
\end{align}
Additionally, the channel $\widetilde{\mathcal{L}}_{A\hat{A}\hat{B}\rightarrow
B}$ remains a one-way LOCC\ channel under these substitutions. Thus, we can
take the Choi operator of the composed channel $\widetilde{\mathcal{L}}%
_{A\hat{A}\hat{B}\rightarrow B}\circ\mathcal{A}_{\hat{A}\hat{B}}^{\rho}$ to be%
\begin{equation}
\Gamma_{AB}\operatorname{Tr}[K_{\hat{A}\hat{B}}\rho_{\hat{A}\hat{B}}%
]+\frac{dI_{AB}-\Gamma_{AB}}{d^{2}-1}\operatorname{Tr}[L_{\hat{A}\hat{B}}%
\rho_{\hat{A}\hat{B}}],
\end{equation}
which corresponds to the following channel:%
\begin{multline}
\sigma\rightarrow\operatorname{Tr}[\rho_{\hat{A}\hat{B}}K_{\hat{A}\hat{B}%
}]\operatorname{id}_{A\rightarrow B}^{d}(\sigma
)\label{eq:final-channel-LOCC-opt}\\
+\operatorname{Tr}[\rho_{\hat{A}\hat{B}}L_{\hat{A}\hat{B}}]\frac{1}{d^{2}%
-1}\sum_{(z,x)\neq(0,0)}W^{z,x}\sigma(W^{z,x})^{\dag},
\end{multline}
which follows because
\begin{equation}
dI_{AB}-\Gamma_{AB} = \sum_{(z,x)\neq(0,0)}W^{z,x}_B\Gamma_{AB}(W^{z,x})^{\dag}_B.    
\end{equation}

Let us now incorporate the one-way LOCC constraint and justify the claim in
\eqref{eq:one-way-LOCC-K-form}. The general form of a one-way LOCC channel
$\mathcal{L}_{A\hat{A}\hat{B}\rightarrow B}$, with forward classical
communication from Alice to Bob, is as follows:%
\begin{equation}
\mathcal{L}_{A\hat{A}\hat{B}\rightarrow B}(\omega_{A\hat{A}\hat{B}})=\sum
_{x}\mathcal{D}_{\hat{B}\rightarrow B}^{x}(\operatorname{Tr}_{A\hat{A}%
}[\Lambda_{A\hat{A}}^{x}\omega_{A\hat{A}\hat{B}}]),
\end{equation}
where $\{\Lambda_{A\hat{A}}^{x}\}_{x}$ is a POVM, satisfying $\Lambda
_{A\hat{A}}^{x}\geq0$ for all $x$ and $\sum_{x}\Lambda_{A\hat{A}}^{x}%
=I_{A\hat{A}}$, and $\{\mathcal{D}_{\hat{B}\rightarrow B}^{x}\}_{x}$ is a set
of quantum channels. The Choi operator of such a channel has the form%
\begin{align}
&  \Gamma_{A\hat{A}\hat{B}B}^{\mathcal{L}}\nonumber\\
&  =\mathcal{L}_{A\hat{A}\hat{B}\rightarrow B}(\Gamma_{AA}\otimes\Gamma
_{\hat{A}\hat{A}}\otimes\Gamma_{\hat{B}\hat{B}})\\
&  =\sum_{x}\mathcal{D}_{\hat{B}\rightarrow B}^{x}(\operatorname{Tr}_{A\hat
{A}}[\Lambda_{A\hat{A}}^{x}\left(  \Gamma_{AA}\otimes\Gamma_{\hat{A}\hat{A}%
}\otimes\Gamma_{\hat{B}\hat{B}}\right)  ])\\
&  =\sum_{x}T_{A\hat{A}}(\Lambda_{A\hat{A}}^{x})\otimes\mathcal{D}_{\hat
{B}\rightarrow B}^{x}(\Gamma_{\hat{B}\hat{B}})\\
&  =\sum_{x}T_{A\hat{A}}(\Lambda_{A\hat{A}}^{x})\otimes\Gamma_{\hat{B}%
B}^{\mathcal{D}^{x}},
\end{align}
where $\Gamma_{\hat{B}B}^{\mathcal{D}^{x}}$ is the Choi operator of the
channel $\mathcal{D}_{\hat{B}\rightarrow B}^{x}$. Since the conditions for
being a POVM\ are invariant under a full transpose and since we are performing
an optimization over all one-way LOCC channels, we can take the Choi operator
to be as follows without loss of generality:%
\begin{equation}
\Gamma_{A\hat{A}\hat{B}B}^{\mathcal{L}}=\sum_{x}\Lambda_{A\hat{A}}^{x}%
\otimes\Gamma_{\hat{B}B}^{\mathcal{D}^{x}}.
\end{equation}
After performing a bilateral twirl of the systems $A$ and $B$, the Choi
operator becomes%
\begin{multline}
\widetilde{\mathcal{T}}_{AB}\!\left(  \sum_{x}\Lambda_{A\hat{A}}^{x}%
\otimes\Gamma_{\hat{B}B}^{\mathcal{D}^{x}}\right) \\
=\Phi_{AB}\operatorname{Tr}_{AB}\!\left[  \Phi_{AB}\sum_{x}\Lambda_{A\hat{A}%
}^{x}\otimes\Gamma_{\hat{B}B}^{\mathcal{D}^{x}}\right] \\
+\frac{I_{AB}-\Phi_{AB}}{d^{2}-1}\operatorname{Tr}_{AB}\!\left[  \left(
I_{AB}-\Phi_{AB}\right)  \sum_{x}\Lambda_{A\hat{A}}^{x}\otimes\Gamma_{\hat
{B}B}^{\mathcal{D}^{x}}\right]  .
\end{multline}
Now consider that%
\begin{multline}
\operatorname{Tr}_{AB}\!\left[  \Phi_{AB}\sum_{x}\Lambda_{A\hat{A}}^{x}%
\otimes\Gamma_{\hat{B}B}^{\mathcal{D}^{x}}\right] \\
=\frac{1}{d}\sum_{x}\operatorname{Tr}_{B}\!\left[  \Gamma_{\hat{B}%
B}^{\mathcal{D}^{x}}T_{B}(\Lambda_{B\hat{A}}^{x})\right]  ,
\end{multline}%
\begin{align}
&  \operatorname{Tr}_{AB}\!\left[  I_{AB}\sum_{x}\Lambda_{A\hat{A}}^{x}%
\otimes\Gamma_{\hat{B}B}^{\mathcal{D}^{x}}\right] \nonumber\\
&  =\operatorname{Tr}_{AB}\!\left[  \sum_{x}\Lambda_{A\hat{A}}^{x}%
\otimes\Gamma_{\hat{B}B}^{\mathcal{D}^{x}}\right] \\
&  =\operatorname{Tr}_{A}\!\left[  \sum_{x}\Lambda_{A\hat{A}}^{x}\otimes
I_{\hat{B}}\right] \\
&  =\operatorname{Tr}_{A}\!\left[  I_{A\hat{A}}\otimes I_{\hat{B}}\right] \\
&  =d\,I_{\hat{A}\hat{B}}.
\end{align}
Revisiting \eqref{eq:Kprime-Lprime-channel}, this implies that we can take%
\begin{align}
K_{\hat{A}\hat{B}}^{\prime}  &  =\frac{1}{d}\sum_{x}\operatorname{Tr}%
_{B}\!\left[  \Gamma_{\hat{B}B}^{\mathcal{D}^{x}}T_{B}(\Lambda_{B\hat{A}}%
^{x})\right]  ,\\
L_{\hat{A}\hat{B}}^{\prime}  &  =d\,I_{\hat{A}\hat{B}}-K_{\hat{A}\hat{B}%
}^{\prime}.
\end{align}
After a rescaling as in \eqref{eq:rescale-K-L}, we set%
\begin{equation}
K_{\hat{A}\hat{B}}=\frac{1}{d^{2}}\sum_{x}\operatorname{Tr}_{B}\!\left[
\Gamma_{\hat{B}B}^{\mathcal{D}^{x}}T_{B}(\Lambda_{B\hat{A}}^{x})\right]  ,
\end{equation}
proceed through the rest of the steps given in
\eqref{eq:K-L-conds-1}--\eqref{eq:final-channel-LOCC-opt}, and arrive at the
claim in \eqref{eq:one-way-LOCC-K-form}.

\subsection{Evaluating normalized diamond distance}

With the reasoning from the previous section, we have reduced the optimization
problem in \eqref{eq:diamond-err-sim-tele}, for simulating the identity
channel, to the following:%
\begin{multline}
e_{\text{1WL}}(\rho_{\hat{A}\hat{B}})=\\
\inf_{\widetilde{\mathcal{L}}_{A\hat{A}\hat{B}\rightarrow B}\in\text{1WL}%
}\frac{1}{2}\left\Vert \widetilde{\mathcal{L}}_{A\hat{A}\hat{B}\rightarrow
B}\circ\mathcal{A}_{\hat{A}\hat{B}}^{\rho}-\operatorname{id}_{A\rightarrow
B}^{d}\right\Vert _{\diamond},
\end{multline}
subject to%
\begin{multline}
\widetilde{\mathcal{L}}_{A\hat{A}\hat{B}\rightarrow B}(\sigma_{A}\otimes
\rho_{\hat{A}\hat{B}})=\operatorname{Tr}[\rho_{\hat{A}\hat{B}}K_{\hat{A}%
\hat{B}}]\operatorname{id}_{A\rightarrow B}^{d}(\sigma_{A})\\
+\operatorname{Tr}[\rho_{\hat{A}\hat{B}}L_{\hat{A}\hat{B}}]\frac{1}{d^{2}%
-1}\sum_{(z,x)\neq(0,0)}W^{z,x}\sigma_{A}(W^{z,x})^{\dag},
\label{eq:restrict-LOCC-opt}%
\end{multline}
and there existing a POVM\ $\{\Lambda_{B\hat{A}}^{x}\}_{x}$ and a set
$\{\mathcal{D}_{\hat{B}\rightarrow B}^{x}\}_{x}$\ of channels such that%
\begin{equation}
K_{\hat{A}\hat{B}}=\frac{1}{d^{2}}\sum_{x}\operatorname{Tr}_{B}[T_{B}%
(\Lambda_{B\hat{A}}^{x})\Gamma_{\hat{B}B}^{\mathcal{D}^{x}}],
\label{eq:1wl-form-K-op}%
\end{equation}
where $\Gamma_{\hat{B}B}^{\mathcal{D}^{x}}$ is the Choi operator of the
channel $\mathcal{D}_{\hat{B}\rightarrow B}^{x}$. That is, there is no need to
optimize over all one-way LOCC channels, but only those satisfying the
constraints in \eqref{eq:restrict-LOCC-opt} and \eqref{eq:1wl-form-K-op}.

We now exploit the form of the optimization of the normalized diamond distance
from \eqref{eq:diamond-d-SDP}, in order to rewrite the optimization problem as%
\begin{equation}
\inf_{\mu,Z_{AB},K_{\hat{A}\hat{B}},L_{\hat{A}\hat{B}}\geq0}\mu,
\end{equation}
subject to%
\begin{align}
\mu I_{A}  &  \geq Z_{A},\\
Z_{AB}  &  \geq\Gamma_{AB}\left(  1-\operatorname{Tr}[K_{\hat{A}\hat{B}}%
\rho_{\hat{A}\hat{B}}]\right) \nonumber\\
&  \qquad-\frac{dI_{AB}-\Gamma_{AB}}{d^{2}-1}\operatorname{Tr}[L_{\hat{A}%
\hat{B}}\rho_{\hat{A}\hat{B}}],\label{eq:Z-op-ineq}\\
K_{\hat{A}\hat{B}}+L_{\hat{A}\hat{B}}  &  =I_{\hat{A}\hat{B}},
\end{align}
with $K_{\hat{A}\hat{B}}$ further subject to having the form in
\eqref{eq:1wl-form-K-op}. Since we are minimizing $\mu$ with respect to $\mu$
and $Z_{AB}$, we can choose $Z_{AB}$ to be the smallest positive semi-definite
operator such that the operator inequality in \eqref{eq:Z-op-ineq} holds. This
smallest positive semi-definite operator is the positive part of the operator
on the right-hand side of the inequality, which is given by%
\begin{equation}
\Gamma_{AB}\left(  1-\operatorname{Tr}[K_{\hat{A}\hat{B}}\rho_{\hat{A}\hat{B}%
}]\right)  .
\end{equation}
This follows because $\Gamma_{AB}$ and $\frac{dI_{AB}-\Gamma_{AB}}{d^{2}-1}$
are each positive semi-definite and orthogonal to each other. Thus, an optimal
solution is%
\begin{equation}
Z_{AB}=\Gamma_{AB}\left(  1-\operatorname{Tr}[K_{\hat{A}\hat{B}}\rho_{\hat
{A}\hat{B}}]\right)  ,
\end{equation}
for which the smallest $\mu$ possible is%
\begin{equation}
\mu=1-\operatorname{Tr}[K_{\hat{A}\hat{B}}\rho_{\hat{A}\hat{B}}],
\label{eq:mu-opt-choice}%
\end{equation}
because%
\begin{align}
Z_{A}  &  =\operatorname{Tr}_{B}[Z_{AB}]\\
&  =\operatorname{Tr}_{B}[\Gamma_{AB}\left(  1-\operatorname{Tr}[K_{\hat
{A}\hat{B}}\rho_{\hat{A}\hat{B}}]\right)  ]\\
&  =I_{A}\left(  1-\operatorname{Tr}[K_{\hat{A}\hat{B}}\rho_{\hat{A}\hat{B}%
}]\right)  .
\end{align}
We then conclude that%
\begin{equation}
e_{\text{$\operatorname{1WL}$}}(\rho_{\hat{A}\hat{B}})=1-\sup_{K_{\hat{A}%
\hat{B}},L_{\hat{A}\hat{B}}\geq0}\operatorname{Tr}[K_{\hat{A}\hat{B}}%
\rho_{\hat{A}\hat{B}}],
\end{equation}
subject to $K_{\hat{A}\hat{B}}+L_{\hat{A}\hat{B}}=I_{\hat{A}\hat{B}}$ and the
following channel $\mathcal{L}_{A\hat{A}\hat{B}\rightarrow B}$ being a one-way
LOCC\ channel:%
\begin{multline}
\mathcal{L}_{A\hat{A}\hat{B}\rightarrow B}(\omega_{A\hat{A}\hat{B}%
})=\operatorname{id}_{A\rightarrow B}^{d}(\operatorname{Tr}_{\hat{A}\hat{B}%
}[K_{\hat{A}\hat{B}}\omega_{A\hat{A}\hat{B}}])\\
+\mathcal{D}_{A\rightarrow B}(\operatorname{Tr}_{\hat{A}\hat{B}}[L_{\hat
{A}\hat{B}}\omega_{A\hat{A}\hat{B}}]).
\end{multline}
As indicated previously, this latter constraint is equivalent to \eqref{eq:1wl-form-K-op}.

\subsection{Evaluating channel infidelity}

Let us recall the symmetries of the identity channel in
\eqref{eq-app:cov-prop-id}, which implies the following symmetry for its Choi
operator:%
\begin{equation}
\Gamma_{AB}=\left(  \overline{\mathcal{U}}_{A}\otimes\mathcal{U}_{B}\right)
(\Gamma_{AB}),
\end{equation}
for every unitary channel $\mathcal{U}(\cdot)=U(\cdot)U^{\dag}$. This implies
that%
\begin{equation}
\Gamma_{AB}=\int dU\ \left(  \overline{\mathcal{U}}_{A}\otimes\mathcal{U}%
_{B}\right)  (\Gamma_{AB}).
\end{equation}
Now applying the semi-definite program in \eqref{eq:SDP-ch-fid-1} for channel
infidelity, we find that%
\begin{equation}
e_{\text{1WL}}^{F}(\rho_{\hat{A}\hat{B}})=1-\left[  \sup_{\lambda
\geq0,L_{A\hat{A}\hat{B}B}\geq0,Q_{AB}}\lambda\right]  ^{2},
\label{eq-app:ch-infid-opt}%
\end{equation}
subject to%
\begin{align}
\lambda I_{A}  &  \leq\operatorname{Re}[\operatorname{Tr}_{B}[Q_{AB}%
]],\label{eq:lam-Q-constraint-ch-inf}\\
\operatorname{Tr}_{B}[L_{A\hat{A}\hat{B}B}]  &  =I_{A\hat{A}\hat{B}},
\end{align}%
\begin{equation}%
\begin{bmatrix}
\Gamma_{AB} & Q_{AB}^{\dag}\\
Q_{AB} & \operatorname{Tr}_{\hat{A}\hat{B}}[T_{\hat{A}\hat{B}}(\rho_{\hat
{A}\hat{B}})L_{A\hat{A}\hat{B}B}]
\end{bmatrix}
\geq0,
\end{equation}
and $L_{A\hat{A}\hat{B}B}$ is the Choi operator for a one-way LOCC channel.
Note that the last constraint is equivalent to%
\begin{multline}
|0\rangle\!\langle0|\otimes\Gamma_{AB}+|0\rangle\!\langle1|\otimes
Q_{AB}^{\dag}+|1\rangle\!\langle0|\otimes Q_{AB}\label{eq:Q-L-constr-ch-infid}%
\\
+|1\rangle\!\langle1|\otimes\operatorname{Tr}_{\hat{A}\hat{B}}[T_{\hat{A}%
\hat{B}}(\rho_{\hat{A}\hat{B}})L_{A\hat{A}\hat{B}B}]\geq0.
\end{multline}
Let $\lambda$, $L_{A\hat{A}\hat{B}B}$, and $Q_{AB}$ be an optimal solution.
Then it follows that $\lambda$, $\left(  \overline{\mathcal{U}}_{A}%
\otimes\mathcal{U}_{B}\right)  (L_{A\hat{A}\hat{B}B})$, and $\left(
\overline{\mathcal{U}}_{A}\otimes\mathcal{U}_{B}\right)  (Q_{AB})$ is an
optimal solution also. This follows because all of the constraints are
satisfied for these choices while still obtaining the same optimal value. To
see this, consider that%
\begin{align}
\lambda I_{A}  &  \leq\operatorname{Re}[\operatorname{Tr}_{B}[Q_{AB}]]\\
\Leftrightarrow\lambda\overline{\mathcal{U}}_{A}(I_{A})  &  \leq
\overline{\mathcal{U}}_{A}(\operatorname{Re}[\operatorname{Tr}_{B}[Q_{AB}]])\\
\Leftrightarrow\lambda I_{A}  &  \leq\operatorname{Re}[\operatorname{Tr}%
_{B}[\overline{\mathcal{U}}_{A}Q_{AB}]]\\
\Leftrightarrow\lambda I_{A}  &  \leq\operatorname{Re}[\operatorname{Tr}%
_{B}[(\overline{\mathcal{U}}_{A}\otimes\mathcal{U}_{B})(Q_{AB})]]
\end{align}%
\begin{align}
\operatorname{Tr}_{B}[L_{A\hat{A}\hat{B}B}]  &  =I_{A\hat{A}\hat{B}}\\
\Leftrightarrow\operatorname{Tr}_{B}[(\overline{\mathcal{U}}_{A}%
\otimes\mathcal{U}_{B})L_{A\hat{A}\hat{B}B}]  &  =I_{A\hat{A}\hat{B}},
\end{align}
and%
\begin{align}
&  |0\rangle\!\langle0|\otimes\Gamma_{AB}+|0\rangle\!\langle1|\otimes
Q_{AB}^{\dag}+|1\rangle\!\langle0|\otimes Q_{AB}\nonumber\\
&  \quad+|1\rangle\!\langle1|\otimes\operatorname{Tr}_{\hat{A}\hat{B}}%
[T_{\hat{A}\hat{B}}(\rho_{\hat{A}\hat{B}})L_{A\hat{A}\hat{B}B}]\geq0\\
&  \Leftrightarrow(\operatorname{id}\otimes(\overline{\mathcal{U}}_{A}%
\otimes\mathcal{U}_{B}))(|0\rangle\!\langle0|\otimes\Gamma_{AB}\nonumber\\
&  \quad+|0\rangle\!\langle1|\otimes Q_{AB}^{\dag}+|1\rangle\!\langle0|\otimes
Q_{AB}\nonumber\\
&  \quad+|1\rangle\!\langle1|\otimes\operatorname{Tr}_{\hat{A}\hat{B}}%
[T_{\hat{A}\hat{B}}(\rho_{\hat{A}\hat{B}})L_{A\hat{A}\hat{B}B}])\geq0\\
&  \Leftrightarrow|0\rangle\!\langle0|\otimes(\overline{\mathcal{U}}%
_{A}\otimes\mathcal{U}_{B})(\Gamma_{AB})\nonumber\\
&  \quad+|0\rangle\!\langle1|\otimes(\overline{\mathcal{U}}_{A}\otimes
\mathcal{U}_{B})(Q_{AB}^{\dag})\nonumber\\
&  \quad+|1\rangle\!\langle0|\otimes(\overline{\mathcal{U}}_{A}\otimes
\mathcal{U}_{B})(Q_{AB})\nonumber\\
&  \quad+|1\rangle\!\langle1|\otimes(\overline{\mathcal{U}}_{A}\otimes
\mathcal{U}_{B})(\operatorname{Tr}_{\hat{A}\hat{B}}[T_{\hat{A}\hat{B}}%
(\rho_{\hat{A}\hat{B}})L_{A\hat{A}\hat{B}B}])\geq0\\
&  \Leftrightarrow|0\rangle\!\langle0|\otimes\Gamma_{AB}+|0\rangle
\!\langle1|\otimes\lbrack(\overline{\mathcal{U}}_{A}\otimes\mathcal{U}%
_{B})(Q_{AB})]^{\dag}\nonumber\\
&  \quad+|1\rangle\!\langle0|\otimes(\overline{\mathcal{U}}_{A}\otimes
\mathcal{U}_{B})(Q_{AB})\nonumber\\
&  \quad+|1\rangle\!\langle1|\otimes(\operatorname{Tr}_{\hat{A}\hat{B}%
}[T_{\hat{A}\hat{B}}(\rho_{\hat{A}\hat{B}})(\overline{\mathcal{U}}_{A}%
\otimes\mathcal{U}_{B})(L_{A\hat{A}\hat{B}B})])\geq0.
\end{align}
Note that $(\overline{\mathcal{U}}_{A}\otimes\mathcal{U}_{B})(L_{A\hat{A}%
\hat{B}B})$ is the Choi operator for a one-way LOCC channel if $L_{A\hat
{A}\hat{B}B}$ is. Also, due to the fact that the objective function is linear
and the constraints are linear operator inequalities, it follows that convex
combinations of solutions are solutions as well. So this implies that if
$\lambda$, $L_{A\hat{A}\hat{B}B}$, and $Q_{AB}$ is an optimal solution, then
so is $\lambda$,%
\begin{align}
\widetilde{L}_{A\hat{A}\hat{B}B}  &  =\int dU\ (\overline{\mathcal{U}}%
_{A}\otimes\mathcal{U}_{B})(L_{A\hat{A}\hat{B}B}),\\
\widetilde{Q}_{AB}  &  =\int dU\ (\overline{\mathcal{U}}_{A}\otimes
\mathcal{U}_{B})(Q_{AB}).
\end{align}
Additionally, $\widetilde{L}_{A\hat{A}\hat{B}B}$ is the Choi operator for a
one-way LOCC channel if $L_{A\hat{A}\hat{B}B}$ is. As argued in
Appendix~\ref{app:exploiting-sym}, $\widetilde{L}_{A\hat{A}\hat{B}B}$ has a
simpler form as%
\begin{multline}
\widetilde{L}_{A\hat{A}\hat{B}B}=\Gamma_{AB}\operatorname{Tr}[K_{\hat{A}%
\hat{B}}\rho_{\hat{A}\hat{B}}]\\
+\frac{dI_{AB}-\Gamma_{AB}}{d^{2}-1}\operatorname{Tr}[L_{\hat{A}\hat{B}}%
\rho_{\hat{A}\hat{B}}],
\end{multline}
with the constraints in \eqref{eq:choi-cond-01}--\eqref{eq:choi-cond-02} on
$\widetilde{L}_{A\hat{A}\hat{B}B}$ simplifying to%
\begin{align}
K_{\hat{A}\hat{B}},L_{\hat{A}\hat{B}}  &  \geq0,\\
K_{\hat{A}\hat{B}}+L_{\hat{A}\hat{B}}  &  =I_{\hat{A}\hat{B}},
\end{align}
such that $K_{\hat{A}\hat{B}}$ and $L_{\hat{A}\hat{B}}$ are measurement
operators for a one-way LOCC channel (i.e., satisfy \eqref{eq:1wl-form-K-op}).
We also find that $\widetilde{Q}_{AB}$ simplifies to%
\begin{equation}
\widetilde{Q}_{AB}=q_{1}\Gamma_{AB}+q_{2}\frac{dI_{AB}-\Gamma_{AB}}{d^{2}-1},
\end{equation}
where $q_{1},q_{2}\in\mathbb{C}$. The constraint in
\eqref{eq:lam-Q-constraint-ch-inf} reduces to%
\begin{equation}
\lambda\leq\operatorname{Re}[q_{1}+q_{2}],
\end{equation}
because%
\begin{equation}
\operatorname{Tr}_{B}[\widetilde{Q}_{AB}]=\left(  q_{1}+q_{2}\right)  I_{A}.
\end{equation}
The constraint in \eqref{eq:Q-L-constr-ch-infid} reduces to%
\begin{multline}
|0\rangle\!\langle0|\otimes\Gamma_{AB}+|0\rangle\!\langle1|\otimes\left[
q_{1}\Gamma_{AB}+q_{2}\frac{dI_{AB}-\Gamma_{AB}}{d^{2}-1}\right]  ^{\dag}\\
+|1\rangle\!\langle0|\otimes\left[  q_{1}\Gamma_{AB}+q_{2}\frac{dI_{AB}%
-\Gamma_{AB}}{d^{2}-1}\right] \\
+|1\rangle\!\langle1|\otimes\left[
\begin{array}
[c]{c}%
\Gamma_{AB}\operatorname{Tr}[K_{\hat{A}\hat{B}}\rho_{\hat{A}\hat{B}}]\\
+\frac{dI_{AB}-\Gamma_{AB}}{d^{2}-1}\operatorname{Tr}[L_{\hat{A}\hat{B}}%
\rho_{\hat{A}\hat{B}}]
\end{array}
\right]  \geq0.
\end{multline}
Exploiting the orthogonality of the operators $\Gamma_{AB}$ and $\frac
{dI_{AB}-\Gamma_{AB}}{d^{2}-1}$, we conclude that the single constraint above
is equivalent to the following two constraints:%
\begin{align}
|0\rangle\!\langle0|+q_{1}^{\ast}|0\rangle\!\langle1|+q_{1}|1\rangle
\!\langle0|+\operatorname{Tr}[K_{\hat{A}\hat{B}}\rho_{\hat{A}\hat{B}%
}]|1\rangle\!\langle1|  &  \geq0,\\
q_{2}^{\ast}|0\rangle\!\langle1|+q_{2}|1\rangle\!\langle0|+\operatorname{Tr}%
[L_{\hat{A}\hat{B}}\rho_{\hat{A}\hat{B}}]|1\rangle\!\langle1|  &  \geq0.
\end{align}
We can rewrite these as the following two matrix inequalities:%
\begin{align}%
\begin{bmatrix}
1 & q_{1}^{\ast}\\
q_{1} & \operatorname{Tr}[K_{\hat{A}\hat{B}}\rho_{\hat{A}\hat{B}}]
\end{bmatrix}
&  \geq0,\\%
\begin{bmatrix}
0 & q_{2}^{\ast}\\
q_{2} & \operatorname{Tr}[L_{\hat{A}\hat{B}}\rho_{\hat{A}\hat{B}}]
\end{bmatrix}
&  \geq0.
\end{align}
Since $\operatorname{Tr}[K_{\hat{A}\hat{B}}\rho_{\hat{A}\hat{B}}%
],\operatorname{Tr}[L_{\hat{A}\hat{B}}\rho_{\hat{A}\hat{B}}]\geq0$, it follows
that the inequalities above hold if and only if%
\begin{equation}
\operatorname{Tr}[K_{\hat{A}\hat{B}}\rho_{\hat{A}\hat{B}}]\geq\left\vert
q_{1}\right\vert ^{2},\qquad q_{2}=0.
\end{equation}
As a consequence, the optimization problem in \eqref{eq-app:ch-infid-opt}
becomes%
\begin{equation}
1-\left[  \sup_{\lambda\geq0,K_{\hat{A}\hat{B}},L_{\hat{A}\hat{B}}\geq
0,q_{1}\in\mathbb{C}}\lambda\right]  ^{2},
\end{equation}
subject to%
\begin{align}
\operatorname{Tr}[K_{\hat{A}\hat{B}}\rho_{\hat{A}\hat{B}}]  &  \geq\left\vert
q_{1}\right\vert ^{2},\\
\lambda &  \leq\operatorname{Re}[q_{1}],\\
K_{\hat{A}\hat{B}}+L_{\hat{A}\hat{B}}  &  =I_{\hat{A}\hat{B}},
\end{align}
and the following channel being one-way LOCC:%
\begin{multline}
\mathcal{L}_{A\hat{A}\hat{B}\rightarrow B}(\omega_{A\hat{A}\hat{B}%
})=\operatorname{id}_{A\rightarrow B}^{d}(\operatorname{Tr}_{\hat{A}\hat{B}%
}[K_{\hat{A}\hat{B}}\omega_{A\hat{A}\hat{B}}])\\
+\mathcal{D}_{A\rightarrow B}(\operatorname{Tr}_{\hat{A}\hat{B}}[L_{\hat
{A}\hat{B}}\omega_{A\hat{A}\hat{B}}]).
\end{multline}
As indicated previously (see \eqref{eq:1wl-form-K-op}), $\mathcal{L}_{A\hat
{A}\hat{B}\rightarrow B}$ is one-way LOCC if there exists a POVM\ $\{\Lambda
_{B\hat{A}}^{x}\}_{x}$ and a set $\{\mathcal{D}_{\hat{B}\rightarrow B}%
^{x}\}_{x}$\ of channels such that%
\begin{equation}
K_{\hat{A}\hat{B}}=\frac{1}{d^{2}}\sum_{x}\operatorname{Tr}_{B}[T_{B}%
(\Lambda_{B\hat{A}}^{x})\Gamma_{\hat{B}B}^{\mathcal{D}^{x}}],
\end{equation}
where $\Gamma_{\hat{B}B}^{\mathcal{D}^{x}}$ is the Choi operator of the
channel $\mathcal{D}_{\hat{B}\rightarrow B}^{x}$. Since we are trying to
maximize $\lambda$ subject to these constraints, we should then set
$\lambda=q_{1}=\sqrt{\operatorname{Tr}[K_{\hat{A}\hat{B}}\rho_{\hat{A}\hat{B}%
}]}$. This concludes the proof.

\section{Proof of Proposition~\ref{prop:simplified-SDP-two-ext-sim-id}}

\label{app:SDP-id-sim-simplify-sym}

Some of the reasoning here is similar conceptually to that given in
Appendix~\ref{app:proof-errors-equal}, but the details in many places are
rather different. We divide our proof into the subsections given below.

\subsection{Twirled two-extendible channel is optimal}

\label{app:twirled-2e-opt}

We begin by considering two-extendible channels exclusively and then later
bring in the PPT constraints. Our first goal is to prove that a two-extendible
channel of the form in \eqref{eq:twirled-2-ext-1}, with extension channel in
\eqref{eq:twirled-2-ext-2}, is optimal for performing the minimization in
\eqref{eq:sim-err-2-ext-gen-ch} if $\mathcal{N}_{A\rightarrow B}$ is the
identity channel $\operatorname{id}_{A\rightarrow B}^{d}$. To this end, let
$\mathcal{K}_{A\hat{A}\hat{B}\rightarrow B}$ be an arbitrary two-extendible
channel, meaning that there exists an extension channel $\mathcal{M}_{A\hat
{A}\hat{B}_{1}\hat{B}_{2}\rightarrow B_{1}B_{2}}$ satisfying the constraints
in \eqref{eq:perm-cov-ch}--\eqref{eq:marg-ch-ch}, i.e.,%
\begin{align}
\mathcal{F}_{B_{1}B_{2}}\circ\mathcal{M}_{A\hat{A}\hat{B}_{1}\hat{B}%
_{2}\rightarrow B_{1}B_{2}}  &  =\mathcal{M}_{A\hat{A}\hat{B}_{1}\hat{B}%
_{2}\rightarrow B_{1}B_{2}}\circ\mathcal{F}_{\hat{B}_{1}\hat{B}_{2}},\\
\operatorname{Tr}_{B_{2}}\circ\mathcal{M}_{A\hat{A}\hat{B}_{1}\hat{B}%
_{2}\rightarrow B_{1}B_{2}}  &  =\mathcal{K}_{A\hat{A}\hat{B}_{1}\rightarrow
B_{1}}\circ\operatorname{Tr}_{\hat{B}_{2}}.
\end{align}
For an arbitrary unitary channel $\mathcal{U}$, it then follows that the
channel%
\begin{equation}\label{eq:K_under_unitary}
\mathcal{K}_{A\hat{A}\hat{B}\rightarrow B}^{\mathcal{U}}\equiv\mathcal{U}%
_{B}^{\dag}\circ\mathcal{K}_{A\hat{A}\hat{B}\rightarrow B}\circ\mathcal{U}_{A}%
\end{equation}
is two-extendible with extension channel%
\begin{equation}\label{eq:M_under_unitary}
\mathcal{M}_{A\hat{A}\hat{B}_{1}\hat{B}_{2}\rightarrow B_{1}B_{2}%
}^{\mathcal{U}}\equiv(\mathcal{U}_{B_{1}}^{\dag}\otimes\mathcal{U}_{B_{2}%
}^{\dag})\circ\mathcal{M}_{A\hat{A}\hat{B}_{1}\hat{B}_{2}\rightarrow
B_{1}B_{2}}\circ\mathcal{U}_{A}%
\end{equation}
and achieves the same value of the normalized diamond distance that
$\mathcal{K}_{A\hat{A}\hat{B}\rightarrow B}$ does. That is,%
\begin{multline}
\left\Vert \mathcal{K}_{A\hat{A}\hat{B}\rightarrow B}\circ\mathcal{A}_{\hat
{A}\hat{B}}^{\rho}-\operatorname{id}_{A\rightarrow B}^{d}\right\Vert
_{\diamond}\\
=\left\Vert \mathcal{K}_{A\hat{A}\hat{B}\rightarrow B}^{\mathcal{U}}%
\circ\mathcal{A}_{\hat{A}\hat{B}}^{\rho}-\operatorname{id}_{A\rightarrow
B}^{d}\right\Vert _{\diamond}.
\end{multline}
This equality follows by the same reasoning used to justify
\eqref{eq:-diamond-d-unitary-equiv}. The claim that $\mathcal{K}_{A\hat{A}%
\hat{B}\rightarrow B}^{\mathcal{U}}$ is two-extendible follows because the
extension channel $\mathcal{M}_{A\hat{A}\hat{B}_{1}\hat{B}_{2}\rightarrow
B_{1}B_{2}}^{\mathcal{U}}$ satisfies%
\begin{align}
&  \mathcal{F}_{B_{1}B_{2}}\circ\mathcal{M}_{A\hat{A}\hat{B}_{1}\hat{B}%
_{2}\rightarrow B_{1}B_{2}}^{\mathcal{U}}\nonumber\\
&  =\mathcal{F}_{B_{1}B_{2}}\circ(\mathcal{U}_{B_{1}}^{\dag}\otimes
\mathcal{U}_{B_{2}}^{\dag})\circ\mathcal{M}_{A\hat{A}\hat{B}_{1}\hat{B}%
_{2}\rightarrow B_{1}B_{2}}\circ\mathcal{U}_{A}\label{eq:two-ext-justify-1}\\
&  =\mathcal{F}_{B_{1}B_{2}}\circ(\mathcal{U}_{B_{1}}^{\dag}\otimes
\mathcal{U}_{B_{2}}^{\dag})\circ\mathcal{F}_{B_{1}B_{2}}^{\dag}\nonumber\\
&  \qquad\circ\mathcal{F}_{B_{1}B_{2}}\circ\mathcal{M}_{A\hat{A}\hat{B}%
_{1}\hat{B}_{2}\rightarrow B_{1}B_{2}}\circ\mathcal{U}_{A}\\
&  =(\mathcal{U}_{B_{1}}^{\dag}\otimes\mathcal{U}_{B_{2}}^{\dag}%
)\circ\mathcal{F}_{B_{1}B_{2}}\circ\mathcal{M}_{A\hat{A}\hat{B}_{1}\hat{B}%
_{2}\rightarrow B_{1}B_{2}}\circ\mathcal{U}_{A}\\
&  =(\mathcal{U}_{B_{1}}^{\dag}\otimes\mathcal{U}_{B_{2}}^{\dag}%
)\circ\mathcal{M}_{A\hat{A}\hat{B}_{1}\hat{B}_{2}\rightarrow B_{1}B_{2}}%
\circ\mathcal{F}_{\hat{B}_{1}\hat{B}_{2}}\circ\mathcal{U}_{A}\\
&  =\mathcal{M}_{A\hat{A}\hat{B}_{1}\hat{B}_{2}\rightarrow B_{1}B_{2}%
}^{\mathcal{U}}\circ\mathcal{F}_{\hat{B}_{1}\hat{B}_{2}}.
\label{eq:two-ext-justify-1-last}%
\end{align}
Additionally,%
\begin{align}
&  \operatorname{Tr}_{B_{2}}\circ\mathcal{M}_{A\hat{A}\hat{B}_{1}\hat{B}%
_{2}\rightarrow B_{1}B_{2}}^{\mathcal{U}}\nonumber\\
&  =\operatorname{Tr}_{B_{2}}\circ(\mathcal{U}_{B_{1}}^{\dag}\otimes
\mathcal{U}_{B_{2}}^{\dag})\circ\mathcal{M}_{A\hat{A}\hat{B}_{1}\hat{B}%
_{2}\rightarrow B_{1}B_{2}}\circ\mathcal{U}_{A}\label{eq:two-ext-justify-2}\\
&  =\mathcal{U}_{B_{1}}^{\dag}\circ\operatorname{Tr}_{B_{2}}\circ
\mathcal{M}_{A\hat{A}\hat{B}_{1}\hat{B}_{2}\rightarrow B_{1}B_{2}}%
\circ\mathcal{U}_{A}\\
&  =\mathcal{U}_{B_{1}}^{\dag}\circ\mathcal{K}_{A\hat{A}\hat{B}_{1}\rightarrow
B_{1}}\circ\operatorname{Tr}_{\hat{B}_{2}}\circ\mathcal{U}_{A}\\
&  =\mathcal{K}_{A\hat{A}\hat{B}\rightarrow B}^{\mathcal{U}}\circ
\operatorname{Tr}_{\hat{B}_{2}}. \label{eq:two-ext-justify-2-last}%
\end{align}
Then consider that%
\begin{multline}
\left\Vert \mathcal{K}_{A\hat{A}\hat{B}\rightarrow B}\circ\mathcal{A}_{\hat
{A}\hat{B}}^{\rho}-\operatorname{id}_{A\rightarrow B}^{d}\right\Vert
_{\diamond}\label{eq:convexity-dd-2-ext}\\
\geq\left\Vert \widetilde{\mathcal{K}}_{A\hat{A}\hat{B}\rightarrow B}%
\circ\mathcal{A}_{\hat{A}\hat{B}}^{\rho}-\operatorname{id}_{A\rightarrow
B}^{d}\right\Vert _{\diamond},
\end{multline}
where $\widetilde{\mathcal{K}}_{A\hat{A}\hat{B}\rightarrow B}$ is the
following twirled two-extendible channel:%
\begin{equation}
\widetilde{\mathcal{K}}_{A\hat{A}\hat{B}\rightarrow B}\coloneqq\int
dU\ \mathcal{K}_{A\hat{A}\hat{B}\rightarrow B}^{\mathcal{U}},
\label{eq:twirled-2-ext-1}%
\end{equation}
with extension channel%
\begin{equation}
\widetilde{\mathcal{M}}_{A\hat{A}\hat{B}_{1}\hat{B}_{2}\rightarrow B_{1}B_{2}%
}\coloneqq\int dU\ \mathcal{M}_{A\hat{A}\hat{B}_{1}\hat{B}_{2}\rightarrow
B_{1}B_{2}}^{\mathcal{U}}. \label{eq:twirled-2-ext-2}%
\end{equation}
The inequality in \eqref{eq:convexity-dd-2-ext} follows from the same
reasoning given for \eqref{eq:convexity-dd-1wl}. Furthermore, $\widetilde
{\mathcal{K}}_{A\hat{A}\hat{B}\rightarrow B}$ is two-extendible with extension
$\widetilde{\mathcal{M}}_{A\hat{A}\hat{B}_{1}\hat{B}_{2}\rightarrow B_{1}%
B_{2}}$ because the extension channel $\widetilde{\mathcal{M}}_{A\hat{A}%
\hat{B}_{1}\hat{B}_{2}\rightarrow B_{1}B_{2}}$ satisfies%
\begin{align}
&  \mathcal{F}_{B_{1}B_{2}}\circ\widetilde{\mathcal{M}}_{A\hat{A}\hat{B}%
_{1}\hat{B}_{2}\rightarrow B_{1}B_{2}}\nonumber\\
&  =\mathcal{F}_{B_{1}B_{2}}\circ\int dU\ \mathcal{M}_{A\hat{A}\hat{B}_{1}%
\hat{B}_{2}\rightarrow B_{1}B_{2}}^{\mathcal{U}}\\
&  =\int dU\ \mathcal{F}_{B_{1}B_{2}}\circ\mathcal{M}_{A\hat{A}\hat{B}_{1}%
\hat{B}_{2}\rightarrow B_{1}B_{2}}^{\mathcal{U}}\\
&  =\int dU\ \mathcal{M}_{A\hat{A}\hat{B}_{1}\hat{B}_{2}\rightarrow B_{1}%
B_{2}}^{\mathcal{U}}\circ\mathcal{F}_{\hat{B}_{1}\hat{B}_{2}}\\
&  =\widetilde{\mathcal{M}}_{A\hat{A}\hat{B}_{1}\hat{B}_{2}\rightarrow
B_{1}B_{2}}\circ\mathcal{F}_{\hat{B}_{1}\hat{B}_{2}},
\end{align}
and%
\begin{align}
&  \operatorname{Tr}_{B_{2}}\circ\widetilde{\mathcal{M}}_{A\hat{A}\hat{B}%
_{1}\hat{B}_{2}\rightarrow B_{1}B_{2}}\nonumber\\
&  =\operatorname{Tr}_{B_{2}}\circ\int dU\ \mathcal{M}_{A\hat{A}\hat{B}%
_{1}\hat{B}_{2}\rightarrow B_{1}B_{2}}^{\mathcal{U}}\\
&  =\int dU\ \operatorname{Tr}_{B_{2}}\circ\mathcal{M}_{A\hat{A}\hat{B}%
_{1}\hat{B}_{2}\rightarrow B_{1}B_{2}}^{\mathcal{U}}\\
&  =\int dU\ \mathcal{K}_{A\hat{A}\hat{B}\rightarrow B}^{\mathcal{U}}%
\circ\operatorname{Tr}_{\hat{B}_{2}}\\
&  =\widetilde{\mathcal{K}}_{A\hat{A}\hat{B}\rightarrow B}\circ
\operatorname{Tr}_{\hat{B}_{2}}.
\end{align}
In the above, we made use of
\eqref{eq:two-ext-justify-1}--\eqref{eq:two-ext-justify-1-last} and \eqref{eq:two-ext-justify-2}--\eqref{eq:two-ext-justify-2-last}.

As a consequence of \eqref{eq:convexity-dd-2-ext}, it follows that it suffices
to minimize the following objective function%
\begin{equation}
\left\Vert \mathcal{K}_{A\hat{A}\hat{B}\rightarrow B}\circ\mathcal{A}_{\hat
{A}\hat{B}}^{\rho}-\operatorname{id}_{A\rightarrow B}^{d}\right\Vert
_{\diamond}%
\end{equation}
with respect to two-extendible channels obeying the symmetries in
\eqref{eq:twirled-2-ext-1} and \eqref{eq:twirled-2-ext-2}. Thus, our goal is
to characterize the form of two-extendible channels possessing the symmetries
in \eqref{eq:twirled-2-ext-1} and \eqref{eq:twirled-2-ext-2}. To this end, we
consider channels of the form in \eqref{eq:twirled-2-ext-2} that obey the
following constraints:%
\begin{align}
\mathcal{F}_{B_{1}B_{2}}\circ\widetilde{\mathcal{M}}_{A\hat{A}\hat{B}_{1}%
\hat{B}_{2}\rightarrow B_{1}B_{2}}  &  =\widetilde{\mathcal{M}}_{A\hat{A}%
\hat{B}_{1}\hat{B}_{2}\rightarrow B_{1}B_{2}}\circ\mathcal{F}_{\hat{B}_{1}%
\hat{B}_{2}},\label{eq:2-ext-constr-1}\\
\operatorname{Tr}_{B_{2}}\circ\widetilde{\mathcal{M}}_{A\hat{A}\hat{B}_{1}%
\hat{B}_{2}\rightarrow B_{1}B_{2}}  &  =\widetilde{\mathcal{K}}_{A\hat{A}%
\hat{B}\rightarrow B}\circ\operatorname{Tr}_{\hat{B}_{2}}.
\label{eq:2-ext-constr-2}%
\end{align}
Equivalently, we work with the Choi operator of such channels.

First consider that a channel obeying the following symmetry:%
\begin{multline}
\widetilde{\mathcal{M}}_{A\hat{A}\hat{B}_{1}\hat{B}_{2}\rightarrow B_{1}B_{2}%
}=\\
\int dU\ (\mathcal{U}_{B_{1}}^{\dag}\otimes\mathcal{U}_{B_{2}}^{\dag}%
)\circ\widetilde{\mathcal{M}}_{A\hat{A}\hat{B}_{1}\hat{B}_{2}\rightarrow
B_{1}B_{2}}\circ\mathcal{U}_{A}%
\end{multline}
is equivalent to its Choi operator obeying the following symmetry:%
\begin{equation}
\widetilde{M}_{A\hat{A}\hat{B}_{1}\hat{B}_{2}B_{1}B_{2}}=\widetilde
{\mathcal{T}}_{AB_{1}B_{2}}(\widetilde{M}_{A\hat{A}\hat{B}_{1}\hat{B}_{2}%
B_{1}B_{2}}), \label{eq:choi-op-two-ext-symmetrized}%
\end{equation}
where%
\begin{equation}
\widetilde{\mathcal{T}}_{AB_{1}B_{2}}(\cdot)\coloneqq\int dU\ (\overline
{\mathcal{U}}_{A}\otimes\mathcal{U}_{B_{1}}\otimes\mathcal{U}_{B_{2}})(\cdot).
\end{equation}
It was shown in \cite[Section~VI-A]{EW01} and \cite{JV13} that the tripartite
twirling channel $\widetilde{\mathcal{T}}_{AB_{1}B_{2}}$ has the following
effect on an arbitrary operator $X_{AB_{1}B_{2}}$:%
\begin{multline}
\widetilde{\mathcal{T}}_{AB_{1}B_{2}}(X_{AB_{1}B_{2}})= \label{eq:3-way-twirl}%
\\
\sum_{i\in\left\{  +,-,0,1,2,3\right\}  }\operatorname{Tr}_{AB_{1}B_{2}%
}[S_{AB_{1}B_{2}}^{i}X_{AB_{1}B_{2}}]\frac{S_{AB_{1}B_{2}}^{i}}%
{\operatorname{Tr}[S_{AB_{1}B_{2}}^{g(i)}]},
\end{multline}
where%
\begin{equation}\label{eq:S_indexing}
g(i)\coloneqq\left\{
\begin{array}
[c]{cc}%
i & \text{if }i\in\left\{  +,-\right\} \\
0 & \text{if }i\in\left\{  0,1,2,3\right\}
\end{array}
\right.  ,
\end{equation}
and%
\begin{align}
S_{AB_{1}B_{2}}^{+}  &  \coloneqq\frac{1}{2}\left[
\begin{array}
[c]{c}%
I_{AB_{1}B_{2}}+V_{B_{1}B_{2}}\\
-\left(  \frac{V_{AB_{1}}^{T_{A}}+V_{AB_{2}}^{T_{A}}+V_{AB_{1}B_{2}}^{T_{A}%
}+V_{B_{2}B_{1}A}^{T_{A}}}{d+1}\right)
\end{array}
\right]  ,\label{eq:S-op-+}\\
S_{AB_{1}B_{2}}^{-}  &  \coloneqq\frac{1}{2}\left[
\begin{array}
[c]{c}%
I_{AB_{1}B_{2}}-V_{B_{1}B_{2}}\\
-\left(  \frac{V_{AB_{1}}^{T_{A}}+V_{AB_{2}}^{T_{A}}-V_{AB_{1}B_{2}}^{T_{A}%
}-V_{B_{2}B_{1}A}^{T_{A}}}{d-1}\right)
\end{array}
\right]  ,\label{eq:S-op--}\\
S_{AB_{1}B_{2}}^{0}  &  \coloneqq\frac{1}{d^{2}-1}\left[
\begin{array}
[c]{c}%
d\left(  V_{AB_{1}}^{T_{A}}+V_{AB_{2}}^{T_{A}}\right) \\
-\left(  V_{AB_{1}B_{2}}^{T_{A}}+V_{B_{2}B_{1}A}^{T_{A}}\right)
\end{array}
\right]  ,\label{eq:S-op-0}\\
S_{AB_{1}B_{2}}^{1}  &  \coloneqq\frac{1}{d^{2}-1}\left[
\begin{array}
[c]{c}%
d\left(  V_{AB_{1}B_{2}}^{T_{A}}+V_{B_{2}B_{1}A}^{T_{A}}\right) \\
-\left(  V_{AB_{1}}^{T_{A}}+V_{AB_{2}}^{T_{A}}\right)
\end{array}
\right]  ,\label{eq:S-op-1}\\
S_{AB_{1}B_{2}}^{2}  &  \coloneqq\frac{1}{\sqrt{d^{2}-1}}\left(  V_{AB_{1}%
}^{T_{A}}-V_{AB_{2}}^{T_{A}}\right)  ,\label{eq:S-op-2}\\
S_{AB_{1}B_{2}}^{3}  &  \coloneqq\frac{i}{\sqrt{d^{2}-1}}\left(
V_{AB_{1}B_{2}}^{T_{A}}-V_{B_{2}B_{1}A}^{T_{A}}\right)  , \label{eq:S-op-3}%
\end{align}
and the permutation operators $V_{AB_{1}}$, $V_{AB_{2}}$, $V_{B_{1}B_{2}}$,
$V_{AB_{1}B_{2}}$, and $V_{B_{2}B_{1}A}$ are given by%
\begin{align}
V_{AB_{1}}  &  \coloneqq\sum_{i,j}|i\rangle\!\langle j|_{A}\otimes
|j\rangle\!\langle i|_{B_{1}}\otimes I_{B_{2}},\\
V_{AB_{2}}  &  \coloneqq\sum_{i,j}|i\rangle\!\langle j|_{A}\otimes I_{B_{1}%
}\otimes|j\rangle\!\langle i|_{B_{2}},\\
V_{B_{1}B_{2}}  &  \coloneqq\sum_{i,j}I_{A}\otimes|i\rangle\!\langle
j|_{B_{1}}\otimes|j\rangle\!\langle i|_{B_{2}}\\
V_{AB_{1}B_{2}}  &  \coloneqq\sum_{i,j,k}|k\rangle\!\langle i|_{A}%
\otimes|i\rangle\!\langle j|_{B_{1}}\otimes|j\rangle\!\langle k|_{B_{2}},\\
V_{B_{2}B_{1}A}  &  \coloneqq\sum_{i,j,k}|j\rangle\!\langle i|_{A}%
\otimes|k\rangle\!\langle j|_{B_{1}}\otimes|i\rangle\!\langle k|_{B_{2}},
\end{align}
so that%
\begin{align}
V_{AB_{1}}^{T_{A}}  &  =\sum_{i,j}|j\rangle\!\langle i|_{A}\otimes
|j\rangle\!\langle i|_{B_{1}}\otimes I_{B_{2}}%
,\label{eq:partial-transpose-perm-ops-1}\\
V_{AB_{2}}^{T_{A}}  &  =\sum_{i,j}|j\rangle\!\langle i|_{A}\otimes I_{B_{1}%
}\otimes|j\rangle\!\langle i|_{B_{2}},\\
V_{AB_{1}B_{2}}^{T_{A}}  &  =\sum_{i,j,k}|i\rangle\!\langle k|_{A}%
\otimes|i\rangle\!\langle j|_{B_{1}}\otimes|j\rangle\!\langle k|_{B_{2}},\\
V_{B_{2}B_{1}A}^{T_{A}}  &  =\sum_{i,j,k}|i\rangle\!\langle j|_{A}%
\otimes|k\rangle\!\langle j|_{B_{1}}\otimes|i\rangle\!\langle k|_{B_{2}}.
\label{eq:partial-transpose-perm-ops-4}%
\end{align}
We note that%
\begin{align}
\operatorname{Tr}[S_{AB_{1}B_{2}}^{+}]  &  =\frac{d\left(  d+2\right)  \left(
d-1\right)  }{2},\\
\operatorname{Tr}[S_{AB_{1}B_{2}}^{-}]  &  =\frac{d\left(  d-2\right)  \left(
d+1\right)  }{2},\\
\operatorname{Tr}[S_{AB_{1}B_{2}}^{0}]  &  =2d.\label{S0_trace}
\end{align}
It is known that $S_{AB_{1}B_{2}}^{+}$, $S_{AB_{1}B_{2}}^{-}$, and
$S_{AB_{1}B_{2}}^{0}$ are orthogonal projectors and satisfy
\begin{equation}
S_{AB_{1}B_{2}}^{+}+S_{AB_{1}B_{2}}^{-}+S_{AB_{1}B_{2}}^{0}=I_{AB_{1}B_{2}},
\end{equation}
while $S_{AB_{1}B_{2}}^{1}$, $S_{AB_{1}B_{2}}^{2}$, $S_{AB_{1}B_{2}}^{3}$ are
Pauli-like operators acting on the subspace onto which $S_{AB_{1}B_{2}}^{0}$
projects (see \cite[Section~VI-A]{EW01} and \cite{JV13}). All $S_{AB_{1}B_{2}%
}^{i}$ operators are Hermitian and satisfy the following for their
Hilbert--Schmidt inner product:%
\begin{align}
\left\langle S_{AB_{1}B_{2}}^{i},S_{AB_{1}B_{2}}^{j}\right\rangle  &
=\operatorname{Tr}[(S_{AB_{1}B_{2}}^{i})^{\dag}S_{AB_{1}B_{2}}^{j}%
]\label{eq:in-prod-rel-1}\\
&  =\operatorname{Tr}[S_{AB_{1}B_{2}}^{i}S_{AB_{1}B_{2}}^{j}]\\
&  =\operatorname{Tr}[S_{AB_{1}B_{2}}^{g(i)}]\delta_{i,j}.
\label{eq:in-prod-rel-3}%
\end{align}

As a consequence of the fact that $S_{AB_{1}B_{2}}^{-}$ is a projector with
trace equal to $\frac{d\left(  d-2\right)  \left(  d+1\right)  }{2}$, it
follows that $S_{AB_{1}B_{2}}^{-}=0$ for the case of $d=2$, so that this
operator and its associated operators $M_{\hat{A}\hat{B}_{1}\hat{B}_{2}%
}^{\prime-}$ and $M_{\hat{A}\hat{B}_{1}\hat{B}_{2}}^{-}$ given below are not
involved for the case of $d=2$ (when simulating a qubit identity channel).

By applying \eqref{eq:3-way-twirl} to \eqref{eq:choi-op-two-ext-symmetrized},
we find that%
\begin{equation}
\widetilde{M}_{A\hat{A}\hat{B}_{1}\hat{B}_{2}B_{1}B_{2}}=\sum_{i\in\left\{
+,-,0,1,2,3\right\}  }M_{\hat{A}\hat{B}_{1}\hat{B}_{2}}^{\prime i}\otimes
\frac{S_{AB_{1}B_{2}}^{i}}{\operatorname{Tr}[S_{AB_{1}B_{2}}^{g(i)}]},
\label{eq:sym-form-for-Choi-op-2-ext}%
\end{equation}
where, for $i\in\left\{  +,-,0,1,2,3\right\}  $,%
\begin{equation}
M_{\hat{A}\hat{B}_{1}\hat{B}_{2}}^{\prime i}\coloneqq\operatorname{Tr}%
_{AB_{1}B_{2}}[S_{AB_{1}B_{2}}^{i}\widetilde{M}_{A\hat{A}\hat{B}_{1}\hat
{B}_{2}B_{1}B_{2}}].
\end{equation}
Thus, our goal from here is to determine conditions on the $M_{\hat{A}\hat
{B}_{1}\hat{B}_{2}}^{\prime i}$ operators such that $\widetilde{M}_{A\hat
{A}\hat{B}_{1}\hat{B}_{2}B_{1}B_{2}}$ corresponds to a Choi operator for a
channel $\widetilde{\mathcal{M}}_{A\hat{A}\hat{B}_{1}\hat{B}_{2}\rightarrow
B_{1}B_{2}}$ satisfying \eqref{eq:2-ext-constr-1} and
\eqref{eq:2-ext-constr-2}. These conditions are the same as those specified in
\eqref{eq:obj-func-gen-two-ext-sim}, \eqref{eq:sim-two-ext-TP},
\eqref{eq:sim-two-ext-2EXT-constr}, and \eqref{eq:sim-two-ext-constr-ext}, and
we repeat them here for convenience:%
\begin{align}
\widetilde{M}_{A\hat{A}\hat{B}_{1}\hat{B}_{2}B_{1}B_{2}}  &  \geq
0,\label{eq:CP-map-Choi-PSD-2-ext}\\
\operatorname{Tr}_{B_{1}B_{2}}[\widetilde{M}_{A\hat{A}\hat{B}_{1}\hat{B}%
_{2}B_{1}B_{2}}]  &  =I_{A\hat{A}\hat{B}_{1}\hat{B}_{2}},
\label{eq:TP-map-Choi-PSD-2-ext}%
\end{align}%
\begin{multline}
\operatorname{Tr}_{B_{2}}[\widetilde{M}_{A\hat{A}\hat{B}_{1}\hat{B}_{2}%
B_{1}B_{2}}]\label{eq:marginal-Choi-PSD-2-ext}\\
=\frac{1}{d_{\hat{B}}}\operatorname{Tr}_{\hat{B}_{2}B_{2}}[\widetilde
{M}_{A\hat{A}\hat{B}_{1}\hat{B}_{2}B_{1}B_{2}}]\otimes I_{\hat{B}_{2}},
\end{multline}%
\begin{equation}
\widetilde{M}_{A\hat{A}\hat{B}_{1}\hat{B}_{2}B_{1}B_{2}}=(\mathcal{F}_{\hat
{B}_{1}\hat{B}_{2}}\otimes\mathcal{F}_{B_{1}B_{2}})(\widetilde{M}_{A\hat
{A}\hat{B}_{1}\hat{B}_{2}B_{1}B_{2}}). \label{eq:perm-cov-Choi-2-ext}%
\end{equation}

\subsection{Complete positivity condition}

\label{Complete_positivity_tel}

We begin by considering the condition in \eqref{eq:CP-map-Choi-PSD-2-ext}, as
applied to \eqref{eq:sym-form-for-Choi-op-2-ext}. By exploiting the previously
stated facts that $S_{AB_{1}B_{2}}^{+}$, $S_{AB_{1}B_{2}}^{-}$, and
$S_{AB_{1}B_{2}}^{0}$ are orthogonal projectors and $S_{AB_{1}B_{2}}^{1}$,
$S_{AB_{1}B_{2}}^{2}$, $S_{AB_{1}B_{2}}^{3}$ are Pauli-like operators acting
on the subspace onto which $S_{AB_{1}B_{2}}^{0}$ projects, it follows that
\eqref{eq:CP-map-Choi-PSD-2-ext} holds if and only if%
\begin{align}
M_{\hat{A}\hat{B}_{1}\hat{B}_{2}}^{\prime+}  &  \geq0,
\label{eq:M+-PSD-CP-map}\\
M_{\hat{A}\hat{B}_{1}\hat{B}_{2}}^{\prime-}  &  \geq0,
\label{eq:M--PSD-CP-map}\\
\sum_{i\in\left\{  0,1,2,3\right\}  }M_{\hat{A}\hat{B}_{1}\hat{B}_{2}}^{\prime
i}\otimes\frac{S_{AB_{1}B_{2}}^{i}}{\operatorname{Tr}[S_{AB_{1}B_{2}}%
^{g(i)}]}  &  \geq0. \label{eq:CP-conditions-last}%
\end{align}
We can now exploit the inner product relations in
\eqref{eq:in-prod-rel-1}--\eqref{eq:in-prod-rel-3} to conclude that the map%
\begin{equation}
X\rightarrow\sum_{i\in\left\{  0,1,2,3\right\}  }\frac{S^{i}}{2d}%
\operatorname{Tr}\left[  \frac{\sigma^{i}}{2}X\right]
\label{eq:isometric-map-sig-to-S}%
\end{equation}
from a qubit input to the systems $AB_{1}B_{2}$ is an isometry with inverse%
\begin{equation}
Y\rightarrow\sum_{i\in\left\{  0,1,2,3\right\}  }\frac{\sigma^{i}}%
{2}\operatorname{Tr}\left[  \frac{S^{i}}{2d}Y\right]  ,
\label{eq:isometric-map-sig-to-S-2}%
\end{equation}
where $\{\sigma^{i}\}_{i}$ is the set of Pauli operators:
\begin{align}
\sigma_0 & \coloneqq \begin{bmatrix} 1 & 0 \\ 0 & 1\end{bmatrix},
& \sigma_1 &\coloneqq \begin{bmatrix} 0 & 1 \\ 1 & 0\end{bmatrix},\\
\sigma_2 & \coloneqq \begin{bmatrix} 0 & -i \\ i & 0\end{bmatrix},
& \sigma_3 &\coloneqq \begin{bmatrix} 1 & 0 \\ 0 & -1\end{bmatrix}.
\end{align}
Exploiting this map,
we find that \eqref{eq:CP-conditions-last} holds if and only if%
\begin{equation}
\sum_{i\in\left\{  0,1,2,3\right\}  }M_{\hat{A}\hat{B}_{1}\hat{B}_{2}}^{\prime
i}\otimes\frac{\sigma^{i}}{\operatorname{Tr}[S_{AB_{1}B_{2}}^{g(i)}]}\geq0,
\end{equation}
the latter being equivalent to%
\begin{equation}%
\begin{bmatrix}
M^{0}+M^{3} & M^{1}-iM^{2}\\
M^{1}+iM^{2} & M^{0}-M^{3}%
\end{bmatrix}
\geq0, \label{eq:M0-3-conds-CP}%
\end{equation}
after defining%
\begin{equation}
M_{\hat{A}\hat{B}_{1}\hat{B}_{2}}^{i}\coloneqq\frac{2M_{\hat{A}\hat{B}_{1}%
\hat{B}_{2}}^{\prime i}}{\operatorname{Tr}[S_{AB_{1}B_{2}}^{g(i)}]}%
=\frac{M_{\hat{A}\hat{B}_{1}\hat{B}_{2}}^{\prime i}}{d},
\label{eq:M-ops-redef-1}%
\end{equation}
for $i\in\left\{  0,1,2,3\right\}  $ and using the shorthand $M^{i}\equiv
M_{\hat{A}\hat{B}_{1}\hat{B}_{2}}^{i}$. Defining
\begin{equation}
M_{\hat{A}\hat{B}_{1}\hat{B}_{2}}^{i}\coloneqq\frac{M_{\hat{A}\hat{B}_{1}%
\hat{B}_{2}}^{\prime i}}{d} \label{eq:M-ops-redef-2}%
\end{equation}
for $i\in\left\{  +,-\right\}  $, we conclude from \eqref{eq:M+-PSD-CP-map},
\eqref{eq:M--PSD-CP-map}, and \eqref{eq:M0-3-conds-CP} that
\eqref{eq:CP-map-Choi-PSD-2-ext} holds if and only if the following operator
inequalities hold%
\begin{align}
M_{\hat{A}\hat{B}_{1}\hat{B}_{2}}^{+}  &  \geq0, \label{eq:PSD-constr-M-ops-1}%
\\
M_{\hat{A}\hat{B}_{1}\hat{B}_{2}}^{-}  &  \geq0,\\%
\begin{bmatrix}
M^{0}+M^{3} & M^{1}-iM^{2}\\
M^{1}+iM^{2} & M^{0}-M^{3}%
\end{bmatrix}
&  \geq0. \label{eq:PSD-constr-M-ops-3}%
\end{align}
From \eqref{eq:PSD-constr-M-ops-3}, we conclude that $M_0 \geq 0$ because it implies $M^{0}+M^{3}\geq 0$ and $M^{0}-M^{3} \geq 0$, which in turn implies $M_0 \geq 0$. We have thus justified the constraints on $M^+, M^-, M^0, M^1, M^2, M^3$ in \eqref{eq:obj-funct-simple-sdp-tele} and \eqref{eq:M-matrix-constr-tele}.

\subsection{Trace preservation condition}\label{sec_trace_preserve_tel}

We now consider the condition in \eqref{eq:TP-map-Choi-PSD-2-ext}, as applied
to \eqref{eq:sym-form-for-Choi-op-2-ext}. Consider that%
\begin{align}
I_{A\hat{A}\hat{B}_{1}\hat{B}_{2}}  &  =\operatorname{Tr}_{B_{1}B_{2}%
}[\widetilde{M}_{A\hat{A}\hat{B}_{1}\hat{B}_{2}B_{1}B_{2}}%
]\label{eq:TP-cond-1-1}\\
&  =\sum_{i\in\left\{  +,-,0,1,2,3\right\}  }M_{\hat{A}\hat{B}_{1}\hat{B}_{2}%
}^{\prime i}\otimes\frac{\operatorname{Tr}_{B_{1}B_{2}}[S_{AB_{1}B_{2}}^{i}%
]}{\operatorname{Tr}[S_{AB_{1}B_{2}}^{g(i)}]}. \label{eq:TP-cond-1-2}%
\end{align}
Thus, we need to calculate $\operatorname{Tr}_{B_{1}B_{2}}[S_{AB_{1}B_{2}}%
^{i}]$ for all $i\in\left\{  +,-,0,1,2,3\right\}  $. To do so, we first find
that%
\begin{align}
\operatorname{Tr}_{B_{1}B_{2}}[I_{AB_{1}B_{2}}]  &  =d^{2}I_{A},\\
\operatorname{Tr}_{B_{1}B_{2}}[V_{B_{1}B_{2}}]  &  =dI_{A},\\
\operatorname{Tr}_{B_{1}B_{2}}[V_{AB_{1}}^{T_{A}}]  &  =dI_{A},\\
\operatorname{Tr}_{B_{1}B_{2}}[V_{AB_{2}}^{T_{A}}]  &  =dI_{A},\\
\operatorname{Tr}_{B_{1}B_{2}}[V_{AB_{1}B_{2}}^{T_{A}}]  &  =I_{A},\\
\operatorname{Tr}_{B_{1}B_{2}}[V_{B_{2}B_{1}A}^{T_{A}}]  &  =I_{A},
\end{align}
which implies from \eqref{eq:S-op-+}--\eqref{eq:S-op-3}, that%
\begin{align}
\operatorname{Tr}_{B_{1}B_{2}}[S_{AB_{1}B_{2}}^{+}]  &  =\frac{1}{2}\left[
\begin{array}
[c]{c}%
d^{2}I_{A}+dI_{A}\\
-\left(  \frac{dI_{A}+dI_{A}+I_{A}+I_{A}}{d+1}\right)
\end{array}
\right] \\
&  =\frac{\left(  d+2\right)  \left(  d-1\right)  }{2}I_{A},\\
\operatorname{Tr}_{B_{1}B_{2}}[S_{AB_{1}B_{2}}^{-}]  &  =\frac{1}{2}\left[
\begin{array}
[c]{c}%
d^{2}I_{A}-dI_{A}\\
-\left(  \frac{dI_{A}+dI_{A}-I_{A}-I_{A}}{d-1}\right)
\end{array}
\right] \\
&  =\frac{\left(  d-2\right)  \left(  d+1\right)  }{2}I_{A},\\
\operatorname{Tr}_{B_{1}B_{2}}[S_{AB_{1}B_{2}}^{0}]  &  =\frac{1}{d^{2}%
-1}\left[
\begin{array}
[c]{c}%
d\left(  dI_{A}+dI_{A}\right) \\
-\left(  I_{A}+I_{A}\right)
\end{array}
\right] \\
&  =2I_{A},\\
\operatorname{Tr}_{B_{1}B_{2}}[S_{AB_{1}B_{2}}^{1}]  &  =\frac{1}{d^{2}%
-1}\left[
\begin{array}
[c]{c}%
d\left(  I_{A}+I_{A}\right) \\
-\left(  dI_{A}+dI_{A}\right)
\end{array}
\right] \\
&  =0,\\
\operatorname{Tr}_{B_{1}B_{2}}[S_{AB_{1}B_{2}}^{2}]  &  =\frac{1}{\sqrt
{d^{2}-1}}\left(  dI_{A}-dI_{A}\right) \\
&  =0,\\
\operatorname{Tr}_{B_{1}B_{2}}[S_{AB_{1}B_{2}}^{3}]  &  =\frac{i}{\sqrt
{d^{2}-1}}\left(  I_{A}-I_{A}\right) \\
&  =0.
\end{align}
Combining with \eqref{eq:TP-cond-1-1}--\eqref{eq:TP-cond-1-2}, we conclude
that%
\begin{align}
&  I_{A\hat{A}\hat{B}_{1}\hat{B}_{2}}\nonumber\\
&  =M_{\hat{A}\hat{B}_{1}\hat{B}_{2}}^{\prime+}\otimes\frac{\operatorname{Tr}%
_{B_{1}B_{2}}[S_{AB_{1}B_{2}}^{+}]}{\operatorname{Tr}[S_{AB_{1}B_{2}}^{g(+)}%
]}\nonumber\\
&  \qquad+M_{\hat{A}\hat{B}_{1}\hat{B}_{2}}^{\prime-}\otimes\frac
{\operatorname{Tr}_{B_{1}B_{2}}[S_{AB_{1}B_{2}}^{-}]}{\operatorname{Tr}%
[S_{AB_{1}B_{2}}^{g(-)}]}\nonumber\\
&  \qquad+M_{\hat{A}\hat{B}_{1}\hat{B}_{2}}^{\prime0}\otimes\frac
{\operatorname{Tr}_{B_{1}B_{2}}[S_{AB_{1}B_{2}}^{0}]}{\operatorname{Tr}%
[S_{AB_{1}B_{2}}^{g(0)}]}\\
&  =M_{\hat{A}\hat{B}_{1}\hat{B}_{2}}^{\prime+}\otimes\frac{\frac{\left(
d+2\right)  \left(  d-1\right)  }{2}I_{A}}{\frac{d\left(  d+2\right)  \left(
d-1\right)  }{2}}\nonumber\\
&  \qquad+M_{\hat{A}\hat{B}_{1}\hat{B}_{2}}^{\prime-}\otimes\frac
{\frac{\left(  d-2\right)  \left(  d+1\right)  }{2}I_{A}}{\frac{d\left(
d-2\right)  \left(  d+1\right)  }{2}}\nonumber\\
&  \qquad+M_{\hat{A}\hat{B}_{1}\hat{B}_{2}}^{\prime0}\otimes\frac{2I_{A}}%
{2d}\\
&  =\frac{1}{d}\left(  M_{\hat{A}\hat{B}_{1}\hat{B}_{2}}^{\prime+}+M_{\hat
{A}\hat{B}_{1}\hat{B}_{2}}^{\prime-}+M_{\hat{A}\hat{B}_{1}\hat{B}_{2}}%
^{\prime0}\right)  \otimes I_{A},
\end{align}
which is equivalent to%
\begin{equation}
I_{\hat{A}\hat{B}_{1}\hat{B}_{2}}=\frac{1}{d}\left(  M_{\hat{A}\hat{B}_{1}%
\hat{B}_{2}}^{\prime+}+M_{\hat{A}\hat{B}_{1}\hat{B}_{2}}^{\prime-}+M_{\hat
{A}\hat{B}_{1}\hat{B}_{2}}^{\prime0}\right)  .
\end{equation}
Observing \eqref{eq:M-ops-redef-1} and \eqref{eq:M-ops-redef-2}, we conclude
that%
\begin{equation}
M_{\hat{A}\hat{B}_{1}\hat{B}_{2}}^{+}+M_{\hat{A}\hat{B}_{1}\hat{B}_{2}}%
^{-}+M_{\hat{A}\hat{B}_{1}\hat{B}_{2}}^{0}=I_{\hat{A}\hat{B}_{1}\hat{B}_{2}}.
\label{eq:meas-ops-conclusion}%
\end{equation}
We then conclude that \eqref{eq:TP-map-Choi-PSD-2-ext} holds if and only if
\eqref{eq:meas-ops-conclusion} holds. This justifies the constraint in \eqref{eq:TP-constr-tele}.

Note that we can interpret $\{M_{\hat
{A}\hat{B}_{1}\hat{B}_{2}}^{+},M_{\hat{A}\hat{B}_{1}\hat{B}_{2}}^{-}%
,M_{\hat{A}\hat{B}_{1}\hat{B}_{2}}^{0}\}$ as a POVM, due to the constraints in
\eqref{eq:meas-ops-conclusion} and
\eqref{eq:PSD-constr-M-ops-1}--\eqref{eq:PSD-constr-M-ops-3}, with the last
implying that $M_{\hat{A}\hat{B}_{1}\hat{B}_{2}}^{0}\geq0$, as indicated after~\eqref{eq:PSD-constr-M-ops-3}.

\subsection{Non-signaling condition}

\label{app:non-sig-tele}

We now consider the condition in \eqref{eq:marginal-Choi-PSD-2-ext}, as
applied to \eqref{eq:sym-form-for-Choi-op-2-ext}. Consider that%
\begin{multline}
\operatorname{Tr}_{B_{2}}[\widetilde{M}_{A\hat{A}\hat{B}_{1}\hat{B}_{2}%
B_{1}B_{2}}]\label{eq:1st-part-no-sig}\\
=\sum_{i\in\left\{  +,-,0,1,2,3\right\}  }M_{\hat{A}\hat{B}_{1}\hat{B}_{2}%
}^{\prime i}\otimes\frac{\operatorname{Tr}_{B_{2}}[S_{AB_{1}B_{2}}^{i}%
]}{\operatorname{Tr}[S_{AB_{1}B_{2}}^{g(i)}]}%
\end{multline}
Thus, we need to calculate $\operatorname{Tr}_{B_{2}}[S_{AB_{1}B_{2}}^{i}]$
for all $i\in\left\{  +,-,0,1,2,3\right\}  $. To do so, we first find that%
\begin{align}
\operatorname{Tr}_{B_{2}}[I_{AB_{1}B_{2}}]  &  =dI_{AB_{1}},\\
\operatorname{Tr}_{B_{2}}[V_{B_{1}B_{2}}]  &  =I_{AB_{1}},\\
\operatorname{Tr}_{B_{2}}[V_{AB_{1}}^{T_{A}}]  &  =d\Gamma_{AB_{1}}=d^{2}%
\Phi_{AB_{1}},\\
\operatorname{Tr}_{B_{2}}[V_{AB_{2}}^{T_{A}}]  &  =I_{AB_{1}},\\
\operatorname{Tr}_{B_{2}}[V_{AB_{1}B_{2}}^{T_{A}}]  &  =\Gamma_{AB_{1}}%
=d\Phi_{AB_{1}},\\
\operatorname{Tr}_{B_{2}}[V_{B_{2}B_{1}A}^{T_{A}}]  &  =\Gamma_{AB_{1}}%
=d\Phi_{AB_{1}},
\end{align}
which implies from \eqref{eq:S-op-+}--\eqref{eq:S-op-3}, that%
\begin{align}
&  \operatorname{Tr}_{B_{2}}[S_{AB_{1}B_{2}}^{+}]\nonumber\\
&  =\frac{1}{2}\left[
\begin{array}
[c]{c}%
dI_{AB_{1}}+I_{AB_{1}}\\
-\left(  \frac{d^{2}\Phi_{AB_{1}}+I_{AB_{1}}+d\Phi_{AB_{1}}+d\Phi_{AB_{1}}%
}{d+1}\right)
\end{array}
\right] \nonumber\\
&  =\frac{1}{2}\left[  \left(  d+1-\frac{1}{d+1}\right)  I_{AB_{1}}-\left(
\frac{d^{2}+2d}{d+1}\right)  \Phi_{AB_{1}}\right] \nonumber\\
&  =\frac{d\left(  d+2\right)  }{2\left(  d+1\right)  }\left(  I_{AB_{1}}%
-\Phi_{AB_{1}}\right)  ,\\
&  \operatorname{Tr}_{B_{2}}[S_{AB_{1}B_{2}}^{-}]\nonumber\\
&  =\frac{1}{2}\left[
\begin{array}
[c]{c}%
dI_{AB_{1}}-I_{AB_{1}}\\
-\left(  \frac{d^{2}\Phi_{AB_{1}}+I_{AB_{1}}-d\Phi_{AB_{1}}-d\Phi_{AB_{1}}%
}{d-1}\right)
\end{array}
\right] \nonumber\\
&  =\frac{1}{2}\left[  \left(  d-1-\frac{1}{d-1}\right)  I_{AB_{1}}-\left(
\frac{d^{2}-2d}{d-1}\right)  \Phi_{AB_{1}}\right] \nonumber\\
&  =\frac{d\left(  d-2\right)  }{2\left(  d-1\right)  }\left(  I_{AB_{1}}%
-\Phi_{AB_{1}}\right)  ,\\
&  \operatorname{Tr}_{B_{2}}[S_{AB_{1}B_{2}}^{0}]\nonumber\\
&  =\frac{1}{d^{2}-1}\left[
\begin{array}
[c]{c}%
d\left(  d^{2}\Phi_{AB_{1}}+I_{AB_{1}}\right) \\
-\left(  d\Phi_{AB_{1}}+d\Phi_{AB_{1}}\right)
\end{array}
\right] \nonumber\\
&  =\frac{1}{d^{2}-1}\left[  dI_{AB_{1}}+\left(  d^{3}-2d\right)  \Phi
_{AB_{1}}\right] \nonumber\\
&  =\frac{1}{d^{2}-1}\left[  d\left(  I_{AB_{1}}-\Phi_{AB_{1}}\right)
+\left(  d^{3}-d\right)  \Phi_{AB_{1}}\right] \nonumber\\
&  =\frac{d}{d^{2}-1}\left(  I_{AB_{1}}-\Phi_{AB_{1}}\right)  +d\Phi_{AB_{1}},
\end{align}%
\begin{align}
&  \operatorname{Tr}_{B_{2}}[S_{AB_{1}B_{2}}^{1}]\nonumber\\
&  =\frac{1}{d^{2}-1}\left[
\begin{array}
[c]{c}%
d\left(  d\Phi_{AB_{1}}+d\Phi_{AB_{1}}\right) \\
-\left(  d^{2}\Phi_{AB_{1}}+I_{AB_{1}}\right)
\end{array}
\right] \nonumber\\
&  =\frac{1}{d^{2}-1}\left[  d^{2}\Phi_{AB_{1}}-I_{AB_{1}}\right] \nonumber\\
&  =\frac{1}{d^{2}-1}\left[  \left(  d^{2}-1\right)  \Phi_{AB_{1}}-\left(
I_{AB_{1}}-\Phi_{AB_{1}}\right)  \right] \nonumber\\
&  =\Phi_{AB_{1}}-\frac{1}{d^{2}-1}\left(  I_{AB_{1}}-\Phi_{AB_{1}}\right)
,\\
&  \operatorname{Tr}_{B_{2}}[S_{AB_{1}B_{2}}^{2}]\nonumber\\
&  =\frac{1}{\sqrt{d^{2}-1}}\left(  d^{2}\Phi_{AB_{1}}-I_{AB_{1}}\right)
\nonumber\\
&  =\frac{1}{\sqrt{d^{2}-1}}\left(  \left(  d^{2}-1\right)  \Phi_{AB_{1}%
}-\left(  I_{AB_{1}}-\Phi_{AB_{1}}\right)  \right) \nonumber\\
&  =\sqrt{d^{2}-1}\Phi_{AB_{1}}-\frac{I_{AB_{1}}-\Phi_{AB_{1}}}{\sqrt{d^{2}%
-1}},\\
&  \operatorname{Tr}_{B_{2}}[S_{AB_{1}B_{2}}^{3}]\nonumber\\
&  =\frac{i}{\sqrt{d^{2}-1}}\left(  d\Phi_{AB_{1}}-d\Phi_{AB_{1}}\right)
\nonumber\\
&  =0,
\end{align}
We thus conclude from \eqref{eq:1st-part-no-sig}\ and the above that%
\begin{align}
&  \operatorname{Tr}_{B_{2}}[\widetilde{M}_{A\hat{A}\hat{B}_{1}\hat{B}%
_{2}B_{1}B_{2}}]\nonumber\\
&  =\sum_{i\in\left\{  +,-,0,1,2,3\right\}  }M_{\hat{A}\hat{B}_{1}\hat{B}_{2}%
}^{\prime i}\otimes\frac{\operatorname{Tr}_{B_{2}}[S_{AB_{1}B_{2}}^{i}%
]}{\operatorname{Tr}[S_{AB_{1}B_{2}}^{g(i)}]}\\
&  =M_{\hat{A}\hat{B}_{1}\hat{B}_{2}}^{\prime+}\otimes\frac{d\left(
d+2\right)  }{2\left(  d+1\right)  }\frac{\left(  I_{AB_{1}}-\Phi_{AB_{1}%
}\right)  }{\frac{d\left(  d+2\right)  \left(  d-1\right)  }{2}}\nonumber\\
&  +M_{\hat{A}\hat{B}_{1}\hat{B}_{2}}^{\prime-}\otimes\frac{\frac{d\left(
d-2\right)  }{2\left(  d-1\right)  }\left(  I_{AB_{1}}-\Phi_{AB_{1}}\right)
}{\frac{d\left(  d-2\right)  \left(  d+1\right)  }{2}}\nonumber\\
&  +M_{\hat{A}\hat{B}_{1}\hat{B}_{2}}^{\prime0}\otimes\frac{\frac{d}{d^{2}%
-1}\left(  I_{AB_{1}}-\Phi_{AB_{1}}\right)  +d\Phi_{AB_{1}}}{2d}\nonumber\\
&  +M_{\hat{A}\hat{B}_{1}\hat{B}_{2}}^{\prime1}\otimes\frac{\Phi_{AB_{1}%
}-\frac{1}{d^{2}-1}\left(  I_{AB_{1}}-\Phi_{AB_{1}}\right)  }{2d}\nonumber\\
&  +M_{\hat{A}\hat{B}_{1}\hat{B}_{2}}^{\prime2}\otimes\frac{\sqrt{d^{2}-1}%
\Phi_{AB_{1}}-\frac{I_{AB_{1}}-\Phi_{AB_{1}}}{\sqrt{d^{2}-1}}}{2d}%
\end{align}%
\begin{align}
&  =M_{\hat{A}\hat{B}_{1}\hat{B}_{2}}^{+}\otimes\frac{d}{\left(  d+1\right)
\left(  d-1\right)  }\left(  I_{AB_{1}}-\Phi_{AB_{1}}\right) \nonumber\\
&  +M_{\hat{A}\hat{B}_{1}\hat{B}_{2}}^{-}\otimes\frac{d}{\left(  d-1\right)
\left(  d+1\right)  }\left(  I_{AB_{1}}-\Phi_{AB_{1}}\right) \nonumber\\
&  +M_{\hat{A}\hat{B}_{1}\hat{B}_{2}}^{0}\otimes\left(  \frac{d}{2\left(
d^{2}-1\right)  }\left(  I_{AB_{1}}-\Phi_{AB_{1}}\right)  +\frac{d}{2}%
\Phi_{AB_{1}}\right) \nonumber\\
&  +M_{\hat{A}\hat{B}_{1}\hat{B}_{2}}^{1}\otimes\left(  \frac{1}{2}%
\Phi_{AB_{1}}-\frac{1}{2\left(  d^{2}-1\right)  }\left(  I_{AB_{1}}%
-\Phi_{AB_{1}}\right)  \right) \nonumber\\
&  +M_{\hat{A}\hat{B}_{1}\hat{B}_{2}}^{2}\otimes\left(  \frac{\sqrt{d^{2}-1}%
}{2}\Phi_{AB_{1}}-\frac{I_{AB_{1}}-\Phi_{AB_{1}}}{2\sqrt{d^{2}-1}}\right)
\end{align}%
\begin{align}
&  =M_{\hat{A}\hat{B}_{1}\hat{B}_{2}}^{+}\otimes d\left(  \frac{I_{AB_{1}%
}-\Phi_{AB_{1}}}{d^{2}-1}\right) \nonumber\\
&  +M_{\hat{A}\hat{B}_{1}\hat{B}_{2}}^{-}\otimes d\left(  \frac{I_{AB_{1}%
}-\Phi_{AB_{1}}}{d^{2}-1}\right) \nonumber\\
&  +M_{\hat{A}\hat{B}_{1}\hat{B}_{2}}^{0}\otimes\left(  \frac{d}{2}\left(
\frac{I_{AB_{1}}-\Phi_{AB_{1}}}{d^{2}-1}\right)  +\frac{d}{2}\Phi_{AB_{1}%
}\right) \nonumber\\
&  +M_{\hat{A}\hat{B}_{1}\hat{B}_{2}}^{1}\otimes\left(  \frac{1}{2}%
\Phi_{AB_{1}}-\frac{1}{2}\left(  \frac{I_{AB_{1}}-\Phi_{AB_{1}}}{d^{2}%
-1}\right)  \right) \nonumber\\
&  +M_{\hat{A}\hat{B}_{1}\hat{B}_{2}}^{2}\otimes\left(
\begin{array}
[c]{c}%
\frac{\sqrt{d^{2}-1}}{2}\Phi_{AB_{1}}\\
-\frac{\sqrt{d^{2}-1}}{2}\left(  \frac{I_{AB_{1}}-\Phi_{AB_{1}}}{d^{2}%
-1}\right)
\end{array}
\right)
\end{align}%
\begin{equation}
=P_{\hat{A}\hat{B}_{1}\hat{B}_{2}}^{\prime}\otimes\Phi_{AB_{1}}+Q_{\hat{A}%
\hat{B}_{1}\hat{B}_{2}}^{\prime}\otimes\left(  \frac{I_{AB_{1}}-\Phi_{AB_{1}}%
}{d^{2}-1}\right)  , \label{eq:no-sig-reduce-1}%
\end{equation}
where%
\begin{equation}
P_{\hat{A}\hat{B}_{1}\hat{B}_{2}}^{\prime}\coloneqq\frac{1}{2}\left[
\begin{array}
[c]{c}%
dM_{\hat{A}\hat{B}_{1}\hat{B}_{2}}^{0}+M_{\hat{A}\hat{B}_{1}\hat{B}_{2}}^{1}\\
+\sqrt{d^{2}-1}M_{\hat{A}\hat{B}_{1}\hat{B}_{2}}^{2}%
\end{array}
\right]  , \label{eq:Q-op-2-ext}%
\end{equation}%
\begin{multline}
Q_{\hat{A}\hat{B}_{1}\hat{B}_{2}}^{\prime}\coloneqq\label{eq:P-op-2-ext}\\
\frac{1}{2}\left[
\begin{array}
[c]{c}%
2d\left(  M_{\hat{A}\hat{B}_{1}\hat{B}_{2}}^{+}+M_{\hat{A}\hat{B}_{1}\hat
{B}_{2}}^{-}\right)  +dM_{\hat{A}\hat{B}_{1}\hat{B}_{2}}^{0}\\
-M_{\hat{A}\hat{B}_{1}\hat{B}_{2}}^{1}-\sqrt{d^{2}-1}M_{\hat{A}\hat{B}_{1}%
\hat{B}_{2}}^{2}%
\end{array}
\right] \\
=dI_{\hat{A}\hat{B}_{1}\hat{B}_{2}}-P_{\hat{A}\hat{B}_{1}\hat{B}_{2}}^{\prime
},
\end{multline}
where the last equality follows from \eqref{eq:meas-ops-conclusion}. Now
consider that%
\begin{multline}
\frac{1}{d_{\hat{B}}}\operatorname{Tr}_{\hat{B}_{2}B_{2}}[\widetilde{M}%
_{A\hat{A}\hat{B}_{1}\hat{B}_{2}B_{1}B_{2}}]\otimes I_{\hat{B}_{2}}\\
=\sum_{i\in\left\{  +,-,0,1,2,3\right\}  }\frac{1}{d_{\hat{B}}}%
\operatorname{Tr}_{\hat{B}_{2}}[M_{\hat{A}\hat{B}_{1}\hat{B}_{2}}^{\prime
i}]\otimes I_{\hat{B}_{2}}\otimes\frac{\operatorname{Tr}_{B_{2}}%
[S_{AB_{1}B_{2}}^{i}]}{\operatorname{Tr}[S_{AB_{1}B_{2}}^{g(i)}]},
\end{multline}
which reduces to%
\begin{multline}
\frac{1}{d_{\hat{B}}}\operatorname{Tr}_{\hat{B}_{2}}[P_{\hat{A}\hat{B}_{1}%
\hat{B}_{2}}^{\prime}]\otimes I_{\hat{B}_{2}}\otimes\Phi_{AB_{1}%
}\label{eq:no-sig-reduce-2}\\
+\frac{1}{d_{\hat{B}}}\operatorname{Tr}_{\hat{B}_{2}}[Q_{\hat{A}\hat{B}%
_{1}\hat{B}_{2}}^{\prime}]\otimes I_{\hat{B}_{2}}\otimes\left(  \frac
{I_{AB_{1}}-\Phi_{AB_{1}}}{d^{2}-1}\right)  ,
\end{multline}
Then the constraint in \eqref{eq:marginal-Choi-PSD-2-ext} reduces to the
following two constraints:%
\begin{align}
P_{\hat{A}\hat{B}_{1}\hat{B}_{2}}^{\prime}  &  =\frac{1}{d_{\hat{B}}%
}\operatorname{Tr}_{\hat{B}_{2}}[P_{\hat{A}\hat{B}_{1}\hat{B}_{2}}^{\prime
}]\otimes I_{\hat{B}_{2}},\label{eq:P'-non-sig-const-nonred}\\
Q_{\hat{A}\hat{B}_{1}\hat{B}_{2}}^{\prime}  &  =\frac{1}{d_{\hat{B}}%
}\operatorname{Tr}_{\hat{B}_{2}}[Q_{\hat{A}\hat{B}_{1}\hat{B}_{2}}^{\prime
}]\otimes I_{\hat{B}_{2}}, \label{eq:Q'-non-sig-constr-redundant}%
\end{align}
by applying \eqref{eq:no-sig-reduce-1}\ and \eqref{eq:no-sig-reduce-2},\ and
using the orthogonality of $\Phi_{AB_{1}}$ and $\frac{I_{AB_{1}}-\Phi_{AB_{1}%
}}{d^{2}-1}$. The second constraint in \eqref{eq:Q'-non-sig-constr-redundant}
is redundant, following from the first one in
\eqref{eq:P'-non-sig-const-nonred}, because%
\begin{equation}
Q_{\hat{A}\hat{B}_{1}\hat{B}_{2}}^{\prime}=dI_{\hat{A}\hat{B}_{1}\hat{B}_{2}%
}-P_{\hat{A}\hat{B}_{1}\hat{B}_{2}}^{\prime}.
\end{equation}
This justifies the constraint in \eqref{eq:no-sig-cond-tele-simple-sdp}, up to an inconsequential scale factor.

\subsection{Permutation covariance condition}

\label{Sec_perm_cov_tel}

We now consider the condition in \eqref{eq:perm-cov-Choi-2-ext}, as applied to
\eqref{eq:sym-form-for-Choi-op-2-ext}. Consider that%
\begin{multline}
(\mathcal{F}_{\hat{B}_{1}\hat{B}_{2}}\otimes\mathcal{F}_{B_{1}B_{2}%
})(\widetilde{M}_{A\hat{A}\hat{B}_{1}\hat{B}_{2}B_{1}B_{2}})\\
=\sum_{i\in\left\{  +,-,0,1,2,3\right\}  }\mathcal{F}_{\hat{B}_{1}\hat{B}_{2}%
}(M_{\hat{A}\hat{B}_{1}\hat{B}_{2}}^{\prime i})\otimes\frac{\mathcal{F}%
_{B_{1}B_{2}}(S_{AB_{1}B_{2}}^{i})}{\operatorname{Tr}[S_{AB_{1}B_{2}}^{g(i)}%
]}.
\end{multline}
Observing from
\eqref{eq:partial-transpose-perm-ops-1}--\eqref{eq:partial-transpose-perm-ops-4}\ that%
\begin{align}
\mathcal{F}_{B_{1}B_{2}}(I_{AB_{1}B_{2}})  &  =I_{AB_{1}B_{2}},\\
\mathcal{F}_{B_{1}B_{2}}(V_{B_{1}B_{2}})  &  =V_{B_{1}B_{2}},\\
\mathcal{F}_{B_{1}B_{2}}(V_{AB_{1}}^{T_{A}})  &  =V_{AB_{2}}^{T_{A}},\\
\mathcal{F}_{B_{1}B_{2}}(V_{AB_{2}}^{T_{A}})  &  =V_{AB_{1}}^{T_{A}},\\
\mathcal{F}_{B_{1}B_{2}}(V_{AB_{1}B_{2}}^{T_{A}})  &  =V_{B_{2}B_{1}A}^{T_{A}%
},\\
\mathcal{F}_{B_{1}B_{2}}(V_{B_{2}B_{1}A}^{T_{A}})  &  =V_{AB_{1}B_{2}}^{T_{A}%
},
\end{align}
we conclude from \eqref{eq:S-op-+}--\eqref{eq:S-op-3} that%
\begin{align}
\mathcal{F}_{B_{1}B_{2}}(S_{AB_{1}B_{2}}^{+})  &  =S_{AB_{1}B_{2}}^{+},\\
\mathcal{F}_{B_{1}B_{2}}(S_{AB_{1}B_{2}}^{-})  &  =S_{AB_{1}B_{2}}^{-},\\
\mathcal{F}_{B_{1}B_{2}}(S_{AB_{1}B_{2}}^{0})  &  =S_{AB_{1}B_{2}}^{0},\\
\mathcal{F}_{B_{1}B_{2}}(S_{AB_{1}B_{2}}^{1})  &  =S_{AB_{1}B_{2}}^{1},\\
\mathcal{F}_{B_{1}B_{2}}(S_{AB_{1}B_{2}}^{2})  &  =-S_{AB_{1}B_{2}}^{2},\\
\mathcal{F}_{B_{1}B_{2}}(S_{AB_{1}B_{2}}^{3})  &  =-S_{AB_{1}B_{2}}^{3}.
\end{align}
This implies that the condition in \eqref{eq:perm-cov-Choi-2-ext} is
equivalent to%
\begin{multline}
\sum_{i\in\left\{  +,-,0,1,2,3\right\}  }M_{\hat{A}\hat{B}_{1}\hat{B}_{2}%
}^{\prime i}\otimes\frac{S_{AB_{1}B_{2}}^{i}}{\operatorname{Tr}[S_{AB_{1}%
B_{2}}^{g(i)}]}\label{eq:perm-cov-cond-S-ops}\\
=\sum_{i\in\left\{  +,-,0,1,2,3\right\}  }\mathcal{F}_{\hat{B}_{1}\hat{B}_{2}%
}(M_{\hat{A}\hat{B}_{1}\hat{B}_{2}}^{\prime i})\otimes\frac{\left(  -1\right)
^{f(i)}S_{AB_{1}B_{2}}^{i}}{\operatorname{Tr}[S_{AB_{1}B_{2}}^{g(i)}]},
\end{multline}
where%
\begin{equation}
f(i)=\left\{
\begin{array}
[c]{cc}%
0 & \text{if }i\in\left\{  +,-,0,1\right\} \\
1 & \text{if }i\in\left\{  2,3\right\}
\end{array}
\right.  .
\end{equation}
Due to the fact that $S_{AB_{1}B_{2}}^{+}$, $S_{AB_{1}B_{2}}^{-}$, and
$S_{AB_{1}B_{2}}^{0}$ are orthogonal projectors, and $S_{AB_{1}B_{2}}^{1}$,
$S_{AB_{1}B_{2}}^{2}$, and $S_{AB_{1}B_{2}}^{3}$ only act on the subspace onto
which $S_{AB_{1}B_{2}}^{0}$ projects, the condition in
\eqref{eq:perm-cov-cond-S-ops} is equivalent to the following three
conditions:%
\begin{align}
M_{\hat{A}\hat{B}_{1}\hat{B}_{2}}^{+}  &  =\mathcal{F}_{\hat{B}_{1}\hat{B}%
_{2}}(M_{\hat{A}\hat{B}_{1}\hat{B}_{2}}^{+}),\\
M_{\hat{A}\hat{B}_{1}\hat{B}_{2}}^{-}  &  =\mathcal{F}_{\hat{B}_{1}\hat{B}%
_{2}}(M_{\hat{A}\hat{B}_{1}\hat{B}_{2}}^{-}),
\end{align}%
\begin{multline}%
\begin{bmatrix}
M^{0}+M^{3} & M^{1}-iM^{2}\\
M^{1}+iM^{2} & M^{0}-M^{3}%
\end{bmatrix}
\label{eq:swap-conds-paulis}\\
=%
\begin{bmatrix}
\mathcal{F}(M^{0})-\mathcal{F}(M^{3}) & \mathcal{F}(M^{1})+i\mathcal{F}%
(M^{2})\\
\mathcal{F}(M^{1})-i\mathcal{F}(M^{2}) & \mathcal{F}(M^{0})+\mathcal{F}(M^{3})
\end{bmatrix}
,
\end{multline}
where we have employed the abbreviation $\mathcal{F}\equiv\mathcal{F}_{\hat
{B}_{1}\hat{B}_{2}}$ and the isometry in
\eqref{eq:isometric-map-sig-to-S}--\eqref{eq:isometric-map-sig-to-S-2}. By considering that \eqref{eq:swap-conds-paulis} is equivalent to
\begin{align}
M^{0}+M^{3} & = \mathcal{F}(M^{0})-\mathcal{F}(M^{3}), \\
M^{1}-iM^{2} &= \mathcal{F}(M^{1})+i\mathcal{F}%
(M^{2}), \\ 
M^{1}+iM^{2} &= \mathcal{F}(M^{1})-i\mathcal{F}(M^{2}), \\
M^{0}-M^{3} &= \mathcal{F}(M^{0})+\mathcal{F}(M^{3}),
\end{align}
and adding and subtracting these equations, we conclude that 
\eqref{eq:swap-conds-paulis} is actually equivalent to the
following four conditions:%
\begin{align}
M_{\hat{A}\hat{B}_{1}\hat{B}_{2}}^{0}  &  =\mathcal{F}_{\hat{B}_{1}\hat{B}%
_{2}}(M_{\hat{A}\hat{B}_{1}\hat{B}_{2}}^{0}),\\
M_{\hat{A}\hat{B}_{1}\hat{B}_{2}}^{1}  &  =\mathcal{F}_{\hat{B}_{1}\hat{B}%
_{2}}(M_{\hat{A}\hat{B}_{1}\hat{B}_{2}}^{1}),\\
M_{\hat{A}\hat{B}_{1}\hat{B}_{2}}^{2}  &  =-\mathcal{F}_{\hat{B}_{1}\hat
{B}_{2}}(M_{\hat{A}\hat{B}_{1}\hat{B}_{2}}^{2}),\\
M_{\hat{A}\hat{B}_{1}\hat{B}_{2}}^{3}  &  =-\mathcal{F}_{\hat{B}_{1}\hat
{B}_{2}}(M_{\hat{A}\hat{B}_{1}\hat{B}_{2}}^{3}).
\end{align}
This justifies the constraints in \eqref{eq:perm-cov-constr-tele-1}--\eqref{eq:perm-cov-constr-tele-2}.

At this point, let us summarize the results of Appendices~\ref{Complete_positivity_tel} through \ref{Sec_perm_cov_tel}: we have reduced the constraints in
\eqref{eq:CP-map-Choi-PSD-2-ext}--\eqref{eq:perm-cov-Choi-2-ext}\ to the
following ones:%
\[%
\begin{bmatrix}
M^{0}+M^{3} & M^{1}-iM^{2}\\
M^{1}+iM^{2} & M^{0}-M^{3}%
\end{bmatrix}
\geq0,
\]%
\begin{align}
M_{\hat{A}\hat{B}_{1}\hat{B}_{2}}^{+}  &  \geq0,\\
M_{\hat{A}\hat{B}_{1}\hat{B}_{2}}^{-}  &  \geq0,\\
I_{\hat{A}\hat{B}_{1}\hat{B}_{2}}  &  =M_{\hat{A}\hat{B}_{1}\hat{B}_{2}}%
^{+}+M_{\hat{A}\hat{B}_{1}\hat{B}_{2}}^{-}+M_{\hat{A}\hat{B}_{1}\hat{B}_{2}%
}^{0}\\
P_{\hat{A}\hat{B}_{1}\hat{B}_{2}}^{\prime}  &  =\operatorname{Tr}_{\hat{B}%
_{2}}[P_{\hat{A}\hat{B}_{1}\hat{B}_{2}}^{\prime}]\otimes\pi_{\hat{B}_{2}},\\
M_{\hat{A}\hat{B}_{1}\hat{B}_{2}}^{i}  &  =\mathcal{F}_{\hat{B}_{1}\hat{B}%
_{2}}(M_{\hat{A}\hat{B}_{1}\hat{B}_{2}}^{i})\quad\forall i\in\left\{
+,-,0,1\right\}  ,\\
M_{\hat{A}\hat{B}_{1}\hat{B}_{2}}^{j}  &  =-\mathcal{F}_{\hat{B}_{1}\hat
{B}_{2}}(M_{\hat{A}\hat{B}_{1}\hat{B}_{2}}^{j})\quad\forall j\in\left\{
2,3\right\}  ,
\end{align}
where $P_{\hat{A}\hat{B}_{1}\hat{B}_{2}}^{\prime}$ is defined in \eqref{eq:Q-op-2-ext}.

\subsection{PPT\ constraints}\label{Sec_PPT_constraints}

We now consider the PPT constraints. Suppose that $\mathcal{K}_{A\hat{A}%
\hat{B}\rightarrow B}$ is an arbitrary two-PPT-extendible channel, meaning
that there exists an extension channel $\mathcal{M}_{A\hat{A}\hat{B}_{1}%
\hat{B}_{2}\rightarrow B_{1}B_{2}}$ satisfying the constraints considered
in Appendices~\ref{Complete_positivity_tel} through \ref{Sec_perm_cov_tel}, as well as%
\begin{align}
T_{B_{2}}\circ\mathcal{M}_{A\hat{A}\hat{B}_{1}\hat{B}_{2}\rightarrow
B_{1}B_{2}}\circ T_{\hat{B}_{2}}  &  \in\text{CP},\label{eq:app-ppt-cond-B2}\\
\mathcal{M}_{A\hat{A}\hat{B}_{1}\hat{B}_{2}\rightarrow B_{1}B_{2}}\circ
T_{A\hat{A}}  &  \in\text{CP}.
\end{align}
Then it follows that $\mathcal{K}_{A\hat{A}\hat{B}\rightarrow B}^{\mathcal{U}%
}$ and its extension $\mathcal{M}_{A\hat{A}\hat{B}_{1}\hat{B}_{2}\rightarrow
B_{1}B_{2}}^{\mathcal{U}}$, as defined in \eqref{eq:K_under_unitary} and \eqref{eq:M_under_unitary}, respectively,
satisfy the following conditions%
\begin{align}
T_{B_{2}}\circ\mathcal{M}_{A\hat{A}\hat{B}_{1}\hat{B}_{2}\rightarrow
B_{1}B_{2}}^{\mathcal{U}}\circ T_{\hat{B}_{2}}  &  \in\text{CP},\\
\mathcal{M}_{A\hat{A}\hat{B}_{1}\hat{B}_{2}\rightarrow B_{1}B_{2}%
}^{\mathcal{U}}\circ T_{A\hat{A}}  &  \in\text{CP},
\end{align}
because%
\begin{align}
&  T_{B_{2}}\circ\mathcal{M}_{A\hat{A}\hat{B}_{1}\hat{B}_{2}\rightarrow
B_{1}B_{2}}^{\mathcal{U}}\circ T_{\hat{B}_{2}}\nonumber\\
&  =T_{B_{2}}\circ\left(  \mathcal{U}_{B_{1}}^{\dag}\otimes\mathcal{U}_{B_{2}%
}^{\dag}\right)  \circ\mathcal{M}_{A\hat{A}\hat{B}_{1}\hat{B}_{2}\rightarrow
B_{1}B_{2}}\circ\mathcal{U}_{A}\circ T_{\hat{B}_{2}}\\
&  =T_{B_{2}}\circ\left(  \mathcal{U}_{B_{1}}^{\dag}\otimes\mathcal{U}_{B_{2}%
}^{\dag}\right)  \circ T_{B_{2}}\circ\nonumber\\
&  \qquad T_{B_{2}}\circ\mathcal{M}_{A\hat{A}\hat{B}_{1}\hat{B}_{2}\rightarrow
B_{1}B_{2}}\circ\mathcal{U}_{A}\circ T_{\hat{B}_{2}}\\
&  =\left(  \mathcal{U}_{B_{1}}^{\dag}\otimes\lbrack T_{B_{2}}\circ
\mathcal{U}_{B_{2}}^{\dag}\circ T_{B_{2}}]\right)  \circ\nonumber\\
&  \qquad\lbrack T_{B_{2}}\circ\mathcal{M}_{A\hat{A}\hat{B}_{1}\hat{B}%
_{2}\rightarrow B_{1}B_{2}}\circ T_{\hat{B}_{2}}]\circ\mathcal{U}_{A}%
\end{align}
and the maps $T_{B_{2}}\circ\mathcal{U}_{B_{2}}^{\dag}\circ T_{B_{2}}$ and
$T_{B_{2}}\circ\mathcal{M}_{A\hat{A}\hat{B}_{1}\hat{B}_{2}\rightarrow
B_{1}B_{2}}\circ T_{\hat{B}_{2}}$ are completely positive. Additionally,%
\begin{align}
&  \mathcal{M}_{A\hat{A}\hat{B}_{1}\hat{B}_{2}\rightarrow B_{1}B_{2}%
}^{\mathcal{U}}\circ T_{A\hat{A}}\nonumber\\
&  =\left(  \mathcal{U}_{B_{1}}^{\dag}\otimes\mathcal{U}_{B_{2}}^{\dag
}\right)  \circ\mathcal{M}_{A\hat{A}\hat{B}_{1}\hat{B}_{2}\rightarrow
B_{1}B_{2}}\circ\mathcal{U}_{A}\circ T_{A\hat{A}}\\
&  =\left(  \mathcal{U}_{B_{1}}^{\dag}\otimes\mathcal{U}_{B_{2}}^{\dag
}\right)  \circ\lbrack\mathcal{M}_{A\hat{A}\hat{B}_{1}\hat{B}_{2}\rightarrow
B_{1}B_{2}}\circ T_{A\hat{A}}]\nonumber\\
&  \qquad\circ\lbrack T_{A}\circ\mathcal{U}_{A}\circ T_{A}],
\end{align}
and the maps $\mathcal{M}_{A\hat{A}\hat{B}_{1}\hat{B}_{2}\rightarrow
B_{1}B_{2}}\circ T_{A\hat{A}}$ and $T_{A}\circ\mathcal{U}_{A}\circ T_{A}$ are
completely positive. Thus, $\mathcal{K}_{A\hat{A}\hat{B}\rightarrow
B}^{\mathcal{U}}$ is a two-PPT-extendible channel with extension
$\mathcal{M}_{A\hat{A}\hat{B}_{1}\hat{B}_{2}\rightarrow B_{1}B_{2}%
}^{\mathcal{U}}$. By linearity, the same is true for $\widetilde{\mathcal{K}}_{A\hat{A}\hat{B}\rightarrow
B}$ and $\widetilde{\mathcal{M}}_{A\hat{A}\hat{B}_{1}\hat{B}_{2}\rightarrow B_{1}B_{2}%
}$, as defined in \eqref{eq:twirled-2-ext-1} and \eqref{eq:twirled-2-ext-2}, respectively. Applying the same arguments used to justify
\eqref{eq:convexity-dd-2-ext}, it suffices to perform the diamond-distance
minimization with respect to twirled two-PPT-extendible channels having the
form in \eqref{eq:twirled-2-ext-1}, with extension channel of the form in~\eqref{eq:twirled-2-ext-2}.

Since we have already identified what the twirling symmetry imposes for trace
preservation, complete positivity, non-signaling, and permutation covariance,
it remains to identify what the twirling symmetry imposes for the
PPT\ conditions. Consider that%
\begin{multline}
T_{A\hat{A}}(\widetilde{M}_{A\hat{A}\hat{B}_{1}\hat{B}_{2}B_{1}B_{2}})\\
=\sum_{i\in\left\{  +,-,0,1,2,3\right\}  }T_{\hat{A}}(M_{\hat{A}\hat{B}%
_{1}\hat{B}_{2}}^{\prime i})\otimes\frac{T_{A}(S_{AB_{1}B_{2}}^{i}%
)}{\operatorname{Tr}[S_{AB_{1}B_{2}}^{g(i)}]}.
\end{multline}
Observing from
\eqref{eq:partial-transpose-perm-ops-1}--\eqref{eq:partial-transpose-perm-ops-4}\ that%
\begin{align}
T_{A}(I_{AB_{1}B_{2}})  &  =I_{AB_{1}B_{2}},\\
T_{A}(V_{AB_{1}}^{T_{A}})  &  =V_{AB_{1}},\\
T_{A}(V_{AB_{2}}^{T_{A}})  &  =V_{AB_{2}},\\
T_{A}(V_{B_{1}B_{2}})  &  =V_{B_{1}B_{2}},\\
T_{A}(V_{AB_{1}B_{2}}^{T_{A}})  &  =V_{AB_{1}B_{2}},\\
T_{A}(V_{B_{2}B_{1}A}^{T_{A}})  &  =V_{B_{2}B_{1}A},
\end{align}
we conclude from \eqref{eq:S-op-+}--\eqref{eq:S-op-3} that%
\begin{align}
T_{A}(S_{AB_{1}B_{2}}^{+})  &  =\frac{1}{2}\left[
\begin{array}
[c]{c}%
I_{AB_{1}B_{2}}+V_{B_{1}B_{2}}\\
-\left(  \frac{V_{AB_{1}}+V_{AB_{2}}+V_{AB_{1}B_{2}}+V_{B_{2}B_{1}A}}%
{d+1}\right)
\end{array}
\right]  ,\\
T_{A}(S_{AB_{1}B_{2}}^{-})  &  =\frac{1}{2}\left[
\begin{array}
[c]{c}%
I_{AB_{1}B_{2}}-V_{B_{1}B_{2}}\\
-\left(  \frac{V_{AB_{1}}+V_{AB_{2}}-V_{AB_{1}B_{2}}-V_{B_{2}B_{1}A}}%
{d-1}\right)
\end{array}
\right]  ,\\
T_{A}(S_{AB_{1}B_{2}}^{0})  &  =\frac{1}{d^{2}-1}\left[
\begin{array}
[c]{c}%
d\left(  V_{AB_{1}}+V_{AB_{2}}\right) \\
-\left(  V_{AB_{1}B_{2}}+V_{B_{2}B_{1}A}\right)
\end{array}
\right]  ,\\
T_{A}(S_{AB_{1}B_{2}}^{1})  &  =\frac{1}{d^{2}-1}\left[
\begin{array}
[c]{c}%
d\left(  V_{AB_{1}B_{2}}+V_{B_{2}B_{1}A}\right) \\
-\left(  V_{AB_{1}}+V_{AB_{2}}\right)
\end{array}
\right]  ,\\
T_{A}(S_{AB_{1}B_{2}}^{2})  &  =\frac{1}{\sqrt{d^{2}-1}}\left(  V_{AB_{1}%
}-V_{AB_{2}}\right)  ,\\
T_{A}(S_{AB_{1}B_{2}}^{3})  &  =\frac{i}{\sqrt{d^{2}-1}}\left(  V_{AB_{1}%
B_{2}}-V_{B_{2}B_{1}A}\right)  .
\end{align}
It is helpful for us to make use of the following operators, defined in
\cite[Section~II-A]{EW01}:%
\begin{align}
R_{AB_{1}B_{2}}^{+}  &  \coloneqq\frac{1}{6}\left[
\begin{array}
[c]{c}%
I_{AB_{1}B_{2}}+V_{AB_{1}}+V_{B_{1}B_{2}}\\
+V_{AB_{2}}+V_{AB_{1}B_{2}}+V_{B_{2}B_{1}A}%
\end{array}
\right]  ,\\
R_{AB_{1}B_{2}}^{-}  &  \coloneqq\frac{1}{6}\left[
\begin{array}
[c]{c}%
I_{AB_{1}B_{2}}-V_{AB_{1}}-V_{B_{1}B_{2}}\\
-V_{AB_{2}}+V_{AB_{1}B_{2}}+V_{B_{2}B_{1}A}%
\end{array}
\right]  ,\label{eq:R_minus_def}\\
R_{AB_{1}B_{2}}^{0}  &  \coloneqq\frac{1}{3}\left[  2I_{AB_{1}B_{2}}%
-V_{AB_{1}B_{2}}-V_{B_{2}B_{1}A}\right]  ,\\
R_{AB_{1}B_{2}}^{1}  &  \coloneqq\frac{1}{3}\left[  2V_{B_{1}B_{2}}-V_{AB_{2}%
}-V_{AB_{1}}\right]  ,\\
R_{AB_{1}B_{2}}^{2}  &  \coloneqq\frac{1}{\sqrt{3}}\left[  V_{AB_{1}%
}-V_{AB_{2}}\right]  ,\\
R_{AB_{1}B_{2}}^{3}  &  \coloneqq\frac{i}{\sqrt{3}}\left[  V_{AB_{1}B_{2}%
}-V_{B_{2}B_{1}A}\right]  .
\end{align}
These operators have the following properties:%
\begin{align}
\operatorname{Tr}[R_{AB_{1}B_{2}}^{+}]  &  =\frac{d\left(  d+1\right)  \left(
d+2\right)  }{6},\\
\operatorname{Tr}[R_{AB_{1}B_{2}}^{-}]  &  =\frac{d\left(  d-1\right)  \left(
d-2\right)  }{6},\label{eq:R_minus_trace}\\
\operatorname{Tr}[R_{AB_{1}B_{2}}^{0}]  &  =\frac{2d\left(  d^{2}-1\right)
}{3},\\
\operatorname{Tr}[R_{AB_{1}B_{2}}^{1}]  &  =\operatorname{Tr}[R_{AB_{1}B_{2}%
}^{2}]=\operatorname{Tr}[R_{AB_{1}B_{2}}^{3}]=0.
\end{align}
Furthermore, $R_{AB_{1}B_{2}}^{+}$, $R_{AB_{1}B_{2}}^{-}$, and $R_{AB_{1}%
B_{2}}^{0}$ are orthogonal projectors, satisfying%
\begin{equation}
R_{AB_{1}B_{2}}^{+}+R_{AB_{1}B_{2}}^{-}+R_{AB_{1}B_{2}}^{0}=I_{AB_{1}B_{2}},
\end{equation}
while $R_{AB_{1}B_{2}}^{1}$, $R_{AB_{1}B_{2}}^{2}$, and $R_{AB_{1}B_{2}}^{3}$
are Pauli-like operators acting on the subspace onto which $R_{AB_{1}B_{2}%
}^{0}$ projects. All $R_{AB_{1}B_{2}}^{i}$ operators are Hermitian and satisfy
the following for their Hilbert--Schmidt inner product:%
\begin{align}
\left\langle R_{AB_{1}B_{2}}^{i},R_{AB_{1}B_{2}}^{j}\right\rangle  &
=\operatorname{Tr}[(R_{AB_{1}B_{2}}^{i})^{\dag}R_{AB_{1}B_{2}}^{j}]\\
&  =\operatorname{Tr}[R_{AB_{1}B_{2}}^{i}R_{AB_{1}B_{2}}^{j}]\\
&  =\operatorname{Tr}[R_{AB_{1}B_{2}}^{g(i)}]\delta_{i,j}.
\end{align}
Making the identifications%
\begin{align}
V_{AB_{1}B_{2}}^{+}  &  \equiv I_{AB_{1}B_{2}}, \label{eq:V-ordering-1}\\
V_{AB_{1}B_{2}}^{-}  &  \equiv V_{AB_{1}},\\
V_{AB_{1}B_{2}}^{0}  &  \equiv V_{AB_{2}},\\
V_{AB_{1}B_{2}}^{1}  &  \equiv V_{B_{1}B_{2}},\\
V_{AB_{1}B_{2}}^{2}  &  \equiv V_{AB_{1}B_{2}},\\
V_{AB_{1}B_{2}}^{3}  &  \equiv V_{B_{2}B_{1}A},
\label{eq:V-ordering-6}
\end{align}
we find, for $i\in\left\{  +,-,0,1,2,3\right\}  $, that%
\begin{align}
S_{AB_{1}B_{2}}^{i}  &  =\sum_{j\in\left\{  +,-,0,1,2,3\right\}  }%
[Y]_{i,j}T_{A}(V_{AB_{1}B_{2}}^{j}),\\
R_{AB_{1}B_{2}}^{i}  &  =\sum_{j\in\left\{  +,-,0,1,2,3\right\}  }%
[Z]_{i,j}V_{AB_{1}B_{2}}^{j},
\end{align}
where the matrix $Y$\ with elements $[Y]_{i,j}$ is given by%
\begin{equation}
Y=%
\begin{bmatrix}
\frac{1}{2} & \frac{-1}{2\left(  d+1\right)  } & \frac{-1}{2\left(
d+1\right)  } & \frac{1}{2} & \frac{-1}{2\left(  d+1\right)  } & \frac
{-1}{2\left(  d+1\right)  }\\
\frac{1}{2} & \frac{-1}{2\left(  d-1\right)  } & \frac{-1}{2\left(
d-1\right)  } & -\frac{1}{2} & \frac{1}{2\left(  d-1\right)  } & \frac
{1}{2\left(  d-1\right)  }\\
0 & \frac{d}{d^{2}-1} & \frac{d}{d^{2}-1} & 0 & \frac{-1}{d^{2}-1} & \frac
{-1}{d^{2}-1}\\
0 & \frac{-1}{d^{2}-1} & \frac{-1}{d^{2}-1} & 0 & \frac{d}{d^{2}-1} & \frac
{d}{d^{2}-1}\\
0 & \frac{1}{\sqrt{d^{2}-1}} & \frac{-1}{\sqrt{d^{2}-1}} & 0 & 0 & 0\\
0 & 0 & 0 & 0 & \frac{i}{\sqrt{d^{2}-1}} & \frac{-i}{\sqrt{d^{2}-1}}%
\end{bmatrix}
,
\end{equation}
and the matrix $Z$ with elements $[Z]_{i,j}$ is given by%
\begin{equation}
Z=%
\begin{bmatrix}
\frac{1}{6} & \frac{1}{6} & \frac{1}{6} & \frac{1}{6} & \frac{1}{6} & \frac
{1}{6}\\
\frac{1}{6} & \frac{-1}{6} & \frac{-1}{6} & \frac{-1}{6} & \frac{1}{6} &
\frac{1}{6}\\
\frac{2}{3} & 0 & 0 & 0 & \frac{-1}{3} & \frac{-1}{3}\\
0 & \frac{-1}{3} & \frac{-1}{3} & \frac{2}{3} & 0 & 0\\
0 & \frac{1}{\sqrt{3}} & \frac{-1}{\sqrt{3}} & 0 & 0 & 0\\
0 & 0 & 0 & 0 & \frac{i}{\sqrt{3}} & \frac{-i}{\sqrt{3}}%
\end{bmatrix}
,
\end{equation}
with inverse%
\begin{equation}
Z^{-1}=%
\begin{bmatrix}
1 & 1 & 1 & 0 & 0 & 0\\
1 & -1 & 0 & -\frac{1}{2} & \frac{\sqrt{3}}{2} & 0\\
1 & -1 & 0 & -\frac{1}{2} & -\frac{\sqrt{3}}{2} & 0\\
1 & -1 & 0 & 1 & 0 & 0\\
1 & 1 & -\frac{1}{2} & 0 & 0 & -\frac{i\sqrt{3}}{2}\\
1 & 1 & -\frac{1}{2} & 0 & 0 & \frac{i\sqrt{3}}{2}%
\end{bmatrix}
.
\end{equation}
Now using the fact that%
\begin{equation}
YZ^{-1}=%
\begin{bmatrix}
\frac{d-1}{d+1} & 0 & \frac{d+2}{2\left(  d+1\right)  } & \frac{d+2}{2\left(
d+1\right)  } & 0 & 0\\
0 & \frac{d+1}{d-1} & \frac{d-2}{2\left(  d-1\right)  } & -\frac{d-2}{2\left(
d-1\right)  } & 0 & 0\\
\frac{2}{d+1} & -\frac{2}{d-1} & \frac{1}{d^{2}-1} & -\frac{d}{d^{2}-1} & 0 &
0\\
\frac{2}{d+1} & \frac{2}{d-1} & -\frac{d}{d^{2}-1} & \frac{1}{d^{2}-1} & 0 &
0\\
0 & 0 & 0 & 0 & \frac{\sqrt{3}}{\sqrt{d^{2}-1}} & 0\\
0 & 0 & 0 & 0 & 0 & \frac{\sqrt{3}}{\sqrt{d^{2}-1}}%
\end{bmatrix}
,
\end{equation}
we find, for $i\in\left\{  +,-,0,1,2,3\right\}  $ that%
\begin{align}
&  T_{A}(S_{AB_{1}B_{2}}^{i})\nonumber\\
&  =\sum_{j\in\left\{  +,-,0,1,2,3\right\}  }[Y]_{i,j}V_{AB_{1}B_{2}}^{j}\\
&  =\sum_{j,k\in\left\{  +,-,0,1,2,3\right\}  }[Y]_{i,j}[Z^{-1}]_{j,k}%
R_{AB_{1}B_{2}}^{k}\\
&  =\sum_{k\in\left\{  +,-,0,1,2,3\right\}  }[YZ^{-1}]_{i,k}R_{AB_{1}B_{2}%
}^{k},
\end{align}
and we conclude that%
\begin{align}
&  T_{A\hat{A}}(\widetilde{M}_{A\hat{A}\hat{B}_{1}\hat{B}_{2}B_{1}B_{2}%
})\nonumber\\
&  =\sum_{i\in\left\{  +,-,0,1,2,3\right\}  }\frac{T_{\hat{A}}(M_{\hat{A}%
\hat{B}_{1}\hat{B}_{2}}^{\prime i})}{\operatorname{Tr}[S_{AB_{1}B_{2}}%
^{g(i)}]}\otimes T_{A}(S_{AB_{1}B_{2}}^{i})\\
&  =\sum_{i,k\in\left\{  +,-,0,1,2,3\right\}  }[YZ^{-1}]_{i,k}\frac{T_{\hat
{A}}(M_{\hat{A}\hat{B}_{1}\hat{B}_{2}}^{\prime i})}{\operatorname{Tr}%
[S_{AB_{1}B_{2}}^{g(i)}]}\otimes R_{AB_{1}B_{2}}^{k}\\
&  =\sum_{k\in\left\{  +,-,0,1,2,3\right\}  }G_{\hat{A}\hat{B}_{1}\hat{B}_{2}%
}^{k}\otimes R_{AB_{1}B_{2}}^{k}.\label{eq:G_decomposition}
\end{align}
where%
\begin{equation}
G_{\hat{A}\hat{B}_{1}\hat{B}_{2}}^{k}\coloneqq\sum_{i\in\left\{
+,-,0,1,2,3\right\}  }[YZ^{-1}]_{i,k}\frac{T_{\hat{A}}(M_{\hat{A}\hat{B}%
_{1}\hat{B}_{2}}^{\prime i})}{\operatorname{Tr}[S_{AB_{1}B_{2}}^{g(i)}]}.
\end{equation}
In detail, we find that%
\begin{align}
&  G_{\hat{A}\hat{B}_{1}\hat{B}_{2}}^{+}\nonumber\\
&  =\frac{d-1}{d+1}\frac{T_{\hat{A}}(M_{\hat{A}\hat{B}_{1}\hat{B}_{2}}%
^{\prime+})}{\frac{1}{2}d\left(  d+2\right)  \left(  d-1\right)  }\nonumber\\
&  \quad+\frac{2}{d+1}\frac{T_{\hat{A}}(M_{\hat{A}\hat{B}_{1}\hat{B}_{2}%
}^{\prime0})}{2d}+\frac{2}{d+1}\frac{T_{\hat{A}}(M_{\hat{A}\hat{B}_{1}\hat
{B}_{2}}^{\prime1})}{2d}\\
&  =\frac{2T_{\hat{A}}(M_{\hat{A}\hat{B}_{1}\hat{B}_{2}}^{+})}{\left(
d+1\right)  \left(  d+2\right)  }+\frac{T_{\hat{A}}(M_{\hat{A}\hat{B}_{1}%
\hat{B}_{2}}^{0}+M_{\hat{A}\hat{B}_{1}\hat{B}_{2}}^{1})}{d+1}\\
&  =\frac{1}{d+1}T_{\hat{A}}\!\left(  \frac{2M_{\hat{A}\hat{B}_{1}\hat{B}_{2}%
}^{+}}{d+2}+M_{\hat{A}\hat{B}_{1}\hat{B}_{2}}^{0}+M_{\hat{A}\hat{B}_{1}\hat
{B}_{2}}^{1}\right)  ,
\end{align}%
\begin{align}
&  G_{\hat{A}\hat{B}_{1}\hat{B}_{2}}^{-}\nonumber\\
&  =\frac{d+1}{d-1}\frac{T_{\hat{A}}(M_{\hat{A}\hat{B}_{1}\hat{B}_{2}}%
^{\prime-})}{\frac{1}{2}d\left(  d-2\right)  \left(  d+1\right)  }\nonumber\\
&  \quad-\frac{2}{d-1}\frac{T_{\hat{A}}(M_{\hat{A}\hat{B}_{1}\hat{B}_{2}%
}^{\prime0})}{2d}+\frac{2}{d-1}\frac{T_{\hat{A}}(M_{\hat{A}\hat{B}_{1}\hat
{B}_{2}}^{\prime1})}{2d}\\
&  =\frac{2T_{\hat{A}}(M_{\hat{A}\hat{B}_{1}\hat{B}_{2}}^{-})}{\left(
d-2\right)  \left(  d-1\right)  }+\frac{T_{\hat{A}}(M_{\hat{A}\hat{B}_{1}%
\hat{B}_{2}}^{1}-M_{\hat{A}\hat{B}_{1}\hat{B}_{2}}^{0})}{d-1}\\
&  =\frac{1}{d-1}T_{\hat{A}}\!\left(  \frac{2M_{\hat{A}\hat{B}_{1}\hat{B}_{2}%
}^{-}}{d-2}+M_{\hat{A}\hat{B}_{1}\hat{B}_{2}}^{1}-M_{\hat{A}\hat{B}_{1}\hat
{B}_{2}}^{0}\right)  ,
\end{align}%
\begin{align}
&  G_{\hat{A}\hat{B}_{1}\hat{B}_{2}}^{0}\nonumber\\
&  =\frac{d+2}{2\left(  d+1\right)  }\frac{T_{\hat{A}}(M_{\hat{A}\hat{B}%
_{1}\hat{B}_{2}}^{\prime+})}{\frac{1}{2}d\left(  d+2\right)  \left(
d-1\right)  }\nonumber\\
&  \quad+\frac{d-2}{2\left(  d-1\right)  }\frac{T_{\hat{A}}(M_{\hat{A}\hat
{B}_{1}\hat{B}_{2}}^{\prime-})}{\frac{1}{2}d\left(  d-2\right)  \left(
d+1\right)  }\nonumber\\
&  \quad+\frac{T_{\hat{A}}(M_{\hat{A}\hat{B}_{1}\hat{B}_{2}}^{\prime
0})-dT_{\hat{A}}(M_{\hat{A}\hat{B}_{1}\hat{B}_{2}}^{\prime1})}{2d\left(
d^{2}-1\right)  }\nonumber\\
&  =\frac{T_{\hat{A}}(M_{\hat{A}\hat{B}_{1}\hat{B}_{2}}^{+}+M_{\hat{A}\hat
{B}_{1}\hat{B}_{2}}^{-})}{d^{2}-1}\\
&  \quad+\frac{T_{\hat{A}}(M_{\hat{A}\hat{B}_{1}\hat{B}_{2}}^{0}-dM_{\hat
{A}\hat{B}_{1}\hat{B}_{2}}^{1})}{2\left(  d^{2}-1\right)  }\nonumber\\
&  =\frac{1}{d^{2}-1}T_{\hat{A}}\!\left(  M^{+}+M^{-}+\frac{M^{0}-dM^{1}}%
{2}\right)  ,
\end{align}%
\begin{align}
&  G_{\hat{A}\hat{B}_{1}\hat{B}_{2}}^{1}\nonumber\\
&  =\frac{T_{\hat{A}}(M_{\hat{A}\hat{B}_{1}\hat{B}_{2}}^{+}-M_{\hat{A}\hat
{B}_{1}\hat{B}_{2}}^{-})}{d^{2}-1}\nonumber\\
&  \quad+\frac{T_{\hat{A}}(M_{\hat{A}\hat{B}_{1}\hat{B}_{2}}^{1}-dM_{\hat
{A}\hat{B}_{1}\hat{B}_{2}}^{0})}{2\left(  d^{2}-1\right)  }\\
&  =\frac{1}{d^{2}-1}T_{\hat{A}}\!\left(  M^{+}-M^{-}+\frac{M^{1}-dM^{0}}%
{2}\right)  ,
\end{align}%
\begin{align}
G_{\hat{A}\hat{B}_{1}\hat{B}_{2}}^{2}  &  =\frac{\sqrt{3}T_{\hat{A}}%
(M_{\hat{A}\hat{B}_{1}\hat{B}_{2}}^{2})}{2\sqrt{d^{2}-1}},\\
G_{\hat{A}\hat{B}_{1}\hat{B}_{2}}^{3}  &  =\frac{\sqrt{3}T_{\hat{A}}%
(M_{\hat{A}\hat{B}_{1}\hat{B}_{2}}^{3})}{2\sqrt{d^{2}-1}}.
\end{align}
Then, using the properties of the $R$ operators, we conclude that $T_{A\hat
{A}}(\widetilde{M}_{A\hat{A}\hat{B}_{1}\hat{B}_{2}B_{1}B_{2}})\geq0$ if and
only if%
\begin{align}
G_{\hat{A}\hat{B}_{1}\hat{B}_{2}}^{+}  &  \geq0, \label{eq:G+-PSD}\\
G_{\hat{A}\hat{B}_{1}\hat{B}_{2}}^{-}  &  \geq0,\label{eq:G--PSD}\\
\sum_{k\in\left\{  0,1,2,3\right\}  }G_{\hat{A}\hat{B}_{1}\hat{B}_{2}}%
^{k}\otimes R_{AB_{1}B_{2}}^{k}  &  \geq0.
\end{align}
The final condition is equivalent to%
\begin{equation}%
\begin{bmatrix}
G^{0}+G^{3} & G^{1}-iG^{2}\\
G^{1}+iG^{2} & G^{0}-G^{3}%
\end{bmatrix}
\geq0,
\label{eq:G-Paulis-PSD}
\end{equation}
by reasoning similar to that given for \eqref{eq:M0-3-conds-CP}.
Note that we can scale the condition in \eqref{eq:G+-PSD} by $d+1$ without changing it, the
condition in \eqref{eq:G--PSD} by $d-1$ without changing it, and the condition in \eqref{eq:G-Paulis-PSD} by $d^{2}-1$ without
changing it. This justifies the conditions in \eqref{eq:PPTA-cond-1}--\eqref{eq:PPTA-cond-3}.

We now consider the other PPT condition in \eqref{eq:app-ppt-cond-B2}. Note
that this is equivalent to the following one, which turns out to be simpler
for us to use:%
\begin{equation}
T_{B_{1}}\circ\mathcal{M}_{A\hat{A}\hat{B}_{1}\hat{B}_{2}\rightarrow
B_{1}B_{2}}\circ T_{A\hat{A}\hat{B}_{1}}\in\text{CP}.
\end{equation}
Consider that%
\begin{multline}
T_{A\hat{A}\hat{B}_{1}B_{1}}(\widetilde{M}_{A\hat{A}\hat{B}_{1}\hat{B}%
_{2}B_{1}B_{2}})\\
=\sum_{i\in\left\{  +,-,0,1,2,3\right\}  }T_{\hat{A}\hat{B}_{1}}(M_{\hat
{A}\hat{B}_{1}\hat{B}_{2}}^{\prime i})\otimes\frac{T_{AB_{1}}(S_{AB_{1}B_{2}%
}^{i})}{\operatorname{Tr}[S_{AB_{1}B_{2}}^{g(i)}]}.
\end{multline}
Observing from
\eqref{eq:partial-transpose-perm-ops-1}--\eqref{eq:partial-transpose-perm-ops-4}\ that%
\begin{align}
T_{AB_{1}}(I_{AB_{1}B_{2}})  &  =I_{AB_{1}B_{2}},\\
T_{AB_{1}}(V_{AB_{1}}^{T_{A}})  &  =V_{AB_{1}}^{T_{B_{1}}},\\
T_{AB_{1}}(V_{AB_{2}}^{T_{A}})  &  =V_{AB_{2}},\\
T_{AB_{1}}(V_{B_{1}B_{2}})  &  =V_{B_{1}B_{2}}^{T_{B_{1}}},\\
T_{AB_{1}}(V_{AB_{1}B_{2}}^{T_{A}})  &  =V_{AB_{1}B_{2}}^{T_{B_{1}}},\\
T_{AB_{1}}(V_{B_{2}B_{1}A}^{T_{A}})  &  =V_{B_{2}B_{1}A}^{T_{B_{1}}},
\end{align}
we conclude from \eqref{eq:S-op-+}--\eqref{eq:S-op-3} that%
\begin{align}
T_{AB_{1}}(S_{AB_{1}B_{2}}^{+})  &  =\frac{1}{2}\left[
\begin{array}
[c]{c}%
I_{AB_{1}B_{2}}+V_{B_{1}B_{2}}^{T_{B_{1}}}\\
-\left(  \frac{V_{AB_{1}}^{T_{B_{1}}}+V_{AB_{2}}^{T_{B_{1}}}+V_{AB_{1}B_{2}%
}^{T_{B_{1}}}+V_{B_{2}B_{1}A}^{T_{B_{1}}}}{d+1}\right)
\end{array}
\right]  ,\\
T_{AB_{1}}(S_{AB_{1}B_{2}}^{-})  &  =\frac{1}{2}\left[
\begin{array}
[c]{c}%
I_{AB_{1}B_{2}}-V_{B_{1}B_{2}}^{T_{B_{1}}}\\
-\left(  \frac{V_{AB_{1}}^{T_{B_{1}}}+V_{AB_{2}}-V_{AB_{1}B_{2}}^{T_{B_{1}}%
}-V_{B_{2}B_{1}A}^{T_{B_{1}}}}{d-1}\right)
\end{array}
\right]  ,\\
T_{AB_{1}}(S_{AB_{1}B_{2}}^{0})  &  =\frac{1}{d^{2}-1}\left[
\begin{array}
[c]{c}%
d\left(  V_{AB_{1}}^{T_{B_{1}}}+V_{AB_{2}}\right) \\
-\left(  V_{AB_{1}B_{2}}^{T_{B_{1}}}+V_{B_{2}B_{1}A}^{T_{B_{1}}}\right)
\end{array}
\right]  ,\\
T_{AB_{1}}(S_{AB_{1}B_{2}}^{1})  &  =\frac{1}{d^{2}-1}\left[
\begin{array}
[c]{c}%
d\left(  V_{AB_{1}B_{2}}^{T_{B_{1}}}+V_{B_{2}B_{1}A}^{T_{B_{1}}}\right) \\
-\left(  V_{AB_{1}}^{T_{B_{1}}}+V_{AB_{2}}\right)
\end{array}
\right]  ,\\
T_{AB_{1}}(S_{AB_{1}B_{2}}^{2})  &  =\frac{1}{\sqrt{d^{2}-1}}\left(
V_{AB_{1}}^{T_{B_{1}}}-V_{AB_{2}}\right)  ,\\
T_{AB_{1}}(S_{AB_{1}B_{2}}^{3})  &  =\frac{i}{\sqrt{d^{2}-1}}\left(
V_{AB_{1}B_{2}}^{T_{B_{1}}}-V_{B_{2}B_{1}A}^{T_{B_{1}}}\right)  .
\end{align}
Let us define the following operators:%
\begin{align}
C_{AB_{1}B_{2}}^{+}  &  \coloneqq \frac{1}{2}\left[
\begin{array}
[c]{c}%
W_{AB_{1}B_{2}}^{+}+W_{AB_{1}B_{2}}^{1}\\
-\left(  \frac{W_{AB_{1}B_{2}}^{-}+W_{AB_{1}B_{2}}^{0}+W_{AB_{1}B_{2}}%
^{2}+W_{AB_{1}B_{2}}^{3}}{d+1}\right)
\end{array}
\right]  ,\\
C_{AB_{1}B_{2}}^{-}  &  \coloneqq \frac{1}{2}\left[
\begin{array}
[c]{c}%
W_{AB_{1}B_{2}}^{+}-W_{AB_{1}B_{2}}^{1}\\
-\left(  \frac{W_{AB_{1}B_{2}}^{-}+W_{AB_{1}B_{2}}^{0}-W_{AB_{1}B_{2}}%
^{2}-W_{AB_{1}B_{2}}^{3}}{d-1}\right)
\end{array}
\right]  ,\\
C_{AB_{1}B_{2}}^{0}  &  \coloneqq \frac{1}{d^{2}-1}\left[
\begin{array}
[c]{c}%
d\left(  W_{AB_{1}B_{2}}^{-}+W_{AB_{1}B_{2}}^{0}\right) \\
-\left(  W_{AB_{1}B_{2}}^{2}+W_{AB_{1}B_{2}}^{3}\right)
\end{array}
\right]  ,\\
C_{AB_{1}B_{2}}^{1}  &  \coloneqq \frac{1}{d^{2}-1}\left[
\begin{array}
[c]{c}%
d\left(  W_{AB_{1}B_{2}}^{2}+W_{AB_{1}B_{2}}^{3}\right) \\
-\left(  W_{AB_{1}B_{2}}^{-}+W_{AB_{1}B_{2}}^{0}\right)
\end{array}
\right]  ,\\
C_{AB_{1}B_{2}}^{2}  &  \coloneqq \frac{1}{\sqrt{d^{2}-1}}\left(
W_{AB_{1}B_{2}}^{-}-W_{AB_{1}B_{2}}^{0}\right)  ,\\
C_{AB_{1}B_{2}}^{3}  &  \coloneqq \frac{i}{\sqrt{d^{2}-1}}\left(
W_{AB_{1}B_{2}}^{2}-W_{AB_{1}B_{2}}^{3}\right)  ,
\end{align}
where%
\begin{align}
W_{AB_{1}B_{2}}^{+}  &  \coloneqq I_{AB_{1}B_{2}},\\
W_{AB_{1}B_{2}}^{-}  &  \coloneqq V_{B_{1}B_{2}}^{T_{B_{1}}},\\
W_{AB_{1}B_{2}}^{0}  &  \coloneqq V_{AB_{1}}^{T_{B_{1}}},\\
W_{AB_{1}B_{2}}^{1}  &  \coloneqq V_{AB_{2}},\\
W_{AB_{1}B_{2}}^{2}  &  \coloneqq V_{AB_{1}B_{2}}^{T_{B_{1}}},\\
W_{AB_{1}B_{2}}^{3}  &  \coloneqq V_{B_{2}B_{1}A}^{T_{B_{1}}}.
\end{align}
Observe that the $C$ operators have the same algebraic relations as the $S$
operators, because the  $W$ operators defined above  are related to the original
$V$ operators by the system permutations $B_{1}\rightarrow A$, $B_{2}%
\rightarrow B_{1}$, and $A\rightarrow B_{2}$. Given the above definitions, we
find that, for $i\in\left\{  +,-,0,1,2,3\right\}  $,%
\begin{equation}
C_{AB_{1}B_{2}}^{i}=\sum_{j\in\left\{  +,-,0,1,2,3\right\}  }[Y]_{i,j}%
W_{AB_{1}B_{2}}^{j}.
\end{equation}
Also, observe that, for $j\in\left\{  +,-,0,1,2,3\right\}  $,%
\begin{equation}
W_{AB_{1}B_{2}}^{j}=\sum_{k\in\left\{  +,-,0,1,2,3\right\}  }[P]_{j,k}%
T_{B_{1}}(V_{AB_{1}B_{2}}^{k}),
\end{equation}
where $P$ is the following permutation matrix:%
\begin{equation}
P=%
\begin{bmatrix}
1 & 0 & 0 & 0 & 0 & 0\\
0 & 0 & 0 & 1 & 0 & 0\\
0 & 1 & 0 & 0 & 0 & 0\\
0 & 0 & 1 & 0 & 0 & 0\\
0 & 0 & 0 & 0 & 1 & 0\\
0 & 0 & 0 & 0 & 0 & 1
\end{bmatrix}
,
\end{equation}
and we made use of the ordering in \eqref{eq:V-ordering-1}--\eqref{eq:V-ordering-6}.
Thus,%
\begin{equation}
T_{B_{1}}(V_{AB_{1}B_{2}}^{i})=\sum_{j\in\left\{  +,-,0,1,2,3\right\}
}\left[  P^{-1}Y^{-1}\right]  _{i,j}C_{AB_{1}B_{2}}^{j},
\end{equation}
and this means that%
\begin{align}
&  T_{AB_{1}}(S_{AB_{1}B_{2}}^{i})\nonumber\\
&  =\sum_{j\in\left\{  +,-,0,1,2,3\right\}  }[Y]_{i,j}T_{AB_{1}}%
(T_{A}(V_{AB_{1}B_{2}}^{j}))\\
&  =\sum_{j\in\left\{  +,-,0,1,2,3\right\}  }[Y]_{i,j}T_{B_{1}}(V_{AB_{1}%
B_{2}}^{j})\\
&  =\sum_{j\in\left\{  +,-,0,1,2,3\right\}  }\left[  YP^{-1}Y^{-1}\right]
_{i,j}C_{AB_{1}B_{2}}^{j}.
\end{align}
Now considering that%
\begin{equation}
Y^{-1}=%
\begin{bmatrix}
1 & 1 & 1 & 0 & 0 & 0\\
0 & 0 & \frac{d}{2} & \frac{1}{2} & \frac{1}{2}\sqrt{d^{2}-1} & 0\\
0 & 0 & \frac{d}{2} & \frac{1}{2} & -\frac{1}{2}\sqrt{d^{2}-1} & 0\\
1 & -1 & 0 & 1 & 0 & 0\\
0 & 0 & \frac{1}{2} & \frac{d}{2} & 0 & -\frac{i}{2}\sqrt{d^{2}-1}\\
0 & 0 & \frac{1}{2} & \frac{d}{2} & 0 & \frac{i}{2}\sqrt{d^{2}-1}%
\end{bmatrix}
\end{equation}
we find that%
\begin{multline}
YP^{-1}Y^{-1}=\\%
\begin{bmatrix}
\frac{d}{2\left(  d+1\right)  } & \frac{d+2}{2\left(  d+1\right)  } &
\frac{d\left(  d+2\right)  }{4(d+1)} & -\frac{d+2}{4\left(  d+1\right)  } &
\frac{\left(  d+2\right)  \sqrt{d^{2}-1}}{4(d+1)} & 0\\
\frac{d-2}{2\left(  d-1\right)  } & \frac{d}{2\left(  d-1\right)  } &
-\frac{d\left(  d-2\right)  }{4\left(  d-1\right)  } & \frac{d-2}{4\left(
d-1\right)  } & -\frac{\left(  d-2\right)  \sqrt{d^{2}-1}}{4(d-1)} & 0\\
\frac{d}{d^{2}-1} & -\frac{d}{d^{2}-1} & \frac{d^{2}-2}{2\left(
d^{2}-1\right)  } & \frac{d}{2\left(  d^{2}-1\right)  } & -\frac{d}%
{2\sqrt{d^{2}-1}} & 0\\
-\frac{1}{d^{2}-1} & \frac{1}{d^{2}-1} & \frac{d}{2\left(  d^{2}-1\right)  } &
\frac{2d^{2}-3}{2\left(  d^{2}-1\right)  } & \frac{1}{2\sqrt{d^{2}-1}} & 0\\
-\frac{1}{\sqrt{d^{2}-1}} & \frac{1}{\sqrt{d^{2}-1}} & \frac{d}{2\sqrt
{d^{2}-1}} & -\frac{1}{2\sqrt{d^{2}-1}} & -\frac{1}{2} & 0\\
0 & 0 & 0 & 0 & 0 & 1
\end{bmatrix}
\end{multline}
so that%
\begin{align}
&  T_{A\hat{A}\hat{B}_{1}B_{1}}(\widetilde{M}_{A\hat{A}\hat{B}_{1}\hat{B}%
_{2}B_{1}B_{2}})\nonumber\\
&  =\sum_{i\in\left\{  +,-,0,1,2,3\right\}  }\frac{T_{\hat{A}\hat{B}_{1}%
}(M_{\hat{A}\hat{B}_{1}\hat{B}_{2}}^{\prime i})}{\operatorname{Tr}%
[S_{AB_{1}B_{2}}^{g(i)}]}\otimes T_{AB_{1}}(S_{AB_{1}B_{2}}^{i})\\
&  =\sum_{i,j\in\left\{  +,-,0,1,2,3\right\}  }\frac{T_{\hat{A}\hat{B}_{1}%
}(M_{\hat{A}\hat{B}_{1}\hat{B}_{2}}^{\prime i})}{\operatorname{Tr}%
[S_{AB_{1}B_{2}}^{g(i)}]}\otimes\left[  YP^{-1}Y^{-1}\right]  _{i,j}%
C_{AB_{1}B_{2}}^{j}\\
&  =\sum_{j\in\left\{  +,-,0,1,2,3\right\}  }E_{\hat{A}\hat{B}_{1}\hat{B}_{2}%
}^{j}\otimes C_{AB_{1}B_{2}}^{j},\label{eq:E_decomposition}
\end{align}
where%
\begin{equation}
E_{\hat{A}\hat{B}_{1}\hat{B}_{2}}^{j}\coloneqq\sum_{i\in\left\{
+,-,0,1,2,3\right\}  }\frac{T_{\hat{A}\hat{B}_{1}}(M_{\hat{A}\hat{B}_{1}%
\hat{B}_{2}}^{\prime i})}{\operatorname{Tr}[S_{AB_{1}B_{2}}^{g(i)}]}\left[
YP^{-1}Y^{-1}\right]  _{i,j}.
\end{equation}
In detail, we find that%
\begin{align}
&  E_{\hat{A}\hat{B}_{1}\hat{B}_{2}}^{+}\nonumber\\
&  =\frac{d}{2\left(  d+1\right)  }\frac{T_{\hat{A}\hat{B}_{1}}(M_{\hat{A}%
\hat{B}_{1}\hat{B}_{2}}^{\prime+})}{\frac{1}{2}d\left(  d+2\right)  \left(
d-1\right)  }\nonumber\\
&  \qquad+\frac{d-2}{2\left(  d-1\right)  }\frac{T_{\hat{A}\hat{B}_{1}%
}(M_{\hat{A}\hat{B}_{1}\hat{B}_{2}}^{\prime-})}{\frac{1}{2}d\left(
d-2\right)  \left(  d+1\right)  }\nonumber\\
&  +\frac{d}{d^{2}-1}\frac{T_{\hat{A}\hat{B}_{1}}(M_{\hat{A}\hat{B}_{1}\hat
{B}_{2}}^{\prime0})}{2d}-\frac{1}{d^{2}-1}\frac{T_{\hat{A}\hat{B}_{1}}%
(M_{\hat{A}\hat{B}_{1}\hat{B}_{2}}^{\prime1})}{2d}\nonumber\\
&  \qquad-\frac{1}{\sqrt{d^{2}-1}}\frac{T_{\hat{A}\hat{B}_{1}}(M_{\hat{A}%
\hat{B}_{1}\hat{B}_{2}}^{\prime2})}{2d}\\
&  =\frac{T_{\hat{A}\hat{B}_{1}}\left(  \frac{dM^{+}}{d+2}+M^{-}+\frac
{dM^{0}-M^{1}-\sqrt{d^{2}-1}M^{2}}{2}\right)  }{d^{2}-1},
\end{align}%
\begin{align}
&  E_{\hat{A}\hat{B}_{1}\hat{B}_{2}}^{-}\nonumber\\
&  =\frac{d+2}{2\left(  d+1\right)  }\frac{T_{\hat{A}\hat{B}_{1}}(M_{\hat
{A}\hat{B}_{1}\hat{B}_{2}}^{\prime+})}{\frac{1}{2}d\left(  d+2\right)  \left(
d-1\right)  }\nonumber\\
&  \qquad+\frac{d}{2\left(  d-1\right)  }\frac{T_{\hat{A}\hat{B}_{1}}%
(M_{\hat{A}\hat{B}_{1}\hat{B}_{2}}^{\prime-})}{\frac{1}{2}d\left(  d-2\right)
\left(  d+1\right)  }\nonumber\\
&  -\frac{d}{d^{2}-1}\frac{T_{\hat{A}\hat{B}_{1}}(M_{\hat{A}\hat{B}_{1}\hat
{B}_{2}}^{\prime0})}{2d}+\frac{1}{d^{2}-1}\frac{T_{\hat{A}\hat{B}_{1}}%
(M_{\hat{A}\hat{B}_{1}\hat{B}_{2}}^{\prime1})}{2d}\nonumber\\
&  \qquad+\frac{1}{\sqrt{d^{2}-1}}\frac{T_{\hat{A}\hat{B}_{1}}(M_{\hat{A}%
\hat{B}_{1}\hat{B}_{2}}^{\prime2})}{2d}\\
&  =\frac{T_{\hat{A}\hat{B}_{1}}\left(  M^{+}+\frac{dM^{-}}{d-2}+\frac
{M^{1}-dM^{0}+\sqrt{d^{2}-1}M^{2}}{2}\right)  }{d^{2}-1},
\end{align}%
\begin{align}
&  E_{\hat{A}\hat{B}_{1}\hat{B}_{2}}^{0}\nonumber\\
&  =\frac{d\left(  d+2\right)  }{4\left(  d+1\right)  }\frac{T_{\hat{A}\hat
{B}_{1}}(M_{\hat{A}\hat{B}_{1}\hat{B}_{2}}^{\prime+})}{\frac{1}{2}d\left(
d+2\right)  \left(  d-1\right)  }\nonumber\\
&  \qquad-\frac{d\left(  d-2\right)  }{4\left(  d-1\right)  }\frac{T_{\hat
{A}\hat{B}_{1}}(M_{\hat{A}\hat{B}_{1}\hat{B}_{2}}^{\prime-})}{\frac{1}%
{2}d\left(  d-2\right)  \left(  d+1\right)  }\nonumber\\
&  +\frac{d^{2}-2}{2\left(  d^{2}-1\right)  }\frac{T_{\hat{A}\hat{B}_{1}%
}(M_{\hat{A}\hat{B}_{1}\hat{B}_{2}}^{\prime0})}{2d}+\frac{d}{2\left(
d^{2}-1\right)  }\frac{T_{\hat{A}\hat{B}_{1}}(M_{\hat{A}\hat{B}_{1}\hat{B}%
_{2}}^{\prime1})}{2d}\nonumber\\
&  \qquad+\frac{d}{2\sqrt{d^{2}-1}}\frac{T_{\hat{A}\hat{B}_{1}}(M_{\hat{A}%
\hat{B}_{1}\hat{B}_{2}}^{\prime2})}{2d}\\
&  =\frac{T_{\hat{A}\hat{B}_{1}}\left(  d(M^{+}-M^{-})+\frac{\left(
d^{2}-2\right)  M^{0}+dM^{1}+d\sqrt{d^{2}-1}M^{2}}{2}\right)  }{2\left(
d^{2}-1\right)  },
\end{align}%
\begin{align}
&  E_{\hat{A}\hat{B}_{1}\hat{B}_{2}}^{1}\nonumber\\
&  =-\frac{\left(  d+2\right)  }{4\left(  d+1\right)  }\frac{T_{\hat{A}\hat
{B}_{1}}(M_{\hat{A}\hat{B}_{1}\hat{B}_{2}}^{\prime+})}{\frac{1}{2}d\left(
d+2\right)  \left(  d-1\right)  }\nonumber\\
&  \qquad+\frac{\left(  d-2\right)  }{4\left(  d-1\right)  }\frac{T_{\hat
{A}\hat{B}_{1}}(M_{\hat{A}\hat{B}_{1}\hat{B}_{2}}^{\prime-})}{\frac{1}%
{2}d\left(  d-2\right)  \left(  d+1\right)  }\nonumber\\
&  +\frac{d}{2\left(  d^{2}-1\right)  }\frac{T_{\hat{A}\hat{B}_{1}}(M_{\hat
{A}\hat{B}_{1}\hat{B}_{2}}^{\prime0})}{2d}+\frac{2d^{2}-3}{2\left(
d^{2}-1\right)  }\frac{T_{\hat{A}\hat{B}_{1}}(M_{\hat{A}\hat{B}_{1}\hat{B}%
_{2}}^{\prime1})}{2d}\nonumber\\
&  \qquad-\frac{1}{2\sqrt{d^{2}-1}}\frac{T_{\hat{A}\hat{B}_{1}}(M_{\hat{A}%
\hat{B}_{1}\hat{B}_{2}}^{\prime2})}{2d}\\
&  =\frac{T_{\hat{A}\hat{B}_{1}}\left(  -M^{+}+M^{-}+\frac{dM^{0}+\left(
2d^{2}-3\right)  M^{1}-\sqrt{d^{2}-1}M^{2}}{2}\right)  }{2\left(
d^{2}-1\right)  },
\end{align}%
\begin{align}
&  E_{\hat{A}\hat{B}_{1}\hat{B}_{2}}^{2}\nonumber\\
&  =\frac{\left(  d+2\right)  \sqrt{d^{2}-1}}{4\left(  d+1\right)  }%
\frac{T_{\hat{A}\hat{B}_{1}}(M_{\hat{A}\hat{B}_{1}\hat{B}_{2}}^{\prime+}%
)}{\frac{1}{2}d\left(  d+2\right)  \left(  d-1\right)  }\nonumber\\
&  \qquad-\frac{\left(  d-2\right)  \sqrt{d^{2}-1}}{4\left(  d-1\right)
}\frac{T_{\hat{A}\hat{B}_{1}}(M_{\hat{A}\hat{B}_{1}\hat{B}_{2}}^{\prime-}%
)}{\frac{1}{2}d\left(  d-2\right)  \left(  d+1\right)  }\nonumber\\
&  -\frac{d}{2\sqrt{d^{2}-1}}\frac{T_{\hat{A}\hat{B}_{1}}(M_{\hat{A}\hat
{B}_{1}\hat{B}_{2}}^{\prime0})}{2d}+\frac{1}{2\sqrt{d^{2}-1}}\frac{T_{\hat
{A}\hat{B}_{1}}(M_{\hat{A}\hat{B}_{1}\hat{B}_{2}}^{\prime1})}{2d}\nonumber\\
&  \qquad-\frac{1}{2}\frac{T_{\hat{A}\hat{B}_{1}}(M_{\hat{A}\hat{B}_{1}\hat
{B}_{2}}^{\prime2})}{2d}\\
&  =\frac{T_{\hat{A}\hat{B}_{1}}(M^{+}-M^{-}+\frac{M^{1}-dM^{0}-\sqrt{d^{2}%
-1}M^{2}}{2})}{2\sqrt{d^{2}-1}},
\end{align}%
\begin{equation}
E_{\hat{A}\hat{B}_{1}\hat{B}_{2}}^{3}=\frac{T_{\hat{A}\hat{B}_{1}}(M_{\hat
{A}\hat{B}_{1}\hat{B}_{2}}^{\prime3})}{2d}=\frac{T_{\hat{A}\hat{B}_{1}%
}(M_{\hat{A}\hat{B}_{1}\hat{B}_{2}}^{3})}{2}.
\end{equation}
Then the condition that $T_{A\hat{A}\hat{B}_{1}B_{1}}(\widetilde{M}_{A\hat
{A}\hat{B}_{1}\hat{B}_{2}B_{1}B_{2}})\geq0$ is equivalent to%
\begin{align}
E_{\hat{A}\hat{B}_{1}\hat{B}_{2}}^{+}  &  \geq0, \label{eq:E+-PSD}\\
E_{\hat{A}\hat{B}_{1}\hat{B}_{2}}^{-}  &  \geq0,\label{eq:E--PSD}\\
\sum_{k\in\left\{  0,1,2,3\right\}  }E_{\hat{A}\hat{B}_{1}\hat{B}_{2}}%
^{k}\otimes S_{AB_{1}B_{2}}^{k}  &  \geq0.
\end{align}
The final condition is equivalent to
\begin{equation}%
\begin{bmatrix}
E^{0}+E^{3} & E^{1}-iE^{2}\\
E^{1}+iE^{2} & E^{0}-E^{3}%
\end{bmatrix}
\geq0,
\end{equation}
by reasoning similar to that given for \eqref{eq:M0-3-conds-CP}.
Note that we can scale the conditions in \eqref{eq:E+-PSD} and \eqref{eq:E--PSD}  by $d^{2}-1$ without
changing them. This justifies the conditions in \eqref{eq:PPTB-cond-1}--\eqref{eq:PPTB-cond-3}.

\subsection{Final evaluation of normalized diamond distance objective
function}

\label{Sec_obj_func_tel}

It now remains to evalulate the normalized diamond distance by means of its
semi-definite programming formulation in~\eqref{eq:diamond-d-SDP}:%
\begin{multline}
\frac{1}{2}\left\Vert \widetilde{\mathcal{K}}_{A\hat{A}\hat{B}\rightarrow
B}\circ\mathcal{A}_{\hat{A}\hat{B}}^{\rho}-\operatorname{id}_{A\rightarrow
B}^{d}\right\Vert _{\diamond}\\
=\inf_{\substack{\mu\geq0,\\Z_{AB}\geq0}}\left\{
\begin{array}
[c]{c}%
\mu:\\
\mu I_{A}\geq Z_{A}\\
Z_{AB}\geq\Gamma_{AB}-\operatorname{Tr}_{\hat{A}\hat{B}}[\rho_{\hat{A}\hat{B}%
}^{T}\widetilde{K}_{A\hat{A}\hat{B}B}]
\end{array}
\right\}  .
\end{multline}
Consider that the Choi operator of the channel $\widetilde{\mathcal{K}}%
_{A\hat{A}\hat{B}\rightarrow B}$ is given by%
\begin{align}
\widetilde{K}_{A\hat{A}\hat{B}B}  &  =\frac{1}{d_{\hat{B}}}\operatorname{Tr}%
_{\hat{B}_{2}}[P_{\hat{A}\hat{B}_{1}\hat{B}_{2}}^{\prime}]\otimes\Phi
_{AB}\nonumber\\
&  \qquad+\frac{1}{d_{\hat{B}}}\operatorname{Tr}_{\hat{B}_{2}}[Q_{\hat{A}%
\hat{B}_{1}\hat{B}_{2}}^{\prime}]\otimes\left(  \frac{I_{AB}-\Phi_{AB}}%
{d^{2}-1}\right) \\
&  =\frac{1}{d_{\hat{B}}}\operatorname{Tr}_{\hat{B}_{2}}[P_{\hat{A}\hat{B}%
_{1}\hat{B}_{2}}]\otimes\Gamma_{AB}\nonumber\\
&  \qquad+\frac{1}{d_{\hat{B}}}\operatorname{Tr}_{\hat{B}_{2}}[Q_{\hat{A}%
\hat{B}_{1}\hat{B}_{2}}]\otimes\left(  \frac{dI_{AB}-\Gamma_{AB}}{d^{2}%
-1}\right)  ,
\end{align}
where%
\begin{align}
P_{\hat{A}\hat{B}_{1}\hat{B}_{2}}  &  \coloneqq\frac{1}{d}P_{\hat{A}\hat
{B}_{1}\hat{B}_{2}}^{\prime},\\
Q_{\hat{A}\hat{B}_{1}\hat{B}_{2}}  &  \coloneqq\frac{1}{d}Q_{\hat{A}\hat
{B}_{1}\hat{B}_{2}}^{\prime}=I_{\hat{A}\hat{B}_{1}\hat{B}_{2}}-P_{\hat{A}%
\hat{B}_{1}\hat{B}_{2}},
\end{align}
which implies that the Choi operator of $\widetilde{\mathcal{K}}_{A\hat{A}%
\hat{B}\rightarrow B}\circ\mathcal{A}_{\hat{A}\hat{B}}^{\rho}$ is%
\begin{multline}
\operatorname{Tr}_{\hat{A}\hat{B}}[\rho_{\hat{A}\hat{B}}^{T}\widetilde
{K}_{A\hat{A}\hat{B}B}]=\frac{1}{d_{\hat{B}}}\operatorname{Tr}[\rho_{\hat
{A}\hat{B}_{1}}^{T}P_{\hat{A}\hat{B}_{1}\hat{B}_{2}}]\Gamma_{AB}\\
+\frac{1}{d_{\hat{B}}}\operatorname{Tr}[\rho_{\hat{A}\hat{B}_{1}}^{T}%
Q_{\hat{A}\hat{B}_{1}\hat{B}_{2}}]\left(  \frac{dI_{AB}-\Gamma_{AB}}{d^{2}%
-1}\right)  ,
\end{multline}
Then it follows that%
\begin{multline}
\Gamma_{AB}-\operatorname{Tr}_{\hat{A}\hat{B}}[\rho_{\hat{A}\hat{B}}%
^{T}\widetilde{K}_{A\hat{A}\hat{B}B}]\\
=\left(  1-\frac{1}{d_{\hat{B}}}\operatorname{Tr}[\rho_{\hat{A}\hat{B}_{1}%
}^{T}P_{\hat{A}\hat{B}_{1}\hat{B}_{2}}]\right)  \Gamma_{AB}\\
-\frac{1}{d_{\hat{B}}}\operatorname{Tr}[\rho_{\hat{A}\hat{B}_{1}}^{T}%
Q_{\hat{A}\hat{B}_{1}\hat{B}_{2}}]\left(  \frac{dI_{AB}-\Gamma_{AB}}{d^{2}%
-1}\right)
\end{multline}
By the same reasoning used to justify
\eqref{eq:Z-op-ineq}--\eqref{eq:mu-opt-choice}, we conclude that optimal
choices for $Z_{AB}$ and $\mu$ are%
\begin{align}
Z_{AB}  &  =\mu\Gamma_{AB},\\
\mu &  =1-\frac{1}{d_{\hat{B}}}\operatorname{Tr}[\rho_{\hat{A}\hat{B}_{1}}%
^{T}P_{\hat{A}\hat{B}_{1}\hat{B}_{2}}].
\end{align}

Putting everything together, we conclude that the semi-definite program in
Proposition~\ref{prop:two-ext-sim-err-gen-ch} reduces to the form stated in Proposition~\ref{prop:simplified-SDP-two-ext-sim-id}.
This concludes the proof.

\subsection{The case  $d=2$}

\label{sec:tel_SDP_2d}

Several steps in the previous calculations involve a denominator of $d-2$, making it unsuitable to get lower bounds on the simulation error of an identity channel for qubits. However, it is easy to navigate around this problem.

The trace of $S^-_{AB_1B_2}$ is $d(d-2)(d+1)/2$; hence, it is traceless for $d=2$. Since $S^-_{AB_1B_2}$ is positive semi-definite as well, it follows that $S^-_{AB_1B_2} = 0$. Therefore, $M^{\prime -}_{\hat{A}\hat{B}_1\hat{B}_2}$ is not involved in \eqref{eq:sym-form-for-Choi-op-2-ext} and consequently $M_{\hat{A}\hat{B}_1\hat{B}_2}^-$ is not involved in any of the constraints. This is also equivalent to setting $M_{\hat{A}\hat{B}_1\hat{B}_2}^-$ to $0$.

Similarly, $R_{AB_1B_2}^-$ defined in \eqref{eq:R_minus_def} is positive semi-definite and traceless for $d=2$, implying that $R_{AB_1B_2}^- = 0$ when $d=2$. This means $G^{-}_{\hat{A}\hat{B}_1\hat{B}_2}$ is not involved in \eqref{eq:G_decomposition}, which in turn corresponds to removing \eqref{eq:PPTA-cond-2} from the constraints in the SDP of Proposition~\ref{prop:simplified-SDP-two-ext-sim-id}.

Finally, $C_{AB_1B_2}^- = 0$ because $C_{AB_1B_2}^-$ is a traceless positive semi-definite operator when $d=2$. Therefore, $E_{\hat{A}\hat{B}_1\hat{B}_2}^-$ is not involved in \eqref{eq:E_decomposition}, which in turn corresponds to removing \eqref{eq:PPTB-cond-2} from the constraints in this SDP. 

This leads us to the claim at the end of Proposition~\ref{prop:simplified-SDP-two-ext-sim-id}, that removing \eqref{eq:PPTA-cond-1} and \eqref{eq:PPTB-cond-2} and setting  $M_{\hat{A}\hat{B}_1\hat{B}_2}^- = 0$ gives the SDP for $d=2$. 

\section{Proof of Proposition~\ref{prop:sim-errs-equal-1WL-QEC}}\label{sec:sim-errs-equal-1WL-QEC-proof}

The proof of Proposition~\ref{prop:sim-errs-equal-1WL-QEC}\ bears similarities
with the proof of Proposition~\ref{prop:sim-errs-equal-1WL}, but we give a
somewhat detailed proof here for completeness and clarity. The main idea is
again to simplify the optimization problems in \eqref{eq:diamond-err-sim-QEC} and \eqref{eq:fidelity-err-sim-QEC} by exploiting
the symmetries of the identity channel, as stated in \eqref{eq:identity-ch-symmetries}.

Let us again begin by analyzing the diamond distance. Let $\Upsilon_{(\hat
{A}\rightarrow\hat{B})\rightarrow(A\rightarrow B)}$ be an arbitrary
LOCR\ superchannel, with corresponding bipartite channel $\mathcal{C}%
_{A\hat{B}\rightarrow\hat{A}B}$\ having the form in \eqref{eq:LOCR_superchannel_bipartite_rep}.
That is, let $\mathcal{C}_{A\hat{B}\rightarrow\hat{A}B}$ denote the bipartite channel
that is in direct correspondence with $\Upsilon_{(\hat{A}\rightarrow\hat
{B})\rightarrow(A\rightarrow B)}$; i.e.,%
\begin{equation}
\mathcal{C}_{A\hat{B}\rightarrow\hat{A}B}\coloneqq\sum_{y}p(y)\mathcal{E}%
_{A\rightarrow\hat{A}}^{y}\otimes\mathcal{D}_{\hat{B}\rightarrow B}^{y}.
\end{equation}
Additionally, the notation $\Upsilon_{(\hat{A}\rightarrow\hat{B}%
)\rightarrow(A\rightarrow B)}(\mathcal{N}_{\hat{A}\rightarrow\hat{B}})$
indicates that $\Upsilon_{(\hat{A}\rightarrow\hat{B})\rightarrow(A\rightarrow
B)}$ is a physical transformation of the channel $\mathcal{N}_{\hat
{A}\rightarrow\hat{B}}$ by means of the encoding-decoding scheme
$(p(y),\{\mathcal{E}_{A\rightarrow\hat{A}}^{y}\}_{y},\{\mathcal{D}_{\hat
{B}\rightarrow B}^{y}\}_{y})$, and the transformation takes a channel with
input and output systems $\hat{A}$ and $\hat{B}$, respectively, to a channel
with input and output systems $A$ and $B$, respectively.
Then by the same
reasoning that led to \eqref{eq:convexity-dd-1wl}, we find that%
\begin{multline}
\left\Vert \Upsilon_{(\hat{A}\rightarrow\hat{B})\rightarrow(A\rightarrow
B)}(\mathcal{N}_{\hat{A}\rightarrow\hat{B}})-\operatorname{id}_{A\rightarrow
B}^{d}\right\Vert _{\diamond}\label{eq:twirl-sym-LOCR}\\
\geq\left\Vert \widetilde{\Upsilon}_{(\hat{A}\rightarrow\hat{B})\rightarrow
(A\rightarrow B)}(\mathcal{N}_{\hat{A}\rightarrow\hat{B}})-\operatorname{id}%
_{A\rightarrow B}^{d}\right\Vert _{\diamond},
\end{multline}
where $\widetilde{\Upsilon}_{(\hat{A}\rightarrow\hat{B})\rightarrow
(A\rightarrow B)}$ is the twirled version of the superchannel $\Upsilon
_{(\hat{A}\rightarrow\hat{B})\rightarrow(A\rightarrow B)}$ and is defined as%
\begin{equation}
\widetilde{\Upsilon}_{(\hat{A}\rightarrow\hat{B})\rightarrow(A\rightarrow
B)}\coloneqq \int dU\ \mathcal{U}_{B}^{\dag}\circ\Upsilon_{(\hat{A}%
\rightarrow\hat{B})\rightarrow(A\rightarrow B)}\circ\mathcal{U}_{A}.
\end{equation}
Importantly, $\widetilde{\Upsilon}_{(\hat{A}\rightarrow\hat{B})\rightarrow
(A\rightarrow B)}$ is an LOCR\ channel because the channel twirl above can be
implemented by means of LOCR. Furthermore, observe that%
\begin{equation}
\widetilde{\Upsilon}_{(\hat{A}\rightarrow\hat{B})\rightarrow(A\rightarrow
B)}=\mathcal{U}_{B}^{\dag}\circ\widetilde{\Upsilon}_{(\hat{A}\rightarrow
\hat{B})\rightarrow(A\rightarrow B)}\circ\mathcal{U}_{A}
\label{eq:twirl-symmetry-super-ch-2}%
\end{equation}
for every unitary channel $\mathcal{U}$. Thus, as a consequence of
\eqref{eq:twirl-sym-LOCR}, it suffices to minimize the error with respect to
superchannels having the symmetry in
\eqref{eq:twirl-symmetry-super-ch-2}.\ The corresponding twirled bipartite
channel $\widetilde{\mathcal{C}}_{A\hat{B}\rightarrow\hat{A}B}$ then has the
following form:%
\begin{equation}
\widetilde{\mathcal{C}}_{A\hat{B}\rightarrow\hat{A}B}=\int dU\ \mathcal{U}%
_{B}^{\dag}\circ\mathcal{C}_{A\hat{B}\rightarrow\hat{A}B}\circ\mathcal{U}_{A}%
\end{equation}

Let us determine the form of LOCR\ superchannels possessing this symmetry. Let
$\widetilde{\Upsilon}_{A\hat{B}\hat{A}B}$ denote the Choi operator for
$\widetilde{\Upsilon}_{(\hat{A}\rightarrow\hat{B})\rightarrow(A\rightarrow
B)}$ and $\widetilde{\mathcal{C}}_{A\hat{B}\rightarrow\hat{A}B}$, and let
$\Upsilon_{A\hat{B}\hat{A}B}$ denote the Choi operator for $\Upsilon_{(\hat
{A}\rightarrow\hat{B})\rightarrow(A\rightarrow B)}$ and $\mathcal{C}_{A\hat
{B}\rightarrow\hat{A}B}$. They are related as follows:%
\begin{equation}
\widetilde{\Upsilon}_{A\hat{B}\hat{A}B}=\int dU\ (\mathcal{U}_{A}%
\otimes\overline{\mathcal{U}}_{B})(\Upsilon_{A\hat{B}\hat{A}B}).
\end{equation}
Recalling the form of the bipartite channel $\mathcal{C}_{A\hat{B}%
\rightarrow\hat{A}B}$, we find that its Choi operator $\Upsilon_{A\hat{B}%
\hat{A}B}$ is given by%
\begin{equation}
\Upsilon_{A\hat{B}\hat{A}B}=\sum_{y}p(y)E_{A\hat{A}}^{y}\otimes D_{\hat{B}%
B}^{y},
\end{equation}
where, for all $y$, $E_{A\hat{A}}^{y},D_{\hat{B}B}^{y}\geq0$,
$\operatorname{Tr}_{\hat{A}}[E_{A\hat{A}}^{y}]=I_{A}$, and $\operatorname{Tr}%
_{B}[D_{\hat{B}B}^{y}]=I_{\hat{B}}$. Recall the action of the twirling channel from \eqref{eq:twirl-ch-def}, so that%
\begin{align}
&  \widetilde{\mathcal{T}}_{AB}(E_{A\hat{A}}^{y}\otimes D_{\hat{B}B}%
^{y})\nonumber\\
&  =\Phi_{AB}\otimes\operatorname{Tr}_{AB}[\Phi_{AB}(E_{A\hat{A}}^{y}\otimes
D_{\hat{B}B}^{y})]\nonumber\\
&  \qquad+\frac{I_{AB}-\Phi_{AB}}{d^{2}-1}\otimes\operatorname{Tr}%
_{AB}[\left(  I_{AB}-\Phi_{AB}\right)  (E_{A\hat{A}}^{y}\otimes D_{\hat{B}%
B}^{y})]\\
&  =\Gamma_{AB}\otimes\frac{1}{d}\operatorname{Tr}_{AB}[\Phi_{AB}(E_{A\hat{A}%
}^{y}\otimes D_{\hat{B}B}^{y})]\nonumber\\
&  \qquad+\frac{dI_{AB}-\Gamma_{AB}}{d^{2}-1}\otimes\frac{1}{d}%
\operatorname{Tr}_{AB}[\left(  I_{AB}-\Phi_{AB}\right)  (E_{A\hat{A}}%
^{y}\otimes D_{\hat{B}B}^{y})].
\end{align}
Consider that the Choi operator of the composition $\widetilde{\Upsilon
}_{(\hat{A}\rightarrow\hat{B})\rightarrow(A\rightarrow B)}(\mathcal{N}%
_{\hat{A}\rightarrow\hat{B}})$ is as follows%
\begin{equation}
\operatorname{Tr}_{\hat{A}\hat{B}}[T_{\hat{A}\hat{B}}(\Gamma_{\hat{A}\hat{B}%
}^{\mathcal{N}})\widetilde{\Upsilon}_{A\hat{B}\hat{A}B}]
\end{equation}
by applying the propagation rule in \eqref{eq:propagation rule}. This means that%
\begin{multline}
\operatorname{Tr}_{\hat{A}\hat{B}}[T_{\hat{A}\hat{B}}(\Gamma_{\hat{A}\hat{B}%
}^{\mathcal{N}})\widetilde{\Upsilon}_{A\hat{B}\hat{A}B}]\\
=\Gamma_{AB}\otimes\sum_{y}p(y)\frac{1}{d}\operatorname{Tr}[(\Phi_{AB}\otimes
T_{\hat{A}\hat{B}}(\Gamma_{\hat{A}\hat{B}}^{\mathcal{N}}))(E_{A\hat{A}}%
^{y}\otimes D_{\hat{B}B}^{y})]\\
+\frac{dI_{AB}-\Gamma_{AB}}{d^{2}-1}\otimes\sum_{y}p(y)\frac{1}{d}\times\\
\operatorname{Tr}[(\left(  I_{AB}-\Phi_{AB}\right)  \otimes T_{\hat{A}\hat{B}%
}(\Gamma_{\hat{A}\hat{B}}^{\mathcal{N}}))(E_{A\hat{A}}^{y}\otimes D_{\hat{B}%
B}^{y})].
\end{multline}
Then we find that%
\begin{align}
&  \sum_{y}p(y)\frac{1}{d}\operatorname{Tr}[(\Phi_{AB}\otimes T_{\hat{A}%
\hat{B}}(\Gamma_{\hat{A}\hat{B}}^{\mathcal{N}}))(E_{A\hat{A}}^{y}\otimes
D_{\hat{B}B}^{y})]\nonumber\\
&  =\sum_{y}p(y)\frac{1}{d}\operatorname{Tr}[\Phi_{AB}\operatorname{Tr}%
_{\hat{A}\hat{B}}[T_{\hat{A}\hat{B}}(\Gamma_{\hat{A}\hat{B}}^{\mathcal{N}%
})(E_{A\hat{A}}^{y}\otimes D_{\hat{B}B}^{y})]]\\
&  =\sum_{y}p(y)\frac{1}{d}\operatorname{Tr}[\Phi_{AB}\Gamma_{AB}%
^{\mathcal{D}^{y}\circ\mathcal{N}\circ\mathcal{E}^{y}}]\\
&  =\sum_{y}p(y)\operatorname{Tr}[\Phi_{AB}(\mathcal{D}_{\hat{B}\rightarrow
B}^{y}\circ\mathcal{N}_{\hat{A}\rightarrow\hat{B}}\circ\mathcal{E}%
_{A\rightarrow\hat{A}}^{y})(\Phi_{AB})]\\
&  =E_{F}(\mathcal{N}_{\hat{A}\rightarrow\hat{B}};\{p(y),\mathcal{E}%
_{A\rightarrow\hat{A}}^{y},\mathcal{D}_{\hat{B}\rightarrow B}^{y}\}_{y}).
\end{align}
Consider that%
\begin{align}
\operatorname{Tr}_{AB}[\Phi_{AB}(E_{A\hat{A}}^{y}\otimes D_{\hat{B}B}^{y})]
&  =\operatorname{Tr}_{AB}[\Phi_{AB}(T_{B}(E_{\hat{A}B}^{y})D_{\hat{B}B}%
^{y})]\nonumber\\
&  =\frac{1}{d}\operatorname{Tr}_{B}[T_{B}(E_{\hat{A}B}^{y})D_{\hat{B}B}%
^{y}],\\
\operatorname{Tr}_{AB}[I_{AB}(E_{A\hat{A}}^{y}\otimes D_{\hat{B}B}^{y})]  &
=\operatorname{Tr}_{AB}[E_{A\hat{A}}^{y}\otimes D_{\hat{B}B}^{y}]\\
&  =\operatorname{Tr}_{A}[E_{A\hat{A}}^{y}]\otimes I_{\hat{B}}.
\end{align}
Thus, we find that%
\begin{equation}
\widetilde{\mathcal{T}}_{AB}(\Upsilon_{A\hat{B}\hat{A}B})=\Gamma_{AB}\otimes
K_{\hat{A}\hat{B}}+\frac{dI_{AB}-\Gamma_{AB}}{d^{2}-1}\otimes L_{\hat{A}%
\hat{B}},
\end{equation}
where%
\begin{align}
K_{\hat{A}\hat{B}}  &  \coloneqq \frac{1}{d^{2}}\sum_{y}p(y)\operatorname{Tr}%
_{B}[T_{B}(E_{\hat{A}B}^{y})D_{\hat{B}B}^{y}],\\
L_{\hat{A}\hat{B}}  &  \coloneqq \frac{1}{d}\sum_{y}p(y)\operatorname{Tr}%
_{A}[E_{A\hat{A}}^{y}]\otimes I_{\hat{B}}-K_{\hat{A}\hat{B}}.
\end{align}
Observe that the operator%
\begin{equation}
\tau_{\hat{A}}\coloneqq \frac{1}{d}\sum_{y}p(y)\operatorname{Tr}_{A}%
[E_{A\hat{A}}^{y}]
\end{equation}
is a state, because it is positive semi-definite and has trace equal to one.

Consider that the Choi operator of the composition $\widetilde{\Upsilon
}_{(\hat{A}\rightarrow\hat{B})\rightarrow(A\rightarrow B)}(\mathcal{N}%
_{\hat{A}\rightarrow\hat{B}})$ is as follows%
\begin{equation}
\operatorname{Tr}_{\hat{A}\hat{B}}[T_{\hat{A}\hat{B}}(\Gamma_{\hat{A}\hat{B}%
}^{\mathcal{N}})\widetilde{\Upsilon}_{A\hat{B}\hat{A}B}]
\end{equation}
by applying the propagation rule in \eqref{eq:propagation rule}. This implies that the simulation
channel has the following Choi operator:%
\begin{align}
&  \operatorname{Tr}_{\hat{A}\hat{B}}[T_{\hat{A}\hat{B}}(\Gamma_{\hat{A}%
\hat{B}}^{\mathcal{N}})\widetilde{\Upsilon}_{A\hat{B}\hat{A}B}]\nonumber\\
&  =\Gamma_{AB}\operatorname{Tr}[T_{\hat{A}\hat{B}}(\Gamma_{\hat{A}\hat{B}%
}^{\mathcal{N}})K_{\hat{A}\hat{B}}]\nonumber\\
&  \qquad+\frac{dI_{AB}-\Gamma_{AB}}{d^{2}-1}\operatorname{Tr}[T_{\hat{A}%
\hat{B}}(\Gamma_{\hat{A}\hat{B}}^{\mathcal{N}})L_{\hat{A}\hat{B}}]\\
&  =\Gamma_{AB}\operatorname{Tr}[\Gamma_{\hat{A}\hat{B}}^{\mathcal{N}}%
T_{\hat{A}\hat{B}}(K_{\hat{A}\hat{B}})]\nonumber\\
&  \qquad+\frac{dI_{AB}-\Gamma_{AB}}{d^{2}-1}\operatorname{Tr}[\Gamma_{\hat
{A}\hat{B}}^{\mathcal{N}}T_{\hat{A}\hat{B}}(L_{\hat{A}\hat{B}})].
\end{align}
Observing that the optimization does not change under the substitutions
$T_{\hat{A}\hat{B}}(K_{\hat{A}\hat{B}})\rightarrow K_{\hat{A}\hat{B}}$ and
$T_{\hat{A}\hat{B}}(L_{\hat{A}\hat{B}})\rightarrow L_{\hat{A}\hat{B}}$, we
find that the Choi operator is given by%
\begin{multline}
\operatorname{Tr}_{\hat{A}\hat{B}}[T_{\hat{A}\hat{B}}(\Gamma_{\hat{A}\hat{B}%
}^{\mathcal{N}})\widetilde{\Upsilon}_{A\hat{B}\hat{A}B}]\\
=\Gamma_{AB}\operatorname{Tr}[K_{\hat{A}\hat{B}}\Gamma_{\hat{A}\hat{B}%
}^{\mathcal{N}}]+\frac{dI_{AB}-\Gamma_{AB}}{d^{2}-1}\operatorname{Tr}%
[L_{\hat{A}\hat{B}}\Gamma_{\hat{A}\hat{B}}^{\mathcal{N}}].
\end{multline}
This corresponds to a channel of the following form:%
\begin{equation}
\operatorname{Tr}[K_{\hat{A}\hat{B}}\Gamma_{\hat{A}\hat{B}}^{\mathcal{N}%
}]\operatorname{id}_{A\rightarrow B}^{d}+\operatorname{Tr}[L_{\hat{A}\hat{B}%
}\Gamma_{\hat{A}\hat{B}}^{\mathcal{N}}]\mathcal{D}_{A\rightarrow B},
\label{eq:LOCR-red}%
\end{equation}
where the randomizing channel $\mathcal{D}_{A\rightarrow B}$ is defined in
\eqref{eq:randomizing channel}. The channel in \eqref{eq:LOCR-red} has the interpretation that the
identity channel $\operatorname{id}_{A\rightarrow B}^{d}$ is applied with
probability $\operatorname{Tr}[K_{\hat{A}\hat{B}}\Gamma_{\hat{A}\hat{B}%
}^{\mathcal{N}}]$ and the randomizing channel $\mathcal{D}_{A\rightarrow B}$
is applied with probability $\operatorname{Tr}[L_{\hat{A}\hat{B}}\Gamma
_{\hat{A}\hat{B}}^{\mathcal{N}}]$. The total probability is indeed equal to
one because%
\begin{align}
\operatorname{Tr}[K_{\hat{A}\hat{B}}\Gamma_{\hat{A}\hat{B}}^{\mathcal{N}%
}]+\operatorname{Tr}[L_{\hat{A}\hat{B}}\Gamma_{\hat{A}\hat{B}}^{\mathcal{N}}]
&  =\operatorname{Tr}[(\tau_{\hat{A}}\otimes I_{\hat{B}})\Gamma_{\hat{A}%
\hat{B}}^{\mathcal{N}}]\\
&  =1.
\end{align}
This justifies the claim in \eqref{eq:ent-fid-ch-spec-form}.

The claim in \eqref{eq:dd_fid_err_equality_QEC}, that the simulation errors are equal, follows the same
reasoning given in Appendix \ref{app:proof-errors-equal}.

\section{Proof of Proposition~\ref{prop:simplified-SDP-two-ext-sim-id-QEC}}

\label{app:proof-approx-QEC-SDP-simplify}

The proof of Proposition~\ref{prop:simplified-SDP-two-ext-sim-id-QEC} bears
some similarities with the proof of Proposition~\ref{prop:simplified-SDP-two-ext-sim-id}. However, there are some key
differences that we detail here. As before, symmetry plays a critical role.

Our goal is to minimize the following objective function%
\begin{equation}
\frac{1}{2}\left\Vert \Theta_{(\hat{A}\rightarrow\hat{B})\rightarrow
(A\rightarrow B)}(\mathcal{N}_{\hat{A}\rightarrow\hat{B}})-\operatorname{id}%
_{A\rightarrow B}^{d}\right\Vert _{\diamond}, \label{eq:sim-err-QEC-app}%
\end{equation}
with respect to every two-PPT-extendible non-signaling superchannel
$\Theta\equiv\Theta_{(\hat{A}\rightarrow\hat{B})\rightarrow(A\rightarrow B)}$ (see Section~\ref{sec_2-ppt-ext-non-sig-superchannels}),
where $\mathcal{N}_{\hat{A}\rightarrow\hat{B}}$ is a given channel and
$\operatorname{id}_{A\rightarrow B}^{d}$ is the $d$-dimensional identity
channel. By exploiting the covariance symmetry of the identity channel
$\operatorname{id}_{A\rightarrow B}^{d}$, as given in \eqref{eq:identity-ch-symmetries}, it suffices to
optimize over the set of two-PPT-extendible non-signaling superchannels that
satisfy the same symmetry. Let $\mathcal{K}_{A\hat{B}\rightarrow\hat{A}B}$ be
the bipartite channel corresponding to $\Theta$, and since $\Theta$ is
two-PPT-extendible and non-signaling, this implies that there exists an
extension channel $\mathcal{M}_{A\hat{B}_{1}\hat{B}_{2}\rightarrow\hat{A}%
B_{1}B_{2}}$ satisfying the constraints for two-PPT-extendibility and
non-signaling. Following the same arguments given in Appendix~\ref{app:twirled-2e-opt}, the simulation
error in \eqref{eq:sim-err-QEC-app} is minimized by a symmetrized or twirled
superchannel $\widetilde{\Theta}_{(\hat{A}\rightarrow\hat{B})\rightarrow
(A\rightarrow B)}$, such that its corresponding bipartite channel
$\widetilde{\mathcal{K}}_{A\hat{B}\rightarrow\hat{A}B}$ is twirled, as well as
its extension $\widetilde{\mathcal{M}}_{A\hat{B}_{1}\hat{B}_{2}\rightarrow
\hat{A}B_{1}B_{2}}$, so that%
\begin{equation}
\widetilde{\mathcal{K}}_{A\hat{B}\rightarrow\hat{A}B}\coloneqq \int
dU\ \mathcal{U}_{B}^{\dag}\circ\mathcal{K}_{A\hat{B}\rightarrow\hat{A}B}%
\circ\mathcal{U}_{A},
\end{equation}%
\begin{multline}
\widetilde{\mathcal{M}}_{A\hat{B}_{1}\hat{B}_{2}\rightarrow\hat{A}B_{1}B_{2}%
}\coloneqq\\
\int dU\ (\mathcal{U}_{B_{1}}^{\dag}\otimes\mathcal{U}_{B_{2}}^{\dag}%
)\circ\mathcal{M}_{A\hat{B}_{1}\hat{B}_{2}\rightarrow\hat{A}B_{1}B_{2}}%
\circ\mathcal{U}_{A}.
\end{multline}
The corresponding Choi operators $\widetilde{K}_{A\hat{B}\hat{A}B}$ and
$\widetilde{M}_{A\hat{B}_{1}\hat{B}_{2}\hat{A}B_{1}B_{2}}$ satisfy%
\begin{equation}
\widetilde{K}_{A\hat{B}\hat{A}B}=\int dU\ (\overline{\mathcal{U}}_{A}%
\otimes\mathcal{U}_{B})(K_{A\hat{B}\hat{A}B}),
\end{equation}%
\begin{multline}
\widetilde{M}_{A\hat{B}_{1}\hat{B}_{2}\hat{A}B_{1}B_{2}}=\\
\int dU\ (\overline{\mathcal{U}}_{A}\otimes\mathcal{U}_{B_{1}}\otimes
\mathcal{U}_{B_{2}})(M_{A\hat{B}_{1}\hat{B}_{2}\hat{A}B_{1}B_{2}}).
\end{multline}
As before, this implies that the Choi operator of $\widetilde{M}_{A\hat{B}%
_{1}\hat{B}_{2}\hat{A}B_{1}B_{2}}$ has the following form:%
\begin{equation}
\widetilde{M}_{A\hat{B}_{1}\hat{B}_{2}\hat{A}B_{1}B_{2}}=d\sum_{i\in\left\{
+,-,0,1,2,3\right\}  }M_{\hat{A}\hat{B}_{1}\hat{B}_{2}}^{i}\otimes
\frac{S_{AB_{1}B_{2}}^{i}}{\operatorname{Tr}[S_{AB_{1}B_{2}}^{g(i)}]},
\label{eq:superch-symm-QEC}%
\end{equation}
where the $S_{AB_{1}B_{2}}^{i}$ operators are defined in \eqref{eq:S-op-+}--\eqref{eq:S-op-3} and the
function $g(i)$ in \eqref{eq:S_indexing}. We now consider the individual constraints on
$\widetilde{M}_{A\hat{B}_{1}\hat{B}_{2}\hat{A}B_{1}B_{2}}$, which are imposed
on it such that it is a two-PPT-extendible non-signaling Choi operator.

\subsection{Complete positivity condition}

Complete positivity is equivalent to $\widetilde{M}_{A\hat{B}_{1}\hat{B}%
_{2}\hat{A}B_{1}B_{2}}$ being positive semi-definite. By applying the
development in Appendix~\ref{Complete_positivity_tel}, we conclude that $\widetilde{M}_{A\hat{B}%
_{1}\hat{B}_{2}\hat{A}B_{1}B_{2}}\geq0$ if and only if%
\begin{align}
M_{\hat{A}\hat{B}_{1}\hat{B}_{2}}^{+}  &  \geq0,\\
M_{\hat{A}\hat{B}_{1}\hat{B}_{2}}^{-}  &  \geq0,\\%
\begin{bmatrix}
M^{0}+M^{3} & M^{1}-iM^{2}\\
M^{1}+iM^{2} & M^{0}-M^{3}%
\end{bmatrix}
&  \geq0.
\end{align}
As justified previously in Appendix~\ref{Complete_positivity_tel}, the last inequality implies that $M^0 \geq 0$. This justifies the constraints in \eqref{eq:obj-func-qec-simp}--\eqref{eq:M-matrix-PSD-qec-constr}.

\subsection{Trace preservation condition}

Trace preservation is equivalent to%
\begin{equation}
\operatorname{Tr}_{\hat{A}B_{1}B_{2}}[\widetilde{M}_{A\hat{B}_{1}\hat{B}%
_{2}\hat{A}B_{1}B_{2}}]=I_{A\hat{B}_{1}\hat{B}_{2}},
\end{equation}
which by applying \eqref{eq:superch-symm-QEC}, is equivalent to%
\begin{equation}
\sum_{i\in\left\{  +,-,0,1,2,3\right\}  }d\operatorname{Tr}_{\hat{A}}%
[M_{\hat{A}\hat{B}_{1}\hat{B}_{2}}^{i}]\otimes\frac{\operatorname{Tr}%
_{B_{1}B_{2}}[S_{AB_{1}B_{2}}^{i}]}{\operatorname{Tr}[S_{AB_{1}B_{2}}^{g(i)}%
]}=I_{A\hat{B}_{1}\hat{B}_{2}}.
\end{equation}
Applying the analysis in Appendix~\ref{sec_trace_preserve_tel}, we find that%
\begin{multline}
\sum_{i\in\left\{  +,-,0,1,2,3\right\}  }d\operatorname{Tr}_{\hat{A}}%
[M_{\hat{A}\hat{B}_{1}\hat{B}_{2}}^{i}]\otimes\frac{\operatorname{Tr}%
_{B_{1}B_{2}}[S_{AB_{1}B_{2}}^{i}]}{\operatorname{Tr}[S_{AB_{1}B_{2}}^{g(i)}%
]}\\
=\operatorname{Tr}_{\hat{A}}[M_{\hat{A}\hat{B}_{1}\hat{B}_{2}}^{+}+M_{\hat
{A}\hat{B}_{1}\hat{B}_{2}}^{-}+M_{\hat{A}\hat{B}_{1}\hat{B}_{2}}^{0}]\otimes
I_{A},
\end{multline}
so that trace preservation is equivalent to the following condition:%
\begin{equation}
\operatorname{Tr}_{\hat{A}}[M_{\hat{A}\hat{B}_{1}\hat{B}_{2}}^{+}+M_{\hat
{A}\hat{B}_{1}\hat{B}_{2}}^{-}+M_{\hat{A}\hat{B}_{1}\hat{B}_{2}}^{0}]\otimes
I_{A}=I_{A\hat{B}_{1}\hat{B}_{2}}%
\end{equation}
which in turn is equivalent to%
\begin{equation}
\operatorname{Tr}_{\hat{A}}[M_{\hat{A}\hat{B}_{1}\hat{B}_{2}}^{+}+M_{\hat
{A}\hat{B}_{1}\hat{B}_{2}}^{-}+M_{\hat{A}\hat{B}_{1}\hat{B}_{2}}^{0}%
]=I_{\hat{B}_{1}\hat{B}_{2}}.
\end{equation}
This justifies the constraint in \eqref{eq:TP-constr-simp-QEC}.

\subsection{Non-signaling conditions}

We have two different non-signaling conditions:%
\begin{align}
\operatorname{Tr}_{\hat{A}}[\widetilde{M}_{A\hat{B}_{1}\hat{B}_{2}\hat{A}%
B_{1}B_{2}}] &  =I_{A}\otimes\frac{1}{d}\operatorname{Tr}_{A\hat{A}%
}[\widetilde{M}_{A\hat{B}_{1}\hat{B}_{2}\hat{A}B_{1}B_{2}}%
],\label{eq:a-nonsig-cond}\\
\operatorname{Tr}_{B_{2}}[\widetilde{M}_{A\hat{B}_{1}\hat{B}_{2}\hat{A}%
B_{1}B_{2}}] &  =I_{\hat{B}_{2}}\otimes\frac{1}{d_{\hat{B}}}\operatorname{Tr}%
_{\hat{B}_{2}B_{2}}[\widetilde{M}_{A\hat{B}_{1}\hat{B}_{2}\hat{A}B_{1}B_{2}}].
\label{eq:b-non-sig-cond-qec}
\end{align}
The second condition we have already investigated in Appendix~\ref{app:non-sig-tele}, and we found
that it reduces to the following:%
\begin{align}
P_{\hat{A}\hat{B}_{1}\hat{B}_{2}} &  =\frac{1}{d_{\hat{B}}}\operatorname{Tr}%
_{\hat{B}_{2}}[P_{\hat{A}\hat{B}_{1}\hat{B}_{2}}]\otimes I_{\hat{B}_{2}}, \label{eq:non-sig-P-QEC-marg}\\
Q_{\hat{A}\hat{B}_{1}\hat{B}_{2}} &  =\frac{1}{d_{\hat{B}}}\operatorname{Tr}%
_{\hat{B}_{2}}[Q_{\hat{A}\hat{B}_{1}\hat{B}_{2}}]\otimes I_{\hat{B}_{2}},
\label{eq:non-sig-Q-QEC-marg}
\end{align}
where%
\begin{align}
P_{\hat{A}\hat{B}_{1}\hat{B}_{2}}  & \coloneqq\frac{1}{2d}[dM^{0}+M^{1}%
+\sqrt{d^{2}-1}M^{2}],\\
Q_{\hat{A}\hat{B}_{1}\hat{B}_{2}}  & \coloneqq\frac{1}{2d}\left[
\begin{array}
[c]{c}%
2d\left(  M_{\hat{A}\hat{B}_{1}\hat{B}_{2}}^{+}+M_{\hat{A}\hat{B}_{1}\hat
{B}_{2}}^{-}\right)  +dM_{\hat{A}\hat{B}_{1}\hat{B}_{2}}^{0}\\
-M_{\hat{A}\hat{B}_{1}\hat{B}_{2}}^{1}-\sqrt{d^{2}-1}M_{\hat{A}\hat{B}_{1}%
\hat{B}_{2}}^{2}
\end{array}
\right]  ,
\end{align}
Note that we need to incorporate the extra condition on $Q_{\hat{A}\hat{B}%
_{1}\hat{B}_{2}}$ because \eqref{eq:TP-constr-tele} no longer holds---we instead have \eqref{eq:TP-constr-simp-QEC} for the case of
approximate quantum error correction and we cannot use this to eliminate the variable $Q_{\hat{A}\hat{B}_{1}\hat{B}_{2}}$.

Let us now consider the condition in \eqref{eq:a-nonsig-cond}. Consider that%
\begin{align}
&  \operatorname{Tr}_{\hat{A}}[\widetilde{M}_{A\hat{B}_{1}\hat{B}_{2}\hat
{A}B_{1}B_{2}}]\nonumber\\
&  =\sum_{i\in\left\{  +,-,0,1,2,3\right\}  }\frac{d\operatorname{Tr}_{\hat
{A}}[M_{\hat{A}\hat{B}_{1}\hat{B}_{2}}^{i}]}{\operatorname{Tr}[S_{AB_{1}B_{2}%
}^{g(i)}]}\otimes S_{AB_{1}B_{2}}^{i}\\
&  =\frac{2\operatorname{Tr}_{\hat{A}}[M_{\hat{A}\hat{B}_{1}\hat{B}_{2}}^{+}%
]}{\left(  d+2\right)  \left(  d-1\right)  }\otimes S_{AB_{1}B_{2}}%
^{+}\nonumber\\
&  \qquad+\frac{2\operatorname{Tr}_{\hat{A}}[M_{\hat{A}\hat{B}_{1}\hat{B}_{2}%
}^{-}]}{\left(  d-2\right)  \left(  d+1\right)  }\otimes S_{AB_{1}B_{2}}%
^{-}\nonumber\\
&  \qquad+\sum_{i\in\left\{  0,1,2,3\right\}  }\frac{\operatorname{Tr}%
_{\hat{A}}[M_{\hat{A}\hat{B}_{1}\hat{B}_{2}}^{i}]}{2}\otimes S_{AB_{1}B_{2}%
}^{i}.\label{eq:non-sig-a-particular}%
\end{align}
Also,%
\begin{multline}
I_{A}\otimes\frac{1}{d}\operatorname{Tr}_{A\hat{A}}[\widetilde{M}_{A\hat
{B}_{1}\hat{B}_{2}\hat{A}B_{1}B_{2}}]=\\
\sum_{i\in\left\{  +,-,0,1,2,3\right\}  }\frac{\operatorname{Tr}_{\hat{A}%
}[M_{\hat{A}\hat{B}_{1}\hat{B}_{2}}^{i}]}{\operatorname{Tr}[S_{AB_{1}B_{2}%
}^{g(i)}]}\otimes I_{A}\otimes\operatorname{Tr}_{A}[S_{AB_{1}B_{2}}^{i}].
\end{multline}
Consider that%
\begin{align}
\operatorname{Tr}_{A}[I_{AB_{1}B_{2}}] &  =dI_{B_{1}B_{2}},\\
\operatorname{Tr}_{A}[V_{AB_{1}}^{T_{A}}] &  =I_{B_{1}B_{2}},\\
\operatorname{Tr}_{A}[V_{AB_{2}}^{T_{A}}] &  =I_{B_{1}B_{2}},\\
\operatorname{Tr}_{A}[V_{B_{1}B_{2}}] &  =dV_{B_{1}B_{2}},\\
\operatorname{Tr}_{A}[V_{AB_{1}B_{2}}^{T_{A}}] &  =V_{B_{1}B_{2}},\\
\operatorname{Tr}_{A}[V_{B_{2}B_{1}A}^{T_{A}}] &  =V_{B_{1}B_{2}},
\end{align}
which from \eqref{eq:S-op-+}--\eqref{eq:S-op-3} implies that%
\begin{align}
\operatorname{Tr}_{A}[S_{AB_{1}B_{2}}^{+}] &  =\frac{1}{2}\left[
\begin{array}
[c]{c}%
dI_{B_{1}B_{2}}+dV_{B_{1}B_{2}}\\
-\left(  \frac{I_{B_{1}B_{2}}+I_{B_{1}B_{2}}+V_{B_{1}B_{2}}+V_{B_{1}B_{2}}%
}{d+1}\right)
\end{array}
\right]  \nonumber\\
&  =\frac{\left(  d+2\right)  \left(  d-1\right)  }{d+1}\Pi_{B_{1}B_{2}%
}^{\mathcal{S}},\\
\operatorname{Tr}_{A}[S_{AB_{1}B_{2}}^{-}] &  =\frac{1}{2}\left[
\begin{array}
[c]{c}%
dI_{B_{1}B_{2}}-dV_{B_{1}B_{2}}\\
-\left(  \frac{I_{B_{1}B_{2}}+I_{B_{1}B_{2}}-V_{B_{1}B_{2}}-V_{B_{1}B_{2}}%
}{d-1}\right)
\end{array}
\right]  \nonumber\\
&  =\frac{\left(  d-2\right)  \left(  d+1\right)  }{d-1}\Pi_{B_{1}B_{2}%
}^{\mathcal{A}},\\
\operatorname{Tr}_{A}[S_{AB_{1}B_{2}}^{0}] &  =\frac{1}{d^{2}-1}\left[
\begin{array}
[c]{c}%
d\left(  I_{B_{1}B_{2}}+I_{B_{1}B_{2}}\right)  \\
-\left(  V_{B_{1}B_{2}}+V_{B_{1}B_{2}}\right)
\end{array}
\right]  \nonumber\\
&  =\frac{2\left(  dI_{B_{1}B_{2}}-V_{B_{1}B_{2}}\right)  }{d^{2}%
-1}\nonumber\\
&  =\frac{2}{d+1}\Pi_{B_{1}B_{2}}^{\mathcal{S}}+\frac{2}{d-1}\Pi_{B_{1}B_{2}%
}^{\mathcal{A}},\\
\operatorname{Tr}_{A}[S_{AB_{1}B_{2}}^{1}] &  =\frac{1}{d^{2}-1}\left[
\begin{array}
[c]{c}%
d\left(  V_{B_{1}B_{2}}+V_{B_{1}B_{2}}\right)  \\
-\left(  I_{B_{1}B_{2}}+I_{B_{1}B_{2}}\right)
\end{array}
\right]  \nonumber\\
&  =\frac{2\left(  dV_{B_{1}B_{2}}-I_{B_{1}B_{2}}\right)  }{d^{2}%
-1}\nonumber\\
&  =\frac{2}{d+1}\Pi_{B_{1}B_{2}}^{\mathcal{S}}-\frac{2}{d-1}\Pi_{B_{1}B_{2}%
}^{\mathcal{A}},\\
\operatorname{Tr}_{A}[S_{AB_{1}B_{2}}^{2}] &  =\frac{1}{\sqrt{d^{2}-1}}\left(
I_{B_{1}B_{2}}-I_{B_{1}B_{2}}\right)  =0,\\
\operatorname{Tr}_{A}[S_{AB_{1}B_{2}}^{3}] &  =\frac{i}{\sqrt{d^{2}-1}}\left(
V_{B_{1}B_{2}}-V_{B_{1}B_{2}}\right)  =0,
\end{align}
where we have defined the projections onto the symmetric and antisymmetric subspaces of systems $B_1B_2 $ as
\begin{align}
    \Pi_{B_{1}B_{2}}^{\mathcal{S}} & \coloneqq \frac{I_{B_{1}B_{2}} + V_{B_{1}B_{2}}}{2} \label{eq:sym-sub-proj} \\
    \Pi_{B_{1}B_{2}}^{\mathcal{A}} & \coloneqq \frac{I_{B_{1}B_{2}} - V_{B_{1}B_{2}}}{2}.
    \label{eq:antisym-sub-proj}
\end{align}
Note that
\begin{align}
    \Pi_{B_{1}B_{2}}^{\mathcal{S}} \Pi_{B_{1}B_{2}}^{\mathcal{A}} & = 0, \\
    \operatorname{Tr}[\Pi_{B_{1}B_{2}}^{\mathcal{S}}] & = \frac{d(d+1)}{2}, \\
    \operatorname{Tr}[\Pi_{B_{1}B_{2}}^{\mathcal{A}}] & = \frac{d(d-1)}{2}.
\end{align}
So we find that%
\begin{align}
&  I_{A}\otimes\frac{1}{d}\operatorname{Tr}_{A\hat{A}}[\widetilde{M}_{A\hat
{B}_{1}\hat{B}_{2}\hat{A}B_{1}B_{2}}]\nonumber\\
&  =\frac{2\operatorname{Tr}_{\hat{A}}[M_{\hat{A}\hat{B}_{1}\hat{B}_{2}}^{+}%
]}{d\left(  d+2\right)  \left(  d-1\right)  }\otimes I_{A}\otimes\frac{\left(
d+2\right)  \left(  d-1\right)  }{d+1}\Pi_{B_{1}B_{2}}^{\mathcal{S}%
}\nonumber\\
&  \qquad+\frac{2\operatorname{Tr}_{\hat{A}}[M_{\hat{A}\hat{B}_{1}\hat{B}_{2}%
}^{-}]}{d\left(  d-2\right)  \left(  d+1\right)  }\otimes I_{A}\otimes
\frac{\left(  d-2\right)  \left(  d+1\right)  }{d-1}\Pi_{B_{1}B_{2}%
}^{\mathcal{A}}\nonumber\\
&  \qquad+\frac{\operatorname{Tr}_{\hat{A}}[M_{\hat{A}\hat{B}_{1}\hat{B}_{2}%
}^{0}]}{2d}\otimes I_{A}\otimes\left(  \frac{2}{d+1}\Pi_{B_{1}B_{2}%
}^{\mathcal{S}}+\frac{2}{d-1}\Pi_{B_{1}B_{2}}^{\mathcal{A}}\right)
\nonumber\\
&  \qquad+\frac{\operatorname{Tr}_{\hat{A}}[M_{\hat{A}\hat{B}_{1}\hat{B}_{2}%
}^{1}]}{2d}\otimes I_{A}\otimes\left(  \frac{2}{d+1}\Pi_{B_{1}B_{2}%
}^{\mathcal{S}}-\frac{2}{d-1}\Pi_{B_{1}B_{2}}^{\mathcal{A}}\right)  \\
&  =\frac{2\operatorname{Tr}_{\hat{A}}[M_{\hat{A}\hat{B}_{1}\hat{B}_{2}}^{+}%
]}{d\left(  d+1\right)  }\otimes I_{A}\otimes\Pi_{B_{1}B_{2}}^{\mathcal{S}%
}\nonumber\\
&  \qquad+\frac{2\operatorname{Tr}_{\hat{A}}[M_{\hat{A}\hat{B}_{1}\hat{B}_{2}%
}^{-}]}{d\left(  d-1\right)  }\otimes I_{A}\otimes\Pi_{B_{1}B_{2}%
}^{\mathcal{A}}\nonumber\\
&  +\operatorname{Tr}_{\hat{A}}[M_{\hat{A}\hat{B}_{1}\hat{B}_{2}}^{0}]\otimes
I_{A}\otimes\left(  \frac{1}{d\left(  d+1\right)  }\Pi_{B_{1}B_{2}%
}^{\mathcal{S}}+\frac{1}{d\left(  d-1\right)  }\Pi_{B_{1}B_{2}}^{\mathcal{A}%
}\right)  \nonumber\\
&  +\operatorname{Tr}_{\hat{A}}[M_{\hat{A}\hat{B}_{1}\hat{B}_{2}}^{1}]\otimes
I_{A}\otimes\left(  \frac{1}{d\left(  d+1\right)  }\Pi_{B_{1}B_{2}%
}^{\mathcal{S}}-\frac{1}{d\left(  d-1\right)  }\Pi_{B_{1}B_{2}}^{\mathcal{A}%
}\right)  \\
&  =\frac{\operatorname{Tr}_{\hat{A}}[2M^{+}+M^{0}+M^{1}]}{d\left(
d+1\right)  }\otimes I_{A}\otimes\Pi_{B_{1}B_{2}}^{\mathcal{S}}\\
&  \qquad+\frac{\operatorname{Tr}_{\hat{A}}[2M^{-}+M^{0}-M^{1}]}{d\left(
d-1\right)  }\otimes I_{A}\otimes\Pi_{B_{1}B_{2}}^{\mathcal{A}}.
\end{align}
Now consider that%
\begin{align}
&  I_{A}\otimes\Pi_{B_{1}B_{2}}^{\mathcal{S}}\nonumber\\
&  =\sum_{i\in\left\{  +,-,0,1,2,3\right\}  }\frac{\operatorname{Tr}%
[(I_{A}\otimes\Pi_{B_{1}B_{2}}^{\mathcal{S}})S_{AB_{1}B_{2}}^{i}%
]}{\operatorname{Tr}[S_{AB_{1}B_{2}}^{g(i)}]}\otimes S_{AB_{1}B_{2}}^{i}\\
&  =\sum_{i\in\left\{  +,-,0,1,2,3\right\}  }\frac{\operatorname{Tr}%
[\Pi_{B_{1}B_{2}}^{\mathcal{S}}\operatorname{Tr}_{A}[S_{AB_{1}B_{2}}^{i}%
]]}{\operatorname{Tr}[S_{AB_{1}B_{2}}^{g(i)}]}\otimes S_{AB_{1}B_{2}}^{i}\\
&  =\frac{\left(  d+2\right)  \left(  d-1\right)  }{d+1}\frac
{\operatorname{Tr}[\Pi_{B_{1}B_{2}}^{\mathcal{S}}]}{d\left(  d+2\right)
\left(  d-1\right)  /2}\otimes S_{AB_{1}B_{2}}^{+}\nonumber\\
&  +\frac{1}{2d}\operatorname{Tr}\left[  \Pi_{B_{1}B_{2}}^{\mathcal{S}}\left(
\frac{2}{d+1}\Pi_{B_{1}B_{2}}^{\mathcal{S}}+\frac{2}{d-1}\Pi_{B_{1}B_{2}%
}^{\mathcal{A}}\right)  \right]  \otimes S_{AB_{1}B_{2}}^{0}\nonumber\\
&  +\frac{1}{2d}\operatorname{Tr}\left[  \Pi_{B_{1}B_{2}}^{\mathcal{S}}\left(
\frac{2}{d+1}\Pi_{B_{1}B_{2}}^{\mathcal{S}}-\frac{2}{d-1}\Pi_{B_{1}B_{2}%
}^{\mathcal{A}}\right)  \right]  \otimes S_{AB_{1}B_{2}}^{1}\\
&  =S_{AB_{1}B_{2}}^{+}+\frac{1}{2}\left(  S_{AB_{1}B_{2}}^{0}+S_{AB_{1}B_{2}%
}^{1}\right)  .
\end{align}
We also have that%
\begin{align}
&  I_{A}\otimes\Pi_{B_{1}B_{2}}^{\mathcal{A}}\nonumber\\
&  =\sum_{i\in\left\{  +,-,0,1,2,3\right\}  }\frac{\operatorname{Tr}%
[(I_{A}\otimes\Pi_{B_{1}B_{2}}^{\mathcal{A}})S_{AB_{1}B_{2}}^{i}%
]}{\operatorname{Tr}[S_{AB_{1}B_{2}}^{g(i)}]}\otimes S_{AB_{1}B_{2}}^{i}\\
&  =\sum_{i\in\left\{  +,-,0,1,2,3\right\}  }\frac{\operatorname{Tr}%
[\Pi_{B_{1}B_{2}}^{\mathcal{A}}\operatorname{Tr}_{A}[S_{AB_{1}B_{2}}^{i}%
]]}{\operatorname{Tr}[S_{AB_{1}B_{2}}^{g(i)}]}\otimes S_{AB_{1}B_{2}}^{i}\\
&  =\frac{\left(  d-2\right)  \left(  d+1\right)  }{d-1}\frac
{\operatorname{Tr}[\Pi_{B_{1}B_{2}}^{\mathcal{A}}]}{d\left(  d-2\right)
\left(  d+1\right)  /2}\otimes S_{AB_{1}B_{2}}^{-}\nonumber\\
&  +\frac{1}{2d}\operatorname{Tr}\left[  \Pi_{B_{1}B_{2}}^{\mathcal{A}}\left(
\frac{2}{d+1}\Pi_{B_{1}B_{2}}^{\mathcal{S}}+\frac{2}{d-1}\Pi_{B_{1}B_{2}%
}^{\mathcal{A}}\right)  \right]  \otimes S_{AB_{1}B_{2}}^{0}\nonumber\\
&  +\frac{1}{2d}\operatorname{Tr}\left[  \Pi_{B_{1}B_{2}}^{\mathcal{A}}\left(
\frac{2}{d+1}\Pi_{B_{1}B_{2}}^{\mathcal{S}}-\frac{2}{d-1}\Pi_{B_{1}B_{2}%
}^{\mathcal{A}}\right)  \right]  \otimes S_{AB_{1}B_{2}}^{1}\\
&  =S_{AB_{1}B_{2}}^{-}+\frac{1}{2}\left(  S_{AB_{1}B_{2}}^{0}-S_{AB_{1}B_{2}%
}^{1}\right)  .
\end{align}
We then find that%
\begin{align}
&  I_{A}\otimes\frac{1}{d}\operatorname{Tr}_{A\hat{A}}[\widetilde{M}_{A\hat
{B}_{1}\hat{B}_{2}\hat{A}B_{1}B_{2}}]\nonumber\\
&  =\frac{\operatorname{Tr}_{\hat{A}}[2M^{+}+M^{0}+M^{1}]}{d\left(
d+1\right)  }\otimes I_{A}\otimes\Pi_{B_{1}B_{2}}^{\mathcal{S}}\nonumber\\
&  \qquad+\frac{\operatorname{Tr}_{\hat{A}}[2M^{-}+M^{0}-M^{1}]}{d\left(
d-1\right)  }\otimes I_{A}\otimes\Pi_{B_{1}B_{2}}^{\mathcal{A}}\\
&  =\frac{\operatorname{Tr}_{\hat{A}}[2M^{+}+M^{0}+M^{1}]}{d\left(
d+1\right)  }\otimes\left(  S_{AB_{1}B_{2}}^{+}+\frac{1}{2}\left(
S_{AB_{1}B_{2}}^{0}+S_{AB_{1}B_{2}}^{1}\right)  \right)  \nonumber\\
&  +\frac{\operatorname{Tr}_{\hat{A}}[2M^{-}+M^{0}-M^{1}]}{d\left(
d-1\right)  }\otimes\left(  S_{AB_{1}B_{2}}^{-}+\frac{1}{2}\left(
S_{AB_{1}B_{2}}^{0}-S_{AB_{1}B_{2}}^{1}\right)  \right)  \\
&  =\frac{\operatorname{Tr}_{\hat{A}}[2M^{+}+M^{0}+M^{1}]}{d\left(
d+1\right)  }\otimes S_{AB_{1}B_{2}}^{+}\nonumber\\
&  \qquad+\frac{\operatorname{Tr}_{\hat{A}}[2M^{-}+M^{0}-M^{1}]}{d\left(
d-1\right)  }\otimes S_{AB_{1}B_{2}}^{-}\nonumber\\
&  \qquad+\frac{1}{2d}\left(
\begin{array}
[c]{c}%
\frac{\operatorname{Tr}_{\hat{A}}[2M^{+}+M^{0}+M^{1}]}{d+1}\\
+\frac{\operatorname{Tr}_{\hat{A}}[2M^{-}+M^{0}-M^{1}]}{d-1}%
\end{array}
\right)  \otimes S_{AB_{1}B_{2}}^{0}\nonumber\\
&  \qquad+\frac{1}{2d}\left(
\begin{array}
[c]{c}%
\frac{\operatorname{Tr}_{\hat{A}}[2M^{+}+M^{0}+M^{1}]}{d+1}\\
-\frac{\operatorname{Tr}_{\hat{A}}[2M^{-}+M^{0}-M^{1}]}{d-1}%
\end{array}
\right)  \otimes S_{AB_{1}B_{2}}^{1}%
\end{align}%
\begin{align}
&  =\frac{\operatorname{Tr}_{\hat{A}}[2M^{+}+M^{0}+M^{1}]}{d\left(
d+1\right)  }\otimes S_{AB_{1}B_{2}}^{+}\nonumber\\
&  +\frac{\operatorname{Tr}_{\hat{A}}[2M^{-}+M^{0}-M^{1}]}{d\left(
d-1\right)  }\otimes S_{AB_{1}B_{2}}^{-}\nonumber\\
&  +\frac{\operatorname{Tr}_{\hat{A}}[\left(  d-1\right)  M^{+}+\left(
d+1\right)  M^{-}+dM^{0}-M^{1}]}{d\left(  d^{2}-1\right)  }\otimes
S_{AB_{1}B_{2}}^{0}\nonumber\\
&  +\frac{\operatorname{Tr}_{\hat{A}}[\left(  d-1\right)  M^{+}-\left(
d+1\right)  M^{-}-M^{0}+dM^{1}]}{d\left(  d^{2}-1\right)  }\otimes
S_{AB_{1}B_{2}}^{1}\\
&  =\frac{\operatorname{Tr}_{\hat{A}}[2M^{+}+M^{0}+M^{1}]}{d\left(
d+1\right)  }\otimes S_{AB_{1}B_{2}}^{+}\nonumber\\
&  +\frac{\operatorname{Tr}_{\hat{A}}[2M^{-}+M^{0}-M^{1}]}{d\left(
d-1\right)  }\otimes S_{AB_{1}B_{2}}^{-}\nonumber\\
&  +\frac{dI_{\hat{B}_{1}\hat{B}_{2}}+\operatorname{Tr}_{\hat{A}}[M^{-}%
-M^{+}-M^{1}]}{d\left(  d^{2}-1\right)  }\otimes S_{AB_{1}B_{2}}%
^{0}\nonumber\\
&  +\frac{-I_{\hat{B}_{1}\hat{B}_{2}}+\operatorname{Tr}_{\hat{A}}%
[dM^{+}-dM^{-}+dM^{1}]}{d\left(  d^{2}-1\right)  }\otimes S_{AB_{1}B_{2}}^{1}.
\end{align}
Now considering \eqref{eq:a-nonsig-cond} and \eqref{eq:non-sig-a-particular},
as well as the facts stated after \eqref{S0_trace}, we find that the non-signaling
condition is equivalent to the following conditions:%
\begin{align}
\frac{\operatorname{Tr}_{\hat{A}}[2M^{+}]}{\left(  d+2\right)  \left(
d-1\right)  } &  =\frac{\operatorname{Tr}_{\hat{A}}[2M^{+}+M^{0}+M^{1}%
]}{d\left(  d+1\right)  },\label{eq:alm-final-non-sig-ch-1}\\
\frac{\operatorname{Tr}_{\hat{A}}[2M^{-}]}{\left(  d-2\right)  \left(
d+1\right)  } &  =\frac{\operatorname{Tr}_{\hat{A}}[2M^{-}+M^{0}-M^{1}%
]}{d\left(  d-1\right)  },\label{eq:alm-final-non-sig-ch-2}%
\end{align}%
\begin{multline}
\frac{1}{2}%
\begin{bmatrix}
\operatorname{Tr}_{\hat{A}}[M^{0}+M^{3}] & \operatorname{Tr}_{\hat{A}}%
[M^{1}-iM^{2}]\\
\operatorname{Tr}_{\hat{A}}[M^{1}+iM^{2}] & \operatorname{Tr}_{\hat{A}}%
[M^{0}-M^{3}]
\end{bmatrix}
=\label{eq:non-sig-A-final-sym}\\
\frac{1}{d\left(  d^{2}-1\right)  }%
\begin{bmatrix}
Y^{0} & Y^{1}\\
Y^{1} & Y^{0}%
\end{bmatrix}
.
\end{multline}
where%
\begin{align}
Y^{0} &  \coloneqq dI_{\hat{B}_{1}\hat{B}_{2}}+\operatorname{Tr}_{\hat{A}%
}[M^{-}-M^{+}-M^{1}],\\
Y^{1} &  \coloneqq-I_{\hat{B}_{1}\hat{B}_{2}}+\operatorname{Tr}_{\hat{A}%
}[dM^{+}-dM^{-}+dM^{1}].
\end{align}
Note that the last condition in \eqref{eq:non-sig-A-final-sym} simplifies
because it reduces to%
\begin{align}
\frac{1}{2}\operatorname{Tr}_{\hat{A}}[M^{0}+M^{3}]  & =\frac{Y^{0}}{d\left(
d^{2}-1\right)  }\\
& =\frac{1}{2}\operatorname{Tr}_{\hat{A}}[M^{0}-M^{3}],
\end{align}
which implies that $\operatorname{Tr}_{\hat{A}}[M^{3}]=0$. Similarly,%
\begin{align}
\frac{1}{2}\operatorname{Tr}_{\hat{A}}[M^{1}+iM^{2}]  & =\frac{Y^{1}}{d\left(
d^{2}-1\right)  }\\
& =\operatorname{Tr}_{\hat{A}}[M^{1}-iM^{2}],
\end{align}
implying that $\operatorname{Tr}_{\hat{A}}[M^{2}]=0$. So
\eqref{eq:non-sig-A-final-sym} simplifies to%
\begin{align}
\frac{1}{2} \operatorname{Tr}_{\hat{A}}[M^{0}]  & =\frac{Y^{0}}{d\left(  d^{2}-1\right)
},\label{eq:alm-final-non-sig-ch-3}\\
\frac{1}{2} \operatorname{Tr}_{\hat{A}}[M^{1}]  & =\frac{Y^{1}}{d\left(  d^{2}-1\right)
},\label{eq:alm-final-non-sig-ch-4}\\
\operatorname{Tr}_{\hat{A}}[M^{2}] &  = \operatorname{Tr}_{\hat{A}}[M^{3}]=0.
\end{align}

\subsection{Permutation covariance condition}

Permutation covariance is equivalent to%
\begin{equation}
\widetilde{M}_{A\hat{B}_{1}\hat{B}_{2}\hat{A}B_{1}B_{2}}=(\mathcal{F}_{\hat
{B}_{1}\hat{B}_{2}}\otimes\mathcal{F}_{B_{1}B_{2}})(\widetilde{M}_{A\hat
{B}_{1}\hat{B}_{2}\hat{A}B_{1}B_{2}}),
\end{equation}
which, by applying \eqref{eq:superch-symm-QEC}, is equivalent to%
\begin{multline}
\sum_{i\in\left\{  +,-,0,1,2,3\right\}  }d\mathcal{F}_{\hat{B}_{1}\hat{B}_{2}%
}(M_{\hat{A}\hat{B}_{1}\hat{B}_{2}}^{i})\otimes\frac{\mathcal{F}_{B_{1}B_{2}%
}(S_{AB_{1}B_{2}}^{i})}{\operatorname{Tr}[S_{AB_{1}B_{2}}^{g(i)}]}=\\
\sum_{i\in\left\{  +,-,0,1,2,3\right\}  }dM_{\hat{A}\hat{B}_{1}\hat{B}_{2}%
}^{i}\otimes\frac{S_{AB_{1}B_{2}}^{i}}{\operatorname{Tr}[S_{AB_{1}B_{2}%
}^{g(i)}]}.
\end{multline}
Applying the analysis in Appendix~\ref{Sec_perm_cov_tel}, we find that
permutation covariance is equivalent to the following conditions:%
\begin{align}
M_{\hat{A}\hat{B}_{1}\hat{B}_{2}}^{i}  &  =\mathcal{F}_{\hat{B}_{1}\hat{B}%
_{2}}(M_{\hat{A}\hat{B}_{1}\hat{B}_{2}}^{i})\quad\forall i\in\left\{
+,-,0,1\right\}  ,\\
M_{\hat{A}\hat{B}_{1}\hat{B}_{2}}^{j}  &  =-\mathcal{F}_{\hat{B}_{1}\hat
{B}_{2}}(M_{\hat{A}\hat{B}_{1}\hat{B}_{2}}^{j})\quad\forall j\in\left\{
2,3\right\}  .
\end{align}

\subsection{PPT\ constraints}

We finally consider the PPT constraints, which are equivalent to%
\begin{align}
T_{A\hat{A}}(\widetilde{M}_{A\hat{B}_{1}\hat{B}_{2}\hat{A}B_{1}B_{2}})  &
\geq0,\\
T_{A\hat{A}\hat{B}_{1}B_{1}}(\widetilde{M}_{A\hat{B}_{1}\hat{B}_{2}\hat
{A}B_{1}B_{2}})  &  \geq0.
\end{align}
We have already evaluated these precise conditions in Appendix~\ref{Sec_PPT_constraints}, and we found
that they are equivalent to the conditions listed in \eqref{eq:PPT-begin-QEC}--\eqref{eq:PPTA-end-QEC} and  \eqref{eq:PPTB-constr-qec-simp-1}--\eqref{eq:PPTB-constr-qec-simp-3}.

\subsection{Final evaluation of the normalized diamond distance objective
function}

We now finally evaluate the following objective function%
\begin{equation}
\frac{1}{2}\left\Vert \widetilde{\Theta}_{(\hat{A}\rightarrow\hat
{B})\rightarrow(A\rightarrow B)}(\mathcal{N}_{\hat{A}\rightarrow\hat{B}%
})-\operatorname{id}_{A\rightarrow B}^{d}\right\Vert _{\diamond},
\label{eq:obj-func-SDP-QEC-lower}%
\end{equation}
subject to the constraint that the twirled superchannel $\widetilde{\Theta
}_{(\hat{A}\rightarrow\hat{B})\rightarrow(A\rightarrow B)}$ is
two-PPT-extendible and non-signaling. By employing the SDP for the normalized
diamond distance in \eqref{eq:diamond-d-SDP}, as well as the propagation rule in \eqref{eq:propagation rule}, we find that
\eqref{eq:obj-func-SDP-QEC-lower} is equal to%
\begin{equation}
\inf_{\mu,Z_{AB}\geq0}\left\{
\begin{array}
[c]{c}%
\mu:\mu I_{A}\geq Z_{A},\\
Z_{AB}\geq\Gamma_{AB}-\operatorname{Tr}_{\hat{A}\hat{B}}[T_{\hat{A}\hat{B}%
}(\Gamma_{\hat{A}\hat{B}}^{\mathcal{N}})\widetilde{K}_{A\hat{B}\hat{A}B}]
\end{array}
\right\}  ,
\end{equation}
where $\widetilde{K}_{A\hat{B}\hat{A}B}$ is the Choi operator corresponding to
the twirled superchannel $\widetilde{\Theta}_{(\hat{A}\rightarrow\hat
{B})\rightarrow(A\rightarrow B)}$. Recalling \eqref{eq:b-non-sig-cond-qec}, \eqref{eq:non-sig-P-QEC-marg}--\eqref{eq:non-sig-Q-QEC-marg}, and \eqref{eq:no-sig-reduce-2}, this Choi operator is given by
\begin{multline}
    \widetilde{K}_{A\hat{B}_1\hat{A}B_1} = \frac{1}{d_{\hat{B}}}\operatorname{Tr}_{\hat{B}_{2}}[P_{\hat{A}\hat{B}_{1}%
\hat{B}_{2}}]\otimes\Gamma_{AB_{1}%
}\\
+\frac{1}{d_{\hat{B}}}\operatorname{Tr}_{\hat{B}_{2}}[Q_{\hat{A}\hat{B}%
_{1}\hat{B}_{2}}]\otimes\left(  \frac
{d I_{AB_{1}}-\Gamma_{AB_{1}}}{d^{2}-1}\right)  ,
\end{multline}
which implies that
\begin{multline}
\operatorname{Tr}_{\hat{A}\hat{B}}[T_{\hat{A}\hat{B}}(\Gamma_{\hat{A}\hat{B}}^{\mathcal{N}})\widetilde{K}_{A\hat{B}\hat{A}B}] =\\ 
\frac{1}{d_{\hat{B}}}\operatorname{Tr}[T_{\hat{A}\hat{B}}(\Gamma_{\hat{A}\hat{B}}^{\mathcal{N}})P_{\hat{A}\hat{B}_{1}%
\hat{B}_{2}}]\otimes\Gamma_{AB_{1}%
}\\
+\frac{1}{d_{\hat{B}}}\operatorname{Tr}[T_{\hat{A}\hat{B}}(\Gamma_{\hat{A}\hat{B}}^{\mathcal{N}})Q_{\hat{A}\hat{B}%
_{1}\hat{B}_{2}}]\otimes\left(  \frac
{d I_{AB_{1}}-\Gamma_{AB_{1}}}{d^{2}-1}\right)
\end{multline}
The rest of the analysis follows along the
lines given in Appendix~\ref{Sec_obj_func_tel}, and so we conclude the statement of Proposition~\ref{prop:simplified-SDP-two-ext-sim-id-QEC},
after minimizing the objective function over all two-PPT-extendible and
non-signaling twirled superchannels.

\subsection{The case $d=2$}

We follow the same reasoning from Appendix~\ref{sec:tel_SDP_2d}. Due to the similarities between the SDP in Proposition~\ref{prop:simplified-SDP-two-ext-sim-id} and the SDP in Proposition~\ref{prop:simplified-SDP-two-ext-sim-id-QEC}, it is clear that we should set $M_{\hat{A}\hat{B}_1\hat{B}_2}^- = 0$ and remove the constraints in \eqref{eq:PPT-begin-QEC-2} and \eqref{eq:PPTB-constr-qec-simp-2} when $d=2$. In addition, we also  remove the constraint in \eqref{eq:non-sig-B-simp-QEC}, as it comes from comparing the coefficients of $S_{AB_1B-2}^{-}$, which is equal to zero when $d=2$. This leads to the claim at the end of Proposition~\ref{prop:simplified-SDP-two-ext-sim-id-QEC}.

\section{Proof of Equation~\eqref{eq:sdp-approx-tele-just-ppt}}

\label{app:approx-tele-just-PPT}

In this appendix, we work out the SDP for approximate teleportation when using
PPT\ constraints alone. Here we do not show as many details as previous
appendices because the methods being used are similar. The simulation error
for approximate teleportation, when using a resource state $\rho_{\hat{A}%
\hat{B}}$ and a C-PPT-P channel $\mathcal{P}_{A\hat{A}\hat{B}\rightarrow B}%
$\ for free, is as follows:%
\begin{equation}
\inf_{\mathcal{P}\in\text{PPT}}\frac{1}{2}\left\Vert \mathcal{P}_{A\hat{A}%
\hat{B}\rightarrow B}\circ\mathcal{A}_{\hat{A}\hat{B}}^{\rho}%
-\operatorname{id}_{A\rightarrow B}^{d}\right\Vert _{\diamond},
\label{eq:init-PPT-opt-tele}
\end{equation}
where we have abbreviated the set of C-PPT-P\ channels by PPT. By the same
arguments given in \eqref{eq:twirl-optimal-best-1}--\eqref{eq:convexity-dd-1wl} (i.e., exploiting the unitary covariance
symmetry of the identity channel $\operatorname{id}_{A\rightarrow B}^{d}$), it
suffices to restrict the optimization in \eqref{eq:init-PPT-opt-tele} to twirled C-PPT-P\ channels,
which satisfy%
\begin{equation}
\mathcal{P}_{A\hat{A}\hat{B}\rightarrow B}=\int dU\ \mathcal{U}_{B}^{\dag
}\circ\mathcal{P}_{A\hat{A}\hat{B}\rightarrow B}\circ\mathcal{U}_{A}.
\end{equation}
The Choi operator $P_{AB\hat{A}\hat{B}}$ thus satisfies%
\begin{equation}
P_{AB\hat{A}\hat{B}}=\widetilde{P}_{AB\hat{A}\hat{B}},
\end{equation}
where%
\begin{align}
\widetilde{P}_{AB\hat{A}\hat{B}}  & :=\int dU\ (\mathcal{U}_{A}\otimes
\overline{\mathcal{U}}_{B})(P_{AB\hat{A}\hat{B}})\\
& =\widetilde{\mathcal{T}}_{AB}(P_{AB\hat{A}\hat{B}}).
\end{align}
By invoking the twirling identity in \eqref{eq:twirl-ch-def}, we find that%
\begin{multline}
\widetilde{P}_{AB\hat{A}\hat{B}}=\Phi_{AB}\otimes\operatorname{Tr}_{AB}%
[\Phi_{AB}P_{AB\hat{A}\hat{B}}]\\
+\frac{I_{AB}-\Phi_{AB}}{d^{2}-1}\otimes\operatorname{Tr}_{AB}[\left(
I_{AB}-\Phi_{AB}\right)  P_{AB\hat{A}\hat{B}}].
\end{multline}
Now setting%
\begin{align}
K_{\hat{A}\hat{B}}  & :=\frac{1}{d}\operatorname{Tr}_{AB}[\Phi_{AB}%
P_{AB\hat{A}\hat{B}}],\\
L_{\hat{A}\hat{B}}  & :=\frac{1}{d}\operatorname{Tr}_{AB}[\left(  I_{AB}%
-\Phi_{AB}\right)  P_{AB\hat{A}\hat{B}}],
\end{align}
we can write%
\begin{equation}
\widetilde{P}_{AB\hat{A}\hat{B}}=\Gamma_{AB}\otimes K_{\hat{A}\hat{B}}%
+\frac{dI_{AB}-\Gamma_{AB}}{d^{2}-1}\otimes L_{\hat{A}\hat{B}}.
\end{equation}
We can then determine the conditions on $K_{\hat{A}\hat{B}}$ and $L_{\hat
{A}\hat{B}}$ in order for $\widetilde{P}_{AB\hat{A}\hat{B}}$ to be a
legitimate C-PPT-P\ channel. We know that $\widetilde{P}_{AB\hat{A}\hat{B}}$
is a Choi operator for a C-PPT-P channel (see Section~\ref{sec:C-PPT-P-def}) if the following conditions hold%
\begin{align}
\widetilde{P}_{AB\hat{A}\hat{B}}  & \geq0,\\
\operatorname{Tr}_{B}[\widetilde{P}_{AB\hat{A}\hat{B}}]  & =I_{A\hat{A}\hat
{B}},\\
T_{B\hat{B}}(\widetilde{P}_{AB\hat{A}\hat{B}})  & \geq0.
\end{align}
So we determine what each of these conditions impose on $K_{\hat{A}\hat{B}}$
and $L_{\hat{A}\hat{B}}$. By considering that $\Gamma_{AB}$ is orthogonal to
$\frac{dI_{AB}-\Gamma_{AB}}{d^{2}-1}$, we conclude that $\widetilde{P}%
_{AB\hat{A}\hat{B}}\geq0$ if and only if%
\begin{equation}
K_{\hat{A}\hat{B}},L_{\hat{A}\hat{B}}\geq0.
\end{equation}
Now consider that%
\begin{align}
I_{A\hat{A}\hat{B}}  & =\operatorname{Tr}_{B}[\widetilde{P}_{AB\hat{A}\hat{B}%
}]\notag \\
& =\operatorname{Tr}_{B}[\Gamma_{AB}]\otimes K_{\hat{A}\hat{B}}%
+\operatorname{Tr}_{B}\left[  \frac{dI_{AB}-\Gamma_{AB}}{d^{2}-1}\right]
\otimes L_{\hat{A}\hat{B}}\notag \\
& =I_{A}\otimes K_{\hat{A}\hat{B}}+I_{A}\otimes L_{\hat{A}\hat{B}}\notag \\
& =I_{A}\otimes\left(  K_{\hat{A}\hat{B}}+L_{\hat{A}\hat{B}}\right)  .
\end{align}
The equality $I_{A\hat{A}\hat{B}}=I_{A}\otimes\left(  K_{\hat{A}\hat{B}%
}+L_{\hat{A}\hat{B}}\right)  $ is then equivalent to
\begin{equation}
K_{\hat{A}\hat{B}}+L_{\hat{A}\hat{B}}=I_{\hat{A}\hat{B}}.\label{eq:K+L=I}%
\end{equation}
Finally, by making use of the identity $T_{B}(\Gamma_{AB})=V_{AB}$ and the
definitions in \eqref{eq:sym-sub-proj} and \eqref{eq:antisym-sub-proj} of the symmetric and antisymmetric subspace projectors, respectively,
consider that%
\begin{align}
& T_{B\hat{B}}(\widetilde{P}_{AB\hat{A}\hat{B}})\nonumber\\
& =T_{B}(\Gamma_{AB})\otimes T_{\hat{B}}(K_{\hat{A}\hat{B}})+T_{B}\left(
\frac{dI_{AB}-\Gamma_{AB}}{d^{2}-1}\right)  \otimes T_{\hat{B}}(L_{\hat{A}%
\hat{B}})\\
& =V_{AB}\otimes T_{\hat{B}}(K_{\hat{A}\hat{B}})+\frac{dI_{AB}-V_{AB}}%
{d^{2}-1}\otimes T_{\hat{B}}(L_{\hat{A}\hat{B}})\\
& =\left(  \Pi_{AB}^{\mathcal{S}}-\Pi_{AB}^{\mathcal{A}}\right)  \otimes
T_{\hat{B}}(K_{\hat{A}\hat{B}})\nonumber\\
& \qquad +\frac{d\left(  \Pi_{AB}^{\mathcal{S}}+\Pi_{AB}^{\mathcal{A}}\right)
-\left(  \Pi_{AB}^{\mathcal{S}}-\Pi_{AB}^{\mathcal{A}}\right)  }{d^{2}%
-1}\otimes T_{\hat{B}}(L_{\hat{A}\hat{B}})\\
& =\left(  \Pi_{AB}^{\mathcal{S}}-\Pi_{AB}^{\mathcal{A}}\right)  \otimes
T_{\hat{B}}(K_{\hat{A}\hat{B}})\nonumber\\
& \qquad +\left(  \frac{\Pi_{AB}^{\mathcal{S}}}{d+1}+\frac{\Pi_{AB}^{\mathcal{A}}%
}{d-1}\right)  \otimes T_{\hat{B}}(L_{\hat{A}\hat{B}})\\
& =\Pi_{AB}^{\mathcal{S}}\otimes T_{\hat{B}}\left(  K_{\hat{A}\hat{B}}%
+\frac{1}{d+1}L_{\hat{A}\hat{B}}\right)  \nonumber\\
& \qquad +\Pi_{AB}^{\mathcal{A}}\otimes T_{\hat{B}}\left(  \frac{1}{d-1}L_{\hat
{A}\hat{B}}-K_{\hat{A}\hat{B}}\right)  .
\end{align}
Since $\Pi_{AB}^{\mathcal{S}}$ and $\Pi_{AB}^{\mathcal{A}}$ are projectors and
orthogonal to each other, it follows that $T_{B\hat{B}}(\widetilde{P}%
_{AB\hat{A}\hat{B}})\geq0$ if and only if%
\begin{align}
T_{\hat{B}}\left(  K_{\hat{A}\hat{B}}+\frac{1}{d+1}L_{\hat{A}\hat{B}}\right)
& \geq0, \label{eq:QEC-PPT-K}\\
T_{\hat{B}}\left(  \frac{1}{d-1}L_{\hat{A}\hat{B}}-K_{\hat{A}\hat{B}}\right)
& \geq0.
\label{eq:QEC-PPT-L}
\end{align}
We can simplify these conditions by employing \eqref{eq:K+L=I}. Substituting,
we find that they reduce to%
\begin{align}
0  & \leq T_{\hat{B}}\left(  K_{\hat{A}\hat{B}}+\frac{1}{d+1}L_{\hat{A}\hat
{B}}\right)  \\
& =T_{\hat{B}}\left(  K_{\hat{A}\hat{B}}+\frac{1}{d+1}\left(  I_{\hat{A}%
\hat{B}}-K_{\hat{A}\hat{B}}\right)  \right)  \\
& =T_{\hat{B}}\left(  \frac{I_{\hat{A}\hat{B}}}{d+1}+\frac{d}{d+1}K_{\hat
{A}\hat{B}}\right)  ,
\end{align}
which is equivalent to%
\begin{equation}
-I_{\hat{A}\hat{B}}\leq d\ T_{\hat{B}}\left(  K_{\hat{A}\hat{B}}\right)  .
\end{equation}
Additionally,%
\begin{align}
0  & \leq T_{\hat{B}}\left(  \frac{1}{d-1}L_{\hat{A}\hat{B}}-K_{\hat{A}\hat
{B}}\right)  \\
& =T_{\hat{B}}\left(  \frac{1}{d-1}\left(  I_{\hat{A}\hat{B}}-K_{\hat{A}%
\hat{B}}\right)  -K_{\hat{A}\hat{B}}\right)  \\
& =T_{\hat{B}}\left(  \frac{1}{d-1}I_{\hat{A}\hat{B}}-\frac{d}{d-1}K_{\hat
{A}\hat{B}}\right)  ,
\end{align}
which is equivalent to%
\begin{equation}
d\ T_{\hat{B}}\left(  K_{\hat{A}\hat{B}}\right)  \leq I_{\hat{A}\hat{B}}.
\end{equation}
So we conclude that the operators $K_{\hat{A}\hat{B}}$ and $L_{\hat{A}\hat{B}%
}$ correspond to a C-PPT-P\ channel if%
\begin{align}
K_{\hat{A}\hat{B}},L_{\hat{A}\hat{B}}  & \geq0,\\
K_{\hat{A}\hat{B}}+L_{\hat{A}\hat{B}}  & =I_{\hat{A}\hat{B}},\\
-I_{\hat{A}\hat{B}}  & \leq d\ T_{\hat{B}}\left(  K_{\hat{A}\hat{B}}\right)
\leq I_{\hat{A}\hat{B}}.
\end{align}
Now employing reasoning similar to that in Appendix~\ref{Sec_obj_func_tel}, we conclude that%
\begin{align}
& \inf_{\mathcal{P}\in\text{PPT}}\frac{1}{2}\left\Vert \mathcal{P}_{A\hat
{A}\hat{B}\rightarrow B}\circ\mathcal{A}_{\hat{A}\hat{B}}^{\rho}%
-\operatorname{id}_{A\rightarrow B}^{d}\right\Vert _{\diamond}\nonumber\\
& =1-\sup_{K_{\hat{A}\hat{B}},L_{\hat{A}\hat{B}}\geq0}\left\{
\begin{array}
[c]{c}%
\operatorname{Tr}[K_{\hat{A}\hat{B}}T(\rho_{\hat{A}\hat{B}})]:\\
K_{\hat{A}\hat{B}%
}+L_{\hat{A}\hat{B}}=I_{\hat{A}\hat{B}},\\
-I_{\hat{A}\hat{B}}\leq d\ T_{\hat{B}}\left(  K_{\hat{A}\hat{B}}\right)  \leq
I_{\hat{A}\hat{B}}%
\end{array}
\right\}  \\
& =1-\sup_{K_{\hat{A}\hat{B}}\geq0}\left\{
\begin{array}
[c]{c}%
\operatorname{Tr}[K_{\hat{A}\hat{B}}T(\rho_{\hat{A}\hat{B}})]:\\
K_{\hat{A}\hat{B}}\leq I_{\hat{A}\hat{B}},\\
I_{\hat{A}\hat{B}} \pm d\ T_{\hat{B}}\left(  K_{\hat{A}\hat{B}}\right)  \geq
0
\end{array}
\right\},
\end{align}
where the last simplification follows because the operator $L_{\hat{A}\hat{B}%
}$ does not appear in the objective and thus can be considered a slack
variable. Finally, we can eliminate the transpose on $K_{\hat{A}\hat{B}}$ in the objective function by making the substitution $K_{\hat{A}\hat{B}} \to T(K_{\hat{A}\hat{B}})$ and noticing that
\begin{align}
\operatorname{Tr}[K_{\hat{A}\hat{B}}T(\rho_{\hat{A}\hat{B}})] & = \operatorname{Tr}[T(K_{\hat{A}\hat{B}})\rho_{\hat{A}\hat{B}}], \\
    K_{\hat{A}\hat{B}} \geq 0 \quad & \Leftrightarrow \quad T(K_{\hat{A}\hat{B}}) \geq 0 ,  \\
    K_{\hat{A}\hat{B}}\leq I_{\hat{A}\hat{B}} \quad & \Leftrightarrow \quad T(K_{\hat{A}\hat{B}}) \leq T(I_{\hat{A}\hat{B}}) = I_{\hat{A}\hat{B}} ,
\end{align}
\begin{align}
    & I_{\hat{A}\hat{B}} \pm d\ T_{\hat{B}}\left(  K_{\hat{A}\hat{B}}\right)  \geq 0 \notag \\
     \Leftrightarrow  \quad &
    T(I_{\hat{A}\hat{B}} \pm d\ T_{\hat{B}}\left(  K_{\hat{A}\hat{B}}\right)) \geq 0 \notag \\
     \Leftrightarrow \quad &
    I_{\hat{A}\hat{B}} \pm d\ T_{\hat{B}}\left( T( K_{\hat{A}\hat{B}})\right) \geq 0.
\end{align}

\section{Proof of Equation~\eqref{eq:PPT-SDP-approx-QEC}}

\label{sec:PPT-NS-alone-approx-QEC}

In this appendix, we derive the SDP lower bound on the simulation error of approximate quantum error correction, when using PPT constraints
alone. Ref.~\cite{LM15} already worked this out, but we provide another derivation
here for completeness. The simulation error for approximate quantum error
correction, when using a resource channel $\mathcal{N}_{\hat{A}\rightarrow
\hat{B}}$ and a C-PPT-P, non-signaling superchannel $\Theta_{(\hat
{A}\rightarrow\hat{B})\rightarrow(A\rightarrow B)}$\ for free, is as follows:%
\begin{equation}
\inf_{\Theta\in\text{PPT}\cap\text{NS}}\frac{1}{2}\left\Vert \Theta_{(\hat
{A}\rightarrow\hat{B})\rightarrow(A\rightarrow B)}(\mathcal{N}_{\hat
{A}\rightarrow\hat{B}})-\operatorname{id}_{A\rightarrow B}^{d}\right\Vert
_{\diamond},\label{eq:opt-PPT-qec}%
\end{equation}
where we have abbreviated the set of C-PPT-P, non-signaling superchannels by
PPT$\cap$NS. Some of the arguments in this appendix are almost identical to those in
Appendix~\ref{app:approx-tele-just-PPT}, and so we provide fewer details and explanations here. The Choi operator for the
channel $\Theta_{(\hat{A}\rightarrow\hat{B})\rightarrow(A\rightarrow
B)}(\mathcal{N}_{\hat{A}\rightarrow\hat{B}})$ is given by the propagation rule
in \eqref{eq:propagation rule}:%
\begin{equation}
\operatorname{Tr}_{\hat{A}\hat{B}}[T_{\hat{A}\hat{B}}(\Gamma_{\hat{A}\hat{B}%
}^{\mathcal{N}})\Gamma_{A\hat{B}\hat{A}B}^{\Theta}],
\end{equation}
where $\Gamma_{A\hat{B}\hat{A}B}^{\Theta}$ is the Choi operator for
$\Theta_{(\hat{A}\rightarrow\hat{B})\rightarrow(A\rightarrow B)}$. In order to
be a legitimate C-PPT-P, non-signaling superchannel, the following constraints
should be satisfied:%
\begin{align}
\Gamma_{A\hat{B}\hat{A}B}^{\Theta}  & \geq0,\label{eq:CP-PPT-qec}\\
\operatorname{Tr}_{\hat{A}B}[\Gamma_{A\hat{B}\hat{A}B}^{\Theta}]  &
=I_{A\hat{B}},\label{eq:TP-PPT-qec}\\
\operatorname{Tr}_{B}[\Gamma_{A\hat{B}\hat{A}B}^{\Theta}]  & =\frac{1}%
{d_{\hat{B}}}\operatorname{Tr}_{\hat{B}B}[\Gamma_{A\hat{B}\hat{A}B}^{\Theta
}]\otimes I_{\hat{B}},\label{eq:non-sig-PPT-QEC}\\
\operatorname{Tr}_{\hat{A}}[\Gamma_{A\hat{B}\hat{A}B}^{\Theta}] &  =\frac
{1}{d}\operatorname{Tr}_{A\hat{A}}[\Gamma_{A\hat{B}\hat{A}B}^{\Theta}]\otimes
I_{A},\label{eq:other-non-sig}\\
T_{\hat{B}B}(\Gamma_{A\hat{B}\hat{A}B}^{\Theta})  & \geq
0.\label{eq:PPT-PPT-qec}%
\end{align}
Due to the unitary covariance symmetry of the identity channel to be
simulated, it suffices to restrict the optimization in \eqref{eq:opt-PPT-qec} to superchannels
with Choi operators having the following form:%
\begin{equation}
\widetilde{\Gamma}_{A\hat{B}\hat{A}B}^{\Theta}=\Gamma_{AB}\otimes K_{\hat
{A}\hat{B}}+\frac{dI_{AB}-\Gamma_{AB}}{d^{2}-1}\otimes L_{\hat{A}\hat{B}}.
\end{equation}
The condition in \eqref{eq:CP-PPT-qec} is equivalent to $K_{\hat{A}\hat{B}%
},L_{\hat{A}\hat{B}}\geq0$. The condition in \eqref{eq:TP-PPT-qec}\ implies
that%
\begin{align}
 I_{A\hat{B}}
& =\operatorname{Tr}_{\hat{A}B}[\widetilde{\Gamma}_{A\hat{B}\hat{A}B}^{\Theta
}]\\
& =\operatorname{Tr}_{B}[\Gamma_{AB}]\otimes\operatorname{Tr}_{\hat{A}%
}[K_{\hat{A}\hat{B}}]\nonumber\\
& \qquad +\operatorname{Tr}_{B}\left[  \frac{dI_{AB}-\Gamma_{AB}}{d^{2}-1}\right]
\otimes\operatorname{Tr}_{\hat{A}}[L_{\hat{A}\hat{B}}]\\
& =I_{A}\otimes\operatorname{Tr}_{\hat{A}}[K_{\hat{A}\hat{B}}]+I_{A}%
\otimes\operatorname{Tr}_{\hat{A}}[L_{\hat{A}\hat{B}}]\\
& =I_{A}\otimes\operatorname{Tr}_{\hat{A}}[K_{\hat{A}\hat{B}}+L_{\hat{A}%
\hat{B}}],
\end{align}
which is equivalent to%
\begin{equation}
\operatorname{Tr}_{\hat{A}}[K_{\hat{A}\hat{B}}+L_{\hat{A}\hat{B}}]=I_{\hat{B}%
}.
\end{equation}
Now consider that%
\begin{align}
& \operatorname{Tr}_{B}[\widetilde{\Gamma}_{A\hat{B}\hat{A}B}^{\Theta
}]\nonumber\\
& =\operatorname{Tr}_{B}[\Gamma_{AB}]\otimes K_{\hat{A}\hat{B}}%
+\operatorname{Tr}_{B}\left[  \frac{dI_{AB}-\Gamma_{AB}}{d^{2}-1}\right]
\otimes L_{\hat{A}\hat{B}}\\
& =I_{A}\otimes\left(  K_{\hat{A}\hat{B}}+L_{\hat{A}\hat{B}}\right)  ,
\end{align}
and
\begin{align}
& \operatorname{Tr}_{\hat{B}B}[\Gamma_{A\hat{B}\hat{A}B}^{\Theta}]\otimes
I_{\hat{B}}\nonumber\\
& =\left(
\begin{array}
[c]{c}%
\operatorname{Tr}_{B}[\Gamma_{AB}]\otimes\operatorname{Tr}_{\hat{B}}%
[K_{\hat{A}\hat{B}}]\\
+\operatorname{Tr}_{B}\left[  \frac{dI_{AB}-\Gamma_{AB}}{d^{2}-1}\right]
\otimes\operatorname{Tr}_{\hat{B}}[L_{\hat{A}\hat{B}}]
\end{array}
\right)  \otimes I_{\hat{B}}\\
& =I_{A}\otimes\left(  \operatorname{Tr}_{\hat{B}}[K_{\hat{A}\hat{B}}%
+L_{\hat{A}\hat{B}}]\right)  \otimes I_{\hat{B}}.
\end{align}
So the constraint in \eqref{eq:non-sig-PPT-QEC} is equivalent to%
\begin{equation}
K_{\hat{A}\hat{B}}+L_{\hat{A}\hat{B}}=\operatorname{Tr}_{\hat{B}}[K_{\hat
{A}\hat{B}}+L_{\hat{A}\hat{B}}]\otimes\frac{1}{d_{\hat{B}}}I_{\hat{B}}.
\end{equation}
The other non-signaling constraint in \eqref{eq:other-non-sig} is%
\begin{equation}
\operatorname{Tr}_{\hat{A}}[\Gamma_{A\hat{B}\hat{A}B}^{\Theta}]=\frac{1}%
{d}\operatorname{Tr}_{A\hat{A}}[\Gamma_{A\hat{B}\hat{A}B}^{\Theta}]\otimes
I_{A}.
\end{equation}
Consider that%
\begin{multline}
\operatorname{Tr}_{\hat{A}}[\widetilde{\Gamma}_{A\hat{B}\hat{A}B}^{\Theta
}]=\Gamma_{AB}\otimes\operatorname{Tr}_{\hat{A}}[K_{\hat{A}\hat{B}}%
]\\
+\frac{dI_{AB}-\Gamma_{AB}}{d^{2}-1}\otimes\operatorname{Tr}_{\hat{A}%
}[L_{\hat{A}\hat{B}}],    
\end{multline}
while%
\begin{align}
& \operatorname{Tr}_{A\hat{A}}[\widetilde{\Gamma}_{A\hat{B}\hat{A}B}^{\Theta}]\nonumber\\
& =\operatorname{Tr}_{A}[\Gamma_{AB}]\otimes\operatorname{Tr}_{\hat{A}%
}[K_{\hat{A}\hat{B}}]\nonumber\\
& \qquad+\operatorname{Tr}_{A}\left[  \frac{dI_{AB}-\Gamma_{AB}}{d^{2}%
-1}\right]  \otimes\operatorname{Tr}_{\hat{A}}[L_{\hat{A}\hat{B}}]\\
& =I_{B}\otimes\operatorname{Tr}_{\hat{A}}[K_{\hat{A}\hat{B}}]+I_{A}%
\otimes\operatorname{Tr}_{\hat{A}}[L_{\hat{A}\hat{B}}]\\
& =I_{B}\otimes\left(  \operatorname{Tr}_{\hat{A}}[K_{\hat{A}\hat{B}}%
+L_{\hat{A}\hat{B}}]\right)  .
\end{align}
Then the equality in \eqref{eq:other-non-sig} is equivalent to
\begin{multline}
\Gamma_{AB}\otimes\operatorname{Tr}_{\hat{A}}[K_{\hat{A}\hat{B}}%
]+\frac{dI_{AB}-\Gamma_{AB}}{d^{2}-1}\otimes\operatorname{Tr}_{\hat{A}%
}[L_{\hat{A}\hat{B}}]\\
=\frac{1}{d}I_{AB}\otimes\left(  \operatorname{Tr}_{\hat{A}}[K_{\hat{A}\hat
{B}}+L_{\hat{A}\hat{B}}]\right)  .
\end{multline}
Scaling both sides by $\frac{1}{d}$ gives%
\begin{align}
&  \Phi_{AB}\otimes\operatorname{Tr}_{\hat{A}}[K_{\hat{A}\hat{B}}%
]+\frac{I_{AB}-\Phi_{AB}}{d^{2}-1}\otimes\operatorname{Tr}_{\hat{A}}%
[L_{\hat{A}\hat{B}}]\nonumber\\
&  =\frac{1}{d^{2}}I_{AB}\otimes\left(  \operatorname{Tr}_{\hat{A}}[K_{\hat
{A}\hat{B}}+L_{\hat{A}\hat{B}}]\right)  \\
&  =\frac{1}{d^{2}}\left(  I_{AB}-\Phi_{AB}+\Phi_{AB}\right)  \otimes\left(
\operatorname{Tr}_{\hat{A}}[K_{\hat{A}\hat{B}}+L_{\hat{A}\hat{B}}]\right)  \\
&  =\frac{1}{d^{2}}\left(  I_{AB}-\Phi_{AB}\right)  \otimes\left(
\operatorname{Tr}_{\hat{A}}[K_{\hat{A}\hat{B}}+L_{\hat{A}\hat{B}}]\right)
\nonumber\\
&  \qquad+\frac{1}{d^{2}}\Phi_{AB}\otimes\left(  \operatorname{Tr}_{\hat{A}%
}[K_{\hat{A}\hat{B}}+L_{\hat{A}\hat{B}}]\right)  ,
\end{align}
so that this is equivalent to the following two constraints:%
\begin{align}
\operatorname{Tr}_{\hat{A}}[K_{\hat{A}\hat{B}}]  & =\frac{1}{d^{2}%
}\operatorname{Tr}_{\hat{A}}[K_{\hat{A}\hat{B}}+L_{\hat{A}\hat{B}}],\\
\frac{1}{d^{2}-1}\operatorname{Tr}_{\hat{A}}[L_{\hat{A}\hat{B}}]  & =\frac
{1}{d^{2}}\operatorname{Tr}_{\hat{A}}[K_{\hat{A}\hat{B}}+L_{\hat{A}\hat{B}}],
\end{align}
which in turn are equivalent to%
\begin{equation}
\left(  d^{2}-1\right)  \operatorname{Tr}_{\hat{A}}[K_{\hat{A}\hat{B}%
}]=\operatorname{Tr}_{\hat{A}}[L_{\hat{A}\hat{B}}].
\end{equation}
By applying \eqref{eq:QEC-PPT-K}--\eqref{eq:QEC-PPT-L}, the condition in \eqref{eq:PPT-PPT-qec} is equivalent to%
\begin{align}
T_{\hat{B}}\left(  K_{\hat{A}\hat{B}}+\frac{1}{d+1}L_{\hat{A}\hat{B}}\right)
&  \geq0,\\
T_{\hat{B}}\left(  \frac{1}{d-1}L_{\hat{A}\hat{B}}-K_{\hat{A}\hat{B}}\right)
&  \geq0.
\end{align}
Thus, the optimization in \eqref{eq:opt-PPT-qec} reduces as follows:%
\begin{equation}
1-\sup_{K_{\hat{A}\hat{B}},L_{\hat{A}\hat{B}}\geq0}\left\{
\begin{array}
[c]{c}%
\operatorname{Tr}[T(K_{\hat{A}\hat{B}})\Gamma_{\hat{A}\hat{B}}^{\mathcal{N}}]:\\
\operatorname{Tr}_{\hat{A}}[K_{\hat{A}\hat{B}}+L_{\hat{A}\hat{B}}]=I_{\hat{B}%
},\\
\left(  d^{2}-1\right)  \operatorname{Tr}_{\hat{A}}[K_{\hat{A}\hat{B}%
}]=\operatorname{Tr}_{\hat{A}}[L_{\hat{A}\hat{B}}],\\
K_{\hat{A}\hat{B}}+L_{\hat{A}\hat{B}}=\\
\operatorname{Tr}_{\hat{B}}[K_{\hat{A}\hat{B}}+L_{\hat{A}\hat{B}}]\otimes
\frac{1}{d_{\hat{B}}}I_{\hat{B}},\\
T_{\hat{B}}\left(  K_{\hat{A}\hat{B}}+\frac{1}{d+1}L_{\hat{A}\hat{B}}\right)
\geq0,\\
T_{\hat{B}}\left(  \frac{1}{d-1}L_{\hat{A}\hat{B}}-K_{\hat{A}\hat{B}}\right)
\geq0
\end{array}
\right\}  ,
\end{equation}
which simplifies to
\begin{equation}
1-\sup_{K_{\hat{A}\hat{B}},L_{\hat{A}\hat{B}}\geq0}\left\{
\begin{array}
[c]{c}%
\operatorname{Tr}[T(K_{\hat{A}\hat{B}})\Gamma_{\hat{A}\hat{B}}^{\mathcal{N}}]:\\
d^{2}\operatorname{Tr}_{\hat{A}}[K_{\hat{A}\hat{B}}]=I_{\hat{B}},\\
\left(  d^{2}-1\right)  \operatorname{Tr}_{\hat{A}}[K_{\hat{A}\hat{B}%
}]=\operatorname{Tr}_{\hat{A}}[L_{\hat{A}\hat{B}}],\\
K_{\hat{A}\hat{B}}+L_{\hat{A}\hat{B}}=\\
\operatorname{Tr}_{\hat{B}}[K_{\hat{A}\hat{B}}+L_{\hat{A}\hat{B}}]\otimes
\frac{1}{d_{\hat{B}}}I_{\hat{B}},\\
T_{\hat{B}}\left(  K_{\hat{A}\hat{B}}+\frac{1}{d+1}L_{\hat{A}\hat{B}}\right)
\geq0,\\
T_{\hat{B}}\left(  \frac{1}{d-1}L_{\hat{A}\hat{B}}-K_{\hat{A}\hat{B}}\right)
\geq0
\end{array}
\right\}  .
\end{equation}
The transpose on $K_{\hat{A}\hat{B}}$ in the objective function can be
eliminated by making the substitution $K_{\hat{A}\hat{B}}\rightarrow
T(K_{\hat{A}\hat{B}})$ and noticing that%
\begin{equation}
K_{\hat{A}\hat{B}},L_{\hat{A}\hat{B}}\geq0\quad\Leftrightarrow\quad
T(K_{\hat{A}\hat{B}}),T(L_{\hat{A}\hat{B}})\geq0,
\end{equation}%
\begin{align}
d^{2}\operatorname{Tr}_{\hat{A}}[K_{\hat{A}\hat{B}}]  & =I_{\hat{B}%
}\nonumber\\
\Leftrightarrow\quad d^{2}T_{\hat{B}}(\operatorname{Tr}_{\hat{A}}[K_{\hat
{A}\hat{B}}])  & =T_{\hat{B}}(I_{\hat{B}})\\
\Leftrightarrow\quad d^{2}T_{\hat{B}}(\operatorname{Tr}_{\hat{A}}[T_{\hat{A}%
}(K_{\hat{A}\hat{B}})])  & =I_{\hat{B}}\\
\Leftrightarrow\quad d^{2}(\operatorname{Tr}_{\hat{A}}[T(K_{\hat{A}\hat{B}%
})])  & =I_{\hat{B}},
\end{align}%
\begin{align}
K_{\hat{A}\hat{B}}+L_{\hat{A}\hat{B}}  & =\operatorname{Tr}_{\hat{B}}%
[K_{\hat{A}\hat{B}}+L_{\hat{A}\hat{B}}]\otimes\frac{1}{d_{\hat{B}}}I_{\hat{B}%
}\nonumber\\
\Leftrightarrow\quad T(K_{\hat{A}\hat{B}}+L_{\hat{A}\hat{B}})  & =T\left(
\operatorname{Tr}_{\hat{B}}[K_{\hat{A}\hat{B}}+L_{\hat{A}\hat{B}}]\otimes
\frac{1}{d_{\hat{B}}}I_{\hat{B}}\right)  \end{align}
\begin{align}
\Leftrightarrow\quad & T(K_{\hat{A}\hat{B}})+T(L_{\hat{A}\hat{B}})  \notag \\
& =T_{\hat{A}}(\operatorname{Tr}_{\hat{B}}[K_{\hat{A}\hat{B}}+L_{\hat{A}\hat{B}%
}])\otimes\frac{1}{d_{\hat{B}}}T_{\hat{B}}(I_{\hat{B}})\\
& =T_{\hat{A}}(\operatorname{Tr}_{\hat{B}}[T_{\hat{B}}(K_{\hat{A}\hat{B}%
}+L_{\hat{A}\hat{B}})])\otimes\frac{1}{d_{\hat{B}}}I_{\hat{B}}\\
& =(\operatorname{Tr}_{\hat{B}}[T(K_{\hat{A}\hat{B}})+T(L_{\hat{A}\hat{B}%
})])\otimes\frac{1}{d_{\hat{B}}}I_{\hat{B}},
\end{align}%
\begin{align}
\left(  d^{2}-1\right)  \operatorname{Tr}_{\hat{A}}[K_{\hat{A}\hat{B}}]  &
=\operatorname{Tr}_{\hat{A}}[L_{\hat{A}\hat{B}}]\notag \\
\Leftrightarrow\quad\left(  d^{2}-1\right)  T_{\hat{B}}(\operatorname{Tr}%
_{\hat{A}}[K_{\hat{A}\hat{B}}])  & =T_{\hat{B}}(\operatorname{Tr}_{\hat{A}%
}[L_{\hat{A}\hat{B}}])\notag \\
\Leftrightarrow\quad\left(  d^{2}-1\right)  T_{\hat{B}}(\operatorname{Tr}%
_{\hat{A}}[T_{\hat{A}}(K_{\hat{A}\hat{B}})])  & =T_{\hat{B}}(\operatorname{Tr}%
_{\hat{A}}[T_{\hat{A}}(L_{\hat{A}\hat{B}})])\notag \\
\Leftrightarrow\quad\left(  d^{2}-1\right)  \operatorname{Tr}_{\hat{A}%
}[T(K_{\hat{A}\hat{B}})]  & =\operatorname{Tr}_{\hat{A}}[T(L_{\hat{A}\hat{B}%
})],
\end{align}%
\begin{align}
T_{\hat{B}}\left(  K_{\hat{A}\hat{B}}+\frac{1}{d+1}L_{\hat{A}\hat{B}}\right)
& \geq0\\
\Leftrightarrow\quad T\left(  T_{\hat{B}}\left(  K_{\hat{A}\hat{B}}+\frac
{1}{d+1}L_{\hat{A}\hat{B}}\right)  \right)    & \geq0\\
\Leftrightarrow\quad T_{\hat{B}}\left(  T\left(  K_{\hat{A}\hat{B}}+\frac
{1}{d+1}L_{\hat{A}\hat{B}}\right)  \right)    & \geq0\\
\Leftrightarrow\quad T_{\hat{B}}\left(  T(K_{\hat{A}\hat{B}})+\frac{1}%
{d+1}T(L_{\hat{A}\hat{B}})\right)    & \geq0,
\end{align}
and similarly,%
\begin{align}
T_{\hat{B}}\left(  \frac{1}{d-1}L_{\hat{A}\hat{B}}-K_{\hat{A}\hat{B}}\right)
& \geq0\\
\Leftrightarrow\quad T_{\hat{B}}\left(  \frac{1}{d-1}T(L_{\hat{A}\hat{B}%
})-T(K_{\hat{A}\hat{B}})\right)    & \geq0.
\end{align}
Thus, the optimization reduces to
\begin{equation}
1-\sup_{K_{\hat{A}\hat{B}},L_{\hat{A}\hat{B}}\geq0}\left\{
\begin{array}
[c]{c}%
\operatorname{Tr}[K_{\hat{A}\hat{B}}\Gamma_{\hat{A}\hat{B}}^{\mathcal{N}}]:\\
\operatorname{Tr}_{\hat{A}}[K_{\hat{A}\hat{B}}+L_{\hat{A}\hat{B}}]=I_{\hat{B}%
},\\
K_{\hat{A}\hat{B}}+L_{\hat{A}\hat{B}}=\\
\operatorname{Tr}_{\hat{B}}[K_{\hat{A}\hat{B}}+L_{\hat{A}\hat{B}}]\otimes
\frac{1}{d_{\hat{B}}}I_{\hat{B}},\\
\left(  d^{2}-1\right)  \operatorname{Tr}_{\hat{A}}[K_{\hat{A}\hat{B}%
}]=\operatorname{Tr}_{\hat{A}}[L_{\hat{A}\hat{B}}],\\
T_{\hat{B}}\left(  K_{\hat{A}\hat{B}}+\frac{1}{d+1}L_{\hat{A}\hat{B}}\right)
\geq0,\\
T_{\hat{B}}\left(  \frac{1}{d-1}L_{\hat{A}\hat{B}}-K_{\hat{A}\hat{B}}\right)
\geq0
\end{array}
\right\}
\label{eq:PPT-LM-general}.
\end{equation}

To recover the form in \cite{LM15}, starting from \eqref{eq:PPT-LM-general},
we can introduce a new optimization variable $\sigma_{\hat{A}}$ with the
constraint%
\begin{equation}
\sigma_{\hat{A}}=\frac{1}{d_{\hat{B}}}\operatorname{Tr}_{\hat{B}}[K_{\hat
{A}\hat{B}}+L_{\hat{A}\hat{B}}].
\end{equation}
This operator is positive semi-definite because $K_{\hat{A}\hat{B}}$ and
$L_{\hat{A}\hat{B}}$ are, and it has trace equal to one because%
\begin{align}
\operatorname{Tr}[\sigma_{\hat{A}}]  &  =\frac{1}{d_{\hat{B}}}%
\operatorname{Tr}_{\hat{A}\hat{B}}[K_{\hat{A}\hat{B}}+L_{\hat{A}\hat{B}}]\\
&  =\frac{1}{d_{\hat{B}}}\operatorname{Tr}_{\hat{B}}[I_{\hat{B}}]\\
&  =1.
\end{align}
Thus, the optimization in \eqref{eq:PPT-LM-general} above   is equivalent to
\begin{equation}
1-\sup_{\substack{K_{\hat{A}\hat{B}},L_{\hat{A}\hat{B}},\\\sigma_{\hat{A}}\geq
0}}\left\{
\begin{array}
[c]{c}%
\operatorname{Tr}[K_{\hat{A}\hat{B}}\Gamma_{\hat{A}\hat{B}}^{\mathcal{N}}]:\\
\operatorname{Tr}_{\hat{A}}[K_{\hat{A}\hat{B}}+L_{\hat{A}\hat{B}}]=I_{\hat{B}%
},\\
K_{\hat{A}\hat{B}}+L_{\hat{A}\hat{B}}=\\
\operatorname{Tr}_{\hat{B}}[K_{\hat{A}\hat{B}}+L_{\hat{A}\hat{B}}]\otimes
\frac{1}{d_{\hat{B}}}I_{\hat{B}},\\
\left(  d^{2}-1\right)  \operatorname{Tr}_{\hat{A}}[K_{\hat{A}\hat{B}%
}]=\operatorname{Tr}_{\hat{A}}[L_{\hat{A}\hat{B}}],\\
T_{\hat{B}}\left(  K_{\hat{A}\hat{B}}+\frac{1}{d+1}L_{\hat{A}\hat{B}}\right)
\geq0,\\
T_{\hat{B}}\left(  \frac{1}{d-1}L_{\hat{A}\hat{B}}-K_{\hat{A}\hat{B}}\right)
\geq0,\\
\sigma_{\hat{A}}=\frac{1}{d_{\hat{B}}}\operatorname{Tr}_{\hat{B}}[K_{\hat
{A}\hat{B}}+L_{\hat{A}\hat{B}}],\\
\operatorname{Tr}[\sigma_{\hat{A}}]=1.
\end{array}
\right\}
\end{equation}
However, now we can substitute to find that%
\begin{equation}
K_{\hat{A}\hat{B}}+L_{\hat{A}\hat{B}}=\sigma_{\hat{A}}\otimes I_{\hat{B}},
\end{equation}
which implies that%
\begin{align}
K_{\hat{A}\hat{B}}+\frac{1}{d+1}L_{\hat{A}\hat{B}}  &  =K_{\hat{A}\hat{B}%
}+\frac{1}{d+1}\left(  \sigma_{\hat{A}}\otimes I_{\hat{B}}-K_{\hat{A}\hat{B}%
}\right) \\
&  =\frac{1}{d+1}\left(  dK_{\hat{A}\hat{B}}+\sigma_{\hat{A}}\otimes
I_{\hat{B}}\right)  ,\\
\frac{1}{d-1}L_{\hat{A}\hat{B}}-K_{\hat{A}\hat{B}}  &  =\frac{1}{d-1}\left(
\sigma_{\hat{A}}\otimes I_{\hat{B}}-K_{\hat{A}\hat{B}}\right)  -K_{\hat{A}%
\hat{B}}\\
&  =\frac{1}{d-1}\left(  \sigma_{\hat{A}}\otimes I_{\hat{B}}-dK_{\hat{A}%
\hat{B}}\right)  .
\end{align}
Thus,%
\begin{align}
T_{\hat{B}}\left(  K_{\hat{A}\hat{B}}+\frac{1}{d+1}L_{\hat{A}\hat{B}}\right)
&  \geq0,\\
T_{\hat{B}}\left(  \frac{1}{d-1}L_{\hat{A}\hat{B}}-K_{\hat{A}\hat{B}}\right)
&  \geq0
\end{align}
is equivalent to%
\begin{align}
d\ T_{\hat{B}}(K_{\hat{A}\hat{B}})+\sigma_{\hat{A}}\otimes I_{\hat{B}}  &
\geq0,\\
\sigma_{\hat{A}}\otimes I_{\hat{B}}-d\ T_{\hat{B}}(K_{\hat{A}\hat{B}})  &
\geq0,
\end{align}
which in turn is equivalent to%
\begin{equation}
\sigma_{\hat{A}}\otimes I_{\hat{B}}\pm d\ T_{\hat{B}}(K_{\hat{A}\hat{B}}%
)\geq0.
\end{equation}
Also, the constraint $\operatorname{Tr}_{\hat{A}}[K_{\hat{A}\hat{B}}%
+L_{\hat{A}\hat{B}}]=I_{\hat{B}}$ becomes%
\begin{align}
I_{\hat{B}}  &  =\operatorname{Tr}_{\hat{A}}[K_{\hat{A}\hat{B}}+L_{\hat{A}%
\hat{B}}]\\
&  =\operatorname{Tr}_{\hat{A}}[K_{\hat{A}\hat{B}}+\sigma_{\hat{A}}\otimes
I_{\hat{B}}-K_{\hat{A}\hat{B}}]\\
&  =I_{\hat{B}},
\end{align}
and is thus redudant. The constraint $\left(  d^{2}-1\right)
\operatorname{Tr}_{\hat{A}}[K_{\hat{A}\hat{B}}]=\operatorname{Tr}_{\hat{A}%
}[L_{\hat{A}\hat{B}}]$ becomes%
\begin{align}
\left(  d^{2}-1\right)  \operatorname{Tr}_{\hat{A}}[K_{\hat{A}\hat{B}}]  &
=\operatorname{Tr}_{\hat{A}}[L_{\hat{A}\hat{B}}]\\
&  =\operatorname{Tr}_{\hat{A}}[\sigma_{\hat{A}}\otimes I_{\hat{B}}-K_{\hat
{A}\hat{B}}]\\
&  =I_{\hat{B}}-\operatorname{Tr}_{\hat{A}}[K_{\hat{A}\hat{B}}],
\end{align}
which is equivalent to%
\begin{equation}
d^{2}\operatorname{Tr}_{\hat{A}}[K_{\hat{A}\hat{B}}]=I_{\hat{B}}%
\end{equation}
Employing these observations, the SDP\ above reduces to%
\begin{equation}
1-\sup_{\substack{K_{\hat{A}\hat{B}},L_{\hat{A}\hat{B}},\\\sigma_{\hat{A}}\geq
0}}\left\{
\begin{array}
[c]{c}%
\operatorname{Tr}[K_{\hat{A}\hat{B}}\Gamma_{\hat{A}\hat{B}}^{\mathcal{N}}]:\\
K_{\hat{A}\hat{B}}+L_{\hat{A}\hat{B}}=\sigma_{\hat{A}}\otimes I_{\hat{B}},\\
d^{2}\operatorname{Tr}_{\hat{A}}[K_{\hat{A}\hat{B}}]=I_{\hat{B}},\\
\sigma_{\hat{A}}\otimes I_{\hat{B}}\pm d\ T_{\hat{B}}(K_{\hat{A}\hat{B}}%
)\geq0,\\
\sigma_{\hat{A}}=\frac{1}{d_{\hat{B}}}\operatorname{Tr}_{\hat{B}}[K_{\hat
{A}\hat{B}}+L_{\hat{A}\hat{B}}],\\
\operatorname{Tr}[\sigma_{\hat{A}}]=1.
\end{array}
\right\}  .
\end{equation}
We now notice that the constraint $\sigma_{\hat{A}}=\frac{1}{d_{\hat{B}}%
}\operatorname{Tr}_{\hat{B}}[K_{\hat{A}\hat{B}}+L_{\hat{A}\hat{B}}]$ is a
consequence of the constraint $K_{\hat{A}\hat{B}}+L_{\hat{A}\hat{B}}%
=\sigma_{\hat{A}}\otimes I_{\hat{B}}$, from applying a partial trace over
$\hat{B}$. It is thus redundant and can be eliminated, leading to%
\begin{equation}
1-\sup_{\substack{K_{\hat{A}\hat{B}},L_{\hat{A}\hat{B}},\\\sigma_{\hat{A}}\geq
0}}\left\{
\begin{array}
[c]{c}%
\operatorname{Tr}[K_{\hat{A}\hat{B}}\Gamma_{\hat{A}\hat{B}}^{\mathcal{N}}]:\\
K_{\hat{A}\hat{B}}+L_{\hat{A}\hat{B}}=\sigma_{\hat{A}}\otimes I_{\hat{B}},\\
d^{2}\operatorname{Tr}_{\hat{A}}[K_{\hat{A}\hat{B}}]=I_{\hat{B}},\\
\sigma_{\hat{A}}\otimes I_{\hat{B}}\pm d\ T_{\hat{B}}(K_{\hat{A}\hat{B}}%
)\geq0,\\
\operatorname{Tr}[\sigma_{\hat{A}}]=1.
\end{array}
\right\}  .
\end{equation}
Now we can understand $L_{\hat{A}\hat{B}}$ as a slack variable and replace the
equality constraint $K_{\hat{A}\hat{B}}+L_{\hat{A}\hat{B}}=\sigma_{\hat{A}%
}\otimes I_{\hat{B}}$ with an inequality constraint and eliminate $L_{\hat
{A}\hat{B}}$, leading to the final form in \eqref{eq:PPT-SDP-approx-QEC}:
\begin{equation}
1-\sup_{K_{\hat{A}\hat{B}},\sigma_{\hat{A}}\geq0}\left\{
\begin{array}
[c]{c}%
\operatorname{Tr}[K_{\hat{A}\hat{B}}\Gamma_{\hat{A}\hat{B}}^{\mathcal{N}}]:\\
K_{\hat{A}\hat{B}}\leq\sigma_{\hat{A}}\otimes I_{\hat{B}},\\
d^{2}\operatorname{Tr}_{\hat{A}}[K_{\hat{A}\hat{B}}]=I_{\hat{B}},\\
\sigma_{\hat{A}}\otimes I_{\hat{B}}\pm d\ T_{\hat{B}}(K_{\hat{A}\hat{B}}%
)\geq0,\\
\operatorname{Tr}[\sigma_{\hat{A}}]=1.
\end{array}
\right\}  .
\end{equation}

\end{document}